\documentclass[11pt,a4paper]{article}
\usepackage[english]{babel}
\usepackage{german}
\usepackage{amssymb}
\usepackage{amsfonts}
\usepackage{amsmath}
\usepackage[latin1]{inputenc}
\usepackage{fullpage}
\usepackage{graphicx}
\usepackage{subfigure}
\usepackage{cite}
\usepackage{hyperref,vmargin}

\newtheorem{theorem}{Theorem}
\newtheorem{transformation}{Transformation}

\newtheorem{condition}{Condition}
\newtheorem{conjecture}{Conjecture}
\newtheorem{corollary}{Corollary}

\newtheorem{definition}{Definition}

\newtheorem{observation}{Observation}

\newtheorem{lemma}{Lemma}
\newtheorem{notation}{Notation}

\newtheorem{remark}{Remark}

\newenvironment{proof}[1][Proof]{\noindent\textbf{#1.} }{\ \rule{0.5em}{0.5em}}

\setmarginsrb{3.5cm}{2.25cm}{3.5cm}{3.05cm}{0.3cm}{0.3cm}{-0.3cm}{1.0cm}
\begin{document}

\title{On the Intersection of Tolerance and Cocomparability Graphs}
\author{George B. Mertzios\thanks{%
School of Engineering and Computing Sciences, Durham University, United Kingdom. Email: 
\texttt{george.mertzios@durham.ac.uk}} 
\and Shmuel Zaks\thanks{%
Department of Computer Science, Technion, Haifa, Israel. Email: \texttt{%
zaks@cs.technion.ac.il}\vspace{-0.3cm}}}
\date{\vspace{-0.7cm}}
\maketitle

\begin{abstract}
Tolerance graphs have been extensively studied since their introduction, 
due to their interesting structure and their numerous
applications, as they generalize both interval and permutation graphs in a 
natural way. It has been conjectured by Golumbic, Monma, and Trotter in 1984 
that the intersection of tolerance and cocomparability graphs coincides with bounded tolerance graphs. 
Since cocomparability graphs can be efficiently recognized, 
a positive answer to this conjecture in the general case 
would enable us to efficiently distinguish between tolerance and bounded tolerance graphs, 
although it is NP-complete to recognize each of these classes of graphs separately. 
The conjecture has been proved under some --rather~strong-- \emph{structural}~assumptions on the input graph; 
in particular, it has been proved for complements of trees, 
and later extended to complements of bipartite graphs, 
and these are the only known results so far. 
Furthermore, it is known that the intersection of tolerance and 
cocomparability graphs is contained in the class of trapezoid graphs. 
Our main result in this article is that the above conjecture 
is true for every graph~$G$ that admits a tolerance representation with exactly one unbounded vertex; 
note here that this assumption concerns only the given tolerance \emph{representation}~$R$ of~$G$, 
rather than any structural property of~$G$. 
Moreover, our results imply as a corollary that the conjecture of Golumbic, Monma, and Trotter is true for every graph $G=(V,E)$ 
that has no three independent vertices $a,b,c\in V$ such that $N(a) \subset N(b) \subset N(c)$; 
this is satisfied in particular when $G$ is the complement of a triangle-free graph 
(which also implies the above-mentioned correctness for complements of bipartite graphs). 
Our proofs are constructive, in the sense that, given a tolerance representation~$R$ of a graph~$G$, 
we transform~$R$ into a bounded tolerance representation~$R^{\ast}$ of~$G$. 
Furthermore, we conjecture that any \emph{minimal} tolerance graph~$G$ that is not a bounded tolerance graph, 
has a tolerance representation with exactly one unbounded vertex. 
Our results imply the non-trivial result that, in order to prove the conjecture of Golumbic, Monma, and Trotter, 
it suffices to prove our conjecture.\newline

\end{abstract}

\section{Introduction\label{sec:intro}}

A simple undirected graph~$G=(V,E)$ on~$n$ vertices is called a \emph{tolerance}
graph if there exists a collection~$I=\{I_{u}\ |\ u\in V\}$ of closed
intervals on the real line and a set~$t=\{t_{u}\ |\ u\in V\}$ of positive
numbers, such that for any two vertices $u,v\in V$, $uv\in E$ if and only if 
${|I_{u}\cap I_{v}|\geq \min \{t_{u},t_{v}\}}$. The pair~$\langle I,t\rangle$ 
is called a \emph{tolerance representation} of~$G$. A vertex~$u$ of~$G$ is
called a \emph{bounded vertex} (in a certain tolerance 
representation~$\langle I,t\rangle$ of~$G$) if~$t_{u}\leq |I_{u}|$; 
otherwise, $u$ is called an \emph{unbounded vertex} of~$G$. If $G$ has a tolerance
representation~$\langle I,t\rangle$ where all vertices are bounded, then~$G$ 
is called a \emph{bounded tolerance} graph and~$\langle I,t\rangle $ a 
\emph{bounded tolerance representation} of~$G$.

Tolerance graphs find numerous applications in constrained-based temporal reasoning,
data transmission through networks to efficiently scheduling aircraft and
crews, as well as contributing to genetic analysis and studies of the 
brain~\cite{GolSi02,GolTol04}. 
This class of graphs has been introduced in~1982~\cite{GoMo82} in order to generalize 
some of the well known applications of interval graphs. 
The main motivation was in the context of resource allocation and scheduling problems, 
in which resources, such as rooms and vehicles, can tolerate sharing among users~\cite{GolTol04}. 
Since then, tolerance graphs have attracted many research 
efforts~\cite{GolTol04,GolumbicMonma84,GolSi02,Fel98,KeBe04,Bus06,BFI95,MSZ-Model-SIDMA-09,MSZ-SICOMP-11,HaSh04,NaMa92}, 
as they generalize in a natural way both interval graphs (when all tolerances are equal) 
and permutation graphs~\cite{GoMo82}(when $t_{i}=|I_{i}|$ for every $i=1,2,\ldots ,n$); 
see~\cite{GolTol04} for a detailed survey.

Given an undirected graph~$G=(V,E)$ and a vertex subset $M\subseteq V$, 
$M$ is called a \emph{module} in~$G$, if for every $u,v\in M$ and every $x\in V\setminus M$, 
$x$ is either adjacent in $G$ to both $u$ and $v$ or to none of them. 
Note that $\emptyset$, $V$, and all singletons $\{v\}$, 
where $v\in V$, are trivial modules in $G$. 
A~\emph{comparability} graph is a graph which can be transitively oriented.
A~\emph{cocomparability} graph is a graph whose complement is a
comparability graph. A~\emph{trapezoid} (resp.~\emph{parallelogram} 
and~\emph{permutation}) graph is the intersection graph of trapezoids
(resp.~parallelograms and line segments) between two parallel lines~$L_{1}$
and~$L_{2}$~\cite{Golumbic04}. Such a representation with trapezoids
(resp.~parallelograms and line segments) is called a~\emph{trapezoid} 
(resp.~\emph{parallelogram} and \emph{permutation}) \emph{representation} of 
this graph. A graph is bounded tolerance if and only if it is a parallelogram
graph~\cite{BFI95}. %Langley93
Permutation graphs are a strict subset of parallelogram graphs~\cite{Brandstaedt99}. 
Furthermore, parallelogram graphs
are a strict subset of trapezoid graphs~\cite{Ryan98}, and both are subsets
of cocomparability graphs~\cite{GolTol04,Golumbic04}. On the other hand, not
every tolerance graph is a cocomparability graph~\cite{GolTol04,Golumbic04}.

Cocomparability graphs have received considerable attention in 
the literature, mainly due to their interesting structure that leads to 
efficient algorithms for several NP-hard problems, 
see e.g.~\cite{Corneil99,KratschStewart93,DeogunSteiner94,GolTol04}. 
Furthermore, the intersection of the class of cocomparability graphs with other 
graph classes has interesting properties and coincides with other 
widely known graph classes. 
For instance, their intersection with chordal graphs is the class of interval graphs~\cite{Golumbic04}, 
while their intersection with comparability graphs is the class of permutation graphs~\cite{Pnueli71,Golumbic04}. 
These structural characterizations find also direct algorithmic implications to the recognition problem of interval 
and permutation graphs, respectively, since the class of cocomparability graphs can be recognized efficiently~\cite{Golumbic04,Spinrad03}.
In this context, the following conjecture has been made in 1984~\cite{GolumbicMonma84}:

\vspace{-0.15cm}
\begin{conjecture}[\hspace{0.3pt}\protect\cite{GolumbicMonma84}]
\label{Golumbic-Monma-conjecture}
The intersection of cocomparability graphs with tolerance graphs 
is exactly the class of bounded tolerance graphs.
\end{conjecture}
\vspace{-0.15cm}

Note that the inclusion in one direction is immediate: 
every bounded tolerance graph is a cocomparability graph~\cite{GolTol04,Golumbic04}, 
as well as a tolerance graph by definition. 
Conjecture~\ref{Golumbic-Monma-conjecture} has been proved for complements of trees~\cite{Andreae93}, 
and later extended to complements of bipartite graphs~\cite{Parra94}, 
and these are the only known results so far.
Furthermore, it has been proved that the intersection of tolerance and 
cocomparability graphs is contained in the class of trapezoid graphs~\cite{Fel98}. 
Since cocomparability graphs can be efficiently recognized~\cite{Spinrad03}, 
a positive answer to Conjecture~\ref{Golumbic-Monma-conjecture} 
would enable us to efficiently distinguish between tolerance and bounded tolerance graphs, 
although it is NP-complete to recognize each of these classes of graphs separately~\cite{MSZ-SICOMP-11}. 
Only little is known so far about the separation of tolerance and bounded tolerance graphs; 
a recent work can be found in~\cite{Eisermann11}. 
An intersection model for general tolerance graphs 
has been recently presented in~\cite{MSZ-Model-SIDMA-09}, given by 3-dimensional parallelepipeds. 
This \emph{parallelepiped representation} of tolerance graphs generalizes 
the parallelogram representation of bounded tolerance graphs; 
the main idea is to exploit the third dimension to capture the 
information given by unbounded tolerances. 
Furthermore, this model proved to be a powerful tool for designing 
efficient algorithms for general tolerance graphs~\cite{MSZ-Model-SIDMA-09}.

\vspace{-0.2cm}
\paragraph{Our contribution.}

Our main result in this article is that Conjecture~\ref{Golumbic-Monma-conjecture} is
true for every graph~$G$, for which there exists a tolerance representation with exactly one unbounded vertex. 
Furthermore, we state a new conjecture regarding the \emph{minimal} separating examples between 
tolerance and bounded tolerance graphs (cf.~Conjecture~\ref{minimal-conjecture} below). 
That is, unlike Conjecture~\ref {Golumbic-Monma-conjecture}, this conjecture does not concern 
any other class of graphs, such as cocomparability or trapezoid graphs. 
In order to state Conjecture~\ref{minimal-conjecture}, we first define a graph~$G$ 
to be a \emph{minimally unbounded tolerance} graph, if~$G$ is tolerance
but not bounded tolerance, while~$G$ becomes a bounded tolerance graph if we
remove any vertex of~$G$.

\vspace{-0.15cm}
\begin{conjecture}
\label{minimal-conjecture}Any minimally unbounded tolerance graph has a
tolerance representation with exactly one unbounded vertex.
\end{conjecture}
\vspace{-0.15cm}

Our results imply the non-trivial result that, in order to prove Conjecture~\ref{Golumbic-Monma-conjecture}, 
it suffices to prove Conjecture~\ref{minimal-conjecture}. 
To the best of our knowledge, Conjecture~\ref{minimal-conjecture} is true for all known examples of 
minimally unbounded tolerance graphs in the literature (see e.g.~\cite{GolTol04}).

All our results are based (a) on the 3-dimensional parallelepiped representation of tolerance graphs~\cite{MSZ-Model-SIDMA-09} 
and (b) on the fact that every graph $G$ that is both a tolerance and a 
cocomparability graph, has necessarily a trapezoid representation $R_{T}$~\cite{Fel98}.
Specifically, in order to prove our results, we define three conditions on the 
unbounded vertices of $G$ (in the parallelepiped representation~$R$ of~$G$). 
Condition~\ref{ass1} states that~$R$ has exactly one unbounded vertex. 
Condition~\ref{ass2} states that, for every unbounded vertex~$u$ of~$G$ (in~$R$), 
there exists no unbounded vertex~$v$ whose neighborhood is strictly included 
in the neighborhood of~$u$. Note that both Conditions~\ref{ass1}~and~\ref{ass2} concern 
only the parallelepiped representation~$R$; furthermore, Condition~\ref{ass2} is weaker 
than Condition~\ref{ass1}. 
Then, Condition~\ref{ass3} (which has a more complicated statement, cf.~Section~\ref{structure-subsec}) 
concerns also the position of the unbounded vertices in the trapezoid representation~$R_{T}$ of~$G$, 
and it is weaker than both Conditions~\ref{ass1}~and~\ref{ass2}.

Consider a graph $G$ that is both tolerance and cocomparability, 
and thus $G$ is also a trapezoid graph~\cite{Fel98}, i.e.~$G$ has both a 
parallelepiped representation $R$ and a trapezoid representation $R_T$. 
Assuming that $G$ satisfies Condition~\ref{ass3}, 
we construct a parallelogram representation of $G$, thus proving that $G$ is a
bounded tolerance graph. 
Therefore, since Condition~\ref{ass3} is weaker than both Conditions~\ref{ass1}~and~\ref{ass2}, 
the same result immediately follows by assuming that the graph $G$ satisfies Conditions~\ref{ass1} 
or Condition~\ref{ass2}. 
In particular, this immediately implies our main result of this paper, 
i.e.~that Conjecture~\ref{Golumbic-Monma-conjecture} is true for every graph~$G$ 
that admits a tolerance representation with exactly one unbounded vertex (i.e.~when Condition~\ref{ass1} is satisfied). 
Moreover, our results imply easily (cf.~Corollary~\ref{complement-triangle-free-cor}) that 
Conjecture~\ref{Golumbic-Monma-conjecture} is true for every graph $G=(V,E)$ that has no three independent vertices $a,b,c\in V$ 
such that the neighborhood of~$a$ is strictly included in the neighborhood of~$b$, 
which in turn is strictly included in the neighborhood of~$c$. 
This is a consequence of the fact that, if a graph $G$ has no such triple of vertices $\{a,b,c\}$, 
then Condition~\ref{ass2} is satisfied.
Thus, in particular, Conjecture~\ref{Golumbic-Monma-conjecture} is true for all complements of 
triangle-free graphs (which also implies the above-mentioned correctness for complements of trees~\cite{Andreae93} 
and complements of bipartite graphs~\cite{Parra94}).

The main idea of the proofs is to iteratively ``eliminate'' the unbounded vertices of the parallelepiped representation $R$. 
That is, assuming that the input representation $R$ has $k\geq 1$ unbounded vertices, 
we choose an unbounded vertex $u$ in $R$ and construct a parallelepiped 
representation~$R^{\ast }$ of~$G$ with $k-1$ unbounded vertices; specifically, 
$R^{\ast }$ has the same unbounded vertices as $R$ except for $u$ (which becomes 
bounded in~$R^{\ast }$). 
As a milestone in the above construction of the 
representation~$R^{\ast }$, we construct an induced subgraph~$G_{0}$ of~$G$ that 
includes~$u$, with the property that the vertex set of~$G_{0}\setminus \{u\}$ 
is a module in $G\setminus \{u\}$. The presented techniques are new and provide 
geometrical insight for the graphs that are both tolerance and cocomparability.

\vspace{-0.2cm}
\paragraph{Organization of the paper.}

We first review in Section~\ref{projection-sec} some properties of tolerance
and trapezoid graphs. Then we define the notion of a \emph{projection
representation} of a tolerance graph~$G$, which is an alternative way to
think about a parallelepiped representation of $G$. Furthermore, we
introduce the \emph{right} and \emph{left border properties} of a vertex in
a projection representation, which are crucial for our analysis. 
In Section~\ref{main-sec} we prove our main results. 
Specifically, we first consider in Section~\ref{right-left-subsec} 
the case where the graph $G$ has at least one unbounded vertex $u$ 
with the right or with the left border property in its projection representation, 
and then we consider in Section~\ref{structure-subsec} the case that $G$ has no such unbounded vertex. 
Next we discuss in Section~\ref{general-subsec} how these results reduce Conjecture~\ref{Golumbic-Monma-conjecture} 
to Conjecture~\ref{minimal-conjecture}. 
Finally, we discuss the presented results and further research 
in Section~\ref{conclusion}.

\section{Definitions and basic properties}
\label{projection-sec}

\paragraph{Notation.}
We consider in this article simple undirected graphs with no loops or multiple edges. 
In a graph ${G=(V,E)}$, the edge between vertices $u$ and $v$ is denoted by $uv$, 
and in this case $u$ and $v$ are called \emph{adjacent} in $G$. 
Given a vertex subset ${S \subseteq V}$, $G[S]$ denotes the induced subgraph of $G$ 
on the vertices in $S$. Whenever it is clear from the context, we may not distinguish 
between a vertex set~$S$ and the induced subgraph $G[S]$ of $G$. 
In particular, if $M$ is a module in $G$, we may also say 
that the induced subgraph $G[M]$ is a module in $G$. 
Furthermore, we denote for simplicity the induced subgraph $G[V\setminus S]$ by $G\setminus S$. 
Denote by $N(u)=\{v\in V\ |\ uv\in E\}$ the set of neighbors of a vertex $u$ in $G$, 
and $N[u]=N(u)\cup \{u\}$. For a subset $U$ of vertices of~$G$, 
denote ${N(U)=\bigcup_{u\in U}N(u)\setminus U}$. 
For any~$k$ vertices $u_{1},u_{2},\ldots ,u_{k}$ of $G$, denote for simplicity 
$N[u_{1},u_{2},\ldots ,u_{k}]=N[u_{1}]\cup N[u_{2}]\cup \ldots \cup N[u_{k}]$, 
i.e.~$N[u_{1},u_{2},\ldots ,u_{k}]=N(\{u_{1},u_{2},\ldots ,u_{k}\})\cup
\{u_{1},u_{2},\ldots ,u_{k}\}$. For any two sets $A$ and $B$, we will write 
$A\subseteq B$ if $A$ is included in $B$, and $A\subset B$ if $A$ is strictly
included in $B$.

\medskip

Consider a trapezoid graph $G=(V,E)$ and a trapezoid representation $R_{T}$
of $G$, where for any vertex $u\in V$ the trapezoid corresponding to $u$ in $%
R_{T}$ is denoted by $T_{u}$. Since trapezoid graphs are also
cocomparability graphs~\cite{Golumbic04}, we can define the partial order $%
(V,\ll _{R_{T}})$, such that~$u\ll _{R_{T}}v$, or equivalently $T_{u}\ll
_{R_{T}}T_{v}$, if and only if $T_{u}$ lies completely to the left of $T_{v}$
in $R_{T}$ (and thus also~$uv\notin E$). Note that there are several
trapezoid representations of a particular trapezoid graph~$G$. Given one
such representation $R_{T}$, we can obtain another one $R_{T}^{\prime }$ by 
\emph{vertical axis flipping} of~$R_{T}$, i.e.~$R_{T}^{\prime }$ is the
mirror image of $R_{T}$ along an imaginary line perpendicular to $L_{1}$ and 
$L_{2}$.

Let us now briefly review the parallelepiped representation model of
tolerance graphs~\cite{MSZ-Model-SIDMA-09}. Consider a tolerance graph $%
G=(V,E)$ and let $V_{B}$ and $V_{U}$ denote the set of bounded and unbounded
vertices of $G$ (for a certain tolerance representation), respectively.
Consider now two parallel lines~$L_{1}$ and $L_{2}$ in the plane. For every
vertex $u\in V $, consider a parallelogram $\overline{P}_{u}$ with two of
its lines on~$L_{1}$ and~$L_{2}$, respectively, and $\phi _{u}$ be the
(common) slope of the other two lines of $\overline{P}_{u}$ with $L_{1}$ 
and~$L_{2}$. For every unbounded vertex $u\in V_{U}$, the parallelogram $%
\overline{P}_{u}$ is trivial, i.e.~a line. In the model of~\cite{MSZ-Model-SIDMA-09}, 
every bounded vertex $u\in V_{B}$ corresponds to the
parallelepiped~${P_{u}=\{(x,y,z)\ |\ (x,y)\in \overline{P}_{u},0\leq z\leq
\phi _{u}\}}$ in the 3-dimensional space, while every unbounded vertex~$u\in
V_{U}$ corresponds to the line~${P_{u}=\{(x,y,z)\ |\ (x,y)\in \overline{P}%
_{u},z=\phi _{u}\}}$. The resulting set $\{P_{u}\ |\ u\in V\}$ of
parallelepipeds in the 3-dimensional space constitutes the 
\emph{parallelepiped representation} of $G$. In this model, two vertices $u,v$ are adjacent
if and only if $P_{u}\cap P_{v}\neq \emptyset $. That is,~$R$ is an
intersection model for $G$. 
For more details we refer to~\cite{MSZ-Model-SIDMA-09}.

An example of a tolerance graph $G$ is given in Figure~\ref{fig-projection-1} 
(in this example, $G$ is the induced path $P_{4}=(z,u,v,w)$ with four vertices). 
Furthermore, a parallelepiped representation $R$ is illustrated in Figure~\ref{fig-projection-2}. 
In particular, vertex $w$ is unbounded in the parallelepiped representation $R$, 
while the vertices $z,u,v$ are bounded in $R$.
In the following, let $V_{B}$ and $V_{U}$ denote the sets of bounded and unbounded 
vertices of a tolerance graph $G$ (for a certain parallelepiped representation), respectively.

\begin{definition}[\hspace{0.3pt}\protect\cite{MSZ-Model-SIDMA-09}]
\label{def5}An unbounded vertex $v\in V_{U}$ of a tolerance graph $G$ is
called \emph{inevitable} (in a certain parallelepiped representation $R$),
if making $v$ a bounded vertex in $R$, i.e.~if replacing $P_{v}$ with $%
\{(x,y,z)\ |\ (x,y)\in P_{v},0\leq z\leq \phi _{v}\}$, creates a new edge in~$G$.
\end{definition}

\begin{definition}[\hspace{0.3pt}\protect\cite{MSZ-Model-SIDMA-09}]
\label{def7}A parallelepiped representation $R$ of a tolerance graph $G$ is
called \emph{canonical} if every unbounded vertex in $R$ is inevitable.
\end{definition}
% \vspace{-0.1cm}

For example, the parallelepiped representation of Figure~\ref{fig-projection-2} 
is canonical, since $w$ is the only unbounded vertex and it is inevitable. 
A canonical representation of a tolerance graph~$G$ always exists, 
and can be computed in $O(n\log n)$ time, given a parallelepiped
representation of~$G$, where~$n$ is the number of vertices of $G$~\cite{MSZ-Model-SIDMA-09}.

\begin{figure}[h]
\centering
\subfigure[] { \label{fig-projection-1}
\includegraphics[scale=1]{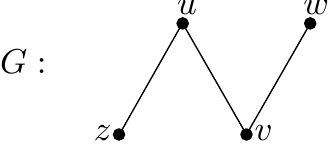}}  \\  %\hspace{0.5cm}
\subfigure[] { \label{fig-projection-2}
\includegraphics[scale=0.74]{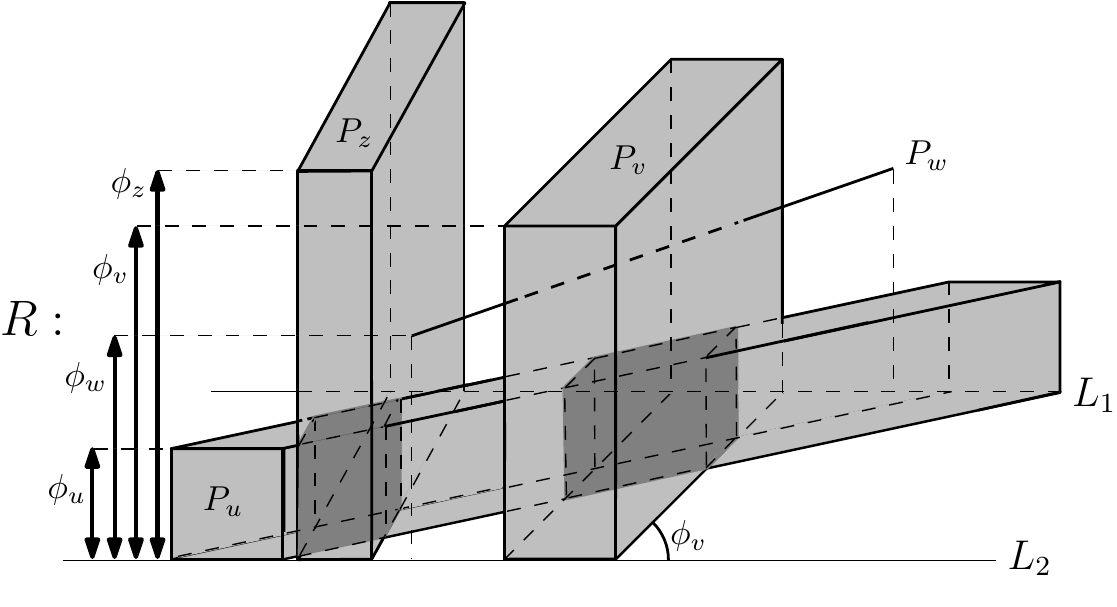}} \hspace{0.05cm}
\subfigure[] { \label{fig-projection-3}
\includegraphics[scale=0.74]{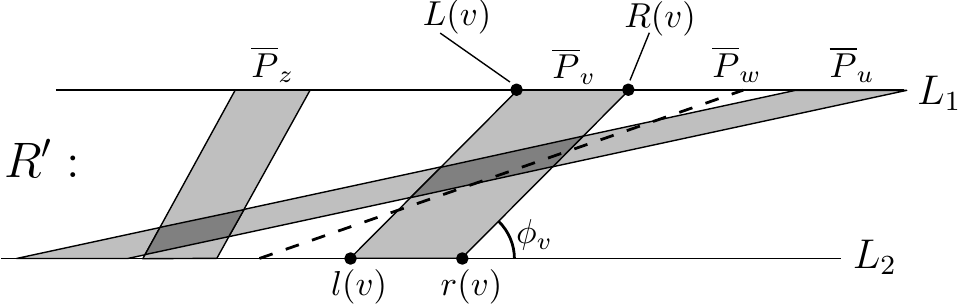}} %\hspace{0.5cm}
%\vspace{-3mm}
\caption{(a) A tolerance graph $G$ (the induced path $P_{4}=(z,u,v,w)$ with four vertices), 
(b)~a parallelepiped representation $R$ of $G$, 
and (c) the corresponding projection representation $R^{\prime}$ of~$G$.}
\label{fig-projection}
\end{figure}

Given a parallelepiped representation $R$ of the tolerance graph $G$, we
define now an alternative representation, as follows. Let $\overline{P}_{u}$
be the projection of $P_{u}$ to the plane $z=0$ for every $u\in V$. Then,
for two bounded vertices $u$ and $v$, $uv\in E$ if and only if $\overline{P}%
_{u}\cap \overline{P}_{v}\neq \emptyset $. Furthermore, for a bounded vertex 
$v$ and an unbounded vertex $\ u$, $uv\in E$ if and only if $\overline{P}%
_{u}\cap \overline{P}_{v}\neq \emptyset $ and $\phi _{v}>\phi _{u}$.
Moreover, two unbounded vertices $u$ and $v$ of $G$ are never adjacent (even
in the case where $\overline{P}_{u}$ intersects~$\overline{P}_{v}$). In the
following, we will call such a representation a \emph{projection
representation} of a tolerance graph. Note that $\overline{P}_{u}$ is a
parallelogram (resp.~a line segment) if $u$ is bounded (resp.~unbounded).
The projection representation that corresponds to the parallelepiped
representation of Figure~\ref{fig-projection-2} is presented in Figure~\ref%
{fig-projection-3}. In the sequel, we will say that a vertex $u$ is \emph{%
adjacent} to a vertex $v$ in a projection representation $R$, if $u$ is
adjacent to $v$ in the tolerance graph $G_{R}$ induced by $R$. Furthermore,
given a tolerance graph $G$, we will call a projection representation $R$ of 
$G$ a \emph{canonical representation} of $G$, if~$R$ is the projection
representation that is implied by a canonical parallelepiped representation
of~$G$. In the example of Figure~\ref{fig-projection}, the projection
representation $R^{\prime }$ is canonical, since the parallelepiped
representation $R$ is canonical as well.

Let $R$ be a projection representation of a tolerance graph $G=(V,E)$. For
every parallelogram~$\overline{P}_{u}$ in $R$, where $u\in V$, we define by $%
l(u)$ and $r(u)$ (resp.~$L(u)$ and $R(u)$) the lower (resp.~upper) left and
right endpoint of $\overline{P}_{u}$, respectively (cf.~the parallelogram~$\overline{P}_{v}$ in Figure~\ref{fig-projection-3}). 
Note that~$l(u)=r(u)$ and $L(u)=R(u)$ for every unbounded vertex $u$. Furthermore, we denote by~$%
\phi _{u}$ the (common) slope of the lines of $\overline{P}_{u}$ in $R$ that
do not lie on $L_{1}$ or on $L_{2}$ 
(cf.~the parallelepiped~$P_{v}$ in Figure~\ref{fig-projection-2} 
and the parallelogram~$\overline{P}_{v}$ in Figure~\ref{fig-projection-3}). 
We assume throughout the paper w.l.o.g.~that all endpoints and all slopes of the parallelograms in a
projection representation are distinct~\cite{GolTol04,IsaakNT03,MSZ-Model-SIDMA-09}. %FishburnTrotter99
For simplicity of the presentation, we will denote in the following $\overline{P}_{u}$ just by 
$P_{u}$ in any projection representation. 
Throughout the paper, given a projection representation $R$, 
we will often need to transform $R$ to another projection representation $R^{\prime}$ 
by moving endpoints of some parallelograms of $R$. 
After such a transformation, we say that the endpoint $a$ on $L\in\{L_{1},L_{2}\}$ lies in $R^{\prime}$ 
\emph{immediately before} (resp.~\emph{immediately after}) the 
endpoint~$b$ on $L$, if there is no other endpoint between~$a$ and~$b$ in 
$R^{\prime}$, and additionally if $a=b-\varepsilon $ (resp.~$a=b+\varepsilon $) on $L$, 
where $\varepsilon >0$ is a sufficiently small positive number. 
Similarly,  given a set $A$ of points on $L\in\{L_{1},L_{2}\}$, we say that $A$ lies in~$R^{\prime}$ 
\emph{immediately before} (resp.~\emph{immediately after}) the 
endpoint $b$ on $L$, if for every $a\in A$ there is no endpoint $c\notin A\cup \{b\}$ between $a$ and $b$ in 
$R^{\prime}$, and additionally if $a\in (b-\varepsilon,b)$ (resp.~$a\in (b,b+\varepsilon)$) on $L$, 
where $\varepsilon >0$ is a sufficiently small positive number.
The exact value of~$\varepsilon >0$ will be chosen each time appropriately, such that certain conditions hold.

Similarly to a trapezoid representation, we can define the relation $\ll
_{R} $ also for a projection representation $R$. Namely, $P_{u}\ll _{R}P_{v}$
if and only if $P_{u}$ lies completely to the left of $P_{v}$ in~$R$.
Otherwise, if neither $P_{u}\ll _{R}P_{v}$ nor $P_{v}\ll _{R}P_{u}$, we will say that 
$P_{u}$ \emph{intersects} $P_{v}$ in $R$, i.e.~${P_{u}\cap P_{v}\neq\emptyset}$ in $R$. 
Furthermore, we define the total order $<_{R}$ on the lines $L_{1}$ and $L_{2}$ in $R$ as follows. 
For two points $a$ and $b$ on $L_{1}$ (resp.~on~$L_{2}$), if $a$ lies to the left of $b$ 
on $L_{1}$ (resp.~on~$L_{2}$), then we will write~$a <_{R} b$. 
Note that, for two vertices $u$ and $v$ of a tolerance
graph $G=(V,E)$, $P_{u}$~may intersect $P_{v}$ in a projection
representation $R$ of $G$, although $u$ is not adjacent to $v$ in~$G$, i.e.~$%
uv\notin E$. Thus, a projection representation $R$ of a tolerance graph $G$
is \emph{not} necessarily an intersection model for $G$.

Let $R$ be a projection representation of a tolerance graph~$G=(V,E)$ 
and~${S\subseteq V}$ be a set of vertices of $G$. We denote by~${R\setminus S}$ the
representation that we obtain by removing the parallelograms~${\{P_{u}\ |\ u\in S\}}$ 
from~$R$. Then,~$R\setminus S$ is a projection representation of
the induced subgraph~${G\setminus S=G[V\setminus S]}$ of~$G$. 
Furthermore, similarly to the trapezoid representations, there are several 
projection representations of a particular tolerance graph $G$. 
In the next two definitions, we correspond to 
every projection 
representation of a tolerance graph $G$ another projection representation 
of the same graph $G$ with special properties.

% \vspace{-0.1cm}
\begin{definition}
\label{reverse-def}Let $R$ be a projection representation. 
The \emph{reverse} representation $\widehat{R}$ of $R$ is obtained 
as the rotation of $R$ by the angle $\pi $.
\end{definition}
% \vspace{-0.1cm}

As an example, given the projection representation $R^{\prime }$ presented
in Figure~\ref{fig-projection-3}, its reverse representation $\widehat{R^{\prime}}$ 
is illustrated in Figure~\ref{fig-reverse}. 
It is easy to see that if $R$ is a projection representation of a tolerance graph $G$, 
then for any two vertices $u$ and $v$ of $G$, $P_{u}\ll _{R}P_{v}$ if and only if 
$P_{v}\ll _{\widehat{R}}P_{u}$, and that $P_{u}\cap P_{v}\neq \emptyset $ in 
$\widehat{R}$ if and only if $P_{u}\cap P_{v}\neq \emptyset $ in $R$. 
Furthermore, the slope $\phi _{u}$ in $\widehat{R}$ equals the slope $\phi _{u}$ in $R$, 
for every vertex $u$ of~$G$. Therefore, reverse representation $\widehat{R}$ of $R$
is also a projection representation of the same graph~$G$.

% \vspace{-0.1cm}
\begin{definition}
\label{squeeze}Let $L_{1}$ and $L_{2}$ be two parallel lines and $\ell $ be
a line segment with endpoints $a_{\ell }$ and $b_{\ell }$ on $L_{1}$ and on $%
L_{2}$, respectively, and $\varepsilon >0$ be arbitrary. A projection
representation $R_{\ell }$ between $L_{1}$ and $L_{2}$ is $\varepsilon $%
\emph{-squeezed} with respect to $\ell $, if all endpoints of $R_{\ell }$ on 
$L_{1}$ and on $L_{2}$ lie in the intervals $[a_{\ell }-\frac{\varepsilon }{2%
},a_{\ell }+\frac{\varepsilon }{2}]$ and $[b_{\ell }-\frac{\varepsilon }{2}%
,b_{\ell }+\frac{\varepsilon }{2}]$, respectively.
\end{definition}

As an example, given the projection representation $R^{\prime }$ presented
in Figure~\ref{fig-projection-3}, the $\varepsilon $-squeezed 
representation $R_{\ell}^{\prime}$ of $R^{\prime }$ with respect to a line $\ell$ 
is illustrated in Figure~\ref{fig-squeezed}. 
It can be easily seen that, given a projection representation $R$ of a tolerance graph~$G$, 
a line segment $\ell$ with endpoints on $L_{1}$ and on $L_{2}$, and any $\varepsilon >0$, 
there clearly exists an $\varepsilon $-squeezed projection representation $R_{\ell }$ of $G$ with
respect to $\ell $; however, we will apply this squeezing operation in a rather delicate way (cf.~the 
proof of Theorem 2) to only some of the parallelograms in a given representation, in order to get 
some desired properties.

\vspace{-0.3cm}
\begin{figure}[bht]
\centering
\subfigure[] { \label{fig-reverse}
\includegraphics[scale=0.9]{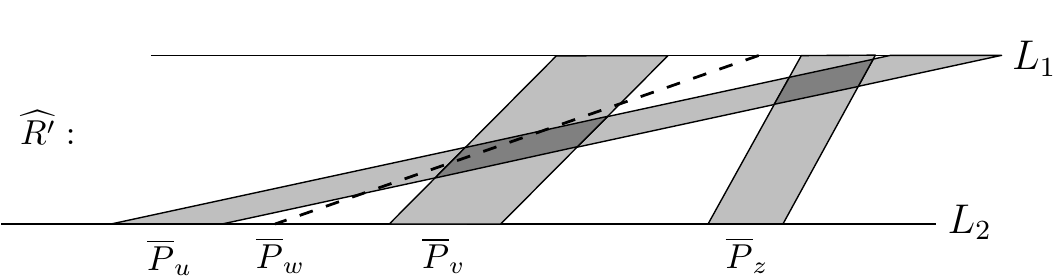}} \hspace{0.5cm}
\subfigure[] { \label{fig-squeezed}
\includegraphics[scale=0.9]{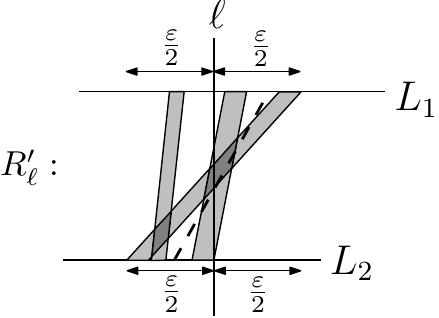}} %\hspace{0.5cm}
%\vspace{-3mm}
\caption{(a) The reverse representation $\widehat{R^{\prime }}$ of the
projection representation $R^{\prime }$ of Figure~\protect\ref%
{fig-projection-3}, and (b) the $\protect\varepsilon $-squeezed
representation $R_{\ell }^{\prime }$ of $R^{\prime }$ with respect to the
line $\ell $.}
\label{fig-reverse-squeezed}
\end{figure}
\vspace{-0.3cm}

\begin{lemma}
\label{unbounded-bounded}Let $G$ be a tolerance graph and $u$ be an
unbounded vertex of $G\ $in a projection representation $R$ of $G$. Then, $%
r(u)<_{R}r(v)$, $L(v)<_{R}L(u)$, and $v$ is a bounded vertex in $R$, for
every $v\in N(u)$.
\end{lemma}

\begin{proof}
Let $v\in N(u)$. Then, since $u$ is unbounded, and since no two unbounded
vertices are adjacent, $v$ is a bounded vertex in $R$ and $\phi _{v}>\phi
_{u}$. Moreover, $P_{u}$ intersects $P_{v}$ in the projection representation 
$R$. Suppose that $r(u)=l(u)>_{R}r(v)$ (resp.~$L(v)>_{R}L(u)=R(u)$). Then,
since $P_{u}$ intersects $P_{v}$ in $R$, it follows that $L(u)=R(u)<_{R}R(v)$
(resp.~$l(v)<_{R}r(u)=l(u)$), and thus $\phi _{v}<\phi _{u}$, which is a
contradiction. Therefore, $r(u)<_{R}r(v)$ and $L(v)<_{R}L(u)$.
\end{proof}

\begin{lemma}
\label{unbounded-hovering}Let $G$ be a tolerance graph and $u$ be an
unbounded vertex of $G\ $in a projection representation $R$ of $G$. Then, $%
l(v)<_{R}l(u)$ and $R(u)<_{R}R(v)$ for every vertex $v\neq u$, such that $%
P_{v}$ intersects $P_{u}$ in $R$ and $\phi _{v}<\phi _{u}$.
\end{lemma}

\begin{proof}
Suppose first that $l(u)<_{R}l(v)$. Then, since by assumption $P_{v}$
intersects $P_{u}$ in $R$, it follows that $L(v)<_{R}L(u)$, and thus $\phi
_{v}>\phi _{u}$ in $R$, which is a contradiction. Thus, $l(v)<_{R}l(u)$.
Similarly, if $R(v)<_{R}R(u)$, then $r(u)<_{R}r(v)$, since $P_{v}$
intersects $P_{u}$ in $R$, and thus $\phi _{v}>\phi _{u}$ in $R$, which is
again a contradiction. Thus, $R(u)<_{R}R(v)$.
\end{proof}

\medskip

In Figure~\ref{fig-reverse} an example for Lemma~\ref{unbounded-bounded}
(resp.~Lemma~\ref{unbounded-hovering}) is illustrated, where $w$ is the
unbounded vertex and $v\in N(w)$ (resp.~$u$ is a vertex, such that $P_{u}$
intersects $P_{w}$ in $R$ and $\phi _{u}<\phi _{w}$).

\begin{lemma}
\label{intersecting-unbounded}Let $G=(V,E)$ be a tolerance graph, $R$ be a
projection representation of $G$, and $u,v$ be two vertices of $G$. If $%
uv\notin E$, $P_{u}$ intersects $P_{v}$ in $R$, and $\phi _{v}<\phi _{u}$ in 
$R$, then $N(u)\subseteq N(v)$.
\end{lemma}

\begin{proof}
Suppose first that $u$ is a bounded vertex in $R$. Then, in both cases where 
$v$ is bounded or unbounded, $u$ is adjacent to $v$ in $R$, since $P_{v}\cap
P_{u}\neq \emptyset $ and $\phi _{v}<\phi _{u}$. This is a contradiction,
since $vu\notin E$, and thus $u$ is an unbounded vertex of $R$. If $v$ is a
bounded vertex, then $l(v)<_{R}l(u)$ and $R(u)<_{R}R(v)$ by Lemma~\ref%
{unbounded-hovering}. Suppose that $v$ is unbounded. If $l(u)<_{R}l(v)$,
then $L(v)<_{R}L(u)$, since $P_{u}$ intersects $P_{v}$ in $R$, and thus $%
\phi _{v}>\phi _{u}$, which is a contradiction to the assumption. Therefore $%
l(v)<_{R}l(u)$, and thus also $R(u)=L(u)<_{R}L(v)=R(v)$, since $P_{u}$
intersects $P_{v}$ in $R$. Summarizing, $l(v)<_{R}r(u)=l(u)$ and $%
R(u)=L(u)<_{R}R(v)$ in both cases where $v$ is bounded and unbounded.
Consider now a vertex $w\in N(u)$. Then, $w$ is a bounded vertex in $R$, $%
r(w)>_{R}r(u)$, and $L(w)<_{R}L(u)$ by Lemma~\ref{unbounded-bounded}.
Furthermore, $\phi _{w}>\phi _{u}>\phi _{v}$. Therefore, $r(w)>_{R}l(v)$ and 
$L(w)<_{R}R(v)$, and thus~$P_{w}$ intersects $P_{v}$ in $R$. Thus, since
also $\phi _{w}>\phi _{v}$, it follows that $w\in N(v)$. Therefore, $%
N(u)\subseteq N(v)$.
\end{proof}

\medskip

In~\cite{GolSi02,MSZ-Model-SIDMA-09} the \emph{hovering set} of an unbounded vertex
in a tolerance graph has been defined. According to these definitions, the
hovering set depends on a particular representation of the tolerance graph.
In the following, we extend this definition to the notion of \emph{covering}
vertices of an arbitrary graph $G$, which is independent of any
representation of $G$.

% \vspace{-0.1cm}
\begin{definition}
\label{hovering}Let $G=(V,E)$ be an arbitrary graph and $u\in V$ be a vertex
of $G$. Then,\vspace{-0.2cm}%
\begin{itemize}%
\item the set $\mathcal{C}(u)=\{v\in V\setminus N[u]\ |\ N(u)\subseteq N(v)\}$ 
is the \emph{covering set} of $u$, and every vertex $v\in \mathcal{C}(u)$% 
is a \emph{covering vertex} of $u$,%\vspace{-0.2cm}%
\item $V_{0}(u)$ is the set of connected components of $G\setminus N[u]$ 
that have at least one covering vertex $v\in \mathcal{C}(u)$ of $u$.%
\end{itemize}
\end{definition}
% \vspace{-0.1cm}

Now, similarly to~\cite{GolSi02}, we state the following auxiliary lemma.

\begin{lemma}
\label{bounded-hovering}Let $G=(V,E)$ be a tolerance graph and $R$ be a
canonical representation of $G$. Then, for every unbounded vertex $u$ of $G$
in $R$, there exists a covering vertex $u^{\ast }$ of $u$ in $G$, such that $%
u^{\ast }$ is bounded in $R$, $P_{u^{\ast }}$ intersects $P_{u}$ in $R$, and 
$\phi _{u^{\ast }}<\phi _{u}$. Thus, in particular $V_{0}(u)\neq \emptyset$.
\end{lemma}

\begin{proof}
Let $u$ be an arbitrary unbounded vertex of $G$ in $R$. Since $R$ is a
canonical representation of $G$, if we make $u$ a bounded vertex in $R$,
then we introduce at last one new adjacency $uu^{\ast }$ in $G$ by
Definitions~\ref{def5} and~\ref{def7}. That is, there exists at least one
vertex $u^{\ast }$, such that $P_{u^{\ast }}$ intersects $P_{u}$ in $R$, $%
\phi _{u^{\ast }}<\phi _{u}$, and $uu^{\ast }\notin E$. Then, Lemma~\ref%
{intersecting-unbounded} implies that $N(u)\subseteq N(u^{\ast })$, i.e.~$%
u^{\ast }$ is a covering vertex of $u$.

Suppose now that every covering vertex $v$ of $u$, such that $P_{v}$
intersects $P_{u}$ in $R$ and $\phi _{v}<\phi _{u}$, is unbounded, and let $%
u^{\ast }$ be the vertex with the smallest slope $\phi _{u^{\ast }}$ among
them in $R$. Then, since $P_{u^{\ast }}$ intersects $P_{u}$ in $R$ and $\phi
_{u^{\ast }}<\phi _{u}$, it follows that $l(u^{\ast })=r(u^{\ast
})<_{R}l(u)=r(u)$ and $L(u^{\ast })=R(u^{\ast })>_{R}L(u)=R(u)$.
Furthermore, since $u^{\ast }$ is assumed to be unbounded, there exists
similarly to the previous paragraph at least one vertex $u^{\ast \ast }$,
such that $P_{u^{\ast \ast }}$ intersects $P_{u^{\ast }}$ in $R$ and $\phi
_{u^{\ast \ast }}<\phi _{u^{\ast }}$, and thus $N(u^{\ast })\subseteq
N(u^{\ast \ast })$ by Lemma~\ref{intersecting-unbounded}. Thus $%
N(u)\subseteq N(u^{\ast \ast })$, since also $N(u)\subseteq N(u^{\ast })$.
Furthermore, $l(u^{\ast \ast })<_{R}l(u^{\ast })$ and $R(u^{\ast
})<_{R}R(u^{\ast \ast })$ by Lemma~\ref{unbounded-hovering}. That is, $%
l(u^{\ast \ast })<_{R}l(u^{\ast })<_{R}l(u)$ and $R(u)<_{R}R(u^{\ast
})<_{R}R(u^{\ast \ast })$, and thus $P_{u^{\ast \ast }}$ intersects $P_{u}$
in $R$. Moreover $uu^{\ast \ast }\notin E$, since $u$ is unbounded and $\phi
_{u^{\ast \ast }}<\phi _{u^{\ast }}<\phi _{u}$.

Summarizing, $u^{\ast \ast }$ is a covering vertex of $u$, $P_{u^{\ast \ast
}}$ intersects $P_{u}$ in $R$ and $\phi _{u^{\ast \ast }}<\phi _{u}$. This
is a contradiction, since $\phi _{u^{\ast \ast }}<\phi _{u^{\ast }}$, and
since $u^{\ast }$ has by assumption the smallest slope $\phi _{u^{\ast }}$
among the covering vertices $v$ of $u$, such that $P_{v}$ intersects $P_{u}$
in $R$ and $\phi _{v}<\phi _{u}$. Therefore, there exists for every
unbounded vertex $u$ at least one covering vertex $u^{\ast }$ of $u$, such
that $P_{u^{\ast }}$ intersects $P_{u}$ in $R$, $\phi _{u^{\ast }}<\phi _{u}$%
, and $u^{\ast }$ is bounded in $R$. 
Furthermore, note that $u^{\ast} \in V_{0}(u)$, and thus $V_{0}(u)\neq \emptyset$.
This completes the proof of the lemma.
\end{proof}

\medskip

In the following, for simplicity of the presentation, we may not distinguish between 
the connected components of $V_{0}(u)$ and the vertex set of these components. 
Note here that $V_{0}(u)\neq \emptyset $ for every unbounded vertex $u$ in a
canonical representation $R$, as we proved in Lemma~\ref{bounded-hovering}. 
In the next definition we introduce the notion of the right (resp.~left) border
property of a vertex $u$ in a projection representation $R$ of a tolerance
graph $G$. This notion is of particular importance for the remainder of the paper.

% \vspace{-0.1cm}
\begin{definition}
\label{right-left-property-def}Let $G=(V,E)$ be a tolerance graph, $u$ be an
arbitrary vertex of $G$, and $R$ be a projection representation of $G$.
Then, $u$ has the \emph{right} (resp.~\emph{left}) \emph{border property} in 
$R$, if there exists \emph{no pair} of vertices $w\in N(u)$ and $x\in V_{0}(u)$,
such that $P_{w}\ll _{R}P_{x}$ (resp.~$P_{x}\ll _{R}P_{w}$).
\end{definition}
% \vspace{-0.1cm}

Observe that, if a vertex $u$ has the left border property in a projection 
representation~$R$ of a tolerance graph $G$, then $u$ has the right border 
property in the reverse representation $\widehat{R}$ of~$R$.
We denote in the following by \textsc{Tolerance} the class of tolerance graphs, 
and we use the corresponding notations for the classes of bounded tolerance, cocomparability, and trapezoid graphs.

Let $G\in$ \textsc{Tolerance }$\cap $ \textsc{Cocomparability}.
Then $G$ is also a trapezoid graph~\cite{Fel98}. Thus, since \textsc{%
Trapezoid }$\subseteq $\textsc{\ Cocomparability}, it follows that \textsc{%
Tolerance }$\cap $ \textsc{Cocomparability }$=$ \textsc{Tolerance }$\cap $ 
\textsc{Trapezoid}. Furthermore, clearly \textsc{Bounded Tolerance} $%
\subseteq $ \textsc{(Tolerance }$\cap $\textsc{\ Trapezoid)}, since \textsc{%
Bounded Tolerance} $\subseteq $ \textsc{Tolerance} and \textsc{Bounded
Tolerance} $\subseteq $ \textsc{Trapezoid}.
%since every bounded tolerance (i.e.~parallelogram) graph is also a tolerance and a trapezoid graph. 
In what follows, we consider a graph 
$G\in$ \textsc{(Tolerance }$\cap $\textsc{\ Trapezoid)\ }$\setminus $ \textsc{Bounded Tolerance}, 
assuming that one exists, and our aim is to get to a contradiction; namely, to prove that 
\textsc{(Tolerance }$\cap $\textsc{\ Trapezoid)\ }$=$ \textsc{Bounded Tolerance}.

Now we state two lemmas that are of crucial importance for the proof of Theorems~\ref{right-property-thm} 
and~\ref{no-property-thm}, (in~Sections~\ref{right-left-subsec} and~\ref{structure-subsec}, respectively). 

% \vspace{-0.1cm}
\begin{lemma}
\label{two-components}Let $G\in$ \textsc{(Tolerance }$\cap $\textsc{\
Trapezoid)\ }$\setminus $ \textsc{Bounded Tolerance} with the smallest
number of vertices and $u$ be a vertex of $G$. Then, either $%
V_{0}(u)=\emptyset $ or $V_{0}(u)$ is connected.
\end{lemma}
% \vspace{-0.1cm}

\begin{proof}
For the sake of contradiction, suppose that $V_{0}(u)$ has at least two
connected components, for some vertex $u$ of $G$. Let $v_{1}$ and $v_{2}$ be
two covering vertices of $u$ that belong to two different connected
components of $V_{0}(u)$. Since $G$ has the smallest number of vertices in
the class \textsc{(Tolerance }$\cap $\textsc{\ Trapezoid)\ }$\setminus $ 
\textsc{Bounded Tolerance}, $G\setminus \{u\}$ is a bounded tolerance graph.
Let $R$ be any parallelogram representation of $G\setminus \{u\}$, and $%
R^{\prime }$ be the representation of $G\setminus N[u]$ obtained by $R$ if
we remove all parallelograms that correspond to vertices of $N(u)$. Since $%
v_{1}$ and $v_{2}$ belong to different connected components of $G\setminus
N[u]$, there is at least one line segment $\ell $ between the connected
components of $v_{1}$ and $v_{2}$ in $G\setminus N[u]$, which does not
intersect any parallelogram of $R^{\prime }$. Since $N_{G}(u)\subseteq
N_{G}(v_{1})$ and $N_{G}(u)\subseteq N_{G}(v_{2})$, and since $\ell $ lies
between $P_{v_{1}}$ and $P_{v_{2}}$ in $R^{\prime }$, it follows that
exactly the parallelograms of the vertices of $N(u)$ intersect $\ell $ in $R$%
. Thus, we can add the trivial parallelogram $P_{u}=\ell $ to $R$, obtaining
thus a parallelogram representation of $G$. Thus, $G$ is a parallelogram
graph, i.e.~a bounded tolerance graph, which is a contradiction to the
assumption. Therefore, either $V_{0}(u)=\emptyset $ or $V_{0}(u)$ is
connected, for any vertex $u$ of $G$. This completes the proof of the lemma.
\end{proof}

\medskip

The next lemma follows now easily by Lemmas~\ref{bounded-hovering} 
and~\ref{two-components}.

% \vspace{-0.1cm}
\begin{lemma}
\label{not-equal}Let $G\in$ \textsc{(Tolerance }$\cap $\textsc{\
Trapezoid)\ }$\setminus $ \textsc{Bounded Tolerance} with the smallest
number of vertices and $v_{1},v_{2}$ be distinct unbounded vertices of $%
G $ in a canonical projection representation~$R$ of $G$. Then $N(v_{1})\neq N(v_{2})$.
\end{lemma}
% \vspace{-0.1cm}

\begin{proof}
Suppose otherwise that $N(v_{1})=N(v_{2})$ for two unbounded vertices $v_{1}$
and $v_{2}$ in $R$, i.e.~$v_{2}$ is a covering vertex of $v_{1}$ and $v_{1}$
is a covering vertex of $v_{2}$. Furthermore, $v_{1}$ is an isolated vertex
in $G\setminus N[v_{2}]$. Recall now by Lemma~\ref{bounded-hovering} that
there exists at least one covering vertex $v_{2}^{\ast }$ of $v_{2}$ in $R$,
such that $v_{2}^{\ast }$ is bounded in $R$. Then, since $v_{1}$ is
unbounded and $v_{2}^{\ast }$ is bounded in $R$, it follows that the
covering vertices $v_{1}$ and $v_{2}^{\ast }$ of $v_{2}$ do not lie in the
same connected component of $G\setminus N[v_{2}]$. That is, $V_{0}(v_{2})$
is not connected, which is a contradiction by Lemma~\ref{two-components}.
Thus, $N(v_{1})\neq N(v_{2})$.
\end{proof}

\section{Main results}
\label{main-sec}

In this section we present our main results. 
Consider a graph $G$ that is both a tolerance and a trapezoid graph, where 
$R$ is a projection representation of $G$. Then, we choose a certain unbounded vertex $u$ in~$R$ 
and we ``eliminate'' $u$ in $R$ in the following sense: assuming that $R$ has $k\geq 1$ unbounded vertices, 
we construct a projection representation~$R^{\ast}$ of $G$ with $k-1$ unbounded vertices, 
where all bounded vertices remain bounded and $u$ is transformed to a bounded vertex. 
In Section~\ref{right-left-subsec} we deal with the case where the unbounded vertex $u$ has the right or 
the left border property in~$R$, while in Section~\ref{structure-subsec} we deal with the case where $u$ has 
neither the left nor the right border property in~$R$. Finally we combine these two results in 
Section~\ref{general-subsec}, in order to eliminate all~$k$ unbounded vertices in~$R$, 
regardless of whether or not they have the right or left border property.

\subsection{The case where $u$ has the right or the left border property}
\label{right-left-subsec}

In this section we consider an arbitrary unbounded vertex $u$ of $G$ in the projection representation~$R$, 
and we assume that $u$ has the right or the left border property in $R$. 
Then, as we prove in the next theorem, there is another projection representation $R^{\ast }$ of $G$, in which 
$u$ has been replaced by a bounded vertex.

% \vspace{-0.1cm}
\begin{theorem}
\label{right-property-thm}Let $G=(V,E)\in$ \textsc{(Tolerance }$\cap $%
\textsc{\ Trapezoid)\ }$\setminus $ \textsc{Bounded Tolerance} with
the smallest number of vertices. Let $R$ be a projection representation of $%
G $ with $k$ unbounded vertices and $u$ be an unbounded vertex in $R$. If $u$
has the right or the left border property in $R$, then there exists a projection
representation $R^{\ast }$ of $G$ with $k-1$ unbounded vertices.
\end{theorem}

\begin{proof}
If $R$ is not a canonical representation of $G$, then there exists a
projection representation $R^{\ast }$ of $G$ with $k-1$ unbounded vertices
by Definition~\ref{def7}. Suppose in the sequel that $R$ is a canonical
representation of $G$. Then, for the unbounded vertex $u$ of $G$ in $R$,
there exists at least one bounded covering vertex $u^{\ast }$ of $u$ by
Lemma~\ref{bounded-hovering}. Therefore $V_{0}(u)\neq \emptyset $, and thus 
$V_{0}(u)$ is connected by Lemma~\ref{two-components}. The proof is done
constructively. Namely, we will construct the projection representations $%
R^{\prime }$, $R^{\prime \prime }$, and $R^{\prime \prime \prime }$, by
applying to $R$ sequentially the %following 
Transformations~\ref{trans1},~\ref{trans2}, and~\ref{trans3}, respectively.
Finally, $R^{\prime \prime \prime } $ is a projection representation of the
same graph $G$ with $k-1$ unbounded vertices, where $u$ is represented as a
bounded vertex in $R^{\prime \prime \prime }$.

For simplicity reasons, we add in $G$ an isolated bounded vertex $t$. This
vertex $t$ corresponds to a parallelogram $P_{t}$, such that $P_{v}\ll
_{R}P_{t}$ for every vertex $v$ of $G$. Recall that $V_{B}$ and $V_{U}$
denote the sets of bounded and unbounded vertices of $G$ in $R$,
respectively (note that $t\in V_{B}$). First, we define for every $w\in N(u)$
the value $L_{0}(w)=\min_{R}\{L(x)\ |\ x\in V_{B}\setminus N(u),P_{w}\ll
_{R}P_{x}\}$. Note that the value $L_{0}(w)$ is well defined for every $w\in
N(u)$, since in particular $t\in V_{B}\setminus N(u)$ and $P_{w}\ll
_{R}P_{t} $. Moreover, for every $w\in N(u)$, $w$ is a bounded vertex and $%
\phi _{w}>\phi _{u}$. For every vertex $x\in V_{B}\setminus N(u)$, such that 
$P_{w}\ll _{R}P_{x}$ for some $w\in N(u)$, it follows that $x\notin V_{0}(u)$
by Definition~\ref{right-left-property-def}, since $u$ has the right border
property in $R$ by assumption. Thus, for every $w\in N(u)$, $%
L_{0}(w)=\min_{R}\{L(x)\ |\ x\in (V_{B}\setminus N(u))\setminus
V_{0}(u),P_{w}\ll _{R}P_{x}\}$. Define now the value $\ell
_{0}=\max_{R}\{l(x)\ |\ x\in V_{0}(u)\}$ and the subset $N_{1}=\{w\in N(u)\
|\ r(w)<_{R}\ell _{0}\}$ of neighbors of $u$.

An example of a projection representation $R$ of a tolerance graph $G$ with
seven vertices is illustrated in Figure~\ref{fig-R-0}. In this figure, the
parallelogram $P_{u}$ of the unbounded vertex $u$ is illustrated by a bold
and dotted line. The transparent parallelograms $P_{w_{1}}$ and $P_{w_{2}}$
correspond to the neighbors $N(u)=\{w_{1},w_{2}\}$ of $u$ in $G$, the light
colored parallelograms $P_{u^{\ast }}$ and $P_{x}$ correspond to the
vertices of $V_{0}(u)=\{u^{\ast },x\}$, and the dark colored parallelograms $%
P_{y}$ and $P_{t}$ correspond to the vertices of $(V\setminus N[u])\setminus
V_{0}(u)=\{y,t\}$. In this example, $L_{0}(w_{1})=L(t)$, $L_{0}(w_{2})=L(y)$, 
and $\ell _{0}=l(x)$, while $N_{1}=\{w_{1},w_{2}\}$.

We construct now the projection representation $R^{\prime} $ from $R$ as follows.

\begin{transformation}
\label{trans1}For every $w\in N_{1}$, move the right line of $P_{w}$
parallel to the right, until either $r(w)$ comes immediately 
after $\ell_{0} $ on $L_{2}$, or $R(w)$ comes immediately before $L_{0}(w)$ on $L_{1}$.
Denote the resulting projection representation by $R^{\prime }$.
\end{transformation}

\begin{figure}[t!]
\centering
\subfigure[] { \label{fig-R-0}
\includegraphics[width=\linewidth]{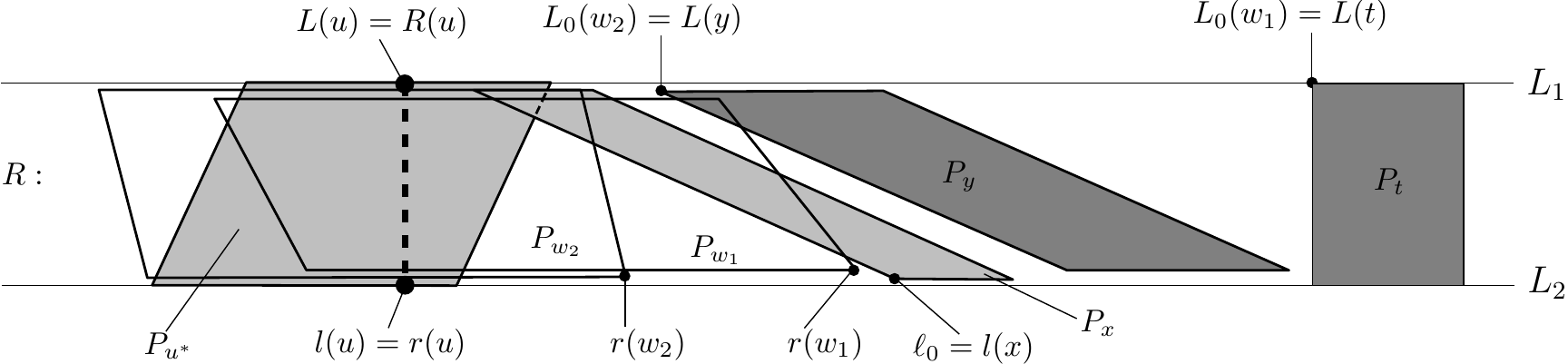}} %[scale=0.9]
\subfigure[] { \label{fig-R-1}
\includegraphics[width=\linewidth]{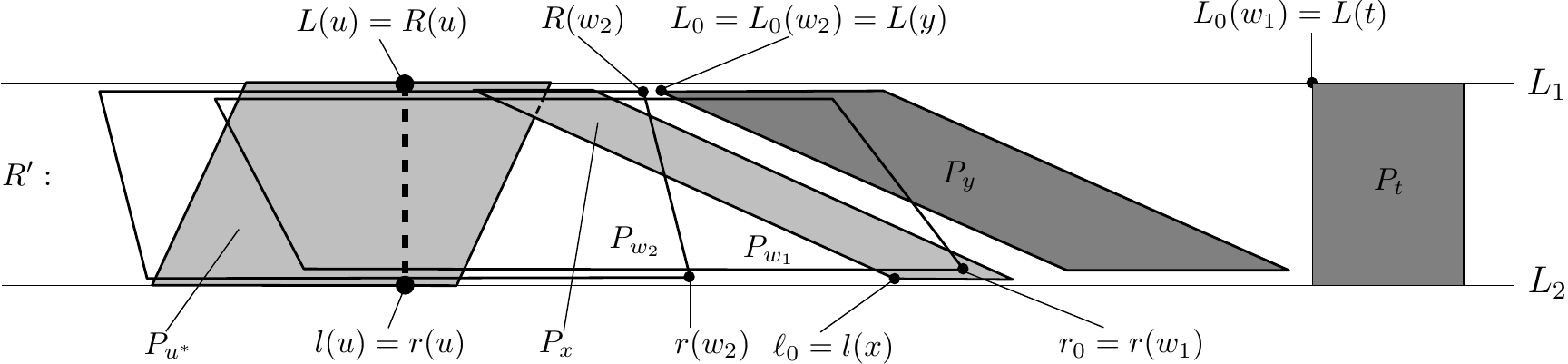}} %[scale=0.9]
\subfigure[] { \label{fig-R-2}
\includegraphics[width=\linewidth]{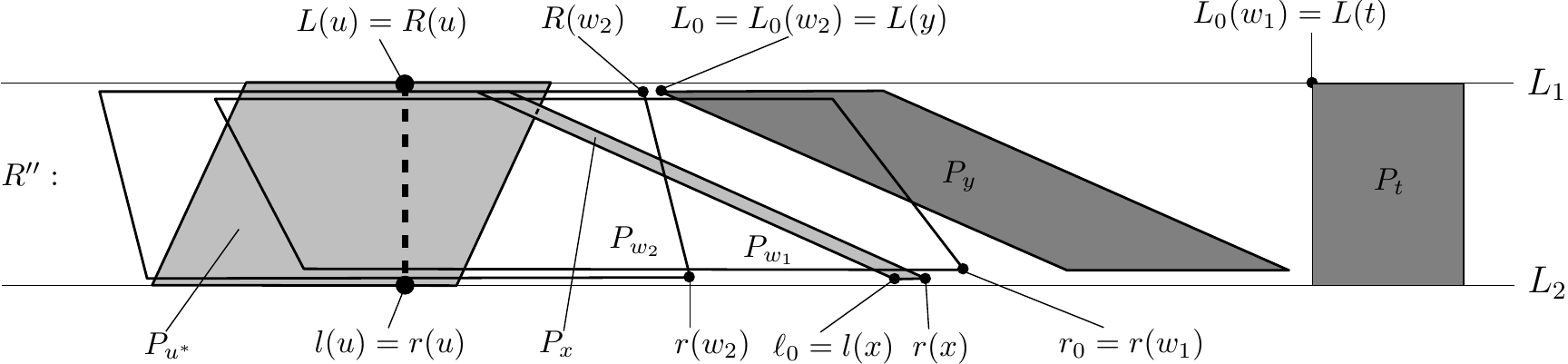}} %[scale=0.9]
\subfigure[] { \label{fig-R-3}
\includegraphics[width=\linewidth]{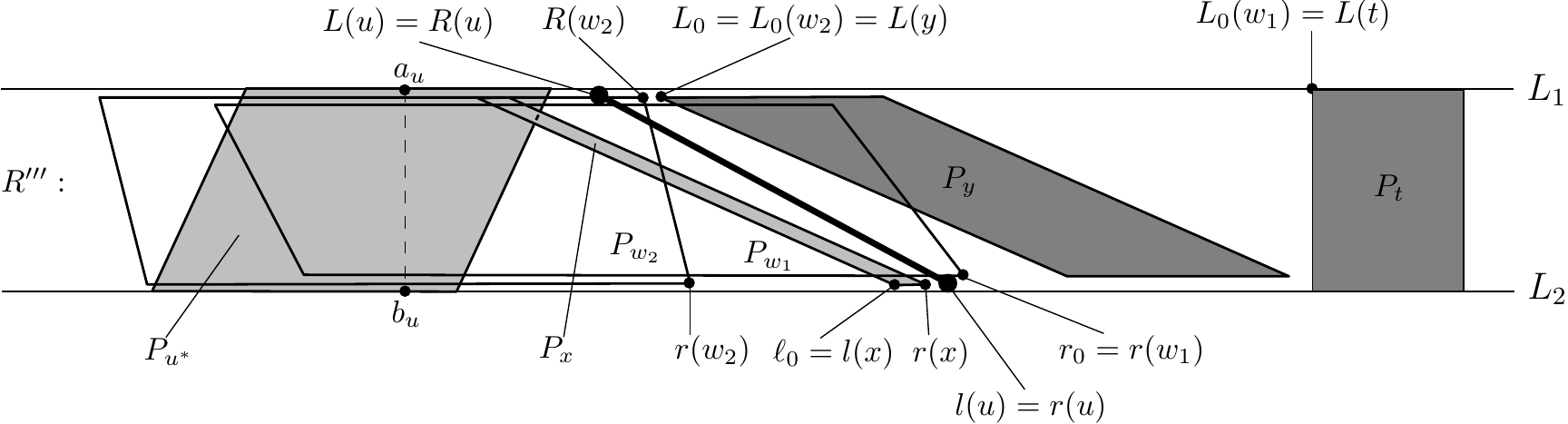}} %[scale=0.9]
\caption{(a) The projection representation $R$ of a tolerance graph $G$ with
seven vertices, and the projection representations (b) $R^{\prime}$ after
Transformation~\protect\ref{trans1}, (c) $R^{\prime\prime}$ after
Transformation~\protect\ref{trans2}, and (d) $R^{\prime\prime\prime}$ after
Transformation~\protect\ref{trans3}.}
\label{fig-transform-123}
\end{figure}

Note that the left lines of all parallelograms do not move during
Transformation~\ref{trans1}. Thus, in particular, the value of $\ell _{0}$
is the same in $R$ and in $R^{\prime }$, i.e.~$\ell _{0}=\max_{R^{\prime
}}\{l(x)\ |\ x\in V_{0}(u)\}$. As we will prove in Lemma~\ref{R'}, the
representation $R^{\prime }$ is a projection representation of the same
graph $G$, and thus the parallelograms of two bounded vertices intersect in $%
R$ if and only if they intersect also in $R^{\prime }$. Therefore, for every 
$w\in N(u)$ the value $L_{0}(w)$ remains the same in $R$ and in $R^{\prime }$, 
i.e.~$L_{0}(w)=\min_{R^{\prime }}\{L(x)\ |\ x\in (V_{B}\setminus N(u))\setminus V_{0}(u),P_{w}\ll _{R^{\prime }}P_{x}\}$ 
for every $w\in N(u)$. Define now the subset $N_{2}=\{w\in N(u)\ |\ \ell _{0}<_{R^{\prime
}}r(w)\} $ of neighbors of $u$. If $N_{2}\neq \emptyset $, we define the
value $r_{0}=\min_{R^{\prime }}\{r(w)\ |\ w\in N_{2}\}$. Then, $%
r_{0}>_{R^{\prime }}r(u)$ by Lemma~\ref{unbounded-bounded}, since $%
N_{2}\subseteq N(u)$. Since the lower right endpoint $r(w)$ of all
parallelograms $P_{w}$ in $R^{\prime } $ is greater than or equal to the
corresponding value $r(w)$ in $R$, it follows that $N(u)\setminus
N_{1}=\{w\in N(u)\ |\ \ell _{0}<_{R}r(w)\}\subseteq \{w\in N(u)\ |\ \ell
_{0}<_{R^{\prime }}r(w)\}=N_{2} $. Thus, $N(u)\setminus N_{2}\subseteq N_{1}$
and $N_{2}\cup (N_{1}\setminus N_{2})=N(u)$.

Define now the value $L_{0}=\min_{R^{\prime }}\{L(x)\ |\ x\in 
(V_{B}\setminus N(u))\setminus V_{0}(u),P_{u}\ll _{R^{\prime }}P_{x}\}$; again, $L_{0}$ is
well defined, since in particular $t\in (V_{B}\setminus N(u))\setminus
V_{0}(u) $ and $P_{u}\ll _{R^{\prime }}P_{t}$. The following property of the
projection representation $R^{\prime }$ can be obtained easily by
Transformation~\ref{trans1}.

\begin{lemma}
\label{property-trans1}For all vertices $w\in N_{1}\setminus N_{2}$, for
which $R(w)<_{R^{\prime }}L_{0}$, the values $R(w)$ lie immediately before $%
L_{0}$ in $R^{\prime }$.
\end{lemma}

\begin{proof}
Let $w\in N_{1}\setminus N_{2}$. By definition of the sets $N_{1}$ and $%
N_{2} $, it follows that $r(w)<_{R}\ell _{0}$ and $r(w)<_{R^{\prime }}\ell
_{0}$ in both $R$ and $R^{\prime }$. Thus, $R(w)$ comes immediately before $%
L_{0}(w)$ in $R^{\prime }$ during Transformation~\ref{trans1}. Consider now
a vertex $x\in (V_{B}\setminus N(u))\setminus V_{0}(u)$, such that $P_{w}\ll
_{R}P_{x}$, i.e.~$r(w)<_{R}l(x)$ and $R(w)<_{R}L(x)$. 
Then $r(u)<_{R}l(x)$, since $r(u)<_{R}r(w)$ by Lemma~\ref{unbounded-bounded}. 
Suppose that $L(x)<_{R}R(u)$. Then, $P_{x}$ intersects $P_{u}$ in $R$ and $\phi _{x}>\phi
_{u}$. Thus, since $x$ is assumed to be bounded, it follows that $x\in N(u)$, 
which is a contradiction. Therefore $R(u)<_{R}L(x)$, and thus $P_{u}\ll
_{R}P_{x}$, since also $r(u)<_{R}l(x)$. Furthermore, also $P_{u}\ll
_{R^{\prime }}P_{x}$, since $P_{u}$ and $P_{x}$ remain the same in both $R$
and $R^{\prime }$. That is, $P_{u}\ll _{R^{\prime }}P_{x}$ for every $x\in
(V_{B}\setminus N(u))\setminus V_{0}(u)$, such that $P_{w}\ll _{R}P_{x}$.
Therefore, it follows by the definitions of $L_{0}$ and of $L_{0}(w)$ that $%
L_{0}\leq L_{0}(w)$. Thus, since $R(w)$ comes immediately before $L_{0}(w)$
in $R^{\prime }$ during Transformation~\ref{trans1}, it follows that either $%
R(w)$ comes immediately before $L_{0}$ in $R^{\prime }$ during
Transformation~\ref{trans1} (in the case where $L_{0}=L_{0}(w)$) or $%
R(w)>_{R^{\prime }}L_{0}$ (in the case where $L_{0}<L_{0}(w)$). This completes the proof of the lemma.
\end{proof}

\medskip

For the example of Figure~\ref{fig-transform-123}, the projection
representation $R^{\prime}$ is illustrated in Figure~\ref{fig-R-1}. In this
figure, $L_{0}=L(y)$ and $r_{0}=r(w_{1})$, while $N_{2}=\{w_{1}\}$ and $%
N_{1}\setminus N_{2}=\{w_{2}\}$.

If $N_{2}=\emptyset $, then we set $R^{\prime \prime }=R^{\prime }$;
otherwise, if $N_{2}\neq \emptyset $, we construct the projection
representation $R^{\prime \prime }$ from $R^{\prime }$ as follows.

\begin{transformation}
\label{trans2}For every $v\in V_{0}(u)\cap V_{B}$, such 
that $r(v)>_{R^{\prime }}r_{0}$, move the right line of $P_{v}$ in $R^{\prime }$
parallel to the left, such that $r(v)$ comes immediately before $r_{0}$ 
in $L_{2}$. Denote the resulting projection representation by $R^{\prime \prime}$.
\end{transformation}

Since by Transformation~\ref{trans2} only some endpoints of vertices ${v\in
V_{0}(u)\cap V_{B}}$ are moved, it follows that the value $L_{0}$ does not
change in $R^{\prime \prime }$, i.e.~${L_{0}=\min_{R^{\prime\prime}}\{L(x)\
|\ x\in (V_{B}\setminus N(u))\setminus V_{0}(u),}$ ${P_{u}\ll _{R^{\prime \prime}}P_{x}\}}$. 
The next property of the projection representation $R^{\prime
\prime }$ follows by Lemma~\ref{property-trans1}.

\begin{corollary}
\label{property-trans2}For all vertices $w\in N_{1}\setminus N_{2}$, for
which $R(w)<_{R^{\prime \prime }}L_{0}$, the values $R(w)$ lie immediately
before $L_{0}$ in $R^{\prime \prime }$.
\end{corollary}

\begin{proof}
Let $x_{0}$ be the vertex of $(V_{B}\setminus N(u))\setminus V_{0}(u)$, such
that $L_{0}=L(x_{0})$. Recall by Lemma~\ref{property-trans1} that for all
vertices $w\in N_{1}\setminus N_{2}$, for which $R(w)<_{R^{\prime }}L_{0}$,
the values $R(w)$ lie immediately before~$L_{0}$ in $R^{\prime }$.
Furthermore, note that the parallelograms of all neighbors $w\in N(u)$ of $u$
do not move by Transformation~\ref{trans2}. Therefore, since also the value $%
L_{0}$ is the same in both $R^{\prime }$ and~$R^{\prime \prime }$, it
suffices to prove that there do not exist vertices $v\in V_{0}(u)\cap V_{B}$
and $w\in N_{1}\setminus N_{2}$, such that $R(w)<_{R^{\prime \prime
}}R(v)<_{R^{\prime \prime }}L_{0}$ in $R^{\prime \prime }$. Suppose
otherwise that $R(w)<_{R^{\prime \prime }}R(v)<_{R^{\prime \prime
}}L_{0}=L(x_{0})$ for two vertices $v\in V_{0}(u)\cap V_{B}$ and $w\in
N_{1}\setminus N_{2}$. Thus, since only the right lines of some
parallelograms $P_{v}$, where $v\in V_{0}(u)\cap V_{B}$, are moved to the
left by Transformation~\ref{trans2}, it follows that $R(w)<_{R^{\prime
}}L_{0}=L(x_{0})<_{R^{\prime }}R(v)$ in $R^{\prime }$. Therefore, in
particular $P_{v}$ intersects $P_{x_{0}}$ in $R^{\prime }$, and thus $v\in
N(x_{0})$, since both $v$ and $x_{0}$ are bounded. Thus $x_{0}\in V_{0}(u)$,
since also $v\in V_{0}(u)$. This is a contradiction, since $x_{0}\in
(V_{B}\setminus N(u))\setminus V_{0}(u)$. This completes the proof of the
corollary.
\end{proof}

\medskip

The projection representation $R^{\prime \prime }$ for the example of Figure~%
\ref{fig-transform-123} is illustrated in Figure~\ref{fig-R-2}. 
We construct now the projection representation $R^{\prime \prime \prime }$ 
from $R^{\prime \prime }$ as follows.

\begin{transformation}
\label{trans3}Move the line $P_{u}$ in $R^{\prime \prime }$, such that its
upper endpoint $L(u)=R(u)$ comes immediately 
before $\min_{R^{\prime \prime}}\{L_{0},R(w)\ |\ w\in N_{1}\setminus N_{2}\}$ 
and its lower endpoint $l(u)=r(u)$ comes immediately 
after $\max_{R^{\prime \prime }}\{r(v)\ |\ v\in V_{0}(u)\cap V_{B}\}$. 
Finally, make $u$ a bounded vertex. Denote the 
resulting projection representation by $R^{\prime \prime \prime }$.
\end{transformation}

The resulting projection representation $R^{\prime \prime \prime }$ has $k-1$ unbounded
vertices, since $u$ is represented in $R^{\prime \prime \prime }$ as a
bounded vertex. The projection representation $R^{\prime \prime \prime}$ for
the example of Figure~\ref{fig-transform-123} is illustrated in Figure~\ref%
{fig-R-3}. In this figure, the new position of the trivial parallelogram
(i.e.~line) $P_{u}$ that corresponds to the (bounded) vertex $u$ is drawn in
bold. Furthermore, for better visibility, the position of $P_{u}$ in the
previous projection representations $R$, $R^{\prime }$, and $R^{\prime
\prime }$ is pointed by a non-bold dashed line; in this figure, $a_{u}$ and $%
b_{u}$ denote the endpoints of this old position of $P_{u}$ on $L_{1}$ and
on $L_{2}$, respectively.

In the following three lemmas, we prove sequentially that $R^{\prime }$, $%
R^{\prime \prime }$, and $R^{\prime \prime \prime }$ are all projection
representations of the same tolerance graph $G$, and thus $R^{\ast
}=R^{\prime \prime \prime }$ is a projection representation of $G$ with $k-1$
unbounded vertices.

\begin{lemma}
\label{R'}$R^{\prime }$ is a projection representation of $G$.
\end{lemma}

\begin{proof}
Denote by $x_{0}$ the vertex of $V_{0}(u)$, such that $\ell _{0}=l(x_{0})$. 
Recall by Lemma~\ref{bounded-hovering} that there exists a covering vertex $u^{\ast}$ of $u$ in $G$, 
such that $u^{\ast}$ is bounded in $R$. 
Since we move the right line of some parallelograms to the right, i.e.~we
increase some parallelograms, all adjacencies of $R$ are kept in $R^{\prime
} $. Suppose that $R^{\prime }$ has the new adjacency $wv$ that is not an
adjacency in $R$, for some $w\in N_{1}$. Therefore, since we perform
parallel movements of lines, i.e.~since every slope $\phi _{z}$ in $%
R^{\prime }$ equals the value of $\phi _{z}$ in $R$ for every vertex $z$ of $%
G$, it follows that $P_{w}\ll _{R}P_{v}$ and $P_{w}$ intersects $P_{v}$ in $%
R^{\prime }$. Thus $v\notin V_{0}(u)$, since $u$ has the right border
property in $R$ by assumption. Furthermore $r(w)<_{R}\ell _{0}=l(x_{0})$,
since $w\in N_{1}$. However, since $x_{0}\in V_{0}(u)$, and since $u$ has
the right border property in $R$, it follows that $P_{w}$ intersects $%
P_{x_{0}}$ in $R$, and thus $L(x_{0})<_{R}R(w)$.

Moreover, $r(u)<_{R}r(w)<_{R}l(x_{0})$ and $L(w)<_{R}L(u)$ by Lemma~\ref%
{unbounded-bounded}. Suppose that $L(x_{0})<_{R}L(u)=R(u)$. Then, $P_{u}$
intersects $P_{x_{0}}$ in $R$ and $\phi _{x_{0}}>\phi _{u}$. Thus, $x_{0}$
is unbounded, since otherwise $x_{0}\in N(u)$, which is a contradiction.
Furthermore, $N(x_{0})\subseteq N(u)$ by Lemma~\ref{intersecting-unbounded},
and thus $x_{0}$ is an isolated vertex of $G\setminus N[u]$. Therefore,
since $x_{0}$ is unbounded and $u^{\ast }$ is bounded in $R$, it follows
that $x_{0}$ and $u^{\ast }$ do not lie in the same connected component of $%
G\setminus N[u]$. That is, $V_{0}(u)$ is not connected, which is a
contradiction. Thus, $L(u)=R(u)<_{R}L(x_{0})$, i.e.~$%
R(u)<_{R}L(x_{0})<_{R}R(w)<_{R}L(v)$ and $r(u)<_{R}r(w)<_{R}l(v)$, which
implies that $P_{u}\ll _{R}P_{v}$, and thus $v\notin N(u)$.

Consider now the projection representation $R^{\prime }$ constructed by
Transformation~\ref{trans1}. Let first $r(w)<_{R^{\prime }}l(v)$. Then,
since $P_{w}$ intersects $P_{v}$ in $R^{\prime }$, it follows that $%
L(v)<_{R^{\prime }}R(w)$, and thus $\phi _{v}>\phi _{w}$. If $v$ is an
unbounded vertex, then $w$ is not adjacent to $v$ in $R^{\prime }$, which is
a contradiction to the assumption. Thus, $v$ is a bounded vertex. Recall
that $P_{w}\ll _{R}P_{v}$ and that $v\notin V_{0}(u)$ and $v\notin N(u)$,
i.e.~$v\in (V_{B}\setminus N(u))\setminus V_{0}(u)$, and thus $L_{0}(w)\leq
_{R}L(v)$ in $R$ by definition of $L_{0}(w)$. Furthermore, since the left
lines of the parallelograms in $R$ do not move during Transformation~\ref%
{trans1}, it remains also $L_{0}(w)\leq _{R^{\prime }}L(v)$ in $R^{\prime }$%
. Therefore, since $R(w)<_{R^{\prime }}L_{0}(w)$ by definition of
Transformation~\ref{trans1}, it follows that $R(w)<_{R^{\prime }}L(v)$,
which is a contradiction, since $L(v)<_{R^{\prime }}R(w)$, as we proved
above in this paragraph.

Let now $l(v)<_{R^{\prime }}r(w)$. Suppose that $l(x_{0})<_{R^{\prime }}l(v)$%
. Then, since $r(w)$ comes in $R^{\prime }$ at most immediately after $\ell
_{0}=l(x_{0})$ on $L_{2}$, it follows that also $r(w)<_{R^{\prime }}l(v)$,
which is a contradiction. Therefore, $l(v)<_{R^{\prime }}l(x_{0})$, and thus
since the left lines of the parallelograms in $R$ do not move during
Transformation~\ref{trans1}, it follows that also $l(v)<_{R}l(x_{0})$.
Furthermore, since $L(x_{0})<_{R}R(w)$ and $P_{w}\ll _{R}P_{v}$, it follows
that $L(x_{0})<_{R}R(w)<_{R}L(v)$, and thus $P_{x_{0}}$ intersects $P_{v}$
in $R$ and $\phi _{x_{0}}>\phi _{v}$. Now, if $x_{0}$ is bounded, then $%
x_{0}v\in E$. Thus, $v\in V_{0}(u)$, since $x_{0}\in V_{0}(u)$ and $v\notin
N(u)$, which is a contradiction. Therefore, $x_{0}$ is unbounded, and thus $%
x_{0}v\notin E$. Then, since $P_{x_{0}}$ intersects $P_{v}$ in $R$ and $\phi
_{x_{0}}>\phi _{v}$, it follows that $N(x_{0})\subseteq N(v)$ by Lemma~\ref%
{intersecting-unbounded}. Recall now that there exists a bounded covering
vertex $u^{\ast }$ of $u$ in $G$, and thus $u^{\ast },x_{0}\in V_{0}(u)$.
Furthermore $u^{\ast }\neq x_{0}$, since $u^{\ast }$ is bounded and $x_{0}$
is unbounded. Therefore, since $V_{0}(u)$ is connected, $x_{0}$ is adjacent
to at least one other vertex $y\in V_{0}(u)$, and thus $y\in N(v)$, since $%
N(x_{0})\subseteq N(v)$. It follows now that $v\in V_{0}(u)$, since $y\in
V_{0}(u)$ and $v\notin N(u)$, which is again a contradiction.

Therefore, $R^{\prime }$ has no new adjacency $wv$ that is not an adjacency
in $R$, for any $w\in N_{1}$, i.e.~$R^{\prime }$ is a projection
representation of $G$. This completes the proof of the lemma.
\end{proof}

\begin{lemma}
\label{R''}$R^{\prime \prime }$ is a projection representation of $G$.
\end{lemma}

\begin{proof}
Denote by $w_{0}$ the vertex of $N_{2}$, such that $r_{0}=r(w_{0})$. Since
we move the right line of some parallelograms to the left, i.e.~we decrease
some parallelograms, no new adjacencies are introduced in $R^{\prime \prime
} $ in comparison to $R^{\prime }$. Suppose that the adjacency $vx$ has been
removed from $R^{\prime }$ in $R^{\prime \prime }$, for some $v\in
V_{0}(u)\cap V_{B}$, where $r(v)>_{R^{\prime }}r_{0}=r(w_{0})$. Therefore,
since we perform parallel movements of lines in $R^{\prime }$, i.e.~since
every slope $\phi _{z}$ in $R^{\prime \prime }$ equals the value of $\phi
_{z}$ in $R^{\prime }$ for every vertex $z$ of $G$, it follows that $%
P_{v}\ll _{R^{\prime \prime }}P_{x}$, while $P_{v}$ intersects $P_{x}$ in $%
R^{\prime }$.

Since $w_{0}\in N(u)$, and since the endpoints of $P_{w_{0}}$ do not move
during Transformation~\ref{trans2}, it follows by Lemma~\ref%
{unbounded-bounded} that $r(u)<_{R^{\prime }}r(w_{0})$ and $r(u)<_{R^{\prime
\prime }}r(w_{0})$. Thus, since $r(v)$ comes in $R^{\prime \prime }$
immediately before $r_{0}=r(w_{0})$, it follows that $r(u)<_{R^{\prime
\prime }}r(v)<_{R^{\prime \prime }}r(w_{0})$. Suppose that $x\in N(u)$.
Then, $L(x)<_{R^{\prime }}L(u)$ by Lemma~\ref{unbounded-bounded}, and thus
also $L(x)<_{R^{\prime \prime }}L(u)$, since the left lines of all
parallelograms do not move during Transformation~\ref{trans2}. Therefore, $%
R(v)<_{R^{\prime \prime }}L(x)<_{R^{\prime \prime }}L(u)=R(u)$, since $%
P_{v}\ll _{R^{\prime \prime }}P_{x}$. That is, $r(u)<_{R^{\prime \prime
}}r(v)$ and $L(v)\leq _{R^{\prime \prime }}R(v)<_{R^{\prime \prime }}R(u)$,
and thus $\phi _{v}>\phi _{u}$ in both $R^{\prime }$ and $R^{\prime \prime }$%
. Furthermore, $L(v)<_{R^{\prime }}R(u)$ (since also $L(v)<_{R^{\prime
\prime }}R(u)$) and $r(u)<_{R^{\prime }}r_{0}=r(w_{0})<_{R^{\prime }}r(v)$,
and thus $P_{v}$ intersects $P_{u}$ in $R^{\prime }$. Therefore, since $v\in
V_{B}$ and $\phi _{v}>\phi _{u}$ in $R^{\prime }$, it follows that $v\in
N(u) $, which is a contradiction. Thus, $\ x\notin N(u)$.

Now, since by assumption $vx\in E$, and since $v\in V_{0}(u)$ and $x\notin
N(u)$, it follows that $x\in V_{0}(u)$, and thus $l(x)\leq _{R}\ell _{0}$ by
definition of $\ell _{0}$. Therefore, since the left lines of all
parallelograms do not move during Transformation~\ref{trans1}, it follows
that also $l(x)\leq _{R^{\prime }}\ell _{0}$. Note that both $r_{0}=r(w_{0})$
and $l(x)$ do not move by Transformation~\ref{trans2}. Therefore, since $%
r(v) $ comes by Transformation~\ref{trans2} in $R^{\prime \prime }$
immediately before $r_{0}$, and since $P_{v}\ll _{R^{\prime \prime }}P_{x}$,
it follows that $r(v)<_{R^{\prime \prime }}r_{0}=r(w_{0})<_{R^{\prime \prime
}}l(x)$. Finally, since both $r(w_{0})$ and $l(x)$ do not move during
Transformation~\ref{trans2}, it follows that also $r(w_{0})<_{R^{\prime
}}l(x)$ in $R^{\prime }$. Thus, since $l(x)\leq _{R^{\prime }}\ell _{0}$, it
follows that $r(w_{0})<_{R^{\prime }}\ell _{0}$ in $R^{\prime }$, which is a
contradiction, since $w_{0}\in N_{2}$. Therefore, no adjacency $vx$ has been
removed from $R^{\prime }$ in $R^{\prime \prime }$, i.e.~$R^{\prime \prime }$
is a projection representation of $G$. This completes the proof of the lemma.
\end{proof}

\begin{lemma}
\label{R'''}$R^{\prime \prime \prime }$ is a projection representation of $G$.
\end{lemma}

\begin{proof}
The proof is done in two parts. In Part 1 we prove that $u$ is adjacent in $%
R^{\prime \prime \prime }$ to all vertices of $N(u)$, while in Part 2 we
prove that $u$ is not adjacent in $R^{\prime \prime \prime }$ to any vertex
of $V\setminus N[u]$.

\medskip

\emph{Part 1.} In this part we prove that $u$ is adjacent in $R^{\prime
\prime \prime }$ to all vertices of $N(u)$. Denote by $a_{u}$ and $b_{u}$
the coordinates of the upper and lower endpoint of $P_{u}$ in the initial
projection representation $R$ on $L_{1}$ and on $L_{2}$, respectively. Then,
since the endpoints of $P_{u}$ do not move by Transformations~\ref{trans1}
and~\ref{trans2}, $a_{u}$ and $b_{u}$ remain the endpoints of $P_{u}$ also
in the representations $R^{\prime }$ and $R^{\prime \prime }$; however, note
that $a_{u}$ and $b_{u}$ are not the endpoints of $P_{u}$ in $R^{\prime
\prime \prime }$. Then, $L(w)<_{R^{\prime \prime }}a_{u}$ for every $w\in
N(u)$ by Lemma~\ref{unbounded-bounded}, and thus also $L(w)<_{R^{\prime
\prime \prime }}a_{u}$ for every $w\in N(u)$, since only the endpoints of $%
P_{u}$ move during Transformation~\ref{trans3}.

Note now that $a_{u}<_{R^{\prime \prime }}L_{0}$, since $L_{0}=\min_{R^{%
\prime \prime }}\{L(x)\ |\ x\in (V_{B}\setminus N(u))\setminus
V_{0}(u),P_{u}\ll _{R^{\prime \prime }}P_{x}\}$. Furthermore, recall by
Corollary~\ref{property-trans2} that for all vertices $w\in N_{1}\setminus
N_{2}$, for which $R(w)<_{R^{\prime \prime }}L_{0}$, the values $R(w)$ lie
immediately before $L_{0}$ in $R^{\prime \prime }$. Therefore, in
particular, $a_{u}<_{R^{\prime \prime }}R(w)$ for every $w\in N_{1}\setminus
N_{2}$, since $a_{u}<_{R^{\prime \prime }}L_{0}$, and thus $L(w)<_{R^{\prime
\prime }}a_{u}<_{R^{\prime \prime }}R(w)$ for every $w\in N_{1}\setminus
N_{2}\subseteq N(u)$ by the previous paragraph. Therefore, since $%
a_{u}<_{R^{\prime \prime }}L_{0}$, and since the upper endpoint $R(u)$ of
the line $P_{u}$ lies in $R^{\prime \prime \prime }$ immediately before $%
\min_{R^{\prime \prime }}\{L_{0},R(w)\ |\ w\in N_{1}\setminus N_{2}\}$, cf.
the statement of Transformation~\ref{trans3}, it follows that also $%
L(w)<_{R^{\prime \prime \prime }}a_{u}<_{R^{\prime \prime \prime
}}R(u)<_{R^{\prime \prime \prime }}R(w)$ for every $w\in N_{1}\setminus
N_{2} $. That is, $L(w)<_{R^{\prime \prime \prime }}R(u)<_{R^{\prime \prime
\prime }}R(w)$ for every $w\in N_{1}\setminus N_{2}$, and thus $P_{u}$
intersects $P_{w}$ in $R^{\prime \prime \prime }$ for every $w\in
N_{1}\setminus N_{2}$. Therefore, since all vertices of $\{u\}\cup
N_{1}\setminus N_{2}$ are bounded in $R^{\prime \prime \prime }$, $u$ is
adjacent in $R^{\prime \prime \prime }$ to all vertices of $N_{1}\setminus
N_{2}$.

Consider now an arbitrary vertex $w\in N_{2}$. Recall that 
$r_{0}=\min_{R^{\prime }}\{r(w)\ |\ w\in N_{2}\}$, i.e.~${r_{0}\leq
_{R^{\prime }}r(w)}$. Thus, since the endpoint $r(w)$ does not move by
Transformation~\ref{trans2}, it follows that also $r_{0}\leq _{R^{\prime
\prime }}r(w)$. Furthermore, by Transformation~\ref{trans2}, $%
r(v)<_{R^{\prime \prime }}r_{0}\leq _{R^{\prime \prime }}r(w)$ for every $v\in V_{0}(u)\cap V_{B}$. 
This holds clearly also in $R^{\prime \prime
\prime }$, i.e.~$r(v)<_{R^{\prime \prime \prime }}r(w)$ for every $v\in
V_{0}(u)\cap V_{B}$ and every $w\in N_{2}$. Since the lower endpoint of the
line $P_{u}$ comes immediately after $\max_{R^{\prime \prime}} \{r(v)\ |\ V_{0}(u)\cap V_{B}\}$
in~$R^{\prime \prime \prime }$, it follows that $r(v)<_{R^{\prime \prime
\prime }}l(u)=r(u)<_{R^{\prime \prime \prime }}r(w)$ for every $v\in
V_{0}(u)\cap V_{B}$ and every $w\in N_{2}$. Thus, since also $%
L(w)<_{R^{\prime \prime \prime }}a_{u}<_{R^{\prime \prime \prime }}R(u)$ for
every $w\in N(u)$, it follows that $P_{u}$ intersects $P_{w}$ in~$R^{\prime
\prime \prime }$ for every $w\in N_{2}$. Therefore, since all vertices of $%
\{u\}\cup N_{2}$ are bounded in $R^{\prime \prime \prime }$, $u$ is adjacent
in~$R^{\prime \prime \prime }$ to all vertices of $N_{2}$. Thus, since $%
N_{2}\cup (N_{1}\setminus N_{2})=N(u)$, $u$ is adjacent in $R^{\prime \prime
\prime }$ to all vertices of~$N(u)$.

\medskip

\emph{Part 2.} In this part we prove that $u$ is not adjacent in $R^{\prime
\prime \prime }$ to any vertex of $V\setminus N[u]$. To this end, recall
first by Lemma~\ref{bounded-hovering} that $u^{\ast }$ is a bounded covering
vertex of $u$ in $G$ (and thus~$u^{\ast }\in V_{0}(u)\cap V_{B}$), such that 
$P_{u}$ intersects $P_{u^{\ast }}$ in $R$ and $\phi _{u^{\ast }}<\phi _{u}$
in $R$. Therefore, $l(u^{\ast })<_{R}l(u)=r(u)$ by Lemma~\ref%
{unbounded-hovering}, and thus also $l(u^{\ast })<_{R^{\prime \prime }}r(u)$%
, since the endpoint $l(u^{\ast })$ remains the same in the representations $%
R$, $R^{\prime }$, and $R^{\prime \prime }$. Recall now that $%
L_{0}=\min_{R^{\prime \prime }}\{L(x)\ |\ x\in (V_{B}\setminus N(u))\setminus
V_{0}(u),P_{u}\ll _{R^{\prime \prime }}P_{x}\}$. Denote by $y_{0}$ the
vertex of $(V_{B}\setminus N(u))\setminus V_{0}(u)$, such that $L_{0}=L(y_{0})$%
, and thus $P_{u}\ll _{R^{\prime \prime }}P_{y_{0}}$. Therefore, since $%
l(u^{\ast })<_{R^{\prime \prime }}r(u)$, it follows that $l(u^{\ast
})<_{R^{\prime \prime }}l(u)<_{R^{\prime \prime }}l(y_{0})$. Since $u^{\ast
}\in V_{0}(u)$ and $y_{0}\notin N(u)\cup V_{0}(u)$, it follows that $u^{\ast
}y_{0}\notin E$. Therefore, since both $u^{\ast }$ and $y_{0}$ are bounded
vertices, $P_{u^{\ast }}$ does not intersect $P_{y_{0}}$ in $R^{\prime
\prime }$, and thus $P_{u^{\ast }}\ll _{R^{\prime \prime }}P_{y_{0}}$, since 
$l(u^{\ast })<_{R^{\prime \prime }}l(y_{0})$. Moreover, since by
Transformation~\ref{trans3} only the line $P_{u}$ is moved, it follows that
also $P_{u^{\ast }}\ll _{R^{\prime \prime \prime }}P_{y_{0}}$.

Since by Transformation~\ref{trans1} only some endpoints of vertices $w\in
N_{1}\subseteq N(u)$ are moved, the value $R(u^{\ast })$ remains the same in 
$R$ and in $R^{\prime }$. Furthermore, $r(u)<_{R^{\prime }}r_{0}$ by
definition of $r_{0}$ and by Lemma~\ref{unbounded-bounded}. Suppose that the
right line of $P_{u^{\ast }}$ is moved during Transformation~\ref{trans2}.
Then, $r(u)<_{R^{\prime }}r_{0}<_{R^{\prime }}r(u^{\ast })$, while $%
r(u^{\ast })$ comes immediately before $r_{0}$ in $R^{\prime \prime }$, i.e.~%
$r(u)<_{R^{\prime \prime }}r(u^{\ast })<_{R^{\prime \prime }}r_{0}$, since $%
r_{0}$ does not move during Transformation~\ref{trans2}. Therefore, since $%
l(u^{\ast })<_{R}l(u)$ by Lemma~\ref{unbounded-hovering} (and thus also $%
l(u^{\ast })<_{R^{\prime \prime }}l(u)$), it follows that $P_{u^{\ast }}$
still intersects $P_{u}$ in $R^{\prime \prime }$.

Denote by $v_{0}$ the vertex of $V_{0}(u)\cap V_{B}$, such that $%
r(v_{0})=\max_{R^{\prime \prime }}\{r(v)\ |\ v\in V_{0}(u)\cap V_{B}\}$, cf.
the statement of Transformation~\ref{trans3}. Since $v_{0}\in V_{0}(u)$ and $%
y_{0}\notin N(u)\cup V_{0}(u)$, it follows that $v_{0}y_{0}\notin E$.
Therefore, since both $v_{0}$ and $y_{0}$ are bounded vertices, either $%
P_{y_{0}}\ll _{R^{\prime \prime }}P_{v_{0}}$ or $P_{v_{0}}\ll _{R^{\prime
\prime }}P_{y_{0}}$. Suppose that $P_{y_{0}}\ll _{R^{\prime \prime
}}P_{v_{0}}$, and thus $P_{u^{\ast }}\ll _{R^{\prime \prime }}P_{y_{0}}\ll
_{R^{\prime \prime }}P_{v_{0}}$. Then, since $u^{\ast },v_{0}\in V_{0}(u)$
and since $V_{0}(u)$ is connected, there exists at least one vertex $v\in
V_{0}(u)$, such that $P_{v}$ intersects $P_{y_{0}}$ in $R^{\prime \prime }$.
Similarly, since $y_{0}\notin N(u)\cup V_{0}(u)$, it follows that $%
vy_{0}\notin E$. Therefore, since $y_{0}$ is a bounded vertex, $v$ must be
an unbounded vertex with $\phi _{v}>\phi _{y_{0}}$, and thus $N(v)\subseteq
N(y_{0})$ by Lemma~\ref{intersecting-unbounded}. Then, $N(v)$ includes at
least one vertex $v^{\prime }\in V_{0}(u)$, and thus $v^{\prime }\in
N(y_{0}) $. Therefore, $y_{0}\in V_{0}(u)$, which is a contradiction. Thus, $%
P_{v_{0}}\ll _{R^{\prime \prime }}P_{y_{0}}$. Moreover, since by
Transformation~\ref{trans3} only the line $P_{u}$ is moved, it follows that
also $P_{v_{0}}\ll _{R^{\prime \prime \prime }}P_{y_{0}}$.

We will prove in the following that $u$ is not adjacent in $R^{\prime \prime
\prime }$ to any vertex $x\notin N(u)$. For the sake of contradiction,
suppose that $P_{x}$ intersects $P_{u}$ in $R^{\prime \prime \prime }$, for
some vertex $x\notin N(u)$. We distinguish in the following the cases
regarding $x$.

\emph{Case 2a.} $x\in V_{B}\setminus N(u)$ (i.e.~$x$ is bounded) and $x\in
V_{0}(u)$. Then, $r(x)\leq _{R^{\prime \prime }}r(v_{0})$ and $r(u^{\ast
})\leq _{R^{\prime \prime }}r(v_{0})$ by definition of $v_{0}$, and thus
also $r(x)\leq _{R^{\prime \prime \prime }}r(v_{0})$ and $r(u^{\ast })\leq
_{R^{\prime \prime \prime }}r(v_{0})$. Therefore, by Transformation~\ref%
{trans3}, $r(x)\leq _{R^{\prime \prime \prime }}r(v_{0})<_{R^{\prime \prime
\prime }}l(u)$, i.e.~$r(x)<_{R^{\prime \prime \prime }}l(u)$, and thus $%
L(u)<_{R^{\prime \prime \prime }}R(x)$, since we assumed that $P_{x}$
intersects $P_{u}$ in $R^{\prime \prime \prime }$. Furthermore, $r(x)\leq
_{R^{\prime \prime \prime }}r(v_{0})<_{R^{\prime \prime \prime }}l(y_{0})$,
i.e.~$r(x)<_{R^{\prime \prime \prime }}l(y_{0})$, since $P_{v_{0}}\ll
_{R^{\prime \prime \prime }}P_{y_{0}}$. Recall by Corollary~\ref%
{property-trans2} that for all vertices $w\in N_{1}\setminus N_{2}$, for
which $R(w)<_{R^{\prime \prime }}L_{0}=L(y_{0})$, the values $R(w)$ lie
immediately before $L_{0}$ in $R^{\prime \prime }$, and thus also in $%
R^{\prime \prime \prime }$. Thus, since $L(u)<_{R^{\prime \prime \prime
}}R(x)$, and since the upper endpoint $L(u)=R(u)$ of $P_{u}$ comes
immediately before $\min \{L_{0},R(w)\ |\ w\in N_{1}\setminus N_{2}\}$ in $%
R^{\prime \prime \prime }$, it follows that $L(u)<_{R^{\prime \prime \prime
}}L_{0}=L(y_{0})<_{R^{\prime \prime \prime }}R(x)$. Therefore, since also $%
r(x)<_{R^{\prime \prime \prime }}l(y_{0})$, $P_{x}$ intersects $P_{y_{0}}$
in $R^{\prime \prime \prime }$, and thus also in $R^{\prime \prime }$. Then $%
xy_{0}\in E$, since both $x$ and $y_{0}$ are bounded, and therefore $%
y_{0}\in V_{0}(u)$, which is a contradiction. It follows that $P_{x}$ does
not intersect $P_{u}$ in $R^{\prime \prime \prime }$ for every $x\in
V_{B}\setminus N(u)$, such that $x\in V_{0}(u)$. In particular, since $%
u^{\ast },v_{0}\in V_{B}\setminus N(u)$ and $u^{\ast },v_{0}\in V_{0}(u)$,
it follows that neither $P_{u^{\ast }}$ nor $P_{v_{0}}$ intersects $P_{u}$
in $R^{\prime \prime \prime }$. Therefore, since $r(u^{\ast })\leq
_{R^{\prime \prime \prime }}r(v_{0})<_{R^{\prime \prime \prime }}l(u)$ by
Transformation~\ref{trans3}, it follows that $P_{u^{\ast }}\ll _{R^{\prime
\prime \prime }}P_{u}$ and $P_{v_{0}}\ll _{R^{\prime \prime \prime }}P_{u}$.

\emph{Case 2b.} $x\in V_{B}\setminus N(u)$ (i.e.~$x$ is bounded) and $%
x\notin V_{0}(u)$. Then, $u^{\ast }x\notin E$, since $u^{\ast }\in V_{0}(u)$%
. Furthermore, since both $x$ and $u^{\ast }$ (resp.~$v_{0}$) are bounded
vertices, $P_{u^{\ast }}$ (resp.~$P_{v_{0}}$) does not intersect $P_{x}$ in $%
R^{\prime \prime \prime }$, i.e.~either $P_{x}\ll _{R^{\prime \prime \prime
}}P_{u^{\ast }}$ or $P_{u^{\ast }}\ll _{R^{\prime \prime \prime }}P_{x}$
(resp.~either $P_{x}\ll _{R^{\prime \prime \prime }}P_{v_{0}}$ or $%
P_{v_{0}}\ll _{R^{\prime \prime \prime }}P_{x}$). If $P_{x}\ll _{R^{\prime
\prime \prime }}P_{u^{\ast }}$ (resp.~$P_{x}\ll _{R^{\prime \prime \prime
}}P_{v_{0}}$), then $P_{x}\ll _{R^{\prime \prime \prime }}P_{u^{\ast }}\ll
_{R^{\prime \prime \prime }}P_{u}$ (resp.~$P_{x}\ll _{R^{\prime \prime
\prime }}P_{v_{0}}\ll _{R^{\prime \prime \prime }}P_{u}$) by the previous
paragraph. This is a contradiction to the assumption that $P_{x}$ intersects 
$P_{u}$ in $R^{\prime \prime \prime }$. Therefore $P_{u^{\ast }}\ll
_{R^{\prime \prime \prime }}P_{x}$ and $P_{v_{0}}\ll _{R^{\prime \prime
\prime }}P_{x}$, and thus also $P_{u^{\ast }}\ll _{R^{\prime \prime }}P_{x}$
and $P_{v_{0}}\ll _{R^{\prime \prime }}P_{x}$. Thus, in particular $%
r(v_{0})<_{R^{\prime \prime \prime }}l(x)$. Furthermore, the lower endpoint $%
l(u)=r(u)$ of $P_{u}$ comes by Transformation~\ref{trans3} immediately after 
$r(v_{0})$ in $R^{\prime \prime \prime }$, and thus $r(v_{0})<_{R^{\prime
\prime \prime }}r(u)<_{R^{\prime \prime \prime }}l(x)$. Then $%
L(x)<_{R^{\prime \prime \prime }}R(u)$, since we assumed that $P_{x}$
intersects $P_{u}$ in $R^{\prime \prime \prime }$.

We distinguish now the cases according to the relative positions of $P_{u}$
and $P_{x}$ in $R^{\prime \prime }$. If $P_{x}\ll _{R^{\prime \prime }}P_{u}$%
, then $P_{u^{\ast }}\ll _{R^{\prime \prime }}P_{x}\ll _{R^{\prime \prime
}}P_{u}$ by the previous paragraph, which is a contradiction, since $%
P_{u^{\ast }}$ intersects $P_{u}$ in $R^{\prime \prime }$, as we proved
above. If $P_{u}\ll _{R^{\prime \prime }}P_{x}$, then $L_{0}\leq _{R^{\prime
\prime }}L(x)$, since $x\in (V_{B}\setminus N(u))\setminus V_{0}(u)$ and $%
L_{0}=\min_{R^{\prime \prime }}\{L(x)\ |\ x\in (V_{B}\setminus N(u))\setminus
V_{0}(u),P_{u}\ll _{R^{\prime \prime }}P_{x}\}$. Thus $R(u)<_{R^{\prime
\prime \prime }}L_{0}\leq _{R^{\prime \prime \prime }}L(x)$ by
Transformation~\ref{trans3}, which is a contradiction, since $%
L(x)<_{R^{\prime \prime \prime }}R(u)$ by the previous paragraph. If $P_{u}$
intersects $P_{x}$ in $R^{\prime \prime }$, then $\phi _{x}<\phi _{u}$ in $%
R^{\prime \prime }$, since $x$ is bounded, $u$ is unbounded, and $x\notin
N(u)$. Therefore, $N(u)\subseteq N(x)$ by Lemma~\ref{intersecting-unbounded}%
, and thus $x$ is a covering vertex of $u$, i.e.~$x\in V_{0}(u)$, which is a
contradiction to the assumption of Case 2b. Thus, $P_{x}$ does not intersect 
$P_{u}$ in $R^{\prime \prime \prime }$, for every $x\in V_{B}\setminus N(u)$, 
such that $x\notin V_{0}(u)$.

\emph{Case 2c.} $x\in V_{U}$ (i.e.~$x$ is unbounded), such that $\phi
_{x}<\phi _{u}$ in $R^{\prime \prime \prime }$. Then, since both $P_{x}$ and 
$P_{u}$ are lines in $R^{\prime \prime \prime }$, it follows that $%
l(x)<_{R^{\prime \prime \prime }}l(u)$ and $R(x)>_{R^{\prime \prime \prime
}}R(u)$. Thus, by Transformation~\ref{trans3}, $l(x)<_{R^{\prime \prime
\prime }}r(v_{0})<_{R^{\prime \prime \prime }}l(u)$ and $R(u)<_{R^{\prime
\prime \prime }}L_{0}=L(y_{0})<_{R^{\prime \prime \prime }}R(x)$. Since $%
P_{v_{0}}\ll _{R^{\prime \prime \prime }}P_{y_{0}}$, it follows that $P_{x}$
intersects both $P_{v_{0}}$ and $P_{y_{0}}$ in $R^{\prime \prime \prime }$
(and thus also in $R^{\prime \prime })$, and that $\phi _{x}<\phi _{v_{0}}$
and $\phi _{x}<\phi _{y_{0}}$. Therefore, since both $v_{0}$ and $y_{0}$ are
bounded, it follows that $x\in N(v_{0})$ and $x\in N(y_{0})$. Thus $%
x,y_{0}\in V_{0}(u)$, since $v_{0}\in V_{0}(u)$. This is a contradiction,
since $y_{0}\notin V_{0}(u)$ by definition of $y_{0}$. It follows that $%
P_{x} $ does not intersect $P_{u}$ in $R^{\prime \prime \prime }$ for every $%
x\in V_{U}$, for which $\phi _{x}<\phi _{u}$ in $R^{\prime \prime \prime }$.

\medskip

Summarizing, due to Part 1 and due to Cases 2a, 2b, and 2c of Part 2, it
follows that $P_{u}$ intersects in $R^{\prime \prime \prime }$ only the
parallelograms $P_{z}$, for every $z\in N(u)$, and possibly some trivial
parallelograms (lines) $P_{x}$, where $x\in V_{U}$ and $\phi _{x}>\phi _{u}$
in $R^{\prime \prime \prime }$. However, since $\phi _{x}>\phi _{u}$ in $%
R^{\prime \prime \prime }$ for all these vertices $x$, it follows that $u$
is not adjacent to these vertices in $R^{\prime \prime \prime }$. Thus $%
R^{\prime \prime \prime }$ is a projection representation of $G$, since $%
R^{\prime \prime }$ is a projection representation of $G$ by Lemma~\ref{R''}. 
This completes the proof of the lemma.
\end{proof}

\medskip

Thus, $R^{\ast }=R^{\prime \prime \prime }$ is a projection representation
of $G$ with $k-1$ unbounded vertices. This completes the proof of Theorem~%
\ref{right-property-thm}.
\end{proof}

\subsection{The case where $u$ has neither the left nor the right border property}
\label{structure-subsec}

In this section we consider graphs in \textsc{(Tolerance }$\cap $\textsc{\
Trapezoid)\ }$\setminus $ \textsc{Bounded Tolerance} that admit a
projection representation, in which there is no unbounded vertex $u$ with
the right or the left border property. 
The proof of the main Theorem~\ref{no-property-thm} of this section is based 
on the fact that~$G$ has simultaneously a projection representation $R$ and a trapezoid representation~$R_{T}$. 
In this theorem we choose a certain unbounded vertex $u$ of $G$ 
and we prove that there is another projection representation~$R^{\ast }$ of $G$, 
in which $u$ has been replaced by a bounded vertex.
First, we introduce in the following the notion of \emph{neighborhood maximality} for unbounded 
vertices in a tolerance graph.

% \vspace{-0.1cm}
\begin{definition}
\label{unbounded-maximal}Let~$G$ be a tolerance graph, $R$ be a projection
representation of~$G$, and~$u$ be an unbounded vertex in~$R$. Then,~$u$ is 
\emph{unbounded-maximal} if there exists no unbounded vertex~$v$ in~$R$,
such that~$N(u)\subset N(v)$.
\end{definition}
% \vspace{-0.1cm}

This notion of an unbounded-maximal vertex will be used in Lemma~\ref{min-angle-2}, 
in order to obtain for an arbitrary tolerance graph $G$ a projection representation with a special property. 
Before we present Lemma~\ref{min-angle-2}, we first present the next auxiliary lemma.

\begin{lemma}
\label{min-angle-1}Let $G$ be a tolerance graph, $R$ be a projection
representation of $G$, and $u$ be an unbounded vertex of $G$ in $R$, such
that $u$ is unbounded-maximal. Then, there exists a projection
representation $R^{\ast }$ of $G$ with the same unbounded vertices, such
that $\phi _{u}<\phi _{v}$ for every unbounded vertex $v\neq u$, for which $%
N(v)\subset N(u)$.
\end{lemma}

\begin{proof}
First, recall that we can assume w.l.o.g.~that all slopes of the
parallelograms in a projection representation are 
distinct~\cite{GolTol04,IsaakNT03,MSZ-Model-SIDMA-09}. %FishburnTrotter99
We will construct the projection representation~$R^{\ast}$ of~$G$ as follows. 
Let $u$ be an unbounded vertex of~$G$ in $R$, such that $u$ is unbounded-maximal, 
and let $v\neq u$ be an arbitrary unbounded vertex of~$G$ in $R$, 
such that $N(v)\subset N(u)$ and $\phi _{v}<\phi _{u}$. Suppose first that $P_{u}$
intersects $P_{v}$ in $R$. Then, since $uv\notin E$ and $\phi _{v}<\phi _{u}$, 
it follows that $N(u)\subseteq N(v)$ by Lemma~\ref{intersecting-unbounded}, 
which is a contradiction.

Suppose now that $P_{v}$ does not intersect $P_{u}$ in $R$. Let $P_{u}\ll
_{R}P_{v}$, i.e.~$r(u)<_{R}r(v)$ and $L(u)<_{R}L(v)$. Furthermore, let $%
\Delta =r(v)-r(u)$. Since for every $w\in N(v)$, it holds also $w\in N(u)$,
it follows by Lemma~\ref{unbounded-bounded} that $r(u)<_{R}r(v)<_{R}r(w)$
and $L(w)<_{R}L(u)<_{R}L(v)$ for every $w\in N(v)\subset N(u)$. Furthermore, 
$\phi _{w}>\phi _{u}>\phi _{v}$ for every $w\in N(v)\subset N(u)$. We can
now move the upper endpoint $L(v)$ of the line $P_{v}$ in $R$ to the point $%
L(u)+\Delta -\varepsilon $, for a sufficiently small positive number $%
\varepsilon >0$. In the resulting projection representation $R^{\prime }$, $%
\phi _{u}<\phi _{v}$.

We will prove that $R^{\prime }$ is a projection representation of the same
graph $G$. Indeed, consider first a vertex $w\in N(v)$. Then, $%
r(u)<_{R^{\prime }}r(v)<_{R^{\prime }}r(w)$ and $L(w)<_{R^{\prime
}}L(u)<_{R^{\prime }}L(v)=L(u)+\Delta -\varepsilon $. Furthermore, $\phi
_{u}<\phi _{v}<\phi _{w}$, since $\varepsilon >0$ has been chosen to be
sufficiently small. Therefore, $P_{v}$ still intersects $P_{w}$ in $%
R^{\prime }$ and $\phi _{v}<\phi _{w}$ for every $w\in N(v)$, i.e.~$v$
remains adjacent in $R^{\prime }$ to all vertices $w\in N(v)$.

Suppose now that $v$ obtains a new adjacency with a vertex $y$ in $R^{\prime
}$. Then, due to Lemma~\ref{unbounded-bounded}, $y$ is bounded in both $R$
and $R^{\prime }$, $r(v)<_{R^{\prime }}r(y)$ and $L(y)<_{R^{\prime }}L(v)$.
Since the lower endpoint $r(v)$ of $P_{v}$ remains the same in both $R$ and $%
R^{\prime }$, and since the upper endpoint $L(v)$ of $P_{v}$ in $R^{\prime }$
is to the left of the upper endpoint of $P_{v}$ in $R$, it follows that also 
$r(v)<_{R}r(y)$ and $L(y)<_{R}L(v)$, i.e.~$P_{y}$ intersects $P_{v}$ also in 
$R$. Thus, since the slope $\phi _{v}$ in $R$ is smaller than the
corresponding slope $\phi _{v}$ in $R^{\prime }$, it follows that $y$ is
adjacent to $v$ also in $R$, i.e.~$y\in N(v)$, which is a contradiction.
Therefore, $v$ does not obtain any new adjacency in $R^{\prime }$. Thus, $v$
is adjacent in $R^{\prime }$ to exactly the vertices $w\in N(v)$, i.e.~$%
R^{\prime }$ is a projection representation of the same tolerance graph $G$.

The case where $P_{v}\ll _{R}P_{u}$ is symmetric. Namely, in this case let $%
\Delta =L(u)-L(v)$; then, construct the projection representation $R^{\prime
}$ by moving the lower endpoint $r(v)$ of the line $P_{v}$ in $R$ to the
point $r(u)-\Delta +\varepsilon $, for a sufficiently small positive number $%
\varepsilon >0$. Similarly, the resulting projection representation $%
R^{\prime }$ is a projection representation of $G$, while $\phi _{u}<\phi
_{v}$. We repeat the above procedure, as long as there exists an unbounded
vertex $v\neq u$ in $R$, such that $N(v)\subset N(u)$ and $\phi _{v}<\phi
_{u}$. The resulting projection representation $R^{\ast }$ of $G$ satisfies
the conditions of the lemma.
\end{proof}

\medskip

We are now ready to present Lemma~\ref{min-angle-2}.

% \vspace{-0.1cm}
\begin{lemma}
\label{min-angle-2}Let~$G$ be a tolerance graph and~$R$ be a projection
representation of~$G$ with at least one unbounded vertex. Then, there exists
a projection representation~$R^{\ast }$ of~$G$ with the same unbounded
vertices, such that the unbounded vertex~$u$, for which~$\phi _{u}=\min
\{\phi _{x}\ |\ x\in V_{U}\}$ in~$R^{\ast }$, is unbounded-maximal.
\end{lemma}
% \vspace{-0.1cm}

\begin{proof}
Recall that $V_{U}$ denotes the set of unbounded vertices of $G$ in $R$. Let 
$S=\{u\in V_{U}\ |\ u$ is unbounded-maximal$\}$. Furthermore, let $R^{\prime
}$ be the projection representation obtained by applying for every $u\in S$
the procedure described in the proof of Lemma~\ref{min-angle-1}. Then, $%
R^{\prime }$ has the same unbounded vertices $V_{U}$, while $\phi _{u}<\phi
_{v}$ for every $u\in S$ and every unbounded vertex $v\neq u$, for which $%
N(v)\subset N(u)$. We choose now $u$ to be that unbounded vertex, for which 
${\phi _{u}=\min \{\phi _{x}\ |\ x\in S\}}$. Then, $u$ satisfies the conditions
of the lemma.
\end{proof}

\medskip

Assume that there exists a graph  $G\in$ \textsc{(Tolerance }$\cap $\textsc{\ Trapezoid)\ }$\setminus $ \textsc{Bounded Tolerance}, 
and let $G$ have the smallest number of vertices. 
Furthermore, let $R$ and $R_{T}$ be a canonical projection 
and a trapezoid representation of $G$, respectively, 
and $u$ be an arbitrary unbounded vertex of $G$ in $R$. 
Then $V_{0}(u)\neq \emptyset$ by Lemma~\ref{bounded-hovering},
and thus also $V_{0}(u)$ is connected by Lemma~\ref{two-components}. 
Therefore, since $u$ is not adjacent to any vertex of $V_{0}(u)$ by
Definition~\ref{hovering}, either all trapezoids of $V_{0}(u)$ lie to the
left, or all to the right of $T_{u}$ in $R_{T}$. 

Consider first the case where all trapezoids of $V_{0}(u)$ lie to the \emph{left} of $T_{u}$ in $R_{T}$, 
i.e.~$T_{x}\ll _{R_{T}}T_{u}$ for every $x\in V_{0}(u)$. 
Recall by Lemma~\ref{not-equal} that $N(v)\neq N(u)$ for every unbounded vertex $v\neq u$ in~$R$. 
Denote by $Q_{u}=\{v\in V_{U}\ |\ N(v)\subset N(u)\}$ the set of unbounded vertices $v$
of $G$ in $R$, whose neighborhood set is strictly included in the
neighborhood set of $u$. The next lemma follows easily by the
definition of $Q_{u}$.

\begin{lemma}
\label{Qu-1}For every $v\in Q_{u}$, every covering vertex $u^{\ast }$ of $u$ is also a covering vertex of $v$.
Furthermore, $Q_{u}\cap V_{0}(u)=\emptyset $.
\end{lemma}

\begin{proof}
Since $u^{\ast}$ is a covering vertex of $u$ by assumption, 
$u^{\ast} \notin N(u)$ and $N(u)\subseteq N(u^{\ast})$ by Definition~\ref{hovering}. 
Let $v\in Q_{u}$. Then, since $N(v)\subset N(u)$ and $u^{\ast }\notin N(u)$,
it follows that $u^{\ast }\notin N(v)$. Furthermore, $N(v)\subset
N(u)\subseteq N(u^{\ast })$, and thus $u^{\ast }$ is a covering vertex of $v$
by Definition~\ref{hovering}. Suppose now that $v\in V_{0}(u)$. Then, $v$ is
an isolated vertex in $G\setminus N[u]$, since $N(v)\subset N(u)$. Thus,
since $v$ is unbounded and $u^{\ast }$ is bounded, i.e.~$v\neq u^{\ast }$,
it follows that $v$ and $u^{\ast }$ do not lie in the same connected
component of $V_{0}(u)$, i.e.~$V_{0}(u)$ is not connected, which is a
contradiction. Thus, $v\notin V_{0}(u)$ for every $v\in Q_{u}$, i.e.~$%
Q_{u}\cap V_{0}(u)=\emptyset $.
\end{proof}

\medskip

Since no two unbounded vertices are adjacent, it follows in particular that 
$T_{v}$ does not intersect~$T_{u}$ in $R_{T}$, for every $v\in Q_{u}$. 
Therefore, we can partition the set $Q_{u}$ into the two 
subsets~$Q_{1}(u)=\{v\in Q_{u}\ |\ T_{v}\ll _{R_{T}}T_{u}\}$ 
and~$Q_{2}(u)=\{v\in Q_{u}\ |\ T_{u}\ll _{R_{T}}T_{v}\}$. 

Consider now a vertex $v\in Q_{1}(u)\subseteq Q_{u}$. 
Note that for every $x\in V_{0}(u)$, $T_{v}$ does not intersect~$T_{x}$ in $R_{T}$, 
since otherwise $v\in V_{0}(u)$, which is a contradiction by Lemma~\ref{Qu-1}. 
Therefore, since in particular $V_{0}(u)$ is connected by Lemma~\ref{two-components}, 
it follows that for every $x\in V_{0}(u)$, either $T_{v}\ll _{R_{T}}T_{x}$ or $T_{x}\ll _{R_{T}}T_{v}$. 
We will now prove that $T_{v}\ll _{R_{T}}T_{x}$ for every $x\in V_{0}(u)$. 
Suppose otherwise that~${T_{x}\ll _{R_{T}}T_{v}}$ for every $x\in V_{0}(u)$. 
Then, since $v\in Q_{1}(u)$, it follows that $T_{x}\ll_{R_{T}}T_{v}\ll _{R_{T}}T_{u}$ for every $x\in V_{0}(u)$. 
Therefore, since $V_{0}(u)$ includes all covering vertices of $u$ by Definition~\ref{hovering}, it follows that 
$T_{x_{0}}\ll_{R_{T}}T_{v}\ll _{R_{T}}T_{u}$ for every covering vertex $x_{0}$ of $u$. 
Thus, since $N(u)\subseteq N(x_{0})$, it follows that $T_{z}$ intersects $T_{v}$ in $R_{T}$ 
for every $z\in N(u)\subseteq N(x_{0})$. Therefore $N(v)\subseteq N(u)$, 
which is a contradiction, since $v\in Q_{1}(u)\subseteq Q_{u}$. 
Therefore $T_{v}\ll _{R_{T}}T_{x}$ for every~$v\in Q_{1}(u)$ and every~$x\in V_{0}(u)$, 
i.e.~$Q_{1}(u)=\{v\in Q_{u}\ |\ T_{v}\ll _{R_{T}}T_{x}$ for every $x\in V_{0}(u)\}$.

Consider now the case where all trapezoids of $V_{0}(u)$ lie to the \emph{right} of $T_{u}$ in $R_{T}$, 
i.e.~$T_{u}\ll _{R_{T}}T_{x}$ for every $x\in V_{0}(u)$. 
Then, by performing vertical axis flipping of $R_{T}$, 
we partition similarly to the above the set $Q_{u}$ into the sets $Q_{1}(u)$ and $Q_{2}(u)$. 
That is, in this (symmetric) case the sets $Q_{1}(u)$ and $Q_{2}(u)$ will 
be~$Q_{1}(u)=\{v\in Q_{u}\ |\ T_{x}\ll _{R_{T}}T_{v}$ for every $x\in V_{0}(u)\}$ 
and~$Q_{2}(u)=\{v\in Q_{u}\ |\ T_{v}\ll _{R_{T}}T_{u}\}$.

In the following we define three conditions on $G$, regarding the unbounded
vertices of~$G$ in~$R$; the third one depends also on the trapezoid
representation $R_{T}$ of $G$. The second condition is weaker than the first
one, while the third condition is weaker than the other two, as it is stated
in Observation~\ref{ass-obs}. Then, we prove Theorem~\ref{no-property-thm},
assuming that the third condition holds.

\begin{condition}
\label{ass1}The projection representation~$R$ of~$G$ has exactly one unbounded vertex.
\end{condition}\vspace{-0.5cm}%
\begin{condition}%
\label{ass2}For every unbounded vertex~$u$ of~$G$ in~$R$, $Q_{u}=\emptyset$;
namely, all unbounded vertices are unbounded-maximal.
%%%i.e.~there exists no unbounded vertex~$v$, for which~$N(v)\subset N(u)$.%
\end{condition}\vspace{-0.5cm}%
\begin{condition}%
\label{ass3}For every unbounded vertex~$u$ of~$G$ in~$R$,~$%
Q_{2}(u)=\emptyset $, i.e.~$Q_{u}=Q_{1}(u)$.
\end{condition}\vspace{-0.2cm}

The next observation, which connects the above conditions, 
follows easily.\vspace{-0.2cm}

\begin{observation}
\label{ass-obs}Condition~\ref{ass1} implies Condition~\ref{ass2}, and
Condition~\ref{ass2} implies Condition~\ref{ass3}.
\end{observation}
 \vspace{-0.2cm}

In the remainder of the section we assume that Condition~\ref{ass3} holds, 
which is weaker than Conditions~\ref{ass1} and~\ref{ass2}. 
We present now the main theorem of this section. 

% \vspace{-0.1cm}
\begin{theorem}
\label{no-property-thm}Let $G=(V,E)\in$ \textsc{(Tolerance }$\cap $\textsc{%
\ Trapezoid)\ }$\setminus $ \textsc{Bounded Tolerance} with the
smallest number of vertices. Let $R_{T}$ be a trapezoid representation of $G$
and $R$ be a projection representation of $G$ with $k$ unbounded vertices.
Then, assuming that $G$ satisfies Condition~\ref{ass3}, there exists a projection
representation $R^{\ast }$ of $G$ with $k-1$ unbounded vertices.
\end{theorem}

\begin{proof}[Proof (sketch)]
The full proof of the theorem can be found in the Appendix. %%% of~\cite{MZ-Intersection-Arxiv}
The proof is done constructively, by exploiting the fact that~$G$ can be 
represented by both the projection representation~$R$ and the trapezoid 
representation $R_{T}$. 

If at least one unbounded vertex of $G$ in $R$ has the right or the left border property,
there exists a projection representation $R^{\ast}$ of $G$ with $k-1$ unbounded
vertices by Theorem~\ref{right-property-thm}, where all unbounded vertices of $R^{\ast}$ are 
also unbounded vertices in $R$. 
Suppose that every unbounded vertex of~$G$ in~$R$ has
neither the right nor the left border property in $R$. Let $u$ be the
unbounded vertex in $R$, such that $\phi _{u}=\min \{\phi _{x}\ |\ x\in
V_{U}\}$ in $R$; then, we may assume by Lemma~\ref{min-angle-2} 
that $u$ is an unbounded-maximal vertex of $G$. 
By possibly performing vertical axis flipping of $R_{T}$, we may assume 
w.l.o.g.~that all trapezoids of $V_{0}(u)$ lie to the \emph{left} of $T_{u}$ 
in $R_{T}$, i.e.~$T_{x}\ll _{R_{T}}T_{u}$ for every $x\in V_{0}(u)$.

We now construct a projection representation $R^{\ast}$ of the same graph $G$, 
in which $u$ is replaced by a bounded vertex, while all other $k-1 $ unbounded 
vertices of $R$ remain also unbounded in $R^{\ast}$. 
We start by constructing a subgraph $G_{0}$ of $G$, such that $u\in
V(G_{0})$ and all vertices of $V(G_{0})\setminus \{u\}$ are bounded. Then,
we prove that $G_{0}\setminus \{u\}$ is a module in $G\setminus \{u\}$, by
exploiting the fact that~$G$ can be represented by both $R$ and $R_{T}$.
That is, we prove that $N(v)\setminus V(G_{0})=N(v^{\prime })\setminus
V(G_{0})$ for all vertices $v,v^{\prime }\in V(G_{0})\setminus \{u\}$. Furthermore,
we define in a particular way a line segment $\ell $ with endpoints on the 
lines $L_{1}$ and $L_{2}$, respectively. Then, we
replace the parallelograms of the vertices of~$G_{0}$ in $R$ by a particular
projection representation $R_{0}$ of $G_{0}$, which is $\varepsilon $%
-squeezed with respect to the line segment $\ell $. We denote the resulting
projection representation by $R_{\ell }$. 
Then we prove that $R_{\ell }\setminus \{u\}$ is a projection representation of the graph $G\setminus \{u\}$ 
--~although $R_{\ell }$ is not necessarily a projection representation of $G$~-- 
and that $u$ has the right border property in $R_{\ell }$. 
Then, similarly to Transformations~\ref{trans1},~\ref{trans2}, and~\ref{trans3} in the proof of 
Theorem~\ref{right-property-thm}, we apply three other transformations to $R_{\ell }$
(Transformations~4, %%%~\ref{trans4}
5, %%%~\ref{trans5} 
and~6, %%%~\ref{trans6} 
respectively),
obtaining thus the projection representations~$R_{\ell}^{\prime}$,~$R_{\ell}^{\prime\prime}$, 
and~$R_{\ell}^{\prime\prime\prime}$, respectively. 
Then we set $R^{\ast }=R_{\ell }^{\prime \prime \prime }$, and we prove that 
$R^{\ast }$ is a projection representation of the graph~$G$ itself. 
Moreover, $R^{\ast }$ has the same unbounded vertices as $R$ except for $u$ (which became bounded in $R^{\ast}$), 
and thus $R^{\ast }$ has $k-1$ unbounded vertices. This completes the proof of Theorem~\ref{no-property-thm}.
\end{proof}

\medskip

Note that, within the proof of Theorem~\ref{no-property-thm} (see the Appendix), 
we mainly use the facts that $u$ is an unbounded-maximal vertex of $G$ and that the slope $\phi _{u}$
of $u$ is the smallest among all unbounded vertices in $R$. 
On the contrary, the assumption that $G$ satisfies Condition~\ref{ass3} 
is used only for a technical part of the proof, namely that $G_{0}\setminus \{u\}$ is a module in $G\setminus \{u\}$ 
(cf.~Lemma~\ref{module-2} in the Appendix).

\subsection{The general case}
\label{general-subsec}

Recall now that \textsc{Tolerance }$\cap $\textsc{\ Cocomparability} $=$ 
\textsc{Tolerance }$\cap $\textsc{\ Trapezoid} (cf.~the discussion before Lemma~\ref{two-components}). 
The next main theorem follows by recursive application of Theorem~\ref{no-property-thm}.

% \vspace{-0.1cm}
\begin{theorem}
\label{intersection-thm}Let $G=(V,E)\in$ \textsc{(Tolerance }$\cap $%
\textsc{\ Cocomparability)}, $R_{T}$ be a trapezoid representation of $%
G$, and $R$ be a projection representation of $G$. Then, assuming that $G$ satisfies 
one of the Conditions~\ref{ass1},~\ref{ass2}, or~\ref{ass3}, $G$ is a
bounded tolerance graph.
\end{theorem}

\begin{proof}
Since $G=(V,E)\in$ \textsc{(Tolerance }$\cap $\textsc{\ Cocomparability)}, 
it follows that $G$ is also a trapezoid graph~\cite{Fel98}. That is, 
$G\in$ \textsc{(Tolerance }$\cap $\textsc{\ Trapezoid)}. Suppose that 
$G$ is not a bounded tolerance graph. We can assume w.l.o.g.~that $G$ 
has the smallest number of vertices among the graphs in 
\textsc{(Tolerance }$\cap $\textsc{\ Trapezoid)\ }$\setminus $ 
\textsc{Bounded Tolerance}. 
Let $R_{0}$ be a projection representation of $G$ with the smallest possible
number $k_{0}$ of unbounded vertices. Note that $k_{0}\geq 1$; indeed, if
otherwise $k_{0}=0$, then $G$ is a bounded tolerance graph, which is a
contradiction to the assumption on $G$. Suppose that the projection
representation $R$ of $G$ has $k$ unbounded vertices, where $k\geq k_{0}$.
Then, there exists by Theorem~\ref{no-property-thm} a projection
representation~$R^{\ast }$ of $G$ with $k-1$ unbounded vertices. 
In particular, due to the proof of Theorem~\ref{no-property-thm}, 
$R^{\ast }$ has the same unbounded vertices as $R$, except for~$u$ (which became bounded in $R^{\ast}$). 

If Condition~\ref{ass1} holds for the projection representation $R$ of $G$,
i.e.~if $k=k_{0}=1$, then $R^{\ast }$ has no unbounded vertex, 
i.e.~$R^{\ast }$ is a parallelogram representation of $G$. This is
a contradiction to the assumption that $G$ is not a bounded tolerance 
(i.e.~parallelogram) graph. 
If Condition~\ref{ass2} holds for~$R$, then it also holds for $R^{\ast }$, 
since all unbounded vertices of $R^{\ast }$ are also unbounded vertices of~$R$. 
Similarly, if Condition~\ref{ass3} holds for $R$ and $R_{T}$, 
then it follows directly that it holds also for the pair~$R^{\ast}$ and $R_{T}$ 
of representations of $G$ (since for every unbounded vertex $u$ in $R^{\ast}$, 
the set $Q_{2}(u)$ depends only on the trapezoid representation $R_T$).

Therefore, we can apply iteratively $k-k_{0}+1$ times the constructive proof
of Theorem~\ref{no-property-thm}, obtaining eventually a projection
representation $R^{\ast \ast }$ of $G$ with $k_{0}-1$ unbounded vertices.
This is a contradiction to the minimality of $k_{0}$. Therefore, $G$ is a
bounded tolerance graph. This completes the proof of the theorem.
\end{proof}

\medskip

As an immediate implication of Theorem~\ref{intersection-thm}, 
we prove in the next corollary that Conjecture~\ref{Golumbic-Monma-conjecture} is true 
in particular for every graph $G$ that has no three independent vertices $a,b,c$ such 
that $N(a)\subset N(b)\subset N(c)$, since Condition~\ref{ass2} is guaranteed to be true for every such graph $G$. 
Therefore the conjecture is also true for the complements of triangle-free graphs. 
Thus, since in particular no bipartite graph has a triangle, 
the next corollary immediately implies the correctness of Conjecture~\ref{Golumbic-Monma-conjecture} 
for the complements of trees and of bipartite graphs, 
which were the only known results until now~\cite{Andreae93,Parra94}.

\begin{corollary}
\label{complement-triangle-free-cor}
Let $G=(V,E)\in$ \textsc{(Tolerance }$\cap $\textsc{\ Cocomparability)}. 
Suppose that there do not exist three independent vertices $a, b, c \in V$ 
such that ${N(a)\subset N(b)\subset N(c)}$. 
Then, $G$ is a bounded tolerance graph.
\end{corollary}

\begin{proof}
Due to Theorem~\ref{intersection-thm}, it suffices to prove that Condition~\ref{ass2} is true for~$G$, 
with respect to \emph{any possible} canonical (projection) representation $R$ 
and \emph{any} trapezoid representation $R_{T}$ of $G$. 
Let $R$ be a canonical representation of $G$. Suppose that Condition~\ref{ass2} is not true for~$G$. 
Then, there exists an unbounded vertex~${u\in V_{U}}$ such that~${Q_{u}\neq \emptyset}$. 
That is, there exists by the definition of the set $Q_{u}$ an unbounded 
vertex~${v\in V_{U}\setminus \{u\}}$ such that~${N(v)\subset N(u)}$. 
Note that~${v\notin N(u)}$, since no two unbounded vertices are adjacent in~$G$. 
Furthermore, there exists at least one covering vertex~$u^{\ast}$ of~$u$ in~$G$, 
since $V_{0}(u)\neq \emptyset$ (cf.~Lemma~\ref{bounded-hovering}), 
and thus~$u^{\ast}\notin N(u)$ and~$N(u)\subset N(u^{\ast})$. 
Therefore, since~$N(v)\subset N(u)$ and~$u^{\ast}\notin N(u)$, it follows that also~$u^{\ast}\notin N(v)$, 
i.e.~the vertices~$v,u,u^{\ast}$ are independent. 
Moreover~$N(v)\subset N(u)\subset N(u^{\ast})$, which comes in contradiction to the assumption of the lemma. 
Therefore Condition~\ref{ass2} holds for~$G$, and thus $G$ is a bounded tolerance graph by Theorem~\ref{intersection-thm}.
\end{proof}

\medskip

We now formally define the notion of a \emph{minimally unbounded tolerance graph}.

\begin{definition}
\label{minimally-unbounded-def}
Let $G\in$ \textsc{Tolerance }$\setminus $ \textsc{Bounded Tolerance}. 
If $G\setminus \{u\}$ is a bounded tolerance graph for every vertex of $G$, 
then $G$ is a \emph{minimally unbounded tolerance} graph.
\end{definition}

Assume now that Conjecture~\ref{Golumbic-Monma-conjecture} is not true, 
and let $G$ be a counterexample with the smallest number of vertices. 
Then, in particular, $G$ is a tolerance but not a bounded tolerance graph; 
furthermore, since $G$ has the smallest number of vertices, the removal of 
any vertex of $G$ makes it a bounded tolerance graph. 
That is, $G$ is a minimally unbounded tolerance graph by Definition~\ref{minimally-unbounded-def}. 
Now, if our Conjecture~\ref{minimal-conjecture} is true (see Section~\ref{sec:intro}), 
then $G$ has a projection representation~$R$ with exactly one unbounded vertex, 
i.e.~$R$ satisfies Condition~\ref{ass1}. 
Thus, $G$ is a bounded tolerance graph by Theorem~\ref{intersection-thm}, which is a contradiction, 
since $G$ has been assumed to be a counterexample to Conjecture~\ref{Golumbic-Monma-conjecture}.
Thus, we obtain the following theorem.

% \vspace{-0.1cm}
\begin{theorem}
\label{conjectures-reduction-thm}Conjecture~\ref{minimal-conjecture} implies Conjecture~\ref{Golumbic-Monma-conjecture}.
\end{theorem}

Therefore, in order to prove Conjecture~\ref{Golumbic-Monma-conjecture}, 
it suffices to prove Conjecture~\ref{minimal-conjecture}. 
Moreover, to the best of our knowledge, all known examples of minimally unbounded 
tolerance graphs have a tolerance representation with exactly one unbounded vertex; 
for such examples, see e.g.~\cite{GolTol04}.

% \vspace{-0.1cm}
\section{Concluding remarks and open problems\label{conclusion}}

In this article we dealt with the over 25 years old conjecture of~\cite{GolumbicMonma84}, 
which states that if a graph~$G$ is both tolerance and cocomparability, then it is also a bounded tolerance graph. 
Our main result was that this conjecture is true for every graph~$G$ that admits a tolerance representation 
with exactly one unbounded vertex. 
Our proofs are constructive, in the sense that, given a tolerance representation~$R$ of a graph~$G$, 
we transform~$R$ into a bounded tolerance representation~$R^{\ast}$ of~$G$. 
Furthermore, we conjectured that any \emph{minimal} graph~$G$ that is a tolerance but not a bounded tolerance graph, 
has a tolerance representation with exactly one unbounded vertex. 
Our results imply the non-trivial result that, in order to prove the conjecture of~\cite{GolumbicMonma84}, 
it suffices to prove~our conjecture. 
An interesting problem for further research that we leave open is to prove this new~conjecture 
(which, in contrast to one stated in~\cite{GolumbicMonma84}, does not concern 
any other class of graphs, such as cocomparability or trapezoid graphs). 
Since cocomparability graphs can be efficiently recognized~\cite{Spinrad03}, 
a positive answer to this conjecture (and thus also to the conjecture of~\cite{GolumbicMonma84}) 
would enable us to efficiently distinguish between tolerance and bounded tolerance graphs, 
although it is NP-complete to recognize each of these graph classes separately~\cite{MSZ-SICOMP-11}.

{\small 
\bibliographystyle{abbrv}
\bibliography{ref-intersection}
}

\newpage

\section*{Appendix: Proof of Theorem~\protect\ref{no-property-thm}}

\begin{proof}
First, we may assume w.l.o.g.~by the minimality of the number of vertices of 
$G$ that $G$ is connected. If $R$ is not a canonical representation of $G$,
then there exists a projection representation of $G$ with $k-1$ unbounded
vertices by Definition~\ref{def7}. Suppose for the sequel of the proof that $%
R$ is a canonical representation of $G$. If at least one unbounded vertex of 
$G$ in $R$ has the right or the left border property, there exists a
projection representation of $G$ with $k-1$ unbounded vertices by 
Theorem~\ref{right-property-thm}. Suppose in 
the sequel that every unbounded vertex of $G$ in $R$ has neither the right
nor the left border property in $R$. Let $u$ be the unbounded vertex in $R$,
such that $\phi _{u}=\min \{\phi _{x}\ |\ x\in V_{U}\}$ in $R$. The proof is
done constructively, by exploiting the fact that $G$ can be represented by
both the projection representation $R$ and the trapezoid representation $%
R_{T}$. Namely, we will construct a projection representation $R^{\ast }$ of
the same graph $G$, in which $u$ is replaced by a bounded vertex, while all
other $k-1$ unbounded vertices of $R$ remain also unbounded in $R^{\ast }$.

By Lemma~\ref{bounded-hovering}, there exists at least one bounded covering
vertex $u^{\ast }$ of $u$, such that $P_{u^{\ast }}$ intersects $P_{u}$ in $%
R $ and $\phi _{u^{\ast }}<\phi _{u}$. Therefore, $V_{0}(u)\neq \emptyset $,
and thus $V_{0}(u)$ is connected by Lemma~\ref{two-components}. Since $%
V_{0}(u)$ is connected, and since $u$ is not adjacent to any vertex of $%
V_{0}(u)$, it follows that either all trapezoids of $V_{0}(u)$ lie to the
left, or all to the right of $T_{u}$ in $R_{T}$. By possibly performing
vertical axis flipping of $R_{T}$, we may assume w.l.o.g.~that all
trapezoids of $V_{0}(u)$ lie to the left of $T_{u}$ in $R_{T}$, i.e.~$%
T_{x}\ll _{R_{T}}T_{u}$ for every $x\in V_{0}(u)$. Moreover, we may assume
w.l.o.g.~by Lemma~\ref{min-angle-2} that $u$ is an unbounded-maximal vertex
of $G$. Recall by Lemma~\ref{not-equal} that $N(v_{1})\neq N(v_{2})$ for any
two unbounded vertices $v_{1},v_{2}$. Denote now by $Q_{u}=\{v\in V_{U}\ |\
N(v)\subset N(u)\}$. Furthermore, since we assumed that Condition~\ref{ass3}
holds, $Q_{u}=Q_{1}(u)=\{v\in Q_{u}\ |\ T_{v}\ll _{R_{T}}T_{x}$ for every $%
x\in V_{0}(u)\}$.

\subsection*{The vertex sets $D_{1}$, $D_{2}$, $S_{2}$, 
and $\widetilde{X}_{1}$ and the vertex $x_{2}$}

Define the sets $D_{1}(u,R)=\{v\in V_{0}(u)\ |\ P_{v}\ll _{R}P_{u}\}$, $%
D_{2}(u,R)=\{v\in V_{0}(u)\ |\ P_{u}\ll _{R}P_{v}\}$, and $S_{2}(u,R)=\{v\in
V_{0}(u)\ |\ P_{v}\not\ll _{R}P_{u}\}$. Note that $V_{0}(u)=D_{1}(u,R)\cup
S_{2}(u,R)$ and that $D_{2}(u,R)\subseteq S_{2}(u,R)$. For simplicity
reasons, we will refer in the following to the sets $D_{1}(u,R)$, $%
D_{2}(u,R) $, and $S_{2}(u,R)$ just by $D_{1}$, $D_{2}$, and $S_{2}$,
respectively. Note that $Q_{u}\cap D_{1}=\emptyset $, $Q_{u}\cap
D_{2}=\emptyset $, and $Q_{u}\cap S_{2}=\emptyset $, since $%
D_{1},D_{2},S_{2}\subseteq V_{0}(u)$ and by Lemma~\ref{Qu-1}.

Since $u$ does not have the right border property in $R$, there exist by
Definition~\ref{right-left-property-def} vertices $w\in N(u)$ and $x\in
V_{0}(u)$, such that $P_{w}\ll _{R}P_{x}$. Therefore, in particular, $%
r(w)<_{R}l(x)$. Since $u$ is unbounded in $R$, and since $w\in N(u)$, Lemma~%
\ref{unbounded-bounded} implies that $r(u)<_{R}r(w)$, and thus $%
r(u)<_{R}l(x) $. For the sake of contradiction, suppose that $L(x)<_{R}R(u)$%
. Then, $P_{x}$ intersects $P_{u}$ in $R$ and $\phi _{x}>\phi _{u}$. Thus, $%
x $ is unbounded in $R$, since otherwise $x\in N(u)$, which is a
contradiction. Furthermore, $N(x)\subseteq N(u)$ by Lemma~\ref%
{intersecting-unbounded}, and thus $x\in Q_{u}$, which is a contradiction by
Lemma~\ref{Qu-1}, since $x\in V_{0}(u)$. Therefore, $R(u)<_{R}L(x)$, and
thus $P_{u}\ll _{R}P_{x}$, since also $r(u)<_{R}l(x)$. That is, $x\in D_{2}$%
. Since $u$ has not the left border property in $R$, there exist vertices $%
w^{\prime }\in N(u)$ and $y\in V_{0}(u)$, such that $P_{y}\ll
_{R}P_{w^{\prime }}$. Therefore, in the reverse projection representation $%
\widehat{R}$ of $R$, $P_{w^{\prime }}\ll _{\widehat{R}}P_{y}$. Then,
applying the same arguments as above, it follows that $P_{u}\ll _{\widehat{R}%
}P_{y}$, and thus $P_{y}\ll _{R}P_{u}$. That is, $y\in D_{1}$. Summarizing,
both sets $D_{1}$ and $D_{2}\subseteq S_{2}$ are not empty.

Among the vertices of $D_{1}\cup D_{2}$ let $x_{1}$ be such a vertex, that
for every other vertex $x^{\prime }\in D_{1}\cup D_{2}\setminus \{x_{1}\}$,
either $T_{x^{\prime }}$ intersects $T_{x_{1}}$ in the trapezoid
representation $R_{T}$, or $T_{x_{1}}\ll _{R_{T}}T_{x^{\prime }}$. That is,
there exists no vertex $x^{\prime }$ in $D_{1}\cup D_{2}$, whose trapezoid
lies to the left of $T_{x_{1}}$ in $R_{T}$. By possibly building the reverse
project representation $\widehat{R}$ of $R$, we may assume w.l.o.g.~that $%
P_{x_{1}}\ll _{R}P_{u}$, i.e.~$x_{1}\in D_{1}$.

As already mentioned above, since $u$ does not have the right border
property in $R$, there exist vertices $w\in N(u)$ and $x\in D_{2}\subseteq
V_{0}(u)$, such that $P_{w}\ll _{R}P_{x}$. Among the vertices $x\in D_{2}$,
for which $P_{w}\ll _{R}P_{x}$, let $x_{2}$ be such a vertex, that for every
other vertex $x^{\prime }\in D_{2}\setminus \{x_{2}\}$ with $P_{w}\ll
_{R}P_{x^{\prime }}$, either $T_{x^{\prime }}$ intersects $T_{x_{2}}$ in the
trapezoid representation $R_{T}$, or $T_{x_{2}}\ll _{R_{T}}T_{x^{\prime }}$.
That is, there exists no vertex $x^{\prime }$ in $D_{2}$ with $P_{w}\ll
_{R}P_{x^{\prime }}$, whose trapezoid $T_{x^{\prime }}$ lies to the left of $%
T_{x_{2}}$ in $R_{T}$.

Furthermore, $x_{1}x_{2}\notin E$, since $x_{1}\in D_{1}$ and $x_{2}\in
D_{2} $, i.e.~$P_{x_{1}}\ll _{R}P_{u}\ll _{R}P_{x_{2}}$. Therefore, since $%
T_{x}\ll _{R_{T}}T_{u}$ for every $x\in V_{0}(u)$, it follows by the
definition of $x_{1}$ that $T_{x_{1}}\ll _{R_{T}}T_{x_{2}}\ll _{R_{T}}T_{u}$%
. Thus, since $wu\in E$ and $wx_{2}\notin E$, it follows that also $%
T_{x_{1}}\ll _{R_{T}}T_{x_{2}}\ll _{R_{T}}T_{w}$, i.e.~$wx_{1}\notin E$.
That is, $x_{1}$, $x_{2}$, and $w$ are three independent vertices in $G$.

We now construct iteratively the vertex set $\widetilde{X}_{1}\subseteq
D_{1} $ from the vertex $x_{1}$, as follows. Initially, we set $\widetilde{X}%
_{1}=\{x_{1}\}$. If $N(w)\cap N(\widetilde{X}_{1})\subset N(\widetilde{X}%
_{1})$, then set $\widetilde{X}_{1}$ to be equal to $\widetilde{X}_{1}\cup N(%
\widetilde{X}_{1})\setminus N(w)$. Iterate, until finally $N(w)\cap N(%
\widetilde{X}_{1})=N(\widetilde{X}_{1})$. This process terminates, since
every time we strictly augment the current set $\widetilde{X}_{1}$.
Furthermore, at the end of this procedure, $N(\widetilde{X}_{1})\neq
\emptyset $, since otherwise $G$ is not connected, which is a contradiction.
Moreover, the vertices of $\widetilde{X}_{1}$ at every step of this
procedure induce a connected subgraph of $G$.

\begin{lemma}
\label{X1}For the constructed set $\widetilde{X}_{1}$, $\widetilde{X}%
_{1}\subseteq D_{1}$. Furthermore, $P_{x}\ll _{R}P_{w}$ and $T_{x}\ll
_{R_{T}}T_{x_{2}}$ for every $x\in \widetilde{X}_{1}$.
\end{lemma}

\begin{proof}
The proof of the lemma is done by induction on $|\widetilde{X}_{1}|$.
Suppose first that $|\widetilde{X}_{1}|=1$, i.e.~$\widetilde{X}%
_{1}=\{x_{1}\} $. Then, $\{x_{1}\}\subseteq D_{1}$ and $T_{x_{1}}\ll
_{R_{T}}T_{x_{2}}$ by definition of $x_{1}$. We will now prove that also $%
P_{x_{1}}\ll _{R}P_{w}$. Otherwise, suppose first that $P_{w}\ll
_{R}P_{x_{1}}$. Then, since $x_{1}\in D_{1}$, it follows that $P_{w}\ll
_{R}P_{x_{1}}\ll _{R}P_{u}$, and thus $w\notin N(u)$, which is a
contradiction. Thus, either $P_{x_{1}}$ intersects $P_{w}$ in $R$, or $%
P_{x_{1}}\ll _{R}P_{w}$. Suppose that $P_{x_{1}}$ intersects $P_{w}$ in $R$.
Then, $x_{1}$ is unbounded and $\phi _{x_{1}}>\phi _{w}>\phi _{u}$, since $w$
is bounded and $x_{1}w\notin E$. Then, Lemma~\ref{intersecting-unbounded}
implies that $N(x_{1})\subseteq N(w) $. Furthermore, since $T_{x_{1}}\ll
_{R_{T}}T_{x_{2}}\ll _{R_{T}}T_{w}$, it follows that $T_{z}$ intersects $%
T_{x_{2}}$ in $R_{T}$ for every $z\in N(x_{1})\subseteq N(w)$, and thus also 
$N(x_{1})\subseteq N(x_{2})$. Therefore, since $P_{x_{1}}\ll _{R}P_{u}\ll
_{R}P_{x_{2}}$, it follows that for every $z\in N(x_{1})\subseteq N(x_{2})$, 
$z$ is bounded in $R$, $\phi _{u}<\phi _{x_{1}}<\phi _{z}$, and $P_{z}$
intersects $P_{u}$ in $R$. Thus, $N(x_{1})\subseteq N(u)$, i.e.~$x_{1}\in
Q_{u}$, which is a contradiction by Lemma~\ref{Qu-1}, since $x_{1}\in
V_{0}(u)$. It follows that $P_{x_{1}}$ does not intersect $P_{w}$ in $R$,
and thus $P_{x_{1}}\ll _{R}P_{w}$. This proves the induction basis.

For the induction step, suppose that the statement of the lemma holds for
the set $\widetilde{X}_{1}$ constructed after an iteration of the
construction procedure, and let $v\in N(\widetilde{X}_{1})\setminus N(w)$.
Suppose first that $v\in N(u)$, and thus $v$ is bounded in $R$. Then, since
by the induction hypothesis $T_{x}\ll _{R_{T}}T_{x_{2}}\ll _{R_{T}}T_{u}$
for every $x\in \widetilde{X}_{1}$, and since $v\in N(x)\cap N(u)$ for some $%
x\in \widetilde{X}_{1}$, it follows that $T_{v}$ intersects $T_{x_{2}}$ in $%
R_{T}$, and thus $vx_{2}\in E$. On the other hand, since $P_{x}\ll
_{R}P_{w}\ll _{R}P_{x_{2}}$ for every $x\in \widetilde{X}_{1}$ by the
induction hypothesis, and since $v\in N(x)\cap N(x_{2})$ for some $x\in 
\widetilde{X}_{1}$, it follows that $P_{v}$ intersects $P_{w}$ in $R$, and
thus $vw\in E$, since both $v$ and $w$ are bounded. This is a contradiction,
since $v\in N(\widetilde{X}_{1})\setminus N(w)$. Thus, $v\notin N(u)$ for
every $v\in N(\widetilde{X}_{1})\setminus N(w)$. Therefore, since $v\in N(%
\widetilde{X}_{1})$ and $\widetilde{X}_{1}\subseteq V_{0}(u)$, it follows
that $v\in V_{0}(u)$ for every $v\in N(\widetilde{X}_{1})\setminus N(w)$,
and thus the updated set $\widetilde{X}_{1}$ is $\widetilde{X}_{1}\cup N(%
\widetilde{X}_{1})\setminus N(w)\subseteq V_{0}(u)$.

Since $v\in N(x)$ for some $x\in \widetilde{X}_{1}$, and since $P_{x}\ll
_{R}P_{w}$ for every $x\in \widetilde{X}_{1}$ by the induction hypothesis,
it follows that either $P_{v}$ intersects $P_{w}$ in $R$, or $P_{v}\ll
_{R}P_{w}$. Suppose that $P_{v}$ intersects $P_{w}$ in $R$. Then, $v$ is
unbounded and $\phi _{v}>\phi _{w}$, since $v\notin N(w)$ and $w$ is
bounded. Therefore, $N(v)\subseteq N(w)$ by Lemma~\ref%
{intersecting-unbounded}, and thus in particular $x\in N(w)$ for some $x\in 
\widetilde{X}_{1}$, which is a contradiction to the induction hypothesis.
Therefore, $P_{v}$ does not intersect $P_{w}$ in $R$, and thus $P_{v}\ll
_{R}P_{w}$ for every $v\in N(\widetilde{X}_{1})\setminus N(w)$.

We will prove that also $P_{v}\ll _{R}P_{u}$ for every $v\in N(\widetilde{X}%
_{1})\setminus N(w)$. Otherwise, suppose first that $P_{u}\ll _{R}P_{v}$.
Then, since $P_{v}\ll _{R}P_{w}$ by the previous paragraph, it follows that $%
P_{u}\ll _{R}P_{v}\ll _{R}P_{w}$, and thus $w\notin N(u)$, which is a
contradiction. Suppose now that $P_{v}$ intersects $P_{u}$ in $R$. Recall
that $v\notin N(u)$, as we proved above. If $\phi _{u}>\phi _{v}$, then $%
N(u)\subseteq N(v)$ by Lemma~\ref{intersecting-unbounded}, and thus also $%
w\in N(v)$, which is a contradiction, since $v\in N(\widetilde{X}%
_{1})\setminus N(w)$. If $\phi _{u}<\phi _{v}$, then $v$ is unbounded, since
otherwise $v\in N(u)$, which is a contradiction. Furthermore, $N(v)\subseteq
N(u)$ by Lemma~\ref{intersecting-unbounded}, and thus $v\in Q_{u}$, which is
a contradiction by Lemma~\ref{Qu-1}, since $v\in V_{0}(u)$ as we proved
above. Therefore, $P_{v}\ll _{R}P_{u}$, i.e.~$v\in D_{1}$, for every $v\in N(%
\widetilde{X}_{1})\setminus N(w)$, and thus the updated set $\widetilde{X}%
_{1}$ is $\widetilde{X}_{1}\cup N(\widetilde{X}_{1})\setminus N(w)\subseteq
D_{1}$.

Since the updated set $\widetilde{X}_{1}\cup N(\widetilde{X}_{1})\setminus
N(w)$ is a subset of $D_{1}$, i.e.~$x\in V_{0}(u)$ and~$P_{x}\ll _{R}P_{u}$
for every $x\in \widetilde{X}_{1}\cup N(\widetilde{X}_{1})\setminus N(w)$,
it follows in particular that $xx_{2}\notin E$ for every $x\in \widetilde{X}%
_{1}\cup N(\widetilde{X}_{1})\setminus N(w)$, since $P_{u}\ll _{R}P_{x_{2}}$%
. Recall furthermore that the set $\widetilde{X}_{1}\cup N(\widetilde{X}%
_{1})\setminus N(w)$ induces a connected subgraph of $G$. Thus, since $%
T_{x_{1}}\ll _{R_{T}}T_{x_{2}}$, it follows that $T_{x}\ll _{R_{T}}T_{x_{2}}$
for every $x\in \widetilde{X}_{1}\cup N(\widetilde{X}_{1})\setminus N(w)$.
This completes the induction step, and the lemma follows.
\end{proof}

\begin{corollary}
\label{N-X1-tilde-N-u}For the constructed set $\widetilde{X}_{1}$, $N(%
\widetilde{X}_{1})\setminus N(u)\neq \emptyset $.
\end{corollary}

\begin{proof}
Suppose for the sake of contradiction that $N(\widetilde{X}_{1})\setminus
N(u)=\emptyset $, i.e.~$N(\widetilde{X}_{1})\subseteq N(u)$. Since $%
\widetilde{X}_{1}\subseteq D_{1}\subseteq V_{0}(u)$ by Lemma~\ref{X1}, it
follows that $P_{x}\ll _{R}P_{u}$ for every $x\in \widetilde{X}_{1}$, and
thus in particular $x\notin N(u)$ for every $x\in \widetilde{X}_{1}$.
Therefore, since $\widetilde{X}_{1}$ induces a connected subgraph of $G$, it
follows that $\widetilde{X}_{1}$ is a connected component of $G\setminus
N[u] $. Therefore, since $V_{0}(u)$ is connected, it follows that $V_{0}(u)=%
\widetilde{X}_{1}$. This is a contradiction, since $\emptyset \neq
D_{2}\subseteq V_{0}(u)$. Therefore, $N(\widetilde{X}_{1})\setminus N(u)\neq
\emptyset $.
\end{proof}

\medskip

Recall by definition of $x_{2}$ that for every vertex $x^{\prime }\in
D_{2}\setminus \{x_{2}\}$ with $P_{w}\ll _{R}P_{x^{\prime }}$, either $%
T_{x^{\prime }}$ intersects $T_{x_{2}}$ in the trapezoid representation $%
R_{T}$, or $T_{x_{2}}\ll _{R_{T}}T_{x^{\prime }}$. We will now prove in the
following lemma that this property holds actually for all vertices $%
x^{\prime }\in S_{2}\setminus \{x_{2}\}$.

\begin{lemma}
\label{x2-relative-position-in-S2}For every vertex $x^{\prime }\in
S_{2}\setminus \{x_{2}\}$, either $T_{x^{\prime }}$ intersects $T_{x_{2}}$
in the trapezoid representation $R_{T}$, or $T_{x_{2}}\ll
_{R_{T}}T_{x^{\prime }}$.
\end{lemma}

\begin{proof}
Consider an arbitrary vertex $x^{\prime }\in S_{2}\setminus \{x_{2}\}$. If $%
x^{\prime }\in N(x_{2})$, then clearly $T_{x^{\prime }}$ intersects $%
T_{x_{2}}$ in $R_{T}$. Thus, it suffices to consider in the sequel of the
proof only the case where $x^{\prime }\notin N(x_{2})$, i.e.~the case where $%
T_{x^{\prime }}$ does not intersect $T_{x_{2}}$ in $R_{T}$. Suppose for the
sake of contradiction that $T_{x^{\prime }}\ll _{R_{T}}T_{x_{2}}$, i.e.~$%
T_{x^{\prime }}\ll _{R_{T}}T_{x_{2}}\ll _{R_{T}}T_{w}$. Then, in particular, 
$x^{\prime }\notin N(w)$. Furthermore, note that $x^{\prime }\notin N(u)$,
since $x^{\prime }\in S_{2}\subseteq V_{0}(u)$.

Suppose first that $x^{\prime }\in S_{2}\setminus D_{2}$, i.e.~$P_{x^{\prime
}}$ intersects $P_{u}$ in $R$. If $\phi _{x^{\prime }}>\phi _{u}$, then $%
x^{\prime }$ is unbounded, since otherwise $x^{\prime }\in N(u)$ which is a
a contradiction. Furthermore, $N(x^{\prime })\subseteq N(u)$ by Lemma~\ref%
{intersecting-unbounded}, and thus $x^{\prime }\in Q_{u}$, which is a
contradiction by Lemma~\ref{Qu-1}, since $x\in V_{0}(u)$. If $\phi
_{x^{\prime }}<\phi _{u}$, then $N(u)\subseteq N(x^{\prime })$ by Lemma~\ref%
{intersecting-unbounded}, and thus in particular $wx^{\prime }\in E$, which
is a contradiction, since $x^{\prime }\notin N(w)$. Therefore, the lemma
holds for every vertex $x^{\prime }\in S_{2}\setminus D_{2}$.

Suppose now that $x^{\prime }\in D_{2}$, i.e.~$P_{u}\ll _{R}P_{x^{\prime }}$%
. If $P_{w}\ll _{R}P_{x^{\prime }}$, then the lemma follows by definition of 
$x_{2}$. If $P_{x^{\prime }}\ll _{R}P_{w}$, then $P_{u}\ll _{R}P_{x^{\prime
}}\ll _{R}P_{w}$, and thus $w\notin N(u)$, which is a contradiction. Suppose
that $P_{x^{\prime }}$ intersects $P_{w}$ in $R$. Then, $x^{\prime }$ is
unbounded and $\phi _{x^{\prime }}>\phi _{w}>\phi _{u}$, since $w$ is
bounded and $x^{\prime }\notin N(w)$. Note that $P_{x}\ll _{R}P_{u}\ll
_{R}P_{x^{\prime }}$ for every $x\in \widetilde{X}_{1}$, since $x^{\prime
}\in D_{2}$ and $\widetilde{X}_{1}\subseteq D_{1}$ by Lemma~\ref{X1}.
Therefore, $x^{\prime }\notin N(x)$ for every $x\in \widetilde{X}_{1}$, and
thus in particular $x^{\prime }\notin N(x_{1})$, since $x_{1}\in \widetilde{X%
}_{1}$. Therefore, $T_{x^{\prime }}$ does not intersect $T_{x_{1}}$ in $%
R_{T} $, and thus $T_{x_{1}}\ll _{R_{T}}T_{x^{\prime }}$ by definition of $%
x_{1}$. Furthermore, since $\widetilde{X}_{1}$ induces a connected subgraph
of $G$, and since $x^{\prime }\notin N(x)$ for every $x\in \widetilde{X}_{1}$%
, it follows that $T_{x}\ll _{R_{T}}T_{x^{\prime }}$ for every $x\in 
\widetilde{X}_{1}$. Recall now that $T_{x_{2}}\ll _{R_{T}}T_{w}$ and that we
assumed that $T_{x^{\prime }}\ll _{R_{T}}T_{x_{2}}$. That is, $T_{x}\ll
_{R_{T}}T_{x^{\prime }}\ll _{R_{T}}T_{x_{2}}\ll _{R_{T}}T_{w}$ for every $%
x\in \widetilde{X}_{1}$.

Recall that $N(\widetilde{X}_{1})\subseteq N(w)$ by the construction of the
set $\widetilde{X}_{1}$. Therefore, since $T_{x}\ll _{R_{T}}T_{x^{\prime
}}\ll _{R_{T}}T_{w}$ for every $x\in \widetilde{X}_{1}$, it follows that $%
T_{z}$ intersects $T_{x^{\prime }}$ in $R_{T}$ for every $z\in N(\widetilde{X%
}_{1})\subseteq N(w)$, and thus $N(\widetilde{X}_{1})\subseteq N(x^{\prime
}) $. On the other hand, since $P_{x}\ll _{R}P_{u}\ll _{R}P_{x^{\prime }}$
for every $x\in \widetilde{X}_{1}$ in the projection representation $R$, it
follows that $P_{z}$ intersects $P_{u}$ in $R$ for every $z\in N(\widetilde{X%
}_{1})\subseteq N(x^{\prime })$. Furthermore, since $x^{\prime }$ is
unbounded and $\phi _{x^{\prime }}>\phi _{u}$ in $R$, it follows that $z$ is
bounded in $R$ and $\phi _{z}>\phi _{x^{\prime }}>\phi _{u}$ for every $z\in
N(\widetilde{X}_{1})\subseteq N(x^{\prime })$. Therefore, $z\in N(u)$ for
every $z\in N(\widetilde{X}_{1})$, i.e.~$N(\widetilde{X}_{1})\subseteq N(u)$%
, which is a contradiction by Corollary~\ref{N-X1-tilde-N-u}. This completes
the proof of the lemma.
\end{proof}

\subsection*{The vertex sets $C_{u}$, $C_{2}$, $X_{1}$, and $H$}

Let $C_{u}$ be the connected component of $G\setminus Q_{u}\setminus N[%
\widetilde{X}_{1},x_{2}]$, in which $u$ belongs. Note that, in particular, $%
w $ belongs to $C_{u}$, since $wu\in E$, $w\notin Q_{u}$, and $%
wx,wx_{2}\notin E$ for every~$x\in \widetilde{X}_{1}$, and thus $%
C_{u}\setminus \{u\}\neq \emptyset $. Recall that the trapezoids of all
vertices of $V_{0}(u)$ lie to the left of the trapezoid of $u$ in the
trapezoid representation $R_{T}$; $S_{2}$ is exactly the subset of vertices
of~$V_{0}(u)$, whose parallelograms do not lie to the left of the
parallelogram $P_{u}$ of $u$ in $R$. Let~$\widetilde{\widetilde{C}}_{2}$ be
the set of connected components of $G\setminus Q_{u}\setminus N[\widetilde{X}%
_{1}]$, in which the vertices of $S_{2}$ belong. Since~$x_{2}\in S_{2}$,
note that $V(C_{u}\cup \widetilde{\widetilde{C}}_{2})$ induces the set of
connected components of $G\setminus Q_{u}\setminus N[\widetilde{X}_{1}]$, in
which the vertices of~$S_{2}\cup \{u\}$ belong. Furthermore, let $\widetilde{%
C}_{2}=\widetilde{\widetilde{C}}_{2}\setminus N[u,w]\setminus C_{u}$.
Finally, let~$\widetilde{H}$ be the induced subgraph of~$G\setminus
Q_{u}\setminus N[\widetilde{X}_{1}]$ on the vertices of $N[u,w]\cap N(x_{2})$%
. Note now that~$V(C_{u}\cup \widetilde{\widetilde{C}}_{2})=V(C_{u}\cup \widetilde{C}%
_{2}\cup \widetilde{H})$, i.e.~$V(C_{u}\cup \widetilde{C}_{2}\cup \widetilde{%
H})$ also induces the set of connected components of~$G\setminus
Q_{u}\setminus N[\widetilde{X}_{1}]$, in which the vertices of $S_{2}\cup
\{u\}$ belong.

Let $v$ be a vertex of the set $\widetilde{C}_{2}$, and thus $v\notin N(u)$
by the definition of $\widetilde{C}_{2}$. Suppose that~$P_{v}$ intersects $%
P_{u}$ in $R$. If $\phi _{v}>\phi _{u}$, then $v$ is unbounded, since
otherwise $v\in N(u)$, which is a contradiction. Furthermore, $N(v)\subseteq
N(u)$ by Lemma~\ref{intersecting-unbounded}, and thus $v\in Q_{u}$, which is
a contradiction to the definition of $\widetilde{C}_{2}$. If $\phi _{v}<\phi
_{u}$, then $N(u)\subseteq N(v)$ by Lemma~\ref{intersecting-unbounded}, and
thus $w\in N(v)$, which is again a contradiction to the definition of $%
\widetilde{C}_{2}$. Therefore, there is no vertex $v$ of $\widetilde{C}_{2}$%
, such that $P_{v}$ intersects $P_{u}$ in $R$. That is, for every $v\in 
\widetilde{C}_{2}$ either $P_{v}\ll _{R}P_{u}$ or $P_{u}\ll _{R}P_{v}$. Let
now $A_{1},A_{2},\ldots ,A_{k},A_{k+1},\ldots ,A_{\ell }$ be the connected
components of $\widetilde{C}_{2}$, such that $P_{v}\ll _{R}P_{u}$ for every $%
v\in A_{i}$, $i=1,2,\ldots ,k$, and $P_{u}\ll _{R}P_{v}$ for every $v\in
A_{j}$, $j=k+1,k+2,\ldots ,\ell $.

We partition first the set $\{A_{k+1},\ldots ,A_{\ell }\}$ of components
into two possibly empty subsets, namely $\mathcal{B}_{1}$ and $\mathcal{B}%
_{2}$, as follows. A component $A_{j}\in \mathcal{B}_{2}$, $j=k+1,k+2,\ldots
,\ell $, if $A_{j}\cap S_{2}\neq \emptyset $; otherwise, $A_{j}\in \mathcal{B%
}_{1}$. Then, since any component $A_{j}\in \mathcal{B}_{2}$ is a connected
subgraph of $G\setminus N[u]$, and since $A_{j}$ has at least one vertex of $%
S_{2}\subseteq V_{0}(u)$, it follows that $v\in V_{0}(u)$ for every $v\in
A_{j}$, where $A_{j}\in \mathcal{B}_{2}$. Furthermore, $v\in D_{2}$ for
every $v\in A_{j}\in \mathcal{B}_{2}$, since $P_{u}\ll _{R}P_{v}$ for every $%
v\in A_{j}$. Thus, $A_{j}\subseteq D_{2}$ for every component $A_{j}\in 
\mathcal{B}_{2}$, while $A_{j}\cap D_{2}=\emptyset $ for every component $%
A_{j}\in \mathcal{B}_{1}$. That is, in particular the next observation
follows.

\begin{observation}
\label{V(B1)}$V(\mathcal{B}_{1})\subseteq V\setminus Q_{u}\setminus
N[u]\setminus V_{0}(u)$, where $V(\mathcal{B}_{1})=\bigcup\nolimits_{A_{j}%
\in \mathcal{B}_{1}}A_{j}$.
\end{observation}

We partition now the set $\{A_{1},A_{2},\ldots ,A_{k}\}$ of components into
two possibly empty subsets, namely $\mathcal{A}_{1}$ and $\mathcal{A}_{2}$,
as follows. A component $A_{i}\in \mathcal{A}_{2}$, $i=1,2,\ldots ,k$, if $%
\widetilde{H}\subseteq N(x)$ for all vertices $x\in A_{i}$; otherwise, $%
A_{i}\in \mathcal{A}_{1}$. That is, $\mathcal{A}_{2}$ includes exactly those
components $A_{i}$, $i=1,2,\ldots ,k$, for which all vertices of $A_{i}$ are
adjacent to all vertices of $\widetilde{H}$.

We now extend the vertex set $\widetilde{X}_{1}$ to the set $X_{1}=%
\widetilde{X}_{1}\cup V(\mathcal{A}_{1})$, where $V(\mathcal{A}%
_{1})=\bigcup\nolimits_{A_{i}\in \mathcal{A}_{1}}A_{i}$, and define $C_{2}=%
\mathcal{A}_{2}\cup \mathcal{B}_{2}$. Furthermore, similarly to the
definition of $\widetilde{H}$, let $H$ be the induced subgraph of $%
G\setminus Q_{u}\setminus N[X_{1}]$ on the vertices of $N[u,w]\cap N(x_{2})$%
. Note that ${H\subseteq \widetilde{H}}$, since~${\widetilde{X}_{1}\subseteq
X_{1}}$, and thus for every component ${A_{i}\in \mathcal{A}_{2}}$, all
vertices of $A_{i}$ are also adjacent to all vertices of~$H$. Furthermore,
since $X_{1}=\widetilde{X}_{1}\cup V(\mathcal{A}_{1})$, and since no vertex
of $\mathcal{A}_{1}$ is adjacent to any vertex of $\widetilde{X}_{1}$, note
that $N(X_{1})=N(\widetilde{X}_{1})\cup N(V(\mathcal{A}_{1}))$ and that $%
N[X_{1}]=N[\widetilde{X}_{1}]\cup N[V(\mathcal{A}_{1})]$, i.e.~in particular 
${N(\widetilde{X}_{1})\subseteq N(X_{1})}$. Moreover, ${N(X_{1})\neq
\emptyset }$, since~${N(\widetilde{X}_{1})\neq \emptyset }$.

Recall that $V(C_{u}\cup \widetilde{C}_{2}\cup \widetilde{H})$ induces the
set of connected components of~$G\setminus Q_{u}\setminus N[\widetilde{X}%
_{1}]$, in which the vertices of $S_{2}\cup \{u\}$ belong. The next lemma
follows by the definitions of ${C_{u}}$, ${C_{2}}$, and~$H$.

\begin{lemma}
\label{module-1}$V({C_{u}\cup C_{2}\cup H)}$ induces a subgraph of ${%
G\setminus Q_{u}\setminus N[X_{1}]\setminus \mathcal{B}_{1}}$ that includes
all connected components of ${G\setminus Q_{u}\setminus N[X_{1}]\setminus 
\mathcal{B}_{1}}$, in which the vertices of~${S_{2}\cup \{u\}}$ belong.
Furthermore, $N(V({C_{u}\cup C_{2}\cup H}))\subseteq Q_{u}\cup N(X_{1})\cup
V({\mathcal{B}_{1}})$.
\end{lemma}

\begin{proof}
Consider a vertex $v\in N(V(\mathcal{A}_{1}))$. That is, $v\in N(v^{\prime
}) $ and $v\notin V(\mathcal{A}_{1})$, for some vertex $v^{\prime }\in V(%
\mathcal{A}_{1})$, i.e.~$v^{\prime }\in A_{i}$ for some $A_{i}\in \mathcal{A}%
_{1}$. First note that $v^{\prime }\notin N(x_{2})$, since $P_{v^{\prime
}}\ll _{R}P_{u}\ll _{R}P_{x_{2}}$ for every $v^{\prime }\in A_{i}$ by
definition of $\mathcal{A}_{1}$. If $v\in Q_{u}$, then $N(v)\subset N(u)$ by
definition of $Q_{u}$, and thus $v^{\prime }\in N(u)$, which is a
contradiction due to the definition of $\widetilde{C}_{2}$, and since $%
v^{\prime }\in V(\mathcal{A}_{1})\subseteq \widetilde{C}_{2}$. Therefore $%
v\notin Q_{u}$. We will now prove that $v\in N(\widetilde{X}_{1})$ or $v\in 
\widetilde{H}$. To this end, suppose that $v\notin N(\widetilde{X}_{1})$. If 
$v\in \widetilde{C}_{2}$, then $v$ is a vertex of the connected component $%
A_{i}$ of $\widetilde{C}_{2}$, since $v\in N(v^{\prime })$ and $v^{\prime
}\in A_{i}$. This is a contradiction, since $v\notin V(\mathcal{A}_{1})$;
thus $v\notin \widetilde{C}_{2}$. That is, $v^{\prime }\in \widetilde{C}%
_{2}\subseteq \widetilde{\widetilde{C}}_{2}$ and $v\notin \widetilde{C}_{2}$%
. Therefore, since $v\in N(v^{\prime })$ and $v\notin Q_{u}\cup N(\widetilde{%
X}_{1})$, it follows by definitions of $\widetilde{\widetilde{C}}_{2}$ and $%
\widetilde{C}_{2}$ that $v\in C_{u}$ or $v\in N[u,w]$. Let $v\in C_{u}$.
Then, since $v^{\prime }\in N(v)$ and $v^{\prime }\notin N(x_{2})$, it
follows that also $v^{\prime }\in C_{u}$, which is a contradiction by
definition of $\widetilde{C}_{2}$. Let $v\in N[u,w]$. If $v\notin N(x_{2})$,
then $v\in C_{u}$ and $v^{\prime }\in C_{u}$, which is again a
contradiction. If $v\in N(x_{2})$, then $v\in \widetilde{H}$ by definition
of~$\widetilde{H}$. Summarizing, if $v\notin N(\widetilde{X}_{1})$, then $%
v\in \widetilde{H}$. That is, for an arbitrary vertex $v\in N(V(\mathcal{A}%
_{1}))$, either $v\in N(\widetilde{X}_{1})$ or $v\in \widetilde{H}$, i.e.~$%
N(V(\mathcal{A}_{1}))\subseteq N(\widetilde{X}_{1})\cup \widetilde{H}$.

Note by definition of $C_{u}$ and of $\widetilde{C}_{2}$ that $V(C_{u})\cap
V(\widetilde{H})=\emptyset $ and that $V(\widetilde{C}_{2})\cap V(\widetilde{%
H})=\emptyset $. Therefore, it follows by the previous paragraph that $%
V(C_{u})\cap N(V(\mathcal{A}_{1}))\subseteq {V(C_{u})\cap (N(\widetilde{X}%
_{1})\cup \widetilde{H})=\emptyset }$ and that $V(\widetilde{C}_{2})\cap N(V(%
\mathcal{A}_{1}))\subseteq V(\widetilde{C}_{2})\cap (N(\widetilde{X}%
_{1})\cup \widetilde{H})=\emptyset $. Thus, 
\begin{eqnarray}
V(C_{u})\setminus N(V(\mathcal{A}_{1})) &=&V(C_{u})  \label{Cu-eq} \\
V(\widetilde{C}_{2})\setminus N(V(\mathcal{A}_{1})) &=&V(\widetilde{C}_{2})
\label{C2-eq}
\end{eqnarray}%
Recall now that $N(X_{1})=N(\widetilde{X}_{1})\cup N(V(\mathcal{A}_{1}))$.
Therefore, it follows by definition of $H$ that%
\begin{eqnarray}
V(\widetilde{H}) &=&V(\widetilde{H}\setminus N(V(\mathcal{A}_{1})))\cup V(%
\widetilde{H}\cap N(V(\mathcal{A}_{1})))  \label{H-tilde-eq} \\
&=&V(H)\cup V(\widetilde{H}\cap N(V(\mathcal{A}_{1})))  \notag
\end{eqnarray}%
Furthermore, recall that $V(\widetilde{C}_{2})=V(C_{2})\cup V({\mathcal{A}%
_{1}})\cup V({\mathcal{B}_{1}})$ by definition of $C_{2}$, and thus it
follows by (\ref{H-tilde-eq}) that 
\begin{eqnarray}
V(C_{u}\cup \widetilde{C}_{2}\cup \widetilde{H}) &=&V(C_{u})\cup
V(C_{2})\cup V({\mathcal{A}_{1}})\cup V({\mathcal{B}_{1}})
\label{CuC2H-tilde-eq} \\
&&\cup V(H)\cup V(\widetilde{H}\cap N(V(\mathcal{A}_{1})))  \notag
\end{eqnarray}%
Therefore, it follows by (\ref{Cu-eq}), (\ref{C2-eq}), and (\ref%
{CuC2H-tilde-eq}) that 
\begin{equation}
V(C_{u}\cup \widetilde{C}_{2}\cup \widetilde{H})\setminus N[V(\mathcal{A}%
_{1})]\setminus V({\mathcal{B}_{1}})=V(C_{u})\cup V(C_{2})\cup V(H)
\label{CuC2H-eq1}
\end{equation}%
Thus, since $N[X_{1}]=N[\widetilde{X}_{1}]\cup N[V(\mathcal{A}_{1})]$, it
follows that also%
\begin{equation}
V(C_{u}\cup \widetilde{C}_{2}\cup \widetilde{H})\setminus N[X_{1}]\setminus
V({\mathcal{B}_{1}})=V(C_{u}\cup C_{2}\cup H)  \label{CuC2H-eq2}
\end{equation}%
Therefore, since $V(C_{u}\cup \widetilde{C}_{2}\cup \widetilde{H})$ induces
the set of connected components of~$G\setminus Q_{u}\setminus N[\widetilde{X}%
_{1}]$, in which the vertices of $S_{2}\cup \{u\}$ belong, it follows in
particular by (\ref{CuC2H-eq2}) that $V(C_{u}\cup C_{2}\cup H)$ induces a
subgraph of ${G\setminus Q_{u}\setminus N[X_{1}]\setminus \mathcal{B}_{1}}$;
moreover, this subgraph includes all connected components of ${G\setminus
Q_{u}\setminus N[X_{1}]\setminus \mathcal{B}_{1}}$, in which the vertices of~%
${S_{2}\cup \{u\}}$ belong. On the other hand, since $V(C_{u}\cup \widetilde{%
C}_{2}\cup \widetilde{H})$ induces a set of connected components of~$%
G\setminus Q_{u}\setminus N[\widetilde{X}_{1}]$, it follows that $%
N(V(C_{u}\cup \widetilde{C}_{2}\cup \widetilde{H}))\subseteq Q_{u}\cup N(%
\widetilde{X}_{1})$. Therefore, it follows by (\ref{CuC2H-eq2}) that $%
N(V(C_{u}\cup C_{2}\cup H))\subseteq Q_{u}\cup N(X_{1})\cup V({\mathcal{B}%
_{1}})$. This completes the proof of the lemma.
\end{proof}

\medskip

For the sequel of the proof, denote for simplicity ${N_{1}(v)=N(v)\cap
N(X_{1})}$ for every vertex~${v\in V\setminus X_{1}}$. Moreover, $C_{u}$ is
also the connected component of ${G\setminus Q_{u}\setminus N[X_{1},x_{2}]}$
(and not only of ${G\setminus Q_{u}\setminus N[\widetilde{X}_{1},x_{2}]}$),
in which $u$ belongs, as we prove in the next lemma. The next two lemmas
extend Lemma~\ref{X1}.

\begin{lemma}
\label{N(w)-1}For the constructed sets $X_{1}$ and $C_{2}$, $%
N_{1}(w)=N(X_{1})$, $X_{1}\subseteq D_{1}$, and $C_{2}\subseteq V_{0}(u)$.
Furthermore, $C_{u}$ is the connected component of $G\setminus
Q_{u}\setminus N[X_{1},x_{2}]$, in which $u$ belongs.
\end{lemma}

\begin{proof}
Recall first that $N(\widetilde{X}_{1})\subseteq N(w)$ by the construction
of the set $\widetilde{X}_{1}$. Consider an arbitrary component $A_{i}\in 
\mathcal{A}_{1}\cup \mathcal{A}_{2}=\{A_{1},A_{2},\ldots ,A_{k}\}$. Recall
that $v\notin N(x_{2})$ for every $v\in A_{i}$, since $P_{v}\ll _{R}P_{u}\ll
_{R}P_{x_{2}}$. We will prove now that $N(A_{i})\setminus N[\widetilde{X}%
_{1}]\subseteq N(x_{2})$. Suppose otherwise that there exists a vertex $v\in
A_{i}$ and a vertex $v^{\prime }\in N(v)\setminus N[\widetilde{X}_{1}]$,
such that $v^{\prime }\notin A_{i}$ and $v^{\prime }\notin N(x_{2})$. By
definition of $\widetilde{C}_{2}$ it follows that either $v^{\prime }\in
Q_{u}$, or $v^{\prime }\in N[u,w]$, or $v^{\prime }\in C_{u}$. Suppose that $%
v^{\prime }\in Q_{u}$. Then, $N(v^{\prime })\subset N(u)$, and thus $v\in
N(u)$, since $vv^{\prime }\in E$. This is a contradiction, since $P_{v}\ll
_{R}P_{u}$ for every $v\in A_{i}$, where $A_{i}\in \mathcal{A}_{1}\cup 
\mathcal{A}_{2}$. Therefore, either $v^{\prime }\in N[u,w]$ or $v^{\prime
}\in C_{u}$. Then, since $u,w\in C_{u}$ and $v^{\prime }\notin N(x_{2})$, it
follows by the definition of $C_{u}$ that always $v^{\prime }\in C_{u}$.
Thus, $v\in C_{u}$, since $v\in N(v^{\prime })$ and $v\notin N(x_{2})$,
which is a contradiction to definition of $\widetilde{C}_{2}$. Therefore, $%
N(A_{i})\setminus N[\widetilde{X}_{1}]\subseteq N(x_{2})$ for every $%
A_{i}\in \mathcal{A}_{1}\cup \mathcal{A}_{2}$. Therefore, in particular $N(V(%
\mathcal{A}_{1}))\setminus N[\widetilde{X}_{1}]\subseteq N(x_{2})$, and thus 
$(N(V(\mathcal{A}_{1}))\setminus N[\widetilde{X}_{1}])\cap N(x_{2})=N(V(%
\mathcal{A}_{1}))\setminus N[\widetilde{X}_{1}]$.

Recall that if a vertex $v\in N[\widetilde{X}_{1}]$, then $v\notin C_{u}$ by
definition of $C_{u}$. Moreover, as we have proved in the previous
paragraph, if a vertex $v\in N(V(\mathcal{A}_{1}))\setminus N[\widetilde{X}%
_{1}]$, then $v\in N(x_{2})$, and thus again $v\notin C_{u}$ by definition
of $C_{u}$. Therefore, since $X_{1}=\widetilde{X}_{1}\cup V(\mathcal{A}_{1})$%
, it follows that if a vertex $v\in N[X_{1}]$, then $v\notin C_{u}$. That
is, $C_{u}$ is the connected component of $G\setminus Q_{u}\setminus
N[X_{1},x_{2}]$, in which $u$ belongs.

Let $A_{i}\in \mathcal{A}_{1}$. Note that no vertex $v\in A_{i}$ is adjacent
to any vertex of $\widetilde{X}_{1}$. Indeed, otherwise $v\in N(w)$ by
definition of $\widetilde{X}_{1}$, which is a contradiction to the
definition of $\widetilde{C}_{2}$. Since $A_{i}\subseteq \widetilde{C}_{2}$
includes no vertex of $C_{u}$, it follows in particular that $v\notin N(w)$
for every $v\in A_{i}$. Indeed, otherwise $v\in C_{u}$, since also $v\notin
N(x_{2})$, which is a contradiction. Consider now a vertex $z\in
(N(A_{i})\setminus N[\widetilde{X}_{1}])\cap N(x_{2})$, i.e.~$z\in
(N(v)\setminus N[\widetilde{X}_{1}])\cap N(x_{2})$ and $z\notin A_{i}$, for
some $v\in A_{i}$. Suppose first that $P_{v}$ intersects $P_{w}$ in $R$.
Then, $v$ is unbounded and $\phi _{v}>\phi _{w}$, since $w$ is bounded, and
thus $N(v)\subseteq N(w)$ by Lemma~\ref{intersecting-unbounded}. Therefore,
in particular, $z\in N(w)$. Suppose now that $P_{v}$ does not intersect $%
P_{w}$ in $R$. Then, $P_{v}\ll _{R}P_{u}\ll _{R}P_{x_{2}}$ and $P_{v}\ll
_{R}P_{w}\ll _{R}P_{x_{2}}$, since $wu\in E$. Thus, $P_{z}$ intersects $%
P_{w} $ and $P_{u}$ in $R$, since $z\in N(v)\cap N(x_{2})$. If $z$ is
unbounded, then $\phi _{z}>\phi _{u}$, since $\phi _{u}=\min \{\phi _{x}\ |\
x\in V_{U}\}$ in $R$ by assumption. Therefore, $N(z)\subseteq N(u)$ by Lemma~%
\ref{intersecting-unbounded}, and thus $x_{2}\in N(u)$, which is a
contradiction. Therefore, $z$ is bounded, and thus $z\in N(w)$, since $P_{z}$
intersects $P_{w}$ in $R$ and both $z$ and $w$ are bounded. Summarizing, $%
z\in N(w)$ for every $z\in (N(A_{i})\setminus N[\widetilde{X}_{1}])\cap
N(x_{2})$. That is, $(N(A_{i})\setminus N[\widetilde{X}_{1}])\cap
N(x_{2})\subseteq N(w)$ for every $A_{i}\in \mathcal{A}_{1}$, i.e.~$(N(V(%
\mathcal{A}_{1}))\setminus N[\widetilde{X}_{1}])\cap N(x_{2})\subseteq N(w)$%
. Therefore, since $X_{1}=\widetilde{X}_{1}\cup V(\mathcal{A}_{1})$, and
since no vertex of $\mathcal{A}_{1}$ is adjacent to any vertex of $%
\widetilde{X}_{1}$, it follows that 
\begin{eqnarray}
N(X_{1}) &=&N(\widetilde{X}_{1})\cup (N(V(\mathcal{A}_{1}))\setminus N[%
\widetilde{X}_{1}])  \label{X1-eq} \\
&=&N(\widetilde{X}_{1})\cup ((N(V(\mathcal{A}_{1}))\setminus N[\widetilde{X}%
_{1}])\cap N(x_{2}))\subseteq N(w)  \notag
\end{eqnarray}%
since $(N(V(\mathcal{A}_{1}))\setminus N[\widetilde{X}_{1}])\cap
N(x_{2})=N(V(\mathcal{A}_{1}))\setminus N[\widetilde{X}_{1}]$ and $N(%
\widetilde{X}_{1})\subseteq N(w)$. That is, ${N(X_{1})\subseteq N(w)}$, i.e.~%
$N_{1}(w)=N(X_{1})$.

Let now $A_{i}\in \mathcal{A}_{1}\cup \mathcal{A}_{2}$, and let $v\in A_{i}$%
. Suppose first that $P_{x}\ll _{R}P_{v}$ for some $x\in \widetilde{X}_{1}$,
i.e.~$P_{x}\ll _{R}P_{v}\ll _{R}P_{u}\ll _{R}P_{x_{2}}$. Then, since $%
x,x_{2}\in V_{0}(u)$, and since $V_{0}(u)$ is connected, there exists a
vertex $z\in V_{0}(u)$, such that $P_{z}$ intersects $P_{v}$ in $R$. If $%
zv\in E$, then $v\in V_{0}(u)$, and thus $A_{i}\subseteq V_{0}(u)$. Let now $%
zv\notin E$. If $\phi _{z}>\phi _{v}$ then $N(z)\subseteq N(v)$ by Lemma~\ref%
{intersecting-unbounded}. Then, since $z\in V_{0}(u)$, and since $V_{0}(u)$
is connected with at least two vertices, $z$ has at least one neighbor $%
z^{\prime }\in V_{0}(u)$, and thus $z^{\prime }\in N(v)$. Then, $v\in
V_{0}(u)$, and thus $A_{i}\subseteq V_{0}(u)$. On the other hand, if $\phi
_{v}>\phi _{z}$, then $N(v)\subseteq N(z)$ by Lemma~\ref%
{intersecting-unbounded}. Furthermore, $v$ is unbounded, since otherwise $%
zv\in E$, which is a contradiction. If $N(v)\subseteq N(u)$, then $v\in
Q_{u} $, which is a contradiction to the definition of $\widetilde{%
\widetilde{C}}_{2}$. Suppose now that $N(v)\nsubseteq N(u)$, i.e.~$v$ has at
least one neighbor $v^{\prime }\notin N(u)$. Then, $v^{\prime }\in N(z)$,
since $N(v)\subseteq N(z)$. Therefore, $v^{\prime }\in V_{0}(u)$ and $v\in
V_{0}(u)$, and thus $A_{i}\subseteq V_{0}(u)$. Summarizing, if $P_{x}\ll
_{R}P_{v}$ for some $x\in \widetilde{X}_{1}$, then $A_{i}\subseteq V_{0}(u)$.

Suppose now that $P_{v}$ intersects $P_{x}$ in $R$, for some $x\in 
\widetilde{X}_{1}$. Recall that $\widetilde{X}_{1}\subseteq V_{0}(u)$ by
Lemma~\ref{X1}, and thus $x\in V_{0}(u)$. If $vx\in E$, then $v\in V_{0}(u)$%
, and thus $A_{i}\subseteq V_{0}(u)$. Let now $vx\notin E$. Then, it follows
similarly to the previous paragraph that $A_{i}\subseteq V_{0}(u)$.

Suppose finally that $P_{v}\ll _{R}P_{x}$, i.e.~$P_{v}\ll _{R}P_{x}\ll
_{R}P_{u}\ll _{R}P_{x_{2}}$, for some $x\in \widetilde{X}_{1}$. Recall that $%
N(A_{i})\setminus N[\widetilde{X}_{1}]\subseteq N(x_{2})$, and thus for
every vertex $v^{\prime }\in N(v)\setminus N[\widetilde{X}_{1}]$, such that $%
v^{\prime }\notin A_{i}$, it follows that $v^{\prime }\in N(x_{2})$.
Consider such a vertex $v^{\prime }$. Then, $P_{v^{\prime }}$ intersects $%
P_{u}$ and $P_{x}$ in $R$, since $P_{v}\ll _{R}P_{x}\ll _{R}P_{u}\ll
_{R}P_{x_{2}}$. Note that $v^{\prime }\notin N(x)$, since otherwise $%
v^{\prime }\in N(\widetilde{X}_{1})$, which is a contradiction to the
assumption that $v^{\prime }\in N(v)\setminus N[\widetilde{X}_{1}]$. Suppose
that $v^{\prime }\in N(u)$, and thus $v^{\prime }$ is bounded in $R$ and $%
\phi _{v^{\prime }}>\phi _{u}$. Then, since $v^{\prime }\notin N(x)$, it
follows that $x$ is unbounded and $\phi _{x}>\phi _{v^{\prime }}>\phi _{u}$.
Thus, $N(x)\subseteq N(v^{\prime })$ by Lemma~\ref{intersecting-unbounded}.
If $\widetilde{X}_{1}\neq \{x\}$, then $x$ has at least one neighbor $%
x^{\prime }$ in $\widetilde{X}_{1}$ and $x^{\prime }\in N(v^{\prime })$,
since $N(x)\subseteq N(v^{\prime })$. Thus, $v^{\prime }\in N(\widetilde{X}%
_{1})$, which is a contradiction to the assumption that $v^{\prime }\in
N(v)\setminus N[\widetilde{X}_{1}]$. Let $\widetilde{X}_{1}=\{x\}$ and $z\in
N(x)$. Then, $N(x)\subseteq N(w)$ by definition of $\widetilde{X}_{1}$, i.e.~%
$z\in N(w)$. Thus, since $T_{x}\ll _{R_{T}}T_{x_{2}}\ll _{R_{T}}T_{w}$, it
follows that $T_{z}$ intersects $T_{x_{2}}$ in $R_{T}$, i.e.~$z\in N(x_{2})$%
. Thus, $P_{z}$ intersects $P_{u}$ in $R$, since $P_{x}\ll _{R}P_{u}\ll
_{R}P_{x_{2}}$ and $z\in N(x)\cap N(x_{2})$. However, $z$ is bounded and $%
\phi _{z}>\phi _{x}>\phi _{u}$, since $x$ is unbounded. Thus, $zu\in E$,
i.e.~$z\in N(u)$. Since this holds for an arbitrary $z\in N(x)$, it follows
that $N(x)\subseteq N(u)$, and thus $x\in Q_{u}$, which is a contradiction
by Lemma~\ref{Qu-1}, since $\widetilde{X}_{1}=\{x\}\subseteq V_{0}(u)$.
Thus, $v^{\prime }\notin N(u)$ for every vertex $v^{\prime }\in
N(v)\setminus N[\widetilde{X}_{1}]$, such that $v^{\prime }\notin A_{i}$.
Therefore, since $v^{\prime }\in N(x_{2})$ for all such vertices $v^{\prime
} $, and since $x_{2}\in V_{0}(u)$, it follows that $v^{\prime },v\in
V_{0}(u)$, and thus $A_{i}\subseteq V_{0}(u)$.

Summarizing, $A_{i}\subseteq V_{0}(u)$ in every case, and thus $%
A_{i}\subseteq D_{1}$ for every component $A_{i}\in \mathcal{A}_{1}\cup 
\mathcal{A}_{2}$. Furthermore, recall that $\widetilde{X}_{1}\subseteq D_{1}$
by Lemma~\ref{X1}. Thus, since also $A_{i}\subseteq D_{1}$ for every
component $A_{i}\in \mathcal{A}_{1}$, it follows that $X_{1}=\widetilde{X}%
_{1}\cup V(\mathcal{A}_{1})\subseteq D_{1}$.

Recall now that $A_{j}\subseteq D_{2}$ for every component $A_{j}\in 
\mathcal{B}_{2}$, where $k+1\leq j\leq \ell $, and thus $A_{j}\subseteq
V_{0}(u)$ for every $A_{j}\in \mathcal{B}_{2}$. Therefore, since also $%
A_{i}\subseteq V_{0}(u)$ for every $A_{i}\in \mathcal{A}_{2}$, and since $%
C_{2}=\mathcal{A}_{2}\cup \mathcal{B}_{2}$, it follows that $C_{2}\subseteq
V_{0}(u)$. This completes the proof of the lemma.
\end{proof}

\begin{lemma}
\label{N(w)-2}For every $x\in X_{1}$, $T_{x}\ll _{R_{T}}T_{x_{2}}$ and $%
P_{x}\ll _{R}P_{w}$.
\end{lemma}

\begin{proof}
Consider a component $A_{i}\in \mathcal{A}_{1}$. Recall that $v\notin
N(x_{2})$ for every $v\in A_{i}$, since $P_{v}\ll _{R}P_{u}\ll _{R}P_{x_{2}}$%
. Thus, since $A_{i}$ is connected, either $T_{x_{2}}\ll _{R_{T}}T_{v}$ or $%
T_{v}\ll _{R_{T}}T_{x_{2}}$ for every vertex $v\in A_{i}$. Suppose that $%
T_{x_{2}}\ll _{R_{T}}T_{v}$ for every $v\in A_{i}$; let $v\in A_{i}$ be such
a vertex. Since $v\in X_{1}\subseteq V_{0}(u)$ by Lemma~\ref{N(w)-1}, it
follows that $T_{v}\ll _{R_{T}}T_{u}$. Recall that $v\notin N(u)\cup N(w)$
by definition of $\widetilde{C}_{2}$. Therefore, since $w\in N(u)$, it
follows that also $T_{v}\ll _{R_{T}}T_{w}$. Consider now a vertex $z\in 
\widetilde{H}=N[u,w]\cap N(x_{2})\setminus Q_{u}\setminus N[\widetilde{X}%
_{1}]$. Then, since $T_{x_{2}}\ll _{R_{T}}T_{v}\ll _{R_{T}}T_{u}$ and $%
T_{x_{2}}\ll _{R_{T}}T_{v}\ll _{R_{T}}T_{w}$, it follows that $T_{z}$
intersects $T_{v}$ in $R_{T}$, and thus $vz\in E$. Since this holds for
every vertex $v\in A_{i}$ and every vertex $z\in \widetilde{H}$, it follows
that $A_{i}\in \mathcal{A}_{2}$, which is a contradiction. Thus, $T_{v}\ll
_{R_{T}}T_{x_{2}}$ for every vertex $v\in A_{i}$, where $A_{i}\in \mathcal{A}%
_{1}$. Therefore, since also $T_{x}\ll _{R_{T}}T_{x_{2}}$ for every vertex $%
x\in \widetilde{X}_{1}$ by Lemma~\ref{X1}, it follows that $T_{x}\ll
_{R_{T}}T_{x_{2}}$ for every vertex $x\in X_{1}$.

We will prove now that $P_{v}\ll _{R}P_{w}$ for every $v\in A_{i}$, where $%
A_{i}\in \mathcal{A}_{1}$. Otherwise, suppose first that $P_{w}\ll _{R}P_{v}$
for some $v\in A_{i}$. Then, since $P_{v}\ll _{R}P_{u}$ for every $v\in
A_{i} $, it follows that $P_{w}\ll _{R}P_{v}\ll _{R}P_{u}$, and thus $%
w\notin N(u)$, which is a contradiction. Suppose now that $P_{v}$ intersects 
$P_{w}$ in $R $, for some $v\in A_{i}$. Then, since $v\notin N(w)$ by
definition of $\widetilde{C}_{2}$, and since $w$ is bounded, it follows that 
$v$ is unbounded and $\phi _{v}>\phi _{w}>\phi _{u}$. Thus, $N(v)\subseteq
N(w)$ by Lemma~\ref{intersecting-unbounded}. Let now $z\in N(v)\subseteq
N(w) $. Then, since $T_{v}\ll _{R_{T}}T_{x_{2}}\ll _{R_{T}}T_{w}$ (cf.~the
previous paragraph), it follows that $T_{z}$ intersects $T_{x_{2}}$ in $%
R_{T} $, i.e.~$z\in N(x_{2})$. Since this holds for an arbitrary $z\in N(v)$%
, it follows that also $N(v)\subseteq N(x_{2})$. Therefore, since $P_{v}\ll
_{R}P_{u}\ll _{R}P_{x_{2}}$, it follows that $P_{z}$ intersects $P_{u}$ in $%
R $ for every $z\in N(v)\subseteq N(x_{2})$. Furthermore, since $v$ is
unbounded, it follows that $z$ is bounded and $\phi _{z}>\phi _{v}>\phi _{u}$
for every $z\in N(v)$, and thus $N(v)\subseteq N(u)$. That is, $v\in Q_{u}$,
which is a contradiction by Lemma~\ref{Qu-1}, since $v\in A_{i}\subseteq
X_{1}\subseteq V_{0}(u)$. It follows that $P_{v}\ll _{R}P_{w}$ for every $%
v\in A_{i}$, where $A_{i}\in \mathcal{A}_{1}$. Therefore, since also $%
P_{x}\ll _{R}P_{w}$ for every vertex $x\in \widetilde{X}_{1}$ by Lemma~\ref%
{X1}, it follows that $P_{x}\ll _{R}P_{w}$ for every vertex $x\in X_{1}$.
This completes the lemma.
\end{proof}

\begin{lemma}
\label{N-Cu}Denote $N=N(X_{1})=N_{1}(w)$. Then, $N_{1}(u)\subset N$ and $%
N_{1}(x_{2})=N_{1}(v)=N$ for every bounded vertex $v\in C_{u}\setminus \{u\}$
in $R$.
\end{lemma}

\begin{proof}
First note that $N_{1}(u)\subseteq N$, since $N=N(X_{1})$ and $%
N_{1}(u)=N(u)\cap N(X_{1})$ by definition. Recall that $N(\widetilde{X}%
_{1})\subseteq N=N(X_{1})$ and that $N(\widetilde{X}_{1})\setminus N(u)\neq
\emptyset $ by Corollary~\ref{N-X1-tilde-N-u}. Therefore also $N\setminus
N(u)\neq \emptyset $, and thus $N_{1}(u)\subset N$.

Consider a vertex $z\in N$, i.e.~$z\in N(x)\cap N(w)$ for some $x\in X_{1}$
by Lemma~\ref{N(w)-1}. Then, since $T_{x}\ll _{R_{T}}T_{x_{2}}\ll
_{R_{T}}T_{w}$ by Lemma~\ref{N(w)-2}, it follows that $T_{z}$ intersects $%
T_{x_{2}}$ in $R_{T}$. Therefore, $z\in N(x_{2})$, and thus $z\in
N_{1}(x_{2})$. Since this holds for every $z\in N$, it follows that $%
N\subseteq N_{1}(x_{2})$. Thus, since by definition $N_{1}(x_{2})\subseteq N$%
, it follows that $N_{1}(x_{2})=N$.

Consider now a bounded vertex $v\in C_{u}$ in $R$ and a vertex $z\in N$.
Then, $z\in N(x)\cap N(x_{2})$ for some $x\in X_{1}$, since $N_{1}(x_{2})=N$
by the previous paragraph. Recall that $C_{u}$ is connected and that no
vertex of $C_{u}$ is adjacent to $x_{2}$ by the definition of $C_{u}$. Thus,
since $w\in C_{u}$ and $T_{x_{2}}\ll _{R_{T}}T_{w}$, it follows that $%
T_{x_{2}}$ lies in $R_{T}$ to the left of all trapezoids of the vertices of $%
C_{u}$; in particular, Lemma~\ref{N(w)-2} implies that $T_{x}\ll
_{R_{T}}T_{x_{2}}\ll _{R_{T}}T_{v}$ for every $x\in X_{1}$.

Suppose first that $P_{x}\ll _{R}P_{v}\ll _{R}P_{x_{2}}$. Then, $P_{z}$
intersects $P_{v}$ in $R$. Suppose that $z\notin N(v)$. Then, since $v$ is
bounded, it follows that $z$ is unbounded and $\phi _{z}>\phi _{v}$, and
thus~$N(z)\subseteq N(v)$ by Lemma~\ref{intersecting-unbounded}. Therefore,
since $x\in N(z)$, it follows that $x\in N(v)$, i.e.~$v\in N(X_{1})$, which
is a contradiction by Lemma~\ref{N(w)-1}. Thus, $z\in N(v)$.

Suppose now that $P_{v}$ intersects $P_{x}$ (resp.~$P_{x_{2}}$) in $R$.
Recall that, since $v\in C_{u}$, $v\notin N(x)$ by Lemma~\ref{N(w)-1} (resp.~%
$v\notin N(x_{2})$ by definition of $C_{u}$). Thus, either $N(v)\subseteq
N(x)$ or $N(x)\subseteq N(v)$ (resp.~$N(v)\subseteq N(x_{2})$ or $%
N(x_{2})\subseteq N(v)$) by Lemma~\ref{intersecting-unbounded}. If $%
N(v)\subseteq N(x)$ (resp.~$N(v)\subseteq N(x_{2})$), then $v$ is an
isolated vertex in $G\setminus Q_{u}\setminus N[X_{1},x_{2}]$, and thus $%
v\notin C_{u}$, since $v\neq u$, which is a contradiction. If $N(x)\subseteq
N(v)$ (resp.~$N(x_{2})\subseteq N(v)$), then $z\in N(v)$, since in
particular $z\in N(x)$ (resp.~$z\in N(x_{2})$). Note here that this
paragraph holds for both cases, where $v$ is a bounded or an unbounded
vertex in $R$.

Suppose that $P_{x_{2}}\ll _{R}P_{v}$. Then, $v\notin N(u)$ and $v\notin
N(w) $, since $P_{u}\ll _{R}P_{x_{2}}$ and $P_{w}\ll _{R}P_{x_{2}}$.
Furthermore, since $C_{u}$ is connected, there must exist a vertex $%
v^{\prime }$ of $C_{u} $, such that $P_{v^{\prime }}$ intersects $P_{x_{2}}$
in $R$, and a path $P$ from $v^{\prime }$ to $v$, where all intermediate
vertices are $v^{\prime \prime }\in C_{u}$, such that $P_{x_{2}}\ll
_{R}P_{v^{\prime \prime }}$, i.e.~$v^{\prime \prime }\notin N(u)$ and $%
v^{\prime \prime }\notin N(w)$. Recall that $v^{\prime }\notin N(x_{2})$ by
definition of $C_{u}$, since $v^{\prime }\in C_{u}$. Then, since $%
P_{v^{\prime }}$ intersects $P_{x_{2}}$ in $R$, it follows by the previous
paragraph that $z\in N(v^{\prime })$.

Let $v^{\prime }\in N(u)$, and thus $v^{\prime }$ is bounded and $\phi
_{v^{\prime }}>\phi _{u}$. Then, $x_{2}$ is unbounded and $\phi
_{x_{2}}>\phi _{v^{\prime }}>\phi _{u}$, since $v^{\prime }$ is bounded and $%
v^{\prime }\notin N(x_{2})$. Consider now an arbitrary $z^{\prime }\in N$.
Recall that $z^{\prime }\in N(x^{\prime })\cap N(x_{2})$ for some $x^{\prime
}\in X_{1}$, and thus $P_{z^{\prime }}$ intersects $P_{u}$ in $R$, since $%
P_{x^{\prime }}\ll _{R}P_{u}\ll _{R}P_{x_{2}}$. Furthermore, $z^{\prime }$
is bounded and $\phi _{z}>\phi _{x_{2}}>\phi _{u}$, since $x_{2}$ is
unbounded. Thus, $z^{\prime }\in N(u)$. Since this holds for an arbitrary $%
z^{\prime }\in N$, it follows that $N_{1}(u)=N$, which is a contradiction.

Let $v^{\prime }\notin N(u)$. Since $v,v^{\prime }\notin N(u)$, and since $%
v^{\prime \prime }\notin N(u)$ for all intermediate vertices $v^{\prime
\prime }$ of the path $P$, it follows that either $T_{u}\ll
_{R_{T}}T_{v^{\prime }}$ and $T_{u}\ll _{R_{T}}T_{v}$, or $T_{v^{\prime
}}\ll _{R_{T}}T_{u}$ and $T_{v}\ll _{R_{T}}T_{u}$. Recall that $z\in
N(v^{\prime })$. Therefore, if $T_{u}\ll _{R_{T}}T_{v^{\prime }}$, then $%
T_{z}$ intersects $T_{u}$ in $R_{T}$, i.e.~$z\in N(u)$, since in this case $%
T_{x_{2}}\ll _{R_{T}}T_{u}\ll _{R_{T}}T_{v^{\prime }}$ and $z\in N(v^{\prime
})\cap N(x_{2})$. Since this holds for an arbitrary $z\in N$, it follows
that $N_{1}(u)=N$, which is a contradiction. Thus, $T_{v^{\prime }}\ll
_{R_{T}}T_{u}$ and $T_{v}\ll _{R_{T}}T_{u}$. Since $v\notin N(w)$, $T_{w}$
does not intersect $T_{v}$ in $R_{T}$, i.e.~either $T_{w}\ll _{R_{T}}T_{v}$
or $T_{v}\ll _{R_{T}}T_{w}$. If $T_{w}\ll _{R_{T}}T_{v}$, then $T_{w}\ll
_{R_{T}}T_{v}\ll _{R_{T}}T_{u}$, and thus $w\notin N(u)$, which is a
contradiction. Therefore, $T_{v}\ll _{R_{T}}T_{w}$, i.e.~$T_{x_{2}}\ll
_{R_{T}}T_{v}\ll _{R_{T}}T_{w}$. Thus, $T_{z}$ intersects $T_{v}$ in $R_{T}$%
, i.e.~$z\in N(v)$, since $z\in N(x_{2})\cap N(w)$.

Suppose finally that $P_{v}\ll _{R}P_{x}$. Then, $v\notin N(u)$ and $v\notin
N(w)$, since $P_{x}\ll _{R}P_{u}$ and $P_{x}\ll _{R}P_{w}$. Furthermore,
since $C_{u}$ is connected, there must exist a vertex $v^{\prime }$ of $%
C_{u} $, such that $P_{v^{\prime }}$ intersects $P_{x}$ in $R$, and a path $%
P $ from $v^{\prime }$ to $v$, where all intermediate vertices are $%
v^{\prime \prime }\in C_{u}$, such that $P_{v^{\prime \prime }}\ll
_{R}P_{x_{1}}$, i.e.~$v^{\prime \prime }\notin N(u)$ and $v^{\prime \prime
}\notin N(w)$. Recall that $v^{\prime }\notin N(x)$ by Lemma~\ref{N(w)-1},
since $v^{\prime }\in C_{u}$. Then, since $P_{v^{\prime }}$ intersects $%
P_{x} $ in $R$, it follows (similarly to the above case where $P_{v}$
intersects $P_{x}$ in $R$) that $z\in N(v^{\prime })$.

Let $v^{\prime }\in N(u)$, and thus $v^{\prime }$ is bounded and $\phi
_{v^{\prime }}>\phi _{u}$. Then, $x$ is unbounded and $\phi _{x}>\phi
_{v^{\prime }}>\phi _{u}$, since $v^{\prime }$ is bounded and $v^{\prime
}\notin N(x)$. Thus $N(x)\subseteq N(v^{\prime })$ by Lemma~\ref%
{intersecting-unbounded}. Since $x\in X_{1}$, either $x\in \widetilde{X}_{1}$
or $x\in A_{i}$ for some $A_{i}\in \mathcal{A}_{1}$. Let $x\in \widetilde{X}%
_{1}$ (resp.~$x\in A_{i}$ for some $A_{i}\in \mathcal{A}_{1}$). If $%
\widetilde{X}_{1}\neq \{x\}$ (resp.~$A_{i}\neq \{x\}$), then $x$ has at
least one neighbor $x^{\prime }$ in $\widetilde{X}_{1}$ (resp.~in $A_{i}$)
and $x^{\prime }\in N(v^{\prime })$, since $N(x)\subseteq N(v^{\prime })$.
Thus, $v^{\prime }\in N(X_{1})$, which is a contradiction by Lemma~\ref%
{N(w)-1}, since $v^{\prime }\in C_{u}$. If $\widetilde{X}_{1}=\{x\}$ (resp.~$%
A_{i}=\{x\}$), then $\{x\}$ is a connected component of $X_{1}$. Therefore, $%
z^{\prime }\notin X_{1}$ for every neighbor $z^{\prime }\in N(x)$, and thus $%
N(x)\subseteq N(x_{2})$, since $N_{1}(x_{2})=N(X_{1})$, as we proved above.
That is, $P_{z^{\prime }}$ intersects $P_{u}$ for every $z^{\prime }\in N(x)$%
, since in this case $P_{x}\ll _{R}P_{u}\ll _{R}P_{x_{2}}$ and $z^{\prime
}\in N(x)\cap N(x_{2})$. However, $z^{\prime }$ is bounded and $\phi
_{z^{\prime }}>\phi _{x}>\phi _{u}$, since $x$ is unbounded. Thus, $%
z^{\prime }\in N(u)$ for every $z^{\prime }\in N(x)$. That is, $%
N(x)\subseteq N(u)$, and thus $x\in Q_{u}$, which is a contradiction by
Lemma~\ref{Qu-1}, since $x\in X_{1}\subseteq V_{0}(u)$.

Let $v^{\prime }\notin N(u)$. Since $v,v^{\prime }\notin N(u)$, and since $%
v^{\prime \prime }\notin N(u)$ for all intermediate vertices $v^{\prime
\prime }$ of the path $P$, it follows that either $T_{u}\ll
_{R_{T}}T_{v^{\prime }}$ and $T_{u}\ll _{R_{T}}T_{v}$, or $T_{v^{\prime
}}\ll _{R_{T}}T_{u}$ and $T_{v}\ll _{R_{T}}T_{u}$. Recall that $z\in
N(v^{\prime })$. Therefore, if $T_{u}\ll _{R_{T}}T_{v^{\prime }}$, then $%
T_{z}$ intersects $T_{u}$ in $R_{T}$, i.e.~$z\in N(u)$, since in this case $%
T_{x}\ll _{R_{T}}T_{u}\ll _{R_{T}}T_{v^{\prime }}$ and $z\in N(x)\cap
N(v^{\prime })$. Since this holds for an arbitrary $z\in N$, it follows that 
$N_{1}(u)=N$, which is a contradiction. Thus, $T_{v^{\prime }}\ll
_{R_{T}}T_{u}$ and $T_{v}\ll _{R_{T}}T_{u}$. Since $v\notin N(w)$, $T_{w}$
does not intersect $T_{v}$ in $R_{T}$, i.e.~either $T_{w}\ll _{R_{T}}T_{v}$
or $T_{v}\ll _{R_{T}}T_{w}$. If $T_{w}\ll _{R_{T}}T_{v}$, then $T_{w}\ll
_{R_{T}}T_{v}\ll _{R_{T}}T_{u}$, and thus $w\notin N(u)$, which is a
contradiction. Therefore, $T_{v}\ll _{R_{T}}T_{w}$, i.e.~$T_{x}\ll
_{R_{T}}T_{v}\ll _{R_{T}}T_{w}$. Thus, $T_{z}$ intersects $T_{v}$ in $R_{T}$%
, i.e.~$z\in N(v)$, since $z\in N(x)\cap N(w)$.

Summarizing, $z\in N(v)$ for any $z\in N$ and any bounded vertex $v$ of $%
C_{u}$ in $R$, i.e.~$N\subseteq N_{1}(v)$. Then, since $N_{1}(v)\subseteq
N(X_{1})=N$, it follows that $N_{1}(v)=N$ for every bounded vertex $v$ of $%
C_{u}$ in $R$. This completes the proof of the lemma.
\end{proof}

\medskip

The next two lemmas follow easily and will be used in the sequel.

\begin{lemma}
\label{foreigner-not-inside-1}Let $v\in V\setminus Q_{u}\setminus
N[u]\setminus V_{0}(u)$. Then, either $P_{x_{2}}\ll _{R}P_{v}$ or $P_{v}\ll
_{R}P_{x}$ for every $x\in X_{1}$.
\end{lemma}

\begin{proof}
Let $v\in V\setminus Q_{u}\setminus N[u]\setminus V_{0}(u)$. Recall that $%
X_{1}\subseteq V_{0}(u)$ by Lemma~\ref{N(w)-1} and that $x_{2}\in V_{0}(u)$
by definition of $x_{2}$. Suppose first that $P_{v}$ intersects $P_{x}$, for
some $x\in X_{1}$ (resp.~$P_{v}$ intersects $P_{x_{2}}$). If $v\in N(x)$
(resp.~$v\in N(x_{2})$), then $v\in V_{0}(u)$, since also $v\notin N(u)$,
which is a contradiction. Therefore, $v\notin N(x)$ (resp.~$v\notin N(x_{2})$%
). If $\phi _{x}>\phi _{v}$ (resp.~$\phi _{x_{2}}>\phi _{v}$), then $%
N(x)\subseteq N(v)$ (resp.~$N(x_{2})\subseteq N(v)$) by Lemma~\ref%
{intersecting-unbounded}. Then, since $x$ (resp.~$x_{2}$) is not the only
vertex of $V_{0}(u)$, and since $V_{0}(u)$ is connected, it follows that $x$
(resp.~$x_{2}$) is adjacent to another vertex $q\in V_{0}(u)$. Therefore $%
q\in N(v)$, since $N(x)\subseteq N(v)$ (resp.~$N(x_{2})\subseteq N(v)$), and
thus~also $v\in V_{0}(u)$, which is a contradiction. If $\phi _{x}<\phi _{v}$
(resp.~$\phi _{x_{2}}<\phi _{v}$), then $N(v)\subseteq N(x)$ (resp.~$%
N(v)\subseteq N(x_{2})$) by Lemma~\ref{intersecting-unbounded}. Then, in
particular, $v$ is unbounded, since otherwise $v\in N(x)$ (resp.~$v\in
N(x_{2})$), which is a contradiction. Since $v\notin Q_{u}$ by the
assumption on $v$, there exists at least one vertex $z\in N(v)\setminus N(u)$%
. Therefore, $z\in N(x)$ (resp.~$z\in N(x_{2})$), since $N(v)\subseteq N(x)$
(resp.~$N(v)\subseteq N(x_{2})$), and thus $z\in V_{0}(u)$ and $v\in
V_{0}(u) $, which is a contradiction. Thus, $P_{v}$ does not intersect $%
P_{x_{2}}$ or $P_{x}$, for any $x\in X_{1}$.

Suppose now that $P_{x}\ll _{R}P_{v}\ll _{R}P_{x_{2}}$ for some $x\in X_{1}$%
. Then, since $x_{2}\in V_{0}(u)$ and $x\in X_{1}\subseteq V_{0}(u)$, and
since $V_{0}(u)$ is connected, there exists a vertex $y\in V_{0}(u)$, such
that $P_{y}$ intersects $P_{v}$ in $R$. Then $v\notin N(y)$, since otherwise 
$v\in V_{0}(u)$, which is a contradiction. If $\phi _{y}>\phi _{v}$, then $%
N(y)\subseteq N(v)$ by Lemma~\ref{intersecting-unbounded}. Since $V_{0}(u)$
is connected with at least two vertices, there exists at least one neighbor $%
q\in V_{0}(u)$ of $y$. Then $q\in N(v)$, since $N(y)\subseteq N(v)$, and
thus $v\in V_{0}(u)$, which is a contradiction. If $\phi _{y}<\phi _{v}$,
then $N(v)\subseteq N(y)$ by Lemma~\ref{intersecting-unbounded}. Then, in
particular, $v$ is unbounded, since otherwise $v\in N(y)$, which is a
contradiction. Since $v\notin Q_{u}$ by the assumption on $v$, there exists
at least one vertex $z\in N(v)\setminus N(u)$. Therefore, $z\in N(y)$, since 
$N(v)\subseteq N(y)$, and thus $z\in V_{0}(u)$ and $v\in V_{0}(u)$, which is
again a contradiction.

Therefore, if $v\in V\setminus Q_{u}\setminus N[u]\setminus V_{0}(u)$, then
either $P_{x_{2}}\ll _{R}P_{v}$ or $P_{v}\ll _{R}P_{x}$ for every $x\in
X_{1} $. This completes the proof of the lemma.
\end{proof}

\begin{lemma}
\label{V0-middle}For every $v\in V\setminus N[u]\setminus V_{0}(u)$, either $%
T_{x_{2}}\ll _{R_{T}}T_{v}$ or $T_{v}\ll _{R_{T}}T_{x}$ for every $x\in
X_{1} $.
\end{lemma}

\begin{proof}
Let $v\in V\setminus N[u]\setminus V_{0}(u)$. Recall first that $%
X_{1}\subseteq V_{0}(u)$ by Lemma~\ref{N(w)-1} and that $x_{2}\in V_{0}(u)$
by definition of $x_{2}$. If $T_{v}$ intersects $T_{x_{2}}$ or $T_{x}$ for
some $x\in X_{1}$ in $R_{T}$, then $v\in V_{0}(u)$, since $v\notin N[u]$,
which is a contradiction. Thus, $T_{v}$ does not intersect $T_{x_{2}}$ or $%
T_{x}$ in $R_{T}$, for any $x\in X_{1}$. Suppose that $T_{x}\ll
_{R_{T}}T_{v}\ll _{R_{T}}T_{x_{2}}$ for some $x\in X_{1}$. Then, since $%
V_{0}(u)$ is connected, it follows that $T_{z}$ intersects $T_{v}$ in $R_{T}$
for at least one vertex $z\in V_{0}(u)$, and thus also $v\in V_{0}(u)$,
which is again a contradiction. Thus, either $T_{x_{2}}\ll _{R_{T}}T_{v}$ or 
$T_{v}\ll _{R_{T}}T_{x}$ for every $x\in X_{1}$.
\end{proof}

\subsection*{Some properties of the sets $C_{u}$ and $C_{2}$}

In the next three lemmas we prove some basic properties of the vertex sets $%
C_{u}$ and $C_{u}$, which will be mainly used in the sequel of the proof of
Theorem~\ref{no-property-thm}.

\begin{lemma}
\label{Cu-V0}For every vertex $v\in C_{u}\setminus \{u\}$, $v\in
V_{0}(u)\cup N(u)$.
\end{lemma}

\begin{proof}
Consider a vertex $v\in C_{u}\setminus \{u\}$. Then, $v\notin Q_{u}$ by
definition of $C_{u}$. Suppose that $v\notin V_{0}(u)\cup N(u)$, i.e.~$v\in
V\setminus Q_{u}\setminus N[u]\setminus V_{0}(u)$. Then, either $%
P_{x_{2}}\ll _{R}P_{v}$ or $P_{v}\ll _{R}P_{x}$ for every $x\in X_{1}$ by
Lemma~\ref{foreigner-not-inside-1}.

Suppose first that $P_{x_{2}}\ll _{R}P_{v}$. Then, since $C_{u}$ is
connected, and since $P_{u}\ll _{R}P_{x_{2}}$, there must exist a vertex $%
v^{\prime }$ of $C_{u}$, such that $P_{v^{\prime }}$ intersects $P_{x_{2}}$
in $R$, and a path $P$ from $v^{\prime }$ to $v$, where all intermediate
vertices are $v^{\prime \prime }\in C_{u}$, such that $P_{x_{2}}\ll
_{R}P_{v^{\prime \prime }}$. Therefore, since $P_{u}\ll _{R}P_{x_{2}}\ll
_{R}P_{v^{\prime \prime }}$, it follows that $v^{\prime \prime }\notin N(u)$
for all these intermediate vertices. Furthermore, $v^{\prime }\notin
N(x_{2}) $ by definition of $C_{u}$. If $\phi _{x_{2}}<\phi _{v^{\prime }}$,
then $N(v^{\prime })\subseteq N(x_{2})$ by Lemma~\ref{intersecting-unbounded}%
. Therefore, $v^{\prime }$ is an isolated vertex of $G\setminus
Q_{u}\setminus N[X_{1},x_{2}]$, and thus $v^{\prime }\notin C_{u}$, which is
a contradiction. If $\phi _{x_{2}}>\phi _{v^{\prime }}$, then $%
N(x_{2})\subseteq N(v^{\prime })$ by Lemma~\ref{intersecting-unbounded}.
Then, in particular, $x_{2}$ is unbounded, since otherwise $v^{\prime }\in
N(x_{2})$, which is a contradiction. Thus, $\phi _{x_{2}}>\phi _{u}$, since $%
\phi _{u}=\min \{\phi _{x}\ |\ x\in V_{U}\}$. Furthermore, since $%
N_{1}(x_{2})=N$ by Lemma~\ref{N-Cu}, and since $P_{x}\ll _{R}P_{u}\ll
_{R}P_{x_{2}}$ for every $x\in X_{1}$, it follows that $P_{z}$ intersects $%
P_{u}$ in $R$ for every $z\in N$. Moreover, since $x_{2}$ is unbounded, and
since $z\in N(x_{2})$ for every $z\in N$, it follows that $z$ is bounded and 
$\phi _{z}>\phi _{x_{2}}>\phi _{u}$ for every $z\in N$. Therefore, $%
N\subseteq N(u)$, i.e.~$N_{1}(u)=N$, which is a contradiction by Lemma~\ref%
{N-Cu}.

Suppose now that $P_{v}\ll _{R}P_{x}$ for every $x\in X_{1}$. Then, since $%
C_{u}$ is connected, and since $P_{x}\ll _{R}P_{u}$ for every $x\in X_{1}$,
there must exist a vertex $v^{\prime }$ of $C_{u}$, such that $P_{v^{\prime
}}$ intersects $P_{x_{0}}$ in $R$ for some $x_{0}\in X_{1}$, and a path $P$
from $v^{\prime }$ to $v$, where all intermediate vertices are $v^{\prime
\prime }\in C_{u}$, such that $P_{v^{\prime \prime }}\ll _{R}P_{x}$ for
every $x\in X_{1}$. Therefore, since $P_{v^{\prime \prime }}\ll _{R}P_{x}\ll
_{R}P_{u}$ for every $x\in X_{1}$, it follows that $v^{\prime \prime }\notin
N(u)$ for all these intermediate vertices. Furthermore, $v^{\prime }\notin
N(x_{0})$ by Lemma~\ref{N(w)-1}, since $v^{\prime }\in C_{u}$.

Let first $v^{\prime }\notin N(u)$. If $\phi _{x_{0}}<\phi _{v^{\prime }}$,
then $N(v^{\prime })\subseteq N(x_{0})$ by Lemma~\ref{intersecting-unbounded}%
. Therefore, $v^{\prime }$ is an isolated vertex of $G\setminus
Q_{u}\setminus N[X_{1},x_{2}]$, and thus $v^{\prime }\notin C_{u}$, which is
a contradiction. If $\phi _{x_{0}}>\phi _{v^{\prime }}$, then $%
N(x_{0})\subseteq N(v^{\prime })$ by Lemma~\ref{intersecting-unbounded}.
Then, in particular, $x_{0}$ is unbounded, since otherwise $v^{\prime }\in
N(x_{0})$, which is a contradiction. Since $x_{0}\in X_{1}\subseteq V_{0}(u)$%
, and since $x_{0}$ is not the only vertex of $V_{0}(u)$, it follows that $%
x_{0}$ has at least one neighbor $z\in V_{0}(u)$. Thus, $z\in N(v^{\prime })$%
, since $N(x_{0})\subseteq N(v^{\prime })$. Therefore, since $v^{\prime
}\notin N(u)$, it follows that also $v^{\prime }\in V_{0}(u)$. Thus, since $%
v\notin N(u)$ and $v^{\prime \prime }\notin N(u)$ for all intermediate
vertices $v^{\prime \prime }$ of the path $P$, it follows that $v\in
V_{0}(u) $ and $v^{\prime \prime }\in V_{0}(u)$ for all these vertices $%
v^{\prime \prime }$. This is a contradiction to the assumption that $v\notin
V_{0}(u)\cup N(u)$.

Let now $v^{\prime }\in N(u)$. Then, $P_{v^{\prime }}$ intersects $P_{x}$
for every $x\in X_{1}$, since $P_{v^{\prime \prime }}\ll _{R}P_{x}\ll
_{R}P_{u}$ for every $x\in X_{1}$ and for every intermediate vertex $%
v^{\prime \prime }$ of the path $P$. If $\phi _{x}<\phi _{v^{\prime }}$ for
at least one $x\in X_{1}$, then $N(v^{\prime })\subseteq N(x)$ by Lemma~\ref%
{intersecting-unbounded}. Therefore, $v^{\prime }$ is an isolated vertex of $%
G\setminus Q_{u}\setminus N[X_{1},x_{2}]$, and thus $v^{\prime }\notin C_{u}$%
, which is a contradiction. Otherwise, if $\phi _{x}>\phi _{v^{\prime }}$
for every $x\in X_{1}$, then $N(x)\subseteq N(v^{\prime })$ for every $x\in
X_{1}$ by Lemma~\ref{intersecting-unbounded}. Then, in particular, every $%
x\in X_{1}$ is unbounded, since otherwise $v^{\prime }\in N(x)$, which is a
contradiction. Thus, $\phi _{x}>\phi _{u}$ for every $x\in X_{1}$, since $%
\phi _{u}=\min \{\phi _{x}\ |\ x\in V_{U}\}$. Furthermore, since $%
N_{1}(x_{2})=N=N(X_{1})$ by Lemma~\ref{N-Cu}, and since $P_{x}\ll
_{R}P_{u}\ll _{R}P_{x_{2}}$ for every $x\in X_{1}$, it follows that $P_{z}$
intersects $P_{u}$ in $R$ for every $z\in N$. Moreover, since every $x\in
X_{1}$ is unbounded, it follows that for every $z\in N$, $z$ is bounded and $%
\phi _{z}>\phi _{x}>\phi _{u}$ for at least one $x\in X_{1}$. Therefore, $%
N\subseteq N(u)$, i.e.~$N_{1}(u)=N$, which is a contradiction by Lemma~\ref%
{N-Cu}. Summarizing, $v\in V_{0}(u)\cup N(u)$ for every $v\in C_{u}\setminus
\{u\}$.
\end{proof}

\begin{lemma}
\label{N-Cu-unbounded}For every vertex $v\in C_{u}\setminus \{u\}$, $%
N_{1}(v)=N$.
\end{lemma}

\begin{proof}
If $v$ is a bounded vertex in $R$, then the lemma follows by Lemma~\ref{N-Cu}%
. Suppose now that $v$ is unbounded. Then, since $v\notin Q_{u}$ by
definition of $C_{u}$, it follows that there exists at least one vertex $%
y_{v}\in N(v)\setminus N(u)$. Furthermore, there exists at least one vertex $%
y_{u}\in N(u)\setminus N(v)$. Indeed, otherwise $N(u)\subseteq N(v)$, and
thus $N(u)\subset N(v)$ by Lemma~\ref{not-equal}, i.e.~$u$ is not unbounded
maximal, which is a contradiction. Then, both $y_{u}$ and $y_{v}$ are
bounded vertices in $R$, since $u$ and $v$ are unbounded. Furthermore, since 
$uv\notin E$, either $T_{u}\ll _{R_{T}}T_{v}$ or $T_{v}\ll _{R_{T}}T_{u}$.

Let first $T_{u}\ll _{R_{T}}T_{v}$. Since $y_{v}\notin N(u)$, $T_{y_{v}}$
does not intersect $T_{u}$ in $R_{T}$, i.e.~either $T_{y_{v}}\ll
_{R_{T}}T_{u}$ or $T_{u}\ll _{R_{T}}T_{y_{v}}$. If $T_{y_{v}}\ll
_{R_{T}}T_{u}$, then $T_{y_{v}}\ll _{R_{T}}T_{u}\ll _{R_{T}}T_{v}$, and thus 
$y_{v}\notin N(v)$, which is a contradiction. Therefore, $T_{u}\ll
_{R_{T}}T_{y_{v}}$. Moreover, $T_{x}\ll _{R_{T}}T_{x_{2}}\ll
_{R_{T}}T_{u}\ll _{R_{T}}T_{y_{v}}$ for every $x\in X_{1}$ by Lemma~\ref%
{N(w)-2}, and thus in particular $y_{v}\notin N(X_{1})$ and $y_{v}\notin
N(x_{2})$. Suppose that $N_{1}(y_{v})\neq N$. Then, $y_{v}\notin C_{u}$ by
Lemma~\ref{N-Cu}, since $y_{v}$ is bounded. Thus, since $v\in C_{u}$, $%
y_{v}\in N(v)$, and $y_{v}\notin Q_{u}$, it follows by Lemma~\ref{N(w)-1}
that either $y_{v}\in N(X_{1})$ or $y_{v}\in N(x_{2})$, which is a
contradiction. Therefore, $N_{1}(y_{v})=N$. Thus, for every $z\in N$, $T_{z}$
intersects $T_{u}$ in $R_{T}$, i.e.~$z\in N(u)$, since $T_{x_{2}}\ll
_{R_{T}}T_{u}\ll _{R_{T}}T_{y_{v}}$ and $z\in N(x_{2})\cap N(y_{v})$.
Therefore, $N_{1}(u)=N$, which is a contradiction by Lemma~\ref{N-Cu}.

Let now $T_{v}\ll _{R_{T}}T_{u}$. Since $y_{u}\notin N(v)$, $T_{y_{u}}$ does
not intersect $T_{v}$ in $R_{T}$, i.e.~either $T_{y_{u}}\ll _{R_{T}}T_{v}$
or $T_{v}\ll _{R_{T}}T_{y_{u}}$. If $T_{y_{u}}\ll _{R_{T}}T_{v}$, then $%
T_{y_{u}}\ll _{R_{T}}T_{v}\ll _{R_{T}}T_{u}$, and thus $y_{u}\notin N(u)$,
which is a contradiction. Therefore, $T_{v}\ll _{R_{T}}T_{y_{u}}$. Recall
that $C_{u}$ is connected and that no vertex of $C_{u}$ is adjacent to $%
x_{2} $ by the definition of $C_{u}$. Thus, since $u\in C_{u}$ and $%
T_{x_{2}}\ll _{R_{T}}T_{u}$, it follows that $T_{x_{2}}$ lies in $R_{T}$ to
the left of all trapezoids of the vertices of $C_{u}$; in particular, Lemma %
\ref{N(w)-2} implies that $T_{x}\ll _{R_{T}}T_{x_{2}}\ll _{R_{T}}T_{v}\ll
_{R_{T}}T_{y_{u}}$ for every $x\in X_{1}$. Thus, in particular, $y_{u}\notin
N(X_{1})$ and $y_{u}\notin N(x_{2})$. Suppose that $N_{1}(y_{u})\neq N$.
Then, $y_{u}\notin C_{u}$ by Lemma~\ref{N-Cu}, since $y_{u}$ is bounded.
Thus, since $u\in C_{u}$, $y_{u}\in N(u)$, and $y_{u}\notin Q_{u}$, it
follows by Lemma~\ref{N(w)-1} that either $y_{u}\in N(X_{1})$ or $y_{u}\in
N(x_{2})$, which is a contradiction. Thus, $N_{1}(y_{u})=N$. Therefore, for
every $z\in N$, $T_{z}$ intersects $T_{v}$ in $R_{T}$, i.e.~$z\in N(v)$,
since $T_{x_{2}}\ll _{R_{T}}T_{v}\ll _{R_{T}}T_{y_{u}}$ and $z\in
N(x_{2})\cap N(y_{u})$. Thus, $N_{1}(v)=N$. This completes the proof of the
lemma.
\end{proof}

\begin{lemma}
\label{N-C2}For every vertex $v\in C_{2}$, $N_{1}(v)=N$.
\end{lemma}

\begin{proof}
Recall first that $N_{1}(w)=N$ by Lemma~\ref{N(w)-1}. Let $v\in C_{2}$ and $%
x\in X_{1}$. Recall that $v\notin N(w)$ by definition of $\widetilde{C}_{2}$%
, and that $v\notin N(x)$ by definition of $\widetilde{\widetilde{C}}_{2}$,
and thus either $T_{v}\ll _{R_{T}}T_{x}$ or $T_{x}\ll _{R_{T}}T_{v}$. We
will first prove that $T_{x}\ll _{R_{T}}T_{v}$. Recall that $X_{1}=%
\widetilde{X}_{1}\cup V(\mathcal{A}_{1})$.

Consider first the case where $x\in \widetilde{X}_{1}$. Note that $%
T_{x_{1}}\ll _{R_{T}}T_{v}$ for every vertex $v$ of $C_{2}$, due to the
definition of $x_{1}$, and since $v\notin N(x_{1})$ and $C_{2}\subseteq
D_{1}\cup D_{2}\setminus \{x_{1}\}$. Recall also that $\widetilde{X}_{1}$
induces a connected subgraph of $G$ and that $v\notin N[\widetilde{X}_{1}]$
for every vertex $v$ of $C_{2}$ by definition of $C_{2}$. Thus, in this case 
$T_{x}\ll _{R_{T}}T_{v}$ for every $x\in \widetilde{X}_{1}$.

Consider now the case where $x\in A_{i}$, for some $A_{i}\in \mathcal{A}_{1}$%
, where $1\leq i\leq k$. Recall that $C_{2}=\mathcal{A}_{2}\cup \mathcal{B}%
_{2}$. Suppose first that $v\in A_{j}$ for some $A_{j}\in \mathcal{B}_{2}$,
where $k+1\leq j\leq \ell $. Then, $v\in D_{2}$, since $A_{j}\subseteq D_{2}$%
, as we proved above. If $T_{v}\ll _{R_{T}}T_{x}$, then $T_{v}\ll
_{R_{T}}T_{x}\ll _{R_{T}}T_{x_{2}}$ by Lemma~\ref{N(w)-2}, which is a
contradiction by Lemma~\ref{x2-relative-position-in-S2}, since $v\in
D_{2}\subseteq S_{2}$. Thus, $T_{x}\ll _{R_{T}}T_{v}$. Suppose now that $%
v\in A_{p}$, for some $A_{p}\in \mathcal{A}_{2}$, where $1\leq p\leq k$. For
the sake of contradiction, suppose that $T_{v}\ll _{R_{T}}T_{x}$, i.e.~$%
T_{v}\ll _{R_{T}}T_{x}\ll _{R_{T}}T_{x_{2}}$. Thus, since $x\in A_{i}$ and $%
A_{i}\neq A_{p}$, it follows that $T_{v}\ll _{R_{T}}T_{y}\ll
_{R_{T}}T_{x_{2}}$ for every $y\in A_{i}$. Recall by definition of $\mathcal{%
A}_{2}$ that $v$ is adjacent to all vertices $v^{\prime }\in \widetilde{H}$.
Thus, since $v^{\prime }\in N(v)\cap N(x_{2})$ for every $v^{\prime }\in 
\widetilde{H}$, it follows that $T_{v^{\prime }}$ intersects $T_{y}$ in $%
R_{T}$, i.e.~$y\in N(v^{\prime })$, for every $y\in A_{i}$ and every $%
v^{\prime }\in \widetilde{H}$. This is a contradiction by the definition of $%
\mathcal{A}_{1}$, and thus again $T_{x}\ll _{R_{T}}T_{v}$.

Summarizing, $T_{x}\ll _{R_{T}}T_{v}$ for every $v\in C_{2}$ and every $x\in
X_{1}$. Since $v\in V_{0}(u)$ for every $v\in C_{2}$ by Lemma~\ref{N(w)-1},
it follows that $T_{v}\ll _{R_{T}}T_{u}$. Since $v\notin N(w)$ by definition
of $C_{2}$, $T_{v}$ does not intersect $T_{w}$ in $R_{T}$, i.e.~either $%
T_{w}\ll _{R_{T}}T_{v}$ or $T_{v}\ll _{R_{T}}T_{w}$. If $T_{w}\ll
_{R_{T}}T_{v}$, then $T_{w}\ll _{R_{T}}T_{v}\ll _{R_{T}}T_{u}$, and thus $%
w\notin N(u)$, which is a contradiction. Therefore $T_{v}\ll _{R_{T}}T_{w}$,
and thus $T_{x}\ll _{R_{T}}T_{v}\ll _{R_{T}}T_{w}$ for every $x\in X_{1}$.
Consider now a vertex $z\in N=N(X_{1})$. Then, $z\in N(x)\cap N(w)$ for some 
$x\in X_{1}$, since $N_{1}(w)=N=N(X_{1})$ by Lemma~\ref{N(w)-1}. Therefore, $%
T_{z}$ intersects $T_{v}$ in $R_{T}$, i.e.~$z\in N(v)$, since $T_{x}\ll
_{R_{T}}T_{v}\ll _{R_{T}}T_{w}$. Since this holds for every $z\in N$, it
follows that $N_{1}(v)=N$. This completes the proof of the lemma.
\end{proof}

\subsection*{The recursive definition of the vertex subsets $H_{i}$, $i\geq 1 $, of $H$}

In the following, we define a partition of the set $H$ into the subsets $%
H_{1},H_{2},\ldots $.

\begin{definition}
\label{Hi}Denote $H_{0}=N$. Then, $H_{i}=\{x\in H\setminus
\bigcup\nolimits_{j=1}^{i-1}H_{j}\ |\ H_{i-1}\nsubseteq N(x)\}$ for every $%
i\geq 1$.
\end{definition}

It is now easy to see by Definition~\ref{Hi} that either $H_{i}=\emptyset $
for every $i\in \mathbb{N}$, or there exists some $p\in \mathbb{N}$, 
such that $H_{p}\neq \emptyset $ and $H_{i}=\emptyset $ for every $i>p$.
That is, either $\bigcup\nolimits_{i=1}^{\infty }H_{i}=\emptyset $, or $%
\bigcup\nolimits_{i=1}^{\infty }H_{i}=\bigcup\nolimits_{i=1}^{p}H_{i}$, for
some $p\in \mathbb{N}$. Furthermore, $\bigcup\nolimits_{i=1}^{\infty }H_{i}\subseteq H$ by
Definition~\ref{Hi}.

\begin{definition}
\label{Hi-chain}Let $v_{i}\in H_{i}$, for some $i\geq 1$. Then, a sequence $%
(v_{0},v_{1},\ldots ,v_{i-1},v_{i})$ of vertices, such that $v_{j}\in H_{j}$%
, $j=0,1,\ldots ,i-1$, and $v_{j-1}v_{j}\notin E$, $j=1,2,\ldots ,i$, is an $%
H_{i}$\emph{-chain} of $v_{i}$.
\end{definition}

It easy to see by Definition~\ref{Hi} that for every set $H_{i}\neq
\emptyset $, $i\geq 1$, and for every vertex $v_{i}\in H_{i}$, there exists
at least one $H_{i}$-chain of $v_{i}$. The next two lemmas will be used in
the sequel of the proof of Theorem~\ref{no-property-thm}.

\begin{lemma}
\label{alternating-bounded-chain-1}Let $v_{1}\in H_{1}$ and $(v_{0},v_{1})$
be an $H_{1}$-chain of $v_{1}$. Then, $v_{1}$ is a bounded vertex, $%
P_{v_{0}}\ll _{R}P_{v_{1}}$ and $T_{v_{0}}\ll _{R_{T}}T_{v_{1}}$.
\end{lemma}

\begin{proof}
First, we will prove that $v_{1}$ is a bounded vertex in $R$. Suppose
otherwise that $v_{1}$ is unbounded, and thus $v_{1}\notin N(u)$. Suppose
that $P_{v_{1}}$ intersects $P_{u}$ in $R$. Then, $\phi _{v_{1}}>\phi _{u}$,
since $\phi _{u}=\min \{\phi _{x}\ |\ x\in V_{U}\}$, and thus $%
N(v_{1})\subseteq N(u)$ by Lemma~\ref{intersecting-unbounded}. Recall that $%
x_{2}\in N(v_{1})$, since $v_{1}\in H_{1}\subseteq H$, and thus also $%
x_{2}\in N(u)$. Then, $x_{2}\in N(u)$, which is a contradiction. Therefore, $%
P_{v_{1}}$ does not intersect $P_{u}$ in $R$. If $P_{v_{1}}\ll _{R}P_{u}$,
then $P_{v_{1}}\ll _{R}P_{u}\ll _{R}P_{x_{2}}$, and thus $v_{1}\notin
N(x_{2})$, which is a contradiction by definition of $H$. Therefore, $%
P_{u}\ll _{R}P_{v_{1}}$. Furthermore, $x_{2}$ is bounded and $\phi
_{x_{2}}>\phi _{v_{1}}$, since $v_{1}$ is assumed to be unbounded and $%
v_{1}\in N(x_{2})$ by definition of $H$. Recall that $T_{x}\ll
_{R_{T}}T_{x_{2}}\ll _{R_{T}}T_{u}$ for every $x\in X_{1}$ by Lemma~\ref%
{N(w)-2}. Thus, since $v_{1}\in N(x_{2})$, $v_{1}\notin N(u)$, and $%
v_{1}\notin N(x)$ for every $x\in X_{1}$, it follows that also $T_{x}\ll
_{R_{T}}T_{v_{1}}\ll _{R_{T}}T_{u}$ for every $x\in X_{1}$. Moreover, $%
N(u)\nsubseteq N(v_{1})$, since $u$ is unbounded-maximal and by Lemma~\ref%
{not-equal}. Let $y\in N(u)\setminus N(v_{1})$, and thus $y$ is bounded.
Then, $T_{v_{1}}\ll _{R_{T}}T_{y}$, since $T_{v_{1}}\ll _{R_{T}}T_{u}$, and
since $y\in N(u)$ and $y\notin N(v_{1})$. Therefore, $T_{x}\ll
_{R_{T}}T_{v_{1}}\ll _{R_{T}}T_{y}$ for every $x\in X_{1}$, and thus, in
particular $y\notin N(X_{1})$.

Suppose that $N_{1}(y)\neq N$. Then, $y\notin C_{u}$ by Lemma~\ref%
{N-Cu-unbounded}. Thus, since $u\in C_{u}$, $y\in N(u)$, and $y\notin Q_{u}$%
, it follows by Lemma~\ref{N(w)-1} that either $y\in N(X_{1})$ or $y\in
N(x_{2})$. Therefore, $y\in N(x_{2})$, since $y\notin N(X_{1})$, as we have
proved above. Let $z\in N\setminus N_{1}(y)$. Then, $z\in N(x)\cap N(x_{2})$
for some $x\in X_{1}$. Thus, since $P_{x}\ll _{R}P_{u}\ll _{R}P_{x_{2}}$, it
follows that $P_{z}$ intersects $P_{u}$ in $R$. Suppose that $z$ is
unbounded. Then, $\phi _{z}>\phi _{u}$, since $\phi _{u}=\min \{\phi _{x}\
|\ x\in V_{U}\}$, and thus $N(z)\subseteq N(u)$ by Lemma~\ref%
{intersecting-unbounded}. Then, $x_{2}\in N(u)$, which is a contradiction.
Therefore, $z$ is bounded, and thus $P_{y}$ does not intersect $P_{z}$,
since $y$ is also \ bounded and $z\notin N(y)$. That is, either $P_{y}\ll
_{R}P_{z}$ or $P_{z}\ll _{R}P_{y}$.

Suppose first that $P_{y}\ll _{R}P_{z}$. If $P_{y}\ll _{R}P_{x}$, then $%
P_{y}\ll _{R}P_{x}\ll _{R}P_{u}$, and thus $y\notin N(u)$, which is a
contradiction. If $P_{x}\ll _{R}P_{y}$, then $P_{x}\ll _{R}P_{y}\ll
_{R}P_{z} $, and thus $z\notin N(x)$, which is again a contradiction. Thus, $%
P_{y}$ intersects $P_{x}$ in $R$. Recall that $y\notin N(x)$, since $y\notin
N(X_{1})$. Thus, since $y$ is bounded, it follows that $x$ is unbounded and $%
\phi _{x}>\phi _{y}$. Then, $N(x)\subseteq N(y)$ by Lemma~\ref%
{intersecting-unbounded}, and thus $z\in N(y)$, which is a contradiction.

Suppose now that $P_{z}\ll _{R}P_{y}$. Recall that $L(y)<_{R}L(u)$ by Lemma~%
\ref{unbounded-bounded}, since $y\in N(u)$, and thus $%
R(z)<_{R}L(y)<_{R}L(u)<_{R}L(x_{2})$. Therefore, $%
r(u)<_{R}l(x_{2})<_{R}r(z)<_{R}l(y)$, since $z\in N(x_{2})$. That is, $%
L(y)<_{R}L(x_{2})$ and $l(x_{2})<_{R}l(y)$, and thus $\phi _{y}>\phi
_{x_{2}}>\phi _{v_{1}}$ (since $\phi _{x_{2}}>\phi _{v_{1}}$, as we proved
above). If $P_{y}$ intersects $P_{v_{1}}$ in $R$, then $y\in N(v_{1})$,
since $y$ is bounded, which is a contradiction. Therefore, $P_{y}$ does not
intersect $P_{v_{1}}$ in $R$, i.e.~either $P_{v_{1}}\ll _{R}P_{y}$ or $%
P_{y}\ll _{R}P_{v_{1}}$. If $P_{v_{1}}\ll _{R}P_{y}$, then $P_{u}\ll
_{R}P_{v_{1}}\ll _{R}P_{y}$, and thus $y\notin N(u)$, which is a
contradiction. Therefore, $P_{y}\ll _{R}P_{v_{1}}$.

Summarizing, $P_{z}\ll _{R}P_{y}\ll _{R}P_{v_{1}}$, and thus $%
r(z)<_{R}r(y)<_{R}r(v_{1})$. Recall that $v_{1}\in N[u,w]=N(u)\cup N(w)$ by
definition of $H$. Therefore, $v_{1}\in N(w)$, since $v_{1}\notin N(u)$, and
thus $r(v_{1})<_{R}r(w)$ by Lemma~\ref{unbounded-bounded}. Recall that $%
r(w)<_{R}l(x_{2})$, since $P_{w}\ll _{R}P_{x_{2}}$. That is, $%
r(z)<_{R}r(y)<_{R}r(v_{1})<_{R}r(w)<_{R}l(x_{2})$, i.e.~$r(z)<_{R}l(x_{2})$.
On the other hand, $R(z)<_{R}L(y)$, since $P_{z}\ll _{R}P_{y}$. Furthermore, 
$L(y)<_{R}L(u)$ by Lemma~\ref{unbounded-bounded} and since $y\in N(u)$, and $%
L(u)<_{R}L(x_{2})$, since $P_{u}\ll _{R}P_{x_{2}}$. That is, $%
R(z)<_{R}L(y)<_{R}L(u)<_{R}L(x_{2})$, i.e.~$R(z)<_{R}L(x_{2})$. Therefore,
since also $r(z)<_{R}l(x_{2})$, it follows that $P_{z}\ll _{R}P_{x_{2}}$.
This is a contradiction, since $z\in N=N_{1}(x_{2})$ by Lemma~\ref{N-Cu}.
Therefore, $N_{1}(y)=N$.

Since $N_{1}(y)=N$, and since $T_{x}\ll _{R_{T}}T_{v_{1}}\ll _{R_{T}}T_{y}$
for every $x\in X_{1}$, it follows that $T_{z}$ intersects $T_{v_{1}}$ in $%
R_{T}$, i.e.~$z\in N(v_{1})$, for every $z\in N$. Thus $N_{1}(v_{1})=N$,
i.e.~$N=H_{0}\subseteq N(v_{1})$, which is a contradiction by Definition~\ref%
{Hi}, since $v_{1}\in H_{1}$. Therefore, $v_{1}$ is a bounded vertex in $R$.

Recall now that $v_{0}\in N(x_{0})\cap N(x_{2})$ for some $x_{0}\in X_{1}$,
since $v_{0}\in N=N_{1}(x_{2})$ by Lemma~\ref{N-Cu}. Furthermore, $%
v_{1}\notin N(x_{0})$ by definition of $H$, since otherwise $v_{1}\in
N(X_{1})$, which is a contradiction. Suppose that $P_{v_{1}}$ intersects $%
P_{x_{0}}$ in $R$. If $\phi _{v_{1}}>\phi _{x_{0}}$, then $v_{1}\notin
N(x_{0})$, since $v_{1}$ is bounded, which is a contradiction. Thus, $\phi
_{v_{1}}<\phi _{x_{0}}$. Then, $N(x_{0})\subseteq N(v_{1})$ by Lemma~\ref%
{intersecting-unbounded}, and thus $v_{0}\in N(v_{1})$, which is a
contradiction. Therefore, $P_{v_{1}}$ does not intersect $P_{x_{0}}$ in $R$.
If $P_{v_{1}}\ll _{R}P_{x_{0}}$, then $P_{v_{1}}\ll _{R}P_{x_{0}}\ll
_{R}P_{u}\ll _{R}P_{x_{2}}$, and thus $v_{1}\notin N(x_{2})$, which is a
contradiction. Thus, $P_{x_{0}}\ll _{R}P_{v_{1}}$.

Furthermore, $P_{v_{0}}$ intersects $P_{u}$ in $R$, since $P_{x_{0}}\ll
_{R}P_{u}\ll _{R}P_{x_{2}}$ and $v_{0}\in N(x_{0})\cap N(x_{2})$. If $v_{0}$
is unbounded, then $\phi _{v_{0}}>\phi _{u}$, since $\phi _{u}=\min \{\phi
_{x}\ |\ x\in V_{U}\}$, and thus $N(v_{0})\subseteq N(u)$ by Lemma~\ref%
{intersecting-unbounded}. Then, $x_{2}\in N(u)$, which is a contradiction.
Therefore, $v_{0}$ is bounded, and thus $P_{v_{0}}$ does not intersect $%
P_{v_{1}}$ in $R$, since $v_{1}$ is also bounded and $v_{0}\notin N(v_{1})$.
That is, either $P_{v_{1}}\ll _{R}P_{v_{0}}$ or $P_{v_{0}}\ll _{R}P_{v_{1}}$%
. If $P_{v_{1}}\ll _{R}P_{v_{0}}$, then $P_{x_{0}}\ll _{R}P_{v_{1}}\ll
_{R}P_{v_{0}}$, and thus $v_{0}\notin N(x_{0})$, which is a contradiction.
Thus, $P_{v_{0}}\ll _{R}P_{v_{1}}$.

Finally, recall that $T_{x}\ll _{R_{T}}T_{x_{2}}$ for every $x\in X_{1}$ by
Lemma~\ref{N(w)-2}. Therefore, $T_{x}\ll _{R_{T}}T_{v_{1}}$ for every $x\in
X_{1}$, since $v_{1}\in N(x_{2})$ and $v_{1}\notin N(x)$ for every $x\in
X_{1}$. Moreover, $T_{v_{1}}$ does not intersect $T_{v_{0}}$ in $R_{T}$,
since $v_{0}\notin N(v_{1})$. Thus, either $T_{v_{1}}\ll _{R_{T}}T_{v_{0}}$
or $T_{v_{0}}\ll _{R_{T}}T_{v_{1}}$. If $T_{v_{1}}\ll _{R_{T}}T_{v_{0}}$,
then $T_{x}\ll _{R_{T}}T_{v_{1}}\ll _{R_{T}}T_{v_{0}}$ for every $x\in X_{1}$%
, and thus $v_{0}\notin N=N(X_{1})$, which is a contradiction. Thus, $%
T_{v_{1}}\ll _{R_{T}}T_{v_{0}}$. This completes the proof of the lemma.
\end{proof}

\begin{lemma}
\label{alternating-bounded-chain}Let $v_{i}\in H_{i}$, for some $i\geq 2$,
and $(v_{0},v_{1},\ldots ,v_{i})$ be an $H_{i}$-chain of $v_{i}$. Then, for
every $j=1,2,\ldots ,i-1$,

\begin{enumerate}
\item $P_{v_{j-1}}\ll _{R}P_{v_{j}}$ and $T_{v_{j-1}}\ll _{R_{T}}T_{v_{j}}$,
if $j$ is odd,

\item $P_{v_{j}}\ll _{R}P_{v_{j-1}}$ and $T_{v_{j}}\ll _{R_{T}}T_{v_{j-1}}$,
if $j$ is even.
\end{enumerate}
\end{lemma}

\begin{proof}
The proof will be done by induction on $j$. For $j=1$, the induction basis
follows by Lemma~\ref{alternating-bounded-chain-1}. For the induction step,
let $2\leq j<i-1$. Note that $v_{j-2}\in N(v_{j})\setminus N(v_{j-1})$ and $%
v_{j+1}\in N(v_{j-1})\setminus N(v_{j})$. Therefore, $N(v_{j})\nsubseteq
N(v_{j-1})$ and $N(v_{j-1})\nsubseteq N(v_{j})$, and thus $P_{v_{j}}$ does
not intersect $P_{v_{j-1}}$ in $R$ by Lemma~\ref{intersecting-unbounded},
since $v_{j-1}v_{j}\notin E$. Thus, either $P_{v_{j-1}}\ll _{R}P_{v_{j}}$ or 
$P_{v_{j}}\ll _{R}P_{v_{j-1}}$. Furthermore, either $T_{v_{j-1}}\ll
_{R_{T}}T_{v_{j}}$ or $T_{v_{j}}\ll _{R_{T}}T_{v_{j-1}}$, since $%
v_{j-1}v_{j}\notin E$.

Let $j$ be odd, i.e.~$j-1$ is even, and suppose by induction hypothesis that 
$P_{v_{j-1}}\ll _{R}P_{v_{j-2}}$ and $T_{v_{j-1}}\ll _{R_{T}}T_{v_{j-2}}$.
If $P_{v_{j}}\ll _{R}P_{v_{j-1}}$ (resp.~$T_{v_{j}}\ll _{R_{T}}T_{v_{j-1}}$%
), then $P_{v_{j}}\ll _{R}P_{v_{j-2}}$ (resp.~$T_{v_{j}}\ll
_{R_{T}}T_{v_{j-2}}$). Thus, $v_{j}v_{j-2}\notin E$, i.e.~$v_{j}\in H_{j-1}$
by Definition~\ref{Hi}, which is a contradiction. Therefore, $P_{v_{j-1}}\ll
_{R}P_{v_{j}}$ and $T_{v_{j-1}}\ll _{R_{T}}T_{v_{j}}$, if $j$ is odd.

Let now $j$ be even, i.e.~$j-1$ is odd, and suppose by induction hypothesis
that $P_{v_{j-2}}\ll _{R}P_{v_{j-1}}$ and $T_{v_{j-2}}\ll
_{R_{T}}T_{v_{j-1}} $. If $P_{v_{j-1}}\ll _{R}P_{v_{j}}$ (resp.~$%
T_{v_{j-1}}\ll _{R_{T}}T_{v_{j}} $), then $P_{v_{j-2}}\ll _{R}P_{v_{j}}$
(resp.~$T_{v_{j-2}}\ll _{R_{T}}T_{v_{j}}$), and thus $v_{j}v_{j-2}\notin E$,
which is again a contradiction. Therefore, $P_{v_{j}}\ll _{R}P_{v_{j-1}}$
and $T_{v_{j}}\ll _{R_{T}}T_{v_{j-1}}$, if $j$ is even. This completes the
induction step, and thus the lemma follows.
\end{proof}

\medskip

The next lemma, which follows now easily by Lemmas~\ref{N-Cu-unbounded},~\ref%
{N-C2},~\ref{alternating-bounded-chain-1}, and~\ref%
{alternating-bounded-chain}, will be mainly used in the sequel.

\begin{lemma}
\label{N-H-C2-Cu-bounded}All vertices of $N\cup H\cup C_{2}\cup
C_{u}\setminus \{u\}$ are bounded.
\end{lemma}

\begin{proof}
Consider first a vertex $v\in N$. Then, $v\in N(x)\cap N(x_{2})$ for some $%
x\in X_{1}$ by Lemma~\ref{N-C2}. Thus, $P_{v}$ intersects $P_{u}$ in $R$,
since $P_{x}\ll _{R}P_{u}\ll _{R}P_{x_{2}}$. Suppose that $v$ is unbounded.
Then, $\phi _{v}>\phi _{u}$, since $\phi _{u}=\min \{\phi _{x}\ |\ x\in
V_{U}\}$, and thus $N(v)\subseteq N(u)$ by Lemma~\ref{intersecting-unbounded}%
. Then, $x_{2}\in N(u)$, which is a contradiction. Thus, every $v\in N$ is
bounded.

Consider now a vertex $v\in H$. If $v\in H_{1}$, then $v$ is bounded by
Lemma~\ref{alternating-bounded-chain-1}. Suppose that $v\in H\setminus H_{1}$
and that $v$ is unbounded. Then, $\phi _{v}>\phi _{u}$, since $\phi
_{u}=\min \{\phi _{x}\ |\ x\in V_{U}\}$. Furthermore, $H_{0}=N\subseteq N(v)$
by Definition~\ref{Hi}, and thus $N_{1}(v)=N$. If $P_{v}\ll _{R}P_{u}$, then 
$P_{v}\ll _{R}P_{u}\ll _{R}P_{x_{2}}$, and thus $v\notin N(x_{2})$, which is
a contradiction to the definition of $H$. If $P_{v}$ intersects $P_{u}$ in $%
R $, then $N(v)\subseteq N(u)$ by Lemma~\ref{intersecting-unbounded}, since $%
\phi _{v}>\phi _{u}$, and thus $x_{2}\in N(u)$, which is again a
contradiction. Therefore, $P_{u}\ll _{R}P_{v}$, i.e.~$P_{x}\ll _{R}P_{u}\ll
_{R}P_{v}$ for every $x\in X_{1}$, and thus $P_{z}$ intersects $P_{u}$ in $R$
for every $z\in N_{1}(v)=N=N(X_{1})$. However, $z$ is bounded and $\phi
_{z}>\phi _{v}>\phi _{u}$ for every $z\in N_{1}(v)$, since $v$ is unbounded.
Therefore, $N_{1}(v)\subseteq N(u)$, and thus $N_{1}(u)=N$, which is a
contradiction by Lemma~\ref{N-Cu}. Thus, every $v\in H\setminus H_{1}$ is
bounded.

Consider finally a vertex $v\in C_{2}\cup C_{u}\setminus \{u\}$ and suppose
that $v$ is unbounded. Then, similarly to the above, $\phi _{v}>\phi _{u}$,
since $\phi _{u}=\min \{\phi _{x}\ |\ x\in V_{U}\}$. Furthermore, $%
N_{1}(v)=N $ by Lemmas~\ref{N-Cu-unbounded} and~\ref{N-C2}, while also $%
N_{1}(x_{2})=N$ by Lemma~\ref{N-Cu}. Suppose that $P_{v}\ll _{R}P_{u}$, i.e.~%
$P_{v}\ll _{R}P_{u}\ll _{R}P_{x_{2}}$. Then, since $N_{1}(v)=N_{1}(x_{2})=N$%
, $P_{z}$ intersects $P_{u}$ in $R$ for every $z\in N$. Furthermore, $z$ is
bounded and $\phi _{z}>\phi _{v}>\phi _{u}$ for every $z\in N_{1}(v)$, since 
$v$ is unbounded. Therefore, $N_{1}(v)\subseteq N(u)$, and thus $N_{1}(u)=N$%
, which is a contradiction by Lemma~\ref{N-Cu}. Suppose that $P_{v}$
intersects $P_{u}$ in $R$. Then, $N(v)\subseteq N(u)$ by Lemma~\ref%
{intersecting-unbounded}, since $\phi _{v}>\phi _{u}$. Therefore, $%
N(v)\subset N(u)$ by Lemma~\ref{not-equal}, and thus $v\in Q_{u}$, which is
a contradiction to the definitions of $C_{u}$ and $C_{2}$. Suppose that $%
P_{u}\ll _{R}P_{v}$, i.e.~$P_{x}\ll _{R}P_{u}\ll _{R}P_{v}$ for every $x\in
X_{1}$. Then, since $N_{1}(v)=N=N(X_{1})$, $P_{z}$ intersects $P_{u}$ in $R$
for every $z\in N$. Furthermore, $z$ is bounded and $\phi _{z}>\phi
_{v}>\phi _{u}$ for every $z\in N_{1}(v)$, since $v$ is unbounded.
Therefore, $N_{1}(v)\subseteq N(u)$, and thus $N_{1}(u)=N$, which is a
contradiction by Lemma~\ref{N-Cu}. Thus, every $v\in C_{2}\cup
C_{u}\setminus \{u\}$ is bounded. This completes the lemma.
\end{proof}

\begin{lemma}
\label{Cu-H}For every vertex $v\in C_{u}\setminus \{u\}$, it holds $%
H_{i}\subseteq N(v)$ for every $i\geq 1$.
\end{lemma}

\begin{proof}
Let $v$ be a vertex of $C_{u}\setminus \{u\}$. Recall that $N_{1}(v)=N$ by
Lemma~\ref{N-Cu-unbounded}. Consider first the case where $v\in
N[u,w]=N(u)\cup N(w)$. The proof will be done by induction on $i$. For $i=1$%
, consider a vertex $v_{1}\in H_{1}$ and an $H_{1}$-chain $(v_{0},v_{1})$ of 
$v_{1}$, where $v_{0}\in H_{0}=N=N(X_{1})$. Since $v_{0}v_{1}\notin E$,
either $T_{v_{1}}\ll _{R_{T}}T_{v_{0}}$ or $T_{v_{0}}\ll _{R_{T}}T_{v_{1}}$.
Suppose that $T_{v_{1}}\ll _{R_{T}}T_{v_{0}}$. Then, since $T_{x}\ll
_{R_{T}}T_{x_{2}}$ for every $x\in X_{1}$ by Lemma~\ref{N(w)-2}, and since $%
v_{1}\in N(x_{2})\setminus N(x)$ for every $x\in X_{1}$ by definition of $H$%
, it follows that $T_{x}\ll _{R_{T}}T_{v_{1}}$ for every $x\in X_{1}$. That
is, $T_{x}\ll _{R_{T}}T_{v_{1}}\ll _{R_{T}}T_{v_{0}}$ for every $x\in X_{1}$%
, and thus $v_{0}\notin N(x)$ for every $x\in X_{1}$, which is a
contradiction. Thus, $T_{v_{0}}\ll _{R_{T}}T_{v_{1}}$. Furthermore, $%
T_{x_{2}}\ll _{R_{T}}T_{v}$, since $T_{x_{2}}\ll _{R_{T}}T_{u}$ and $C_{u}$
is connected. Suppose that $v_{1}\notin N(v)$. Then, $T_{v_{1}}\ll
_{R_{T}}T_{v}$, since $T_{x_{2}}\ll _{R_{T}}T_{v}$ and $v_{1}\in
N(x_{2})\setminus N(v)$. That is, $T_{v_{0}}\ll _{R_{T}}T_{v_{1}}\ll
_{R_{T}}T_{v}$, and thus $v_{0}\notin N(v)$, which is a contradiction, since 
$N_{1}(v)=N$ and $v_{0}\in N$. Thus, $v_{1}\in N(v)$ for every $v_{1}\in
H_{1}$. This proves the induction basis.

For the induction step, let $i\geq 2$, and suppose that $v^{\prime }\in N(v)$
for every $v^{\prime }\in H_{j}$, where $0\leq j\leq i-1$. Let $v_{i}\in
H_{i}$ and $(v_{0},v_{1},\ldots ,v_{i-2},v_{i-1},v_{i})$ be an $H_{i}$-chain
of $v_{i}$. Note that $v_{i-2}$ exists, since $i\geq 2$, and thus $%
v_{i-1}v_{i-2}\notin E\ $and $v_{i}v_{i-2}\in E$ by Definition~\ref{Hi}. For
the sake of contradiction, suppose that $v_{i}\notin N(v)$. We will now
prove that $P_{v}\ll _{R}P_{x_{2}}$. Otherwise, suppose first that $%
P_{x_{2}}\ll _{R}P_{v}$. Then, $P_{u}\ll _{R}P_{x_{2}}\ll _{R}P_{v}$ and $%
P_{w}\ll _{R}P_{x_{2}}\ll _{R}P_{v}$, and thus $v\notin N[u,w]=N(u)\cup N(w)$%
, which is a contradiction to the assumption on $v$. Suppose now that $P_{v}$
intersects $P_{x_{2}}$ in $R$. Then, either $N(x_{2})\subseteq N(v)$ or $%
N(v)\subseteq N(x_{2})$ by Lemma~\ref{intersecting-unbounded}, since $%
v\notin N(x_{2})$ by the definition of $C_{u}$. If $N(x_{2})\subseteq N(v)$,
then $v_{i}\in N(v)$, since $v_{i}\in N(x_{2})$, which is a contradiction.
Let $N(v)\subseteq N(x_{2})$. Then, since $C_{u}$ is connected and $v\neq u$%
, $v$ is adjacent to at least one vertex $z\in C_{u}$, and thus $z\in
N(x_{2})$, which is a contradiction to the definition of $C_{u}$. Thus, $%
P_{v}\ll _{R}P_{x_{2}}$.

Recall that $v_{i-1}\in N(v)$ by the induction hypothesis. Since $v\in
N(v_{i-1})\setminus N(v_{i})$ and $v_{i-2}\in N(v_{i})\setminus N(v_{i-1})$,
it follows that $P_{v_{i}}$ does not intersect $P_{v_{i-1}}$ in $R$ by Lemma %
\ref{intersecting-unbounded}. Similarly, $P_{v_{i}}$ does not intersect $%
P_{v}$ in $R$, since $x_{2}\in N(v_{i})\setminus N(v)$ and $v_{i-1}\in
N(v)\setminus N(v_{i})$. Thus, since $v_{i-1}\in N(v)$, either $P_{v_{i}}\ll
_{R}P_{v_{i-1}}$ and $P_{v_{i}}\ll _{R}P_{v}$, or $P_{v_{i-1}}\ll
_{R}P_{v_{i}}$ and $P_{v}\ll _{R}P_{v_{i}}$. Suppose that $P_{v_{i}}\ll
_{R}P_{v_{i-1}}$ and $P_{v_{i}}\ll _{R}P_{v}$. Then, $P_{v_{i}}\ll
_{R}P_{v}\ll _{R}P_{x_{2}}$, and thus $v_{i}\notin N(x_{2})$, which is a
contradiction.

Thus, $P_{v_{i-1}}\ll _{R}P_{v_{i}}$ and $P_{v}\ll _{R}P_{v_{i}}$. Recall
now by Lemmas~\ref{alternating-bounded-chain-1} and~\ref%
{alternating-bounded-chain} that either $P_{v_{i-2}}\ll _{R}P_{v_{i-1}}$ or $%
P_{v_{i-1}}\ll _{R}P_{v_{i-2}}$. If $P_{v_{i-2}}\ll _{R}P_{v_{i-1}}$, then $%
P_{v_{i-2}}\ll _{R}P_{v_{i-1}}\ll _{R}P_{v_{i}}$, and thus~$%
v_{i}v_{i-2}\notin E$, which is a contradiction. Therefore, $P_{v_{i-1}}\ll
_{R}P_{v_{i-2}}$. Thus, also $T_{v_{i-1}}\ll _{R_{T}}T_{v_{i-2}}$ and $i$ is
odd, by Lemmas~\ref{alternating-bounded-chain-1} and~\ref%
{alternating-bounded-chain}. Since $v_{i-1}v_{i}\notin E$, either $%
T_{v_{i}}\ll _{R_{T}}T_{v_{i-1}}$ or $T_{v_{i-1}}\ll _{R_{T}}T_{v_{i}}$. If $%
T_{v_{i}}\ll _{R_{T}}T_{v_{i-1}}$, then $T_{v_{i}}\ll _{R_{T}}T_{v_{i-1}}\ll
_{R_{T}}T_{v_{i-2}}$, and thus~$v_{i}v_{i-2}\notin E$, which is a
contradiction. Therefore, $T_{v_{i-1}}\ll _{R_{T}}T_{v_{i}}$, and thus $%
T_{v}\ll _{R_{T}}T_{v_{i}}$, since $v\in N(v_{i-1})\setminus N(v_{i})$.
Recall also that $T_{x_{2}}\ll _{R_{T}}T_{v}$, since $T_{x_{2}}\ll
_{R_{T}}T_{u}$ and $C_{u}$ is connected. That is, $T_{x_{2}}\ll
_{R_{T}}T_{v}\ll _{R_{T}}T_{v_{i}}$, and thus $v_{i}\notin N(x_{2})$, which
is a contradiction. Thus, $v_{i}\in N(v)$. This completes the induction step.

Summarizing, we have proved that $H_{i}\subseteq N(v)$ for every $i\geq 1$
and for every vertex $v\in C_{u}\setminus \{u\}$, such that $v\in N[u,w]$.
This holds in particular for $w$, i.e.~$H_{i}\subseteq N(w)$ for every $%
i\geq 1$, since $w\in N(u)$ is a vertex of $C_{u}\setminus \{u\}$. Consider
now the case where $v\notin N[u,w]$. Then, since $w\in N(u)$, either $%
T_{u}\ll _{R_{T}}T_{v}$ and $T_{w}\ll _{R_{T}}T_{v}$, or $T_{v}\ll
_{R_{T}}T_{u}$ and $T_{v}\ll _{R_{T}}T_{w}$. Suppose that $T_{u}\ll
_{R_{T}}T_{v}$, i.e.~$T_{x}\ll _{R_{T}}T_{x_{2}}\ll _{R_{T}}T_{u}\ll
_{R_{T}}T_{v}$ for every $x\in X_{1}$ by Lemma~\ref{N(w)-2}. Recall that $%
N_{1}(v)=N$ by Lemma~\ref{N-Cu-unbounded}. That is, $T_{z}$ intersects $%
T_{u} $ in $R_{T}$, i.e.~$z\in N(u)$, for every $z\in N_{1}(v)=N$, and thus $%
N_{1}(u)=N$, which is a contradiction by Lemma~\ref{N-Cu}. Thus, $T_{v}\ll
_{R_{T}}T_{u}$ and $T_{v}\ll _{R_{T}}T_{w}$. Then, $T_{x_{2}}\ll
_{R_{T}}T_{v}$, since $T_{x_{2}}\ll _{R_{T}}T_{u}$ and $C_{u}$ is connected.
That is, $T_{x_{2}}\ll _{R_{T}}T_{v}\ll _{R_{T}}T_{w}$. Then, since every $%
z\in H_{i}$, $i\geq 1$, is adjacent to both $x_{2}$ and $w$, it follows that 
$T_{z}$ intersects $T_{v}$ in $R_{T}$, i.e.~$z\in N(v)$, for every $z\in
H_{i}$, where $i\geq 1$. Thus, $H_{i}\subseteq N(v)$ for every $i\geq 1$ and
for every vertex $v\in C_{u}\setminus \{u\}$, such that $v\notin N[u,w]$.
This completes the proof of the lemma.
\end{proof}

\begin{lemma}
\label{C2-H}For every vertex $v\in C_{2}$, it holds $H_{i}\subseteq N(v)$
for every $i\geq 1$.
\end{lemma}

\begin{proof}
Recall that $C_{2}=\mathcal{A}_{2}\cup \mathcal{B}_{2}$, where $%
A_{j}\subseteq D_{2}$ for every $A_{j}\in \mathcal{B}_{2}$, $k+1\leq j\leq
\ell $, and $\mathcal{A}_{2}$ includes exactly those components $A_{i}$, $%
1\leq i\leq k$, for which all vertices of $A_{i}$ are adjacent to all
vertices of $\widetilde{H}$. Therefore, if $v\in A_{i}$ for some component $%
A_{i}\in \mathcal{A}_{2}$, then $H\subseteq \widetilde{H}\subseteq N(v)$ by
definition, and thus $H_{i}\subseteq N(v)$ for every $i\geq 1$.

Let now $v\in A_{j}$, for some $A_{j}\in \mathcal{B}_{2}$, and suppose first
that $v\notin N(x_{2})$. Then, since $v\in D_{2}\subseteq S_{2}\subseteq
V_{0}(u)$, it follows that $T_{v}\ll _{R_{T}}T_{u}$ and that $T_{x_{2}}\ll
_{R_{T}}T_{v}$ by Lemma~\ref{x2-relative-position-in-S2} (since $v\notin
N(x_{2})$), i.e.~$T_{x_{2}}\ll _{R_{T}}T_{v}\ll _{R_{T}}T_{u}$. Moreover, $%
v\notin N(w)$ by definition of $\widetilde{C}_{2}$. Thus, $T_{v}\ll
_{R_{T}}T_{w}$, since $T_{v}\ll _{R_{T}}T_{u}$ and $w\in N(u)\setminus N(v)$%
. That is, $T_{x_{2}}\ll _{R_{T}}T_{v}\ll _{R_{T}}T_{w}$. Let now $z\in
H_{i} $, for some $i\geq 1$. Then, $z\in N(x_{2})$ and $z\in N(w)$ by Lemma~%
\ref{Cu-H}, and thus $T_{z}$ intersects $T_{v}$ in $R_{T}$, i.e.~$z\in N(v)$%
. Therefore, $H_{i}\subseteq N(v)$ for every $i\geq 1$, where $v\notin
N(x_{2}) $.

Suppose now that $v\in N(x_{2})$. We will prove by contradiction that $%
H_{i}\subseteq N(v)$ for every $i\geq 1$. Suppose otherwise that there
exists an index $i\geq 1$, such that $v_{i}\notin N(v)$, for some vertex $%
v_{i}\in H_{i}$. W.l.o.g.~let $i$ be the smallest such index, i.e.~$%
v^{\prime }\in N(v)$ for every $v^{\prime }\in H_{j}$, where $0\leq j\leq
i-1 $ (recall that $H_{0}=N$, and thus $v^{\prime }\in N(v)$ for every $%
v^{\prime }\in H_{0}$ by Lemma~\ref{N-C2}).

Let $(v_{0},v_{1},\ldots ,v_{i-1},v_{i})$ be an $H_{i}$-chain of $v_{i}$. If 
$i=1$, then $P_{v_{1}}$ does not intersect $P_{v_{0}}$ in $R$ by Lemma~\ref%
{alternating-bounded-chain-1}. If $i\geq 2$, then $v_{i-2}\in
N(v_{i})\setminus N(v_{i-1})$ and $v\in N(v_{i-1})\setminus N(v_{i})$;
therefore $N(v_{i-1})\nsubseteq N(v_{i})$ and $N(v_{i})\nsubseteq N(v_{i-1})$%
, and thus $P_{v_{i}}$ does not intersect $P_{v_{i-1}}$ in $R$ by Lemma~\ref%
{intersecting-unbounded}. That is, $P_{v_{i}}$ does not intersect $%
P_{v_{i-1}}$ in $R$ for every $i\geq 1$. Recall now that $v_{i}\in N[u,w]$
by definition of $H$, and that $v\notin N[u,w]$ by definition of $\widetilde{%
C}_{2}$. If $v_{i}\in N(u)$ (resp.~$v_{i}\in N(w)$), then $u\in
N(v_{i})\setminus N(v)$ (resp.~$w\in N(v_{i})\setminus N(v)$). Furthermore, $%
v_{i-1}\in N(v)\setminus N(v_{i})$, i.e.~$N(v_{i})\nsubseteq N(v)$ and $%
N(v)\nsubseteq N(v_{i})$, and thus $P_{v_{i}}$ does not intersect $P_{v}$ in 
$R$ by Lemma~\ref{intersecting-unbounded}. Therefore, since $vv_{i-1}\in E$,
it follows hat either $P_{v_{i-1}}\ll _{R}P_{v_{i}}$ and $P_{v}\ll
_{R}P_{v_{i}}$, or $P_{v_{i}}\ll _{R}P_{v_{i-1}}$ and $P_{v_{i}}\ll
_{R}P_{v} $.

Suppose first that $P_{v_{i-1}}\ll _{R}P_{v_{i}}$ and $P_{v}\ll
_{R}P_{v_{i}} $. Recall that $v_{i}\in N[u,w]$ and that $v\notin N[u,w]$.
Let $v_{i}\in N(u)$ (resp.~$v_{i}\in N(w)$). Then, $P_{v}$ does not
intersect $P_{u}$ (resp.~$P_{w}$) in $R$ by Lemma~\ref%
{intersecting-unbounded}, since $x_{2}\in N(v)\setminus N[u,w]$ and $%
v_{i}\in N(u)\setminus N(v)$ (resp.~$v_{i}\in N(w)\setminus N(v)$). Thus,
since $P_{u}\ll _{R}P_{x_{2}}$ (resp.~$P_{w}\ll _{R}P_{x_{2}}$) and $v\in
N(x_{2})\setminus N(u)$ (resp.~$v\in N(x_{2})\setminus N(w)$), it follows
that $P_{u}\ll _{R}P_{v}$ (resp.~$P_{w}\ll _{R}P_{v}$). That is, $P_{u}\ll
_{R}P_{v}\ll _{R}P_{v_{i}}$ (resp.~$P_{w}\ll _{R}P_{v}\ll _{R}P_{v_{i}}$),
i.e.~$v_{i}\notin N(u)$ (resp.~$v_{i}\notin N(w)$), which is a contradiction.

Suppose now that $P_{v_{i}}\ll _{R}P_{v_{i-1}}$ and $P_{v_{i}}\ll _{R}P_{v}$%
. Then, $i\neq 1$ by Lemma~\ref{alternating-bounded-chain-1}. That is, $%
i\geq 2$, i.e.~$v_{i-2}$ exists. Recall by Lemmas~\ref%
{alternating-bounded-chain-1} and~\ref{alternating-bounded-chain} that
either $P_{v_{i-1}}\ll _{R}P_{v_{i-2}}$ or $P_{v_{i-2}}\ll _{R}P_{v_{i-1}}$.
If $P_{v_{i-1}}\ll _{R}P_{v_{i-2}}$, then $P_{v_{i}}\ll _{R}P_{v_{i-1}}\ll
_{R}P_{v_{i-2}}$, and thus~$v_{i}v_{i-2}\notin E$, which is a contradiction.
Therefore, $P_{v_{i-2}}\ll _{R}P_{v_{i-1}}$, and thus also $T_{v_{i-2}}\ll
_{R_{T}}T_{v_{i-1}}$ and $i$ is even by Lemmas~\ref%
{alternating-bounded-chain-1} and~\ref{alternating-bounded-chain}. Since $%
v_{i-1}v_{i}\notin E$, either $T_{v_{i-1}}\ll _{R_{T}}T_{v_{i}}$ or $%
T_{v_{i}}\ll _{R_{T}}T_{v_{i-1}}$. If $T_{v_{i-1}}\ll _{R_{T}}T_{v_{i}}$,
then $T_{v_{i-2}}\ll _{R_{T}}T_{v_{i-1}}\ll _{R_{T}}T_{v_{i}}$, and thus~$%
v_{i}v_{i-2}\notin E$, which is a contradiction. Therefore, $T_{v_{i}}\ll
_{R_{T}}T_{v_{i-1}}$, and thus also $T_{v_{i}}\ll _{R_{T}}T_{v}$, since $%
v\in N(v_{i-1})\setminus N(v_{i})$. Recall also that $T_{x_{2}}\ll
_{R_{T}}T_{u}$ and $T_{x_{2}}\ll _{R_{T}}T_{w}$. Thus, also $T_{v}\ll
_{R_{T}}T_{u}$ and $T_{v}\ll _{R_{T}}T_{w}$, since $v\in N(x_{2})\setminus
N[u,w]$. That is, $T_{v_{i}}\ll _{R_{T}}T_{v}\ll _{R_{T}}T_{u}$ and $%
T_{v_{i}}\ll _{R_{T}}T_{v}\ll _{R_{T}}T_{w}$, i.e.~$v_{i}\notin N[u,w]$,
which is a contradiction. Thus, $H_{i}\subseteq N(v)$ for every $i\geq 1$.
This completes the proof of the lemma.
\end{proof}

\subsection*{The recursive definition of the vertex subsets $H_{i}^{\prime }$, 
$i\geq 0$, of $H$}

Similarly to Definitions~\ref{Hi} and~\ref{Hi-chain}, we partition in the
following the set $H\setminus \bigcup\nolimits_{i=1}^{\infty }H_{i}$ into
the subsets $H_{0}^{\prime },H_{1}^{\prime },\ldots $.

\begin{definition}
\label{Hi-tilde}Let $H^{\prime }=H\setminus \bigcup\nolimits_{i=1}^{\infty
}H_{i}$ and $H_{0}^{\prime }=\{x\in H^{\prime }\ |\ xv\in E$ for some $v\in
V\setminus Q_{u}\setminus N[u]\setminus V_{0}(u)\}$. Furthermore, $%
H_{i}^{\prime }=\{x\in H^{\prime }\setminus
\bigcup\nolimits_{j=0}^{i-1}H_{j}^{\prime }\ |\ H_{i-1}^{\prime }\nsubseteq
N(x)\}$ for every $i\geq 1$.
\end{definition}

It is now easy to see by Definition~\ref{Hi-tilde} that either $%
H_{i}^{\prime }=\emptyset $ for every $i\in \mathbb{N}\cup \{0\}$, 
or there exists some $p\in \mathbb{N}\cup \{0\}$, 
such that $H_{p}^{\prime }\neq \emptyset $ and $H_{i}^{\prime
}=\emptyset $ for every $i>p$. That is, either $\bigcup\nolimits_{i=0}^{%
\infty }H_{i}^{\prime }=\emptyset $, or $\bigcup\nolimits_{i=0}^{\infty
}H_{i}^{\prime }=\bigcup\nolimits_{i=0}^{p}H_{i}^{\prime }$, 
for some $p\in \mathbb{N}\cup \{0\}$, while $\bigcup\nolimits_{i=0}^{\infty }H_{i}^{\prime }\subseteq
H^{\prime }$ by Definition~\ref{Hi-tilde}. Furthermore, it is easy to
observe by Definitions~\ref{Hi} and~\ref{Hi-tilde} that every vertex of $%
H\setminus \bigcup\nolimits_{i=1}^{\infty }H_{i}\setminus
\bigcup\nolimits_{i=0}^{\infty }H_{i}^{\prime }$ is adjacent to every vertex
of $N(X_{1})\cup \bigcup\nolimits_{i=1}^{\infty }H_{i}\cup
\bigcup\nolimits_{i=0}^{\infty }H_{i}^{\prime }$, and to no vertex of $%
V\setminus Q_{u}\setminus N[u]\setminus V_{0}(u)$.

\begin{definition}
\label{Hi-chain-tilde}Let $v_{i}\in H_{i}^{\prime }$, for some $i\geq 1$.
Then, a sequence $(v_{0},v_{1},\ldots ,v_{i-1},v_{i})$ of vertices, such
that $v_{j}\in H_{j}^{\prime }$, $j=0,1,\ldots ,i-1$, and $%
v_{j-1}v_{j}\notin E$, $j=1,2,\ldots ,i$, is an $H_{i}^{\prime }$\emph{-chain%
} of $v_{i}$.
\end{definition}

It is easy to see by Definition~\ref{Hi-tilde} that for every set $%
H_{i}^{\prime }\neq \emptyset $, $i\geq 1$, and for every vertex $v_{i}\in
H_{i}^{\prime }$, there exists at least one $H_{i}^{\prime }$-chain of $%
v_{i} $. Now, similarly to Lemmas~\ref{alternating-bounded-chain-1} and~\ref%
{alternating-bounded-chain}, we state the following two lemmas.

\begin{lemma}
\label{alternating-bounded-chain-tilde-1}Let $v_{1}\in H_{1}^{\prime }$ and $%
(v_{0},v_{1})$ be an $H_{1}^{\prime }$-chain of $v_{1}$. Then, $%
v_{0},v_{1}\in N(u)$, $P_{v_{1}}\ll _{R}P_{v_{0}}$ and $T_{v_{1}}\ll
_{R_{T}}T_{v_{0}}$.
\end{lemma}

\begin{proof}
First, recall that there exists a bounded covering vertex $u^{\ast }$ of $u$
by Lemma~\ref{bounded-hovering}, and thus $w\in N(u)\subseteq N(u^{\ast })$.
Let $y\in V\setminus Q_{u}\setminus N[u]\setminus V_{0}(u)$ be a vertex,
such that $yv_{0}\in E$; such a vertex $y$ exists by Definition~\ref%
{Hi-tilde}. Then, $y\notin N(w)$, since either $P_{w}\ll _{R}P_{x_{2}}\ll
_{R}P_{y}$ or $P_{y}\ll _{R}P_{x}\ll _{R}P_{w}$ for every $x\in X_{1}$ by
Lemma~\ref{foreigner-not-inside-1}. Consider the trapezoid representation $%
R_{T}$. Then, either $T_{x_{2}}\ll _{R_{T}}T_{y}$ or $T_{y}\ll _{R_{T}}T_{x}$
for every $x\in X_{1}$ by Lemma~\ref{V0-middle}. Suppose that $T_{y}\ll
_{R_{T}}T_{x}$ for every $x\in X_{1}$, i.e.~$T_{y}\ll _{R_{T}}T_{x}\ll
_{R_{T}}T_{x_{2}}$ for every $x\in X_{1}$. Then, since $v_{0}\in N(y)$ and $%
v_{0}\in N(x_{2})$, $T_{v_{0}}$ intersects $T_{x}$ in $R_{T}$ for every $%
x\in X_{1}$, and thus $v_{0}\in N(X_{1})$. This is a contradiction, since $%
v_{0}\in H_{0}^{\prime }\subseteq H$, and since $H$ is an induced subgraph
of $G\setminus Q_{u}\setminus N[X_{1}]$. Thus, $T_{x_{2}}\ll _{R_{T}}T_{y}$.

Since $y\notin N(u)$ by the assumption on $y$, either $T_{y}\ll
_{R_{T}}T_{u} $ or $T_{u}\ll _{R_{T}}T_{y}$. Suppose that $T_{y}\ll
_{R_{T}}T_{u}$, i.e.~$T_{x_{2}}\ll _{R_{T}}T_{y}\ll _{R_{T}}T_{u}$. Then,
also $T_{x_{2}}\ll _{R_{T}}T_{y}\ll _{R_{T}}T_{w}$, since $w\in N(u)$ and $%
w\notin N(y)$. Note that $y\notin N(u^{\ast })$, since otherwise $y\in
V_{0}(u)$, which is a contradiction. Thus, since also $w\in N(u^{\ast })$,
it follows that $T_{x_{2}}\ll _{R_{T}}T_{y}\ll _{R_{T}}T_{u^{\ast }}$. Then,
since $x_{2},u^{\ast }\in V_{0}(u)$, and since $V_{0}(u)$ is connected, $%
T_{y}$ intersects $T_{z}$ for some $z\in V_{0}(u)$, and thus $y\in V_{0}(u)$%
, which is a contradiction. Therefore, $T_{u}\ll _{R_{T}}T_{y}$, i.e.~$%
T_{x_{2}}\ll _{R_{T}}T_{u}\ll _{R_{T}}T_{y}$. Thus, since $v_{0}\in N(x_{2})$
and $v_{0}\in N(y)$, $T_{v_{0}}$ intersects $T_{u}$ in $R_{T}$, i.e.~$%
v_{0}\in N(u)$; in particular, $v_{0}$ is bounded.

Since $v_{1}v_{0}\notin E$, either $T_{v_{0}}\ll _{R_{T}}T_{v_{1}}$ or $%
T_{v_{1}}\ll _{R_{T}}T_{v_{0}}$. Suppose that $T_{v_{0}}\ll
_{R_{T}}T_{v_{1}} $. Recall that $yv_{1}\notin E$ by Definition~\ref%
{Hi-tilde}, since $y\in V\setminus Q_{u}\setminus N[u]\setminus V_{0}(u)$
and $v_{1}\in H_{1}$. That is, either $T_{v_{1}}\ll _{R_{T}}T_{y}$ or $%
T_{y}\ll _{R_{T}}T_{v_{1}}$. If $T_{v_{1}}\ll _{R_{T}}T_{y}$, then $%
T_{v_{0}}\ll _{R_{T}}T_{v_{1}}\ll _{R_{T}}T_{y}$, i.e.~$yv_{0}\notin E$,
which is a contradiction. If $T_{y}\ll _{R_{T}}T_{v_{1}}$, then $%
T_{x_{2}}\ll _{R_{T}}T_{y}\ll _{R_{T}}T_{v_{1}}$, i.e.~$v_{1}\notin N(x_{2})$%
, which is a contradiction. Thus, $T_{v_{1}}\ll _{R_{T}}T_{v_{0}}$.

Consider now the projection representation $R$, and recall that $%
v_{1}v_{0},v_{1}y\notin E$. Furthermore, recall that $v_{0}\notin N(%
\widetilde{X}_{1})$ by definition of $H$, and that either $P_{x_{2}}\ll
_{R}P_{y}$ or $P_{y}\ll _{R}P_{x}$ for every $x\in X_{1}$ by Lemma~\ref%
{foreigner-not-inside-1}. Suppose that $P_{y}\ll _{R}P_{x}$ for every $x\in
X_{1}$, and thus $P_{y}\ll _{R}P_{x}\ll _{R}P_{u}\ll _{R}P_{x_{2}}$ for
every $x\in \widetilde{X}_{1}\subseteq X_{1}$. Then, $P_{v_{0}}$ intersects $%
P_{x}$ in $R$ for every $x\in \widetilde{X}_{1}$, since $v_{0}\in N(y)\cap
N(x_{2})$. Furthermore, $v_{0}x\notin E$ for every $x\in \widetilde{X}_{1}$,
since $v_{0}\notin N(\widetilde{X}_{1})$. Thus, every $x\in \widetilde{X}%
_{1} $ is unbounded and $\phi _{x}>\phi _{v_{0}}>\phi _{u}$, since $v_{0}$
is bounded and $v_{0}\in N(u)$, as we proved above. Moreover, since $%
\widetilde{X}_{1}$ is connected, and since no two unbounded vertices are
adjacent, it follows that $\widetilde{X}_{1}$ has one vertex, i.e.~$%
\widetilde{X}_{1}=\{x_{1}\}$. Thus, $N(x_{1})=N(\widetilde{X}_{1})\subseteq
N(x_{2})$ by Lemma~\ref{N-Cu}, since $\widetilde{X}_{1}\subseteq X_{1}$.
Therefore, $P_{z} $ intersects $P_{u}$ in $R$, for every $z\in N(x_{1})$,
since $P_{x_{1}}\ll _{R}P_{u}\ll _{R}P_{x_{2}}$. Furthermore, $z$ is bounded
and $\phi _{z}>\phi _{x_{1}}>\phi _{u}$ for every $z\in N(x_{1})$, since $%
x_{1}$ is unbounded. That is, $z\in N(u)$ for every $z\in N(x_{1})$, i.e.~$%
N(x_{1})\subseteq N(u)$, and thus $x_{1}$ is an isolated vertex of $%
G\setminus N[u]$. Therefore, since $x_{1}$ is unbounded and $u^{\ast }$ is
bounded in $R$, it follows that $x_{1}$ and $u^{\ast }$ do not lie in the
same connected component of $G\setminus N[u]$. That is, $V_{0}(u)$ is not
connected, which is a contradiction. Thus, $P_{x_{2}}\ll _{R}P_{y}$, i.e.~$%
P_{u}\ll _{R}P_{x_{2}}\ll _{R}P_{y}$.

Suppose that $P_{v_{1}}$ intersects $P_{y}$ in $R$. Then, either $%
N(v_{1})\subseteq N(y)$ or $N(y)\subseteq N(v_{1})$ by Lemma~\ref%
{intersecting-unbounded}, since $v_{1}y\notin E$. If $N(v_{1})\subseteq N(y)$%
, then $x_{2}\in N(y)$, which is a contradiction, since $P_{x_{2}}\ll
_{R}P_{y}$. On the other hand, if $N(y)\subseteq N(v_{1})$, then $v_{0}\in
N(v_{1})$, since $yv_{0}\in E$, which is a contradiction. Thus, $P_{v_{1}}$
does not intersect $P_{y}$ in $R$, i.e.~either $P_{y}\ll _{R}P_{v_{1}}$ or $%
P_{v_{1}}\ll _{R}P_{y}$. If $P_{y}\ll _{R}P_{v_{1}}$, then $P_{x_{2}}\ll
_{R}P_{y}\ll _{R}P_{v_{1}}$, i.e.~$v_{1}\notin N(x_{2})$, which is a
contradiction. Thus, $P_{v_{1}}\ll _{R}P_{y}$.

Suppose that $P_{v_{1}}$ intersects $P_{v_{0}}$ in $R$. Then, $v_{1}$ is
unbounded and $\phi _{v_{1}}>\phi _{v_{0}}>\phi _{u}$, since $v_{0}$ is
bounded and $v_{0}\in N(u)$. Furthermore, note that $N_{1}(v_{1})=N$, since
otherwise $v_{1}\in H_{1}$ by Definition~\ref{Hi}, and thus $v_{1}\notin
H^{\prime }=H\setminus \bigcup\nolimits_{i=1}^{\infty }H_{i}$, which is a
contradiction. Consider now a vertex $z\in N$. Then, $z\in N(x)\cap N(x_{2})$%
, for some $x\in X_{1}$. Furthermore, $z\in N(v_{1})$, since $N_{1}(v_{1})=N$%
; thus, $z$ is bounded and $\phi _{z}>\phi _{v_{1}}>\phi _{u}$, since $v_{1}$
is unbounded. On the other hand, $P_{z}$ intersects $P_{u}$ in $R$, since $%
P_{x}\ll _{R}P_{u}\ll _{R}P_{x_{2}}$ and $z\in N(x)\cap N(x_{2})$. Thus $%
z\in N(u)$, since $z$ is bounded and $\phi _{z}>\phi _{u}$. Since this holds
for an arbitrary $z\in N$, it follows that $N_{1}(u)=N$, which is a
contradiction by Lemma~\ref{N-Cu}. Thus, $P_{v_{1}}$ does not intersect $%
P_{v_{0}}$ in $R$, i.e.~either $P_{v_{0}}\ll _{R}P_{v_{1}}$ or $P_{v_{1}}\ll
_{R}P_{v_{0}}$. If $P_{v_{0}}\ll _{R}P_{v_{1}}$, then $P_{v_{0}}\ll
_{R}P_{v_{1}}\ll _{R}P_{y}$, i.e.~$y\notin N(v_{0})$, which is a
contradiction. Thus, $P_{v_{1}}\ll _{R}P_{v_{0}}$.

Recall that $v_{0}\in N(u)$ as we have proved above, and thus $%
L(v_{0})<_{R}L(u)$ by Lemma~\ref{unbounded-bounded}. Furthermore, $%
R(v_{1})<_{R}L(v_{0})$, since $P_{v_{1}}\ll _{R}P_{v_{0}}$, and thus $%
R(v_{1})<_{R}L(u)$. On the other hand, since $v_{1}\in N(x_{2})$, and since $%
R(v_{1})<_{R}L(u)<_{R}L(x_{2})$, it follows that $l(x_{2})<_{R}r(v_{1})$,
and thus $l(u)<_{R}r(v_{1})$, since $P_{u}\ll _{R}P_{x_{2}}$. Therefore,
since also $R(v_{1})<_{R}L(u)$, $P_{v_{1}}$ intersects $P_{u}$ in $R$ and $%
\phi _{v_{1}}>\phi _{u}$. If $v_{1}\notin N(u)$, then $N(v_{1})\subseteq
N(u) $ by Lemma~\ref{intersecting-unbounded}, and thus $x_{2}\in N(u)$,
since $x_{2}\in N(v_{1})$\ by definition of $H$, which is a contradiction.
Therefore, $v_{1}\in N(u)$. This completes the proof of the lemma.
\end{proof}

\begin{lemma}
\label{alternating-bounded-chain-tilde}Let $v_{i}\in H_{i}^{\prime }$, for
some $i\geq 2$, and $(v_{0},v_{1},\ldots ,v_{i})$ be an 
$H_{i}^{\prime }$-chain of $v_{i}$. Then, for every $j=1,2,\ldots ,i-1$:

\begin{enumerate}
\item $P_{v_{j-1}}\ll _{R}P_{v_{j}}$ and $T_{v_{j-1}}\ll _{R_{T}}T_{v_{j}}$,
if $j$ is even,

\item $P_{v_{j}}\ll _{R}P_{v_{j-1}}$ and $T_{v_{j}}\ll _{R_{T}}T_{v_{j-1}}$,
if $j$ is odd.
\end{enumerate}
\end{lemma}

\begin{proof}
The proof will be done by induction on $j$. For $j=1$, the induction basis
follows by Lemma~\ref{alternating-bounded-chain-tilde-1}. For the induction
step, let $2\leq j<i$. Note that $v_{j-2}\in N(v_{j})\setminus N(v_{j-1})$
and $v_{j+1}\in N(v_{j-1})\setminus N(v_{j})$. Therefore, $%
N(v_{j})\nsubseteq N(v_{j-1})$ and $v_{j+1}\in N(v_{j-1})\nsubseteq N(v_{j})$%
, and thus $P_{v_{j}}$ does not intersect $P_{v_{j-1}}$ in $R$ by Lemma~\ref%
{intersecting-unbounded}, since $v_{j-1}v_{j}\notin E$. Thus, either $%
P_{v_{j-1}}\ll _{R}P_{v_{j}}$ or $P_{v_{j}}\ll _{R}P_{v_{j-1}}$.
Furthermore, clearly either $T_{v_{j-1}}\ll _{R_{T}}T_{v_{j}}$ or $%
T_{v_{j}}\ll _{R_{T}}T_{v_{j-1}}$, since $v_{j-1}v_{j}\notin E$.

Let $j$ be even, i.e.~$j-1$ is odd, and suppose by induction hypothesis that 
$P_{v_{j-1}}\ll _{R}P_{v_{j-2}}$ and $T_{v_{j-1}}\ll _{R_{T}}T_{v_{j-2}}$.
If $P_{v_{j}}\ll _{R}P_{v_{j-1}}$ (resp.~$T_{v_{j}}\ll _{R_{T}}T_{v_{j-1}}$%
), then $P_{v_{j}}\ll _{R}P_{v_{j-2}}$ (resp.~$T_{v_{j}}\ll
_{R_{T}}T_{v_{j-2}}$). Thus, $v_{j}v_{j-2}\notin E$, i.e.~$v_{j}\in
H_{j-1}^{\prime }$ by Definition~\ref{Hi-tilde}, which is a contradiction.
Therefore, $P_{v_{j-1}}\ll _{R}P_{v_{j}}$ and $T_{v_{j-1}}\ll
_{R_{T}}T_{v_{j}}$, if $j$ is even.

Let now $j$ be odd, i.e.~$j-1$ is even, and suppose by induction hypothesis
that $P_{v_{j-2}}\ll _{R}P_{v_{j-1}}$ and $T_{v_{j-2}}\ll
_{R_{T}}T_{v_{j-1}} $. If $P_{v_{j-1}}\ll _{R}P_{v_{j}}$ (resp.~$%
T_{v_{j-1}}\ll _{R_{T}}T_{v_{j}} $), then $P_{v_{j-2}}\ll _{R}P_{v_{j}}$
(resp.~$T_{v_{j-2}}\ll _{R_{T}}T_{v_{j}}$), and thus $v_{j}v_{j-2}\notin E$,
which is again a contradiction. Therefore, $P_{v_{j}}\ll _{R}P_{v_{j-1}}$
and $T_{v_{j}}\ll _{R_{T}}T_{v_{j-1}}$, if $j$ is odd. This completes the
induction step, and thus the lemma follows.
\end{proof}

\begin{lemma}
\label{tilde-neighbors}$H_{i}^{\prime }\subseteq N(u)$, for every $i\geq 0$.
\end{lemma}

\begin{proof}
The proof will be done by induction on $i$. For $i=0$ and $i=1$, the lemma
follows by Lemma~\ref{alternating-bounded-chain-tilde-1}. This proves the
induction basis. For the induction step, let $i\geq 2$. Suppose that $%
v_{i}\notin N(u)$, and let $(v_{0},v_{1},\ldots ,v_{i-2},v_{i-1},v_{i})$ be
an $H_{i}^{\prime }$-chain of $v_{i}$. By the induction hypothesis, $%
v_{j}\in N(u)$ for every $j=0,1,\ldots ,i-1$. Then, in particular, $%
r(u)<_{R}r(v_{i-1})$ and $L(v_{i-1})<_{R}L(u)$ by Lemma~\ref%
{unbounded-bounded}. Furthermore, $v_{i-2}\in N(v_{i})\setminus N(v_{i-1})$
and $u\in N(v_{i-1})\setminus N(v_{i})$, i.e.~$N(v_{i})\nsubseteq N(v_{i-1})$
and $N(v_{i-1})\nsubseteq N(v_{i})$, and thus Lemma~\ref%
{intersecting-unbounded} implies that $P_{v_{i}}$ does not intersect $%
P_{v_{i-1}}$ in $R$, since $v_{i}v_{i-1}\notin E$.

Suppose first that $i$ is odd. Then, $P_{v_{i-2}}\ll _{R}P_{v_{i-1}}$ by
Lemma~\ref{alternating-bounded-chain-tilde}. Thus, since $v_{i}\in
N(v_{i-2}) $, and since $P_{v_{i}}$ does not intersect $P_{v_{i-1}}$ in $R$
by the previous paragraph, it follows that $P_{v_{i}}\ll _{R}P_{v_{i-1}}$.
Therefore, in particular, $R(v_{i})<_{R}L(v_{i-1})<_{R}L(u)$, i.e.~$%
R(v_{i})<_{R}L(u)$. On the other hand, $v_{i}\in N(x_{2})$, and thus $%
T_{v_{i}}$ intersects $T_{x_{2}}$ in $R_{T}$. Therefore, since $%
R(v_{i})<_{R}L(u)<_{R}L(x_{2})$, it follows that $l(x_{2})<_{R}r(v_{i})$.
Furthermore, since $P_{u}\ll _{R}P_{x_{2}}$, it follows that $%
l(u)<_{R}l(x_{2})<_{R}r(v_{i})$. That is, $R(v_{i})<_{R}L(u)$ and $%
l(u)<_{R}r(v_{i})$, i.e.~$P_{v_{i}}$ intersects $P_{u}$ in $R$ and $\phi
_{v_{i}}>\phi _{u}$. If $v_{i}\notin N(u)$, then $N(v_{i})\subseteq N(u)$ by
Lemma~\ref{intersecting-unbounded}, and thus $x_{2}\in N(u)$, which is a
contradiction. Therefore, $v_{i}\in N(u)$ if $i$ is odd.

Suppose now that $i$ is even. Then, $T_{v_{i-1}}\ll _{R_{T}}T_{v_{i-2}}$ by
Lemma~\ref{alternating-bounded-chain-tilde}. Thus, since $v_{i}\in
N(v_{i-2}) $ and $v_{i}\notin N(v_{i-1})$, it follows that $T_{v_{i-1}}\ll
_{R_{T}}T_{v_{i}}$. Recall that $T_{x_{2}}\ll _{R_{T}}T_{u}$. Since we
assumed that $v_{i}\notin N(u)$, either $T_{v_{i}}\ll _{R_{T}}T_{u}$ or $%
T_{u}\ll _{R_{T}}T_{v_{i}}$. If $T_{v_{i}}\ll _{R_{T}}T_{u}$, then $%
T_{v_{i-1}}\ll _{R_{T}}T_{v_{i}}\ll _{R_{T}}T_{u}$, i.e.~$v_{i-1}\notin N(u)$%
, which is a contradiction by the induction hypothesis. If $T_{u}\ll
_{R_{T}}T_{v_{i}}$, then $T_{x_{2}}\ll _{R_{T}}T_{u}\ll _{R_{T}}T_{v_{i}}$,
i.e.~$v_{i}\notin N(x_{2})$, which is a contradiction. Thus, $v_{i}\in N(u)$
if $i$ is even. This completes the induction step and the lemma follows.
\end{proof}

\medskip

Now, similarly to Lemmas~\ref{Cu-H}, and~\ref{C2-H}, we state the following
two lemmas.

\begin{lemma}
\label{Cu-H-tilde}For every vertex $v\in C_{u}\setminus \{u\}$, it holds $%
H_{i}^{\prime }\subseteq N(v)$ for every $i\geq 0$.
\end{lemma}

\begin{proof}
Let $v$ be a vertex of $C_{u}\setminus \{u\}$. Recall that $N_{1}(v)=N$ by
Lemma~\ref{N-Cu-unbounded}. Consider first the case where $v\in N(u)\cup
N(w) $. The proof will be done by induction on $i$. For $i=0$, consider a
vertex $v_{0}\in H_{0}^{\prime }$ and a vertex $y\in V\setminus
Q_{u}\setminus N[u]\setminus V_{0}(u)$, such that $yv_{0}\in E$; such a
vertex $y$ exists by Definition~\ref{Hi-tilde}. Recall that $T_{x_{2}}\ll
_{R_{T}}T_{u}\ll _{R_{T}}T_{y}$ and that $P_{u}\ll _{R}P_{x_{2}}\ll
_{R}P_{y} $\ by the proof of Lemma~\ref{alternating-bounded-chain-tilde-1}.

Let first $v\notin N(u)$ (and thus $v\in N(w)$). If $T_{u}\ll _{R_{T}}T_{v}$%
, i.e.~$T_{x}\ll _{R_{T}}T_{u}\ll _{R_{T}}T_{v}$ for every $x\in X_{1}$,
then $T_{z}$ intersects $T_{u}$ in $R_{T}$ for every $z\in N_{1}(v)=N$.
Thus, $N_{1}(u)=N$, which is a contradiction by Lemma~\ref{N-Cu}. Therefore, 
$T_{v}\ll _{R_{T}}T_{u}$. Furthermore, $T_{x_{2}}\ll _{R_{T}}T_{v}\ll
_{R_{T}}T_{u}$, since $T_{x_{2}}\ll _{R_{T}}T_{u}$, and since $v\in C_{u}$
and $C_{u}$ is connected. That is, $T_{x_{2}}\ll _{R_{T}}T_{v}\ll
_{R_{T}}T_{u}\ll _{R_{T}}T_{y}$. Then, $T_{v_{0}}$ intersects $T_{v}$ in $%
R_{T}$, since $v_{0}\in N(x_{2})\cap N(y)$, i.e.~$v_{0}\in N(v)$.

Let now $v\in N(u)$, and thus $v$ is bounded and $\phi _{v}>\phi _{u}$ in
the projection representation $R$. Suppose that $v\in N(y)$. Then, $P_{v}$
intersects $P_{x_{2}}$ in $R$, since $P_{u}\ll _{R}P_{x_{2}}\ll _{R}P_{y}$,
and since $v\in N(u)$ and $v\in N(y)$. Recall that $v\notin N(x_{2})$, since 
$v\in C_{u}$. Thus, since $v$ is bounded, it follows that $x_{2}$ is
unbounded and $\phi _{x_{2}}>\phi _{v}>\phi _{u}$. Recall that $%
N_{1}(x_{2})=N$ by Lemma~\ref{N-Cu}. Consider now a vertex $z\in N$, i.e.~$%
z\in N(x)\cap N(x_{2})$ for some $x\in X_{1}$. Then, $z$ is bounded and $%
\phi _{z}>\phi _{x_{2}}>\phi _{u}$, since $x_{2}$ is unbounded. Furthermore, 
$P_{z}$ intersects $P_{u}$ in $R$, since $P_{x}\ll _{R}P_{u}\ll
_{R}P_{x_{2}} $ and $z\in N(x)\cap N(x_{2})$, and thus $z\in N(u)$. Since
this holds for an arbitrary $z\in N$, it follows that $N_{1}(u)=N$, which is
a contradiction by Lemma~\ref{N-Cu}. Thus, $v\notin N(y)$. Then, $T_{v}\ll
_{R_{T}}T_{y}$, since $T_{u}\ll _{R_{T}}T_{y}$, and since $v\in N(u)$ and $%
v\notin N(y)$. Furthermore, $T_{x_{2}}\ll _{R_{T}}T_{v}$, since $%
T_{x_{2}}\ll _{R_{T}}T_{u}$, and since $v\in N(u)$ and $v\notin N(x_{2})$.
Therefore, $T_{x_{2}}\ll _{R_{T}}T_{v}\ll _{R_{T}}T_{y}$, and thus $%
T_{v_{0}} $ intersects $T_{v}$ in $R_{T}$, i.e.~$v_{0}\in N(v)$, since $%
v_{0}\in N(x_{2})\cap N(y)$. Summarizing, $v_{0}\in N(v)$ for every vertex $%
v_{0}\in H_{0}^{\prime }$ and for every vertex $v\in C_{u}\setminus \{u\}$,
such that $v\in N(u)\cup N(w)$, i.e.~$H_{0}^{\prime }\subseteq N(v)$ for all
these vertices $v$. This proves the induction basis.

For the induction step, let $i\geq 1$, and suppose that $v^{\prime }\in N(v)$
for every $v^{\prime }\in H_{j}$, where $0\leq j\leq i-1$. Let $v_{i}\in
H_{i}$ and $(v_{0},v_{1},\ldots ,v_{i-2},v_{i-1},v_{i})$ be an $H_{i}$-chain
of $v_{i}$. For the sake of contradiction, suppose that $v_{i}\notin N(v)$.
We will first prove that $P_{v}\ll _{R}P_{x_{2}}$. Otherwise, suppose first
that $P_{x_{2}}\ll _{R}P_{v}$. Then, $P_{u}\ll _{R}P_{x_{2}}\ll _{R}P_{v}$
and $P_{w}\ll _{R}P_{x_{2}}\ll _{R}P_{v}$, and thus $v\notin N(u)\cup N(w)$,
which is a contradiction to the assumption on $v$. Suppose now that $P_{v}$
intersects $P_{x_{2}}$ in $R$. Then, either $N(x_{2})\subseteq N(v)$ or $%
N(v)\subseteq N(x_{2})$ by Lemma~\ref{intersecting-unbounded}, since $%
v\notin N(x_{2})$ by definition of $\widetilde{C}_{2}$. If $%
N(x_{2})\subseteq N(v)$, then $v_{i}\in N(v)$, since $v_{i}\in N(x_{2})$,
which is a contradiction. Let $N(v)\subseteq N(x_{2})$. Then, since $C_{u}$
is connected with at least two vertices, $v$ is adjacent to at least one
vertex $z\in C_{u}$, and thus $z\in N(x_{2})$, which is a contradiction.
Thus, $P_{v}\ll _{R}P_{x_{2}}$.

Recall that $v_{i-1}\in N(v)$ by the induction hypothesis. If $i=1$, $%
P_{v_{1}}$ does not intersect $P_{v_{0}}$ in $R$ by Lemma~\ref%
{alternating-bounded-chain-tilde-1}. If $i\geq 2$, i.e.~if $v_{i-2}$ exists,
then $P_{v_{i}}$ does not intersect $P_{v_{i-1}}$ in $R$ by Lemma~\ref%
{intersecting-unbounded}, since $v_{i-2}\in N(v_{i})\setminus N(v_{i-1})$
and $v\in N(v_{i-1})\setminus N(v_{i})$. Thus, $P_{v_{i}}$ does not
intersect $P_{v_{i-1}}$ in $R$ for every $i\geq 1$. Similarly, $P_{v_{i}}$
does not intersect $P_{v}$ in $R$, since $x_{2}\in N(v_{i})\setminus N(v)$
and $v_{i-1}\in N(v)\setminus N(v_{i})$. Therefore, since $v_{i-1}\in N(v)$,
it follows that either $P_{v_{i}}\ll _{R}P_{v_{i-1}}$ and $P_{v_{i}}\ll
_{R}P_{v}$, or $P_{v_{i-1}}\ll _{R}P_{v_{i}}$ and $P_{v}\ll _{R}P_{v_{i}}$.
Suppose that $P_{v_{i}}\ll _{R}P_{v_{i-1}}$ and $P_{v_{i}}\ll _{R}P_{v}$.
Then, $P_{v_{i}}\ll _{R}P_{v}\ll _{R}P_{x_{2}}$, and thus $v_{i}\notin
N(x_{2})$, which is a contradiction.

Therefore, $P_{v_{i-1}}\ll _{R}P_{v_{i}}$ and $P_{v}\ll _{R}P_{v_{i}}$, and
thus $i\neq 1$ by Lemma~\ref{alternating-bounded-chain-tilde-1}. That is, $%
i\geq 2$, i.e.~$v_{i-2}$ exists. Furthermore, either $P_{v_{i-2}}\ll
_{R}P_{v_{i-1}}$ or $P_{v_{i-1}}\ll _{R}P_{v_{i-2}}$ by Lemma~\ref%
{alternating-bounded-chain-tilde}. If $P_{v_{i-2}}\ll _{R}P_{v_{i-1}}$, then 
$P_{v_{i-2}}\ll _{R}P_{v_{i-1}}\ll _{R}P_{v_{i}}$, and thus $%
v_{i}v_{i-2}\notin E$, which is a contradiction. Therefore $P_{v_{i-1}}\ll
_{R}P_{v_{i-2}}$, and thus also $T_{v_{i-1}}\ll _{R_{T}}T_{v_{i-2}}$ and $i$
is even, by Lemma~\ref{alternating-bounded-chain-tilde}. Furthermore, $%
T_{v_{i-1}}\ll _{R_{T}}T_{v_{i}}$, since $v_{i}\in N(v_{i-2})$ and $%
v_{i}\notin N(v_{i-1})$. Moreover, $T_{v}\ll _{R_{T}}T_{v_{i}}$, since $%
T_{v_{i-1}}\ll _{R_{T}}T_{v_{i}}$, and since $v\in N(v_{i-1})$ and $v\notin
N(v_{i})$. Recall also that $T_{x_{2}}\ll _{R_{T}}T_{v}$, since $%
T_{x_{2}}\ll _{R_{T}}T_{u}$, and since $v\in C_{u}$ and $C_{u}$ is
connected. That is, $T_{x_{2}}\ll _{R_{T}}T_{v}\ll _{R_{T}}T_{v_{i}}$, and
thus $v_{i}\notin N(x_{2})$, which is a contradiction. Thus, $v_{i}\in N(v)$
in the case where $v\in N(u)\cup N(w)$. This completes the induction step.

Summarizing, we have proved that $H_{i}^{\prime }\subseteq N(v)$ for every $%
i\geq 0$ and for every vertex $v\in C_{u}\setminus \{u\}$, such that $v\in
N(u)\cup N(w)$. This holds in particular for $w$, i.e.~$H_{i}^{\prime
}\subseteq N(w)$ for every $i\geq 0$, since $w$ is a vertex of $%
C_{u}\setminus \{u\}$ and $w\in N(u)\subseteq N(u)\cup N(w)$. Consider now
the case where $v\notin N(u)\cup N(w)$. Then, since $w\in N(u)$, either $%
T_{u}\ll _{R_{T}}T_{v}$ and $T_{w}\ll _{R_{T}}T_{v}$, or $T_{v}\ll
_{R_{T}}T_{u}$ and $T_{v}\ll _{R_{T}}T_{w}$. Suppose first that $T_{u}\ll
_{R_{T}}T_{v}$, i.e.~$T_{x}\ll _{R_{T}}T_{x_{2}}\ll _{R_{T}}T_{u}\ll
_{R_{T}}T_{v}$ for every $x\in X_{1}$ by Lemma~\ref{N(w)-2}. Recall that $%
N_{1}(v)=N$ by Lemma~\ref{N-Cu-unbounded}. Then, $T_{z}$ intersects $T_{u}$
in $R_{T}$, i.e.~$z\in N(u)$, for every $z\in N_{1}(v)=N$, and thus $%
N_{1}(u)=N$, which is a contradiction by Lemma~\ref{N-Cu}. Therefore, $%
T_{v}\ll _{R_{T}}T_{u}$ and $T_{v}\ll _{R_{T}}T_{w}$. Furthermore, $%
T_{x_{2}}\ll _{R_{T}}T_{v}$, since $T_{x_{2}}\ll _{R_{T}}T_{u}$, and since $%
v\in C_{u}$ and $C_{u}$ is connected. That is, $T_{x_{2}}\ll
_{R_{T}}T_{v}\ll _{R_{T}}T_{w}$. Then, since every $z\in H_{i}^{\prime }$, $%
i\geq 0$, is adjacent to both $x_{2}$ and $w$, as we proved above, it
follows that $T_{z}$ intersects $T_{v}$ in $R_{T}$, i.e.~$z\in N(v)$, for
every $z\in H_{i}^{\prime }$, where $i\geq 0$. Thus, $H_{i}^{\prime
}\subseteq N(v)$ for every $i\geq 0$ and for every vertex $v\in
C_{u}\setminus \{u\}$, such that $v\notin N(u)\cup N(w)$. This completes the
proof of the lemma.
\end{proof}

\begin{lemma}
\label{C2-H-tilde}For every vertex $v\in C_{2}$, it holds $H_{i}^{\prime
}\subseteq N(v)$ for every $i\geq 0$.
\end{lemma}

\begin{proof}
Recall that $C_{2}=\mathcal{A}_{2}\cup \mathcal{B}_{2}$, where $%
A_{j}\subseteq D_{2}$ for every $A_{j}\in \mathcal{B}_{2}$, $k+1\leq j\leq
\ell $, and $\mathcal{A}_{2}$ includes exactly those components $A_{i}$, $%
1\leq i\leq k$, for which all vertices of $A_{i}$ are adjacent to all
vertices of $\widetilde{H}$. Therefore, if $v\in A_{i}$ for some component $%
A_{i}\in \mathcal{A}_{2}$, then $H^{\prime }\subseteq H\subseteq \widetilde{H%
}\subseteq N(v)$ by definition, and thus $H_{i}^{\prime }\subseteq N(v)$ for
every $i\geq 0$.

Let now $v\in A_{j}$, for some $A_{j}\in \mathcal{B}_{2}$, and thus $v\in
D_{2}$. Suppose first that $v\notin N(x_{2})$. Then, $T_{x_{2}}\ll
_{R_{T}}T_{v}$ by Lemma~\ref{x2-relative-position-in-S2}, and $T_{v}\ll
_{R_{T}}T_{u}$, since $v\in D_{2}\subseteq S_{2}\subseteq V_{0}(u)$.
Moreover, $v\notin N(w)$, since otherwise $v\in C_{u}$, which is a
contradiction to the definition of $C_{2}$. Thus, $T_{v}\ll _{R_{T}}T_{w}$,
since $T_{v}\ll _{R_{T}}T_{u}$, and since $w\in N(u)$ and $w\notin N(v)$.
That is, $T_{x_{2}}\ll _{R_{T}}T_{v}\ll _{R_{T}}T_{w}$. Let now $z\in
H_{i}^{\prime }$, for some $i\geq 0$. Then, $z\in N(x_{2})$ by definition of 
$H^{\prime }$ and $z\in N(w)$ by Lemma~\ref{Cu-H-tilde}, and thus $T_{z}$
intersects $T_{v}$ in $R_{T}$, i.e.~$z\in N(v)$. Therefore, $H_{i}^{\prime
}\subseteq N(v)$ for every $i\geq 0$, in the case where $v\notin N(x_{2})$.

Suppose now that $v\in N(x_{2})$. We will prove by induction on $i$ that $%
H_{i}^{\prime }\subseteq N(v)$ for every $i\geq 0$. For $i=0$, let first $%
v_{0}\in H_{0}^{\prime }$ and $y\in V\setminus Q_{u}\setminus N[u]\setminus
V_{0}(u)$ be a vertex, such that $yv_{0}\in E$; such a vertex $y$ exists by
Definition~\ref{Hi-tilde}. For the sake of contradiction, assume that $%
v_{0}\notin N(v)$. Recall that $v_{0}\in N(u)$ by Lemma~\ref{tilde-neighbors}%
, and thus $v_{0}$ is bounded and $\phi _{v_{0}}>\phi _{u}$. Suppose that $%
P_{v_{0}}$ intersects $P_{v}$ in $R$. Then, $v$ is unbounded and $\phi
_{v}>\phi _{v_{0}}>\phi _{u}$, since $v_{0}$ is bounded and $v_{0}\notin
N(v) $. Recall that $N_{1}(v)=N$ by Lemma~\ref{N-C2}. Consider now a vertex $%
z\in N$, i.e.~$z\in N(x)\cap N(x_{2})$ for some $x\in X_{1}$. Then, $z\in
N(v)$, since $N_{1}(v)=N$, and thus $z$ is bounded and $\phi _{z}>\phi
_{v}>\phi _{u}$, since $v$ is unbounded. On the other hand, $P_{z}$
intersects $P_{u}$ in $R$, since $P_{x}\ll _{R}P_{u}\ll _{R}P_{x_{2}}$ and $%
z\in N(x)\cap N(x_{2})$. Thus, $z\in N(u)$, since $z$ is bounded and $\phi
_{z}>\phi _{u}$. Since this holds for an arbitrary $z\in N$, it follows that 
$N_{1}(u)=N$, which is a contradiction by Lemma~\ref{N-Cu}. Thus, $P_{v_{0}}$
does not intersect $P_{v}$ in $R$, i.e.~either $P_{v}\ll _{R}P_{v_{0}}$ or $%
P_{v_{0}}\ll _{R}P_{v}$.

Let first $P_{v}\ll _{R}P_{v_{0}}$. Suppose that $P_{v}$ intersects $P_{u}$
in $R$. Recall that $v\notin N(u)$, since $v\in C_{2}$, and thus either $%
N(u)\subseteq N(v)$ or $N(v)\subseteq N(u)$ by Lemma~\ref%
{intersecting-unbounded}. If $N(u)\subseteq N(v)$, then $v_{0}\in N(v)$,
which is a contradiction. If $N(v)\subseteq N(u)$, then $x_{2}\in N(u)$,
which is again a contradiction. Thus, $P_{v}$ does not intersect $P_{u}$ in $%
R$, i.e.~either $P_{v}\ll _{R}P_{u}$ or $P_{u}\ll _{R}P_{v}$. If $P_{v}\ll
_{R}P_{u}$, then $P_{v}\ll _{R}P_{u}\ll _{R}P_{x_{2}}$, i.e.~$v\notin
N(x_{2})$, which is a contradiction to the assumption on $v$. Thus, $%
P_{u}\ll _{R}P_{v}$. Moreover, since we assumed that $P_{v}\ll _{R}P_{v_{0}}$%
, it follows that $P_{u}\ll _{R}P_{v}\ll _{R}P_{v_{0}}$, and thus $%
v_{0}\notin N(u)$, which is a contradiction by Lemma~\ref{tilde-neighbors}.

Let now $P_{v_{0}}\ll _{R}P_{v}$. Suppose that $P_{v}$ intersects $P_{y}$ in 
$R$. Recall that $v\in V_{0}(u)$ by Lemma~\ref{N(w)-1}, and thus $vy\notin E$%
, since otherwise $y\in V_{0}(u)$, which is a contradiction. Thus, either $%
N(y)\subseteq N(v)$ or $N(v)\subseteq N(y)$ by Lemma~\ref%
{intersecting-unbounded}. If $N(y)\subseteq N(v)$, then $v_{0}\in N(v)$,
which is a contradiction. If $N(v)\subseteq N(y)$, then $x_{2}\in N(y)$
(since we assumed that $x_{2}\in N(v)$), and thus $y\in V_{0}(u)$, which is
a contradiction. Thus, $P_{v}$ noes not intersect $P_{y}$ in $R$,
i.e.~either $P_{v}\ll _{R}P_{y}$ or $P_{y}\ll _{R}P_{v}$. If $P_{v}\ll
_{R}P_{y}$, then $P_{v_{0}}\ll _{R}P_{v}\ll _{R}P_{y}$, i.e.~$yv_{0}\notin E$%
, which is a contradiction. Suppose that $P_{y}\ll _{R}P_{v}$. Recall that $%
P_{x_{2}}\ll _{R}P_{y}$ by the proof of Lemma~\ref%
{alternating-bounded-chain-tilde-1}. Thus $P_{x_{2}}\ll _{R}P_{y}\ll
_{R}P_{v}$, i.e.~$v\notin N(x_{2})$, which is a contradiction to the
assumption on $v$. Therefore, $v_{0}\in N(v)$, and thus $H_{0}^{\prime
}\subseteq N(v)$. This proves the induction basis.

For the induction step, let $i\geq 1$, and suppose that $v^{\prime }\in N(v)$
for every $v^{\prime }\in H_{j}^{\prime }$, where $0\leq j\leq i-1$. For the
sake of contradiction, assume that $v_{i}\notin N(v)$. Let $%
(v_{0},v_{1},\ldots ,v_{i-1},v_{i})$ be an $H_{i}$-chain of $v_{i}$. If $i=1$%
, $P_{v_{1}}$ does not intersect $P_{v_{0}}$ in $R$ by Lemma~\ref%
{alternating-bounded-chain-tilde-1}. If $i\geq 2$, i.e.~if $v_{i-2}$ exists,
then $P_{v_{i}}$ does not intersect $P_{v_{i-1}}$ in $R$ by Lemma~\ref%
{intersecting-unbounded}, since $v_{i-2}\in N(v_{i})\setminus N(v_{i-1})$
and $v\in N(v_{i-1})\setminus N(v_{i})$. Thus, $P_{v_{i}}$ does not
intersect $P_{v_{i-1}}$ in $R$ for every $i\geq 1$. Recall now that $%
v_{i}\in N[u,w]=N(u)\cup N(w)$, since $v_{i}\in H$, and that $v\notin
N[u,w]=N(u)\cup N(w)$ by definition of $C_{2}$. If $v_{i}\in N(u)$ (resp.~$%
v_{i}\in N(w)$), then $u\in N(v_{i})\setminus N(v)$ (resp.~$w\in
N(v_{i})\setminus N(v)$). Furthermore, $v_{i-1}\in N(v)\setminus N(v_{i})$,
i.e.~$N(v_{i})\nsubseteq N(v)$ and $N(v)\nsubseteq N(v_{i})$, and thus $%
P_{v_{i}}$ does not intersect $P_{v}$ in $R$ by Lemma~\ref%
{intersecting-unbounded}. Therefore, since $v_{i-1}\in N(v)$, it follows
that either $P_{v_{i-1}}\ll _{R}P_{v_{i}}$ and $P_{v}\ll _{R}P_{v_{i}}$ or $%
P_{v_{i}}\ll _{R}P_{v_{i-1}}$ and $P_{v_{i}}\ll _{R}P_{v}$.

Suppose first that $P_{v_{i-1}}\ll _{R}P_{v_{i}}$ and $P_{v}\ll
_{R}P_{v_{i}} $. Recall that $v_{i}\in N(u)$ or $v_{i}\in N(w)$.
Furthermore, recall that $v\in N(x_{2})$ by our assumption on $v$. Let $%
v_{i}\in N(u)$ (resp.~$v_{i}\in N(w)$). Then, $P_{v}$ does not intersect $%
P_{u}$ (resp.~$P_{w}$) in $R$ by Lemma~\ref{intersecting-unbounded}, since $%
x_{2}\in N(v)\setminus N(u) $ (resp.~$x_{2}\in N(v)\setminus N(w)$) and $%
v_{i}\in N(u)\setminus N(v)$ (resp.~$v_{i}\in N(w)\setminus N(v)$).
Therefore, since $P_{u}\ll _{R}P_{x_{2}}$ (resp.~$P_{w}\ll _{R}P_{x_{2}}$)
and $v\in N(x_{2})$, it follows that $P_{u}\ll _{R}P_{v}$ (resp.~$P_{w}\ll
_{R}P_{v}$). That is, $P_{u}\ll _{R}P_{v}\ll _{R}P_{v_{i}}$ (resp.~$P_{w}\ll
_{R}P_{v}\ll _{R}P_{v_{i}}$), i.e.~$v_{i}\notin N(u)$ (resp.~$v_{i}\notin
N(w)$), which is a contradiction.

Suppose now that $P_{v_{i}}\ll _{R}P_{v_{i-1}}$ and $P_{v_{i}}\ll _{R}P_{v}$%
. If $i=1$, then $T_{v_{1}}\ll _{R_{T}}T_{v_{0}}$ by Lemma~\ref%
{alternating-bounded-chain-tilde-1}. If $i\geq 2$, i.e.~if $v_{i-2}$ exists,
then $P_{v_{i-2}}\ll _{R}P_{v_{i-1}}$. Indeed, otherwise $P_{v_{i}}\ll
_{R}P_{v_{i-1}}\ll _{R}P_{v_{i-2}}$, i.e.~$v_{i}v_{i-2}\notin E$, which is a
contradiction. Thus, also $T_{v_{i-2}}\ll _{R_{T}}T_{v_{i-1}}$ and $i$ is
odd by Lemma~\ref{alternating-bounded-chain-tilde}. Therefore, $T_{v_{i}}\ll
_{R_{T}}T_{v_{i-1}}$ if $i\geq 2$, since otherwise $T_{v_{i-2}}\ll
_{R_{T}}T_{v_{i-1}}\ll _{R_{T}}T_{v_{i}}$, i.e.~$v_{i}v_{i-2}\notin E$,
which is a contradiction. That is, $T_{v_{i}}\ll _{R_{T}}T_{v_{i-1}}$ for
all $i\geq 1$. Therefore, since $v\in N(v_{i-1})$ and $v\notin N(v_{i})$, it
follows that $T_{v_{i}}\ll _{R_{T}}T_{v}$. Recall also that $T_{x_{2}}\ll
_{R_{T}}T_{u}$ and $T_{x_{2}}\ll _{R_{T}}T_{w}$. Thus, $T_{v}\ll
_{R_{T}}T_{u}$ and $T_{v}\ll _{R_{T}}T_{w}$, since we assumed that $v\in
N(x_{2})$, and since $v\notin N(u)\cup N(w)$ by definition of $C_{2}$. That
is, $T_{v_{i}}\ll _{R_{T}}T_{v}\ll _{R_{T}}T_{u}$ and $T_{v_{i}}\ll
_{R_{T}}T_{v}\ll _{R_{T}}T_{w}$, i.e.~$v_{i}\notin N(u)\cup N(w)$, which is
a contradiction. Therefore, $v_{i}\in N(v)$, and thus $H_{i}^{\prime
}\subseteq N(v)$. This completes the induction step, and the lemma follows.
\end{proof}

\subsection*{The subgraph $G_{0}$ of $G$}

Let $G_{0}$ be the graph induced in $G$ by the vertices of $C_{u}\cup
C_{2}\cup (H\setminus \bigcup\nolimits_{i=1}^{\infty }H_{i}\setminus
\bigcup\nolimits_{i=0}^{\infty }H_{i}^{\prime })$. Note that $G_{0}$ is an
induced subgraph also of $G\setminus Q_{u}\setminus N[X_{1}]$. Furthermore,
note that every vertex of $G_{0}\setminus \{u\}$ is bounded by to Lemma~\ref%
{N-H-C2-Cu-bounded}. Recall that $C_{2}\subseteq V_{0}(u)$ by Lemma~\ref%
{N(w)-1} and that $C_{u}\setminus \{u\}\subseteq N(u)\cup V_{0}(u)$ by Lemma~%
\ref{Cu-V0}. Consider now a vertex $v\in H\setminus
\bigcup\nolimits_{i=1}^{\infty }H_{i}\setminus
\bigcup\nolimits_{i=0}^{\infty }H_{i}^{\prime }$. If $v\notin N(u)$, then $%
v\in V_{0}(u)$, since $x_{2}\in V_{0}(u)$ and $v\in N(x_{2})$ by definition
of $H$. Thus, the next observation follows.

\begin{observation}
\label{V(G0)}Every vertex of $G_{0}\setminus \{u\}$ is bounded. Furthermore, 
$V(G_{0})\subseteq N[u]\cup V_{0}(u)$.
\end{observation}

\begin{lemma}
\label{module-2}$G_{0}\setminus \{u\}$ is a module in $G\setminus \{u\}$. In
particular, $N(v)\setminus V(G_{0})=N(X_{1})\cup
\bigcup\nolimits_{i=1}^{\infty }H_{i}\cup \bigcup\nolimits_{i=0}^{\infty
}H_{i}^{\prime }$ for every vertex $v\in V(G_{0})\setminus \{u\}$.
\end{lemma}

\begin{proof}
First recall by Lemma~\ref{module-1} that $N(V({C_{u}\cup C_{2}\cup H}%
))\subseteq Q_{u}\cup N(X_{1})\cup V({\mathcal{B}_{1}})$, where $V(\mathcal{B%
}_{1})=\bigcup\nolimits_{A_{j}\in \mathcal{B}_{1}}A_{j}$. Consider a vertex $%
q\in Q_{u}$. Then, since we assumed in the statement of Theorem~\ref%
{no-property-thm} that Condition~\ref{ass3} holds, and since $X_{1}\subseteq
D_{1}\subseteq V_{0}(u)$ by Lemma~\ref{N(w)-1}, it follows that $T_{q}\ll
_{R_{T}}T_{x}\ll _{R_{T}}T_{u}$ for every $x\in X_{1}$. Thus, since $%
N(q)\subset N(u)$ by definition of $Q_{u}$, it follows that $T_{z}$
intersects $T_{x}$ in $R_{T}$ for every $z\in N(q)\subset N(u)$ and every $%
x\in X_{1}$. Therefore, in particular, $N(q)\subseteq N(X_{1})$ for every $%
q\in Q_{u}$. Thus, no vertex $q\in Q_{u}$ is adjacent to any vertex of $V({%
C_{u}\cup C_{2}\cup H})$, since $V({C_{u}\cup C_{2}\cup H})$ induces a
subgraph of $G\setminus Q_{u}\setminus N[X_{1}]\setminus {\mathcal{B}_{1}}$
by Lemma~\ref{module-1}. Thus, $N(V({C_{u}\cup C_{2}\cup H}))\cap
Q_{u}=\emptyset $, i.e.~$N(V({C_{u}\cup C_{2}\cup H}))\subseteq N(X_{1})\cup
V({\mathcal{B}_{1}})$.

Recall that $V(G_{0})={C_{u}\cup C_{2}\cup (}H\setminus
\bigcup\nolimits_{i=1}^{\infty }H_{i}\setminus
\bigcup\nolimits_{i=0}^{\infty }H_{i}^{\prime })$ by definition of $G_{0}$.
Consider now an arbitrary vertex $v\in V(G_{0})\setminus \{u\}$. Then, it
follows by the previous paragraph that 
\begin{equation}
N(v)\setminus V(G_{0})\subseteq N(X_{1})\cup V({\mathcal{B}_{1}})\cup
(\bigcup\nolimits_{i=1}^{\infty }H_{i}\cup \bigcup\nolimits_{i=0}^{\infty
}H_{i}^{\prime })  \label{N(v)-G0-eq-1}
\end{equation}%
We will prove that $N(v)\setminus V(G_{0})=N(X_{1})\cup
(\bigcup\nolimits_{i=1}^{\infty }H_{i}\cup \bigcup\nolimits_{i=0}^{\infty
}H_{i}^{\prime })$. If $v\in {C_{u}\setminus \{u\}}$, then $%
N(X_{1})\subseteq N(v)$, since $N_{1}(v)=N=N(X_{1})$ by Lemma~\ref%
{N-Cu-unbounded}. Similarly, if $v\in {C_{2}}$, then $N(X_{1})\subseteq N(v)$%
, since $N_{1}(v)=N=N(X_{1})$ by Lemma~\ref{N-C2}. If $v\in H\setminus
\bigcup\nolimits_{i=1}^{\infty }H_{i}\setminus
\bigcup\nolimits_{i=0}^{\infty }H_{i}^{\prime }$, then $N=H_{0}\subseteq
N(v) $ by Definition~\ref{Hi} (where $N=N(X_{1})$), since otherwise $v\in
H_{1}$, which is a contradiction. That is, $N(X_{1})\subseteq N(v)$ for
every vertex $v\in V(G_{0})\setminus \{u\}$.

If $v\in {C_{u}\setminus \{u\}}$, then $\bigcup\nolimits_{i=1}^{\infty
}H_{i}\cup \bigcup\nolimits_{i=0}^{\infty }H_{i}^{\prime }\subseteq N(v)$ by
Lemmas~\ref{Cu-H} and~\ref{Cu-H-tilde}. Similarly, if $v\in {C_{2}}$, then $%
\bigcup\nolimits_{i=1}^{\infty }H_{i}\cup \bigcup\nolimits_{i=0}^{\infty
}H_{i}^{\prime }\subseteq N(v)$ by Lemmas~\ref{C2-H} and~\ref{C2-H-tilde}.
If $v\in H\setminus \bigcup\nolimits_{i=1}^{\infty }H_{i}\setminus
\bigcup\nolimits_{i=0}^{\infty }H_{i}^{\prime }$, then $\bigcup%
\nolimits_{i=1}^{\infty }H_{i}\cup \bigcup\nolimits_{i=0}^{\infty
}H_{i}^{\prime }\subseteq N(v)$ by Definitions~\ref{Hi} and~\ref{Hi-tilde}.
Indeed, otherwise $v\in H_{i}$ for some $i\geq 1$, or $v\in H_{i}^{\prime }$
for some $i\geq 0$, which is a contradiction. That is, $\bigcup%
\nolimits_{i=1}^{\infty }H_{i}\cup \bigcup\nolimits_{i=0}^{\infty
}H_{i}^{\prime }\subseteq N(v)$ for every vertex $v\in V(G_{0})\setminus
\{u\}$.

We will now prove that $N(v)\cap V({\mathcal{B}_{1}})=\emptyset $. Suppose
for the sake of contradiction that $v^{\prime }\in N(v)$, for some $%
v^{\prime }\in V({\mathcal{B}_{1}})$. Note that $v^{\prime }\notin N(u)$ by
definition of $\widetilde{C}_{2}$. Let first $v\in {C_{u}\setminus \{u\}}$.
Then, either $v\in V_{0}(u)$ or $v\in N(u)$ by Lemma~\ref{Cu-V0}. If $v\in
V_{0}(u)$, then also $v^{\prime }\in V_{0}(u)$, which is a contradiction by
definition of ${\mathcal{B}_{1}}$. Suppose that $v\in N(u)$. Recall that $%
v^{\prime }\in {V(\mathcal{B}_{1})\subseteq V\setminus Q_{u}\setminus
N[u]\setminus V_{0}(u)}$ by our assumption on $v^{\prime }$ and by
Observation~\ref{V(B1)}. Thus, either $P_{u}\ll _{R}P_{x_{2}}\ll
_{R}P_{v^{\prime }}$ or $P_{v^{\prime }}\ll _{R}P_{x}\ll _{R}P_{u}$ for
every $x\in X_{1}$ by Lemma~\ref{foreigner-not-inside-1}. Therefore $%
P_{u}\ll _{R}P_{x_{2}}\ll _{R}P_{v^{\prime }}$, since $P_{u}\ll
_{R}P_{v^{\prime }}$ for every $v^{\prime }\in V({\mathcal{B}_{1}})$ by
definition of ${\mathcal{B}_{1}}$. Then, since we assumed that $v\in N(u)$
and $v\in N(v^{\prime })$, it follows that $P_{v}$ intersects $P_{x_{2}}$ in 
$R$. Furthermore, $x_{2}\in C_{2}$ is a bounded vertex by Lemma~\ref%
{N-H-C2-Cu-bounded}; $v$ is also a bounded vertex, since $v\in N(u)$.
Therefore $v\in N(x_{2})$, which is a contradiction by definition of $C_{u}$%
. Thus, $N(v)\cap V({\mathcal{B}_{1}})=\emptyset $ for every $v\in {%
C_{u}\setminus \{u\}}$.

Let now $v\in C_{2}$. Then $v\in V_{0}(u)$, since $C_{2}\subseteq V_{0}(u)$
by Lemma~\ref{N(w)-1}, and thus also $v^{\prime }\in V_{0}(u)$, since $%
v^{\prime }\notin N(u)$. This which is a contradiction by definition of ${%
\mathcal{B}_{1}}$. Therefore, $N(v)\cap V({\mathcal{B}_{1}})=\emptyset $ for
every~$v\in {C_{2}}$. Let finally ${v\in H\setminus
\bigcup\nolimits_{i=1}^{\infty }H_{i}\setminus
\bigcup\nolimits_{i=0}^{\infty }H_{i}^{\prime }}$. Recall that ${v^{\prime
}\in {V(\mathcal{B}_{1})\subseteq }V\setminus Q_{u}\setminus N[u]\setminus
V_{0}(u)}$. Thus, since ${v\in H\setminus \bigcup\nolimits_{i=1}^{\infty
}H_{i}}$, and since $vv^{\prime }\in E$, it follows by Definition~\ref%
{Hi-tilde} that $v\in H_{0}^{\prime }$. This is a contradiction to the
assumption that $v\in H\setminus \bigcup\nolimits_{i=1}^{\infty
}H_{i}\setminus \bigcup\nolimits_{i=0}^{\infty }H_{i}^{\prime }$. Therefore, 
$N(v)\cap V({\mathcal{B}_{1}})=\emptyset $ for every $v\in H\setminus
\bigcup\nolimits_{i=1}^{\infty }H_{i}\setminus
\bigcup\nolimits_{i=0}^{\infty }H_{i}^{\prime }$. That is, $N(v)\cap V({%
\mathcal{B}_{1}})=\emptyset $ for every vertex $v\in V(G_{0})\setminus \{u\}$.

Summarizing, $N(X_{1})\cup (\bigcup\nolimits_{i=1}^{\infty }H_{i}\cup
\bigcup\nolimits_{i=0}^{\infty }H_{i}^{\prime })\subseteq N(v)$ and $%
N(v)\cap V({\mathcal{B}_{1}})=\emptyset $ for every vertex~${v\in }${$%
V(G_{0})\setminus \{u\}$}. Therefore, it follows by (\ref{N(v)-G0-eq-1})
that 
\begin{equation}
N(v)\setminus V(G_{0})=N(X_{1})\cup (\bigcup\nolimits_{i=1}^{\infty
}H_{i}\cup \bigcup\nolimits_{i=0}^{\infty }H_{i}^{\prime })
\label{N(v)-G0-eq-2}
\end{equation}%
for every vertex $v\in V(G_{0})\setminus \{u\}$. Thus, in particular, $%
G_{0}\setminus \{u\}$ is a module in $G\setminus \{u\}$, since every vertex
of $G_{0}\setminus \{u\}$ has the same neighbors in $G\setminus G_{0}$. This
completes the proof of the lemma.
\end{proof}

\medskip

Now let $G_{0}^{\prime }=G[V(G_{0})\cup \{u^{\ast }\}]$. Then, since $%
u^{\ast }\in V_{0}(u)$ and $V(G_{0})\subseteq N[u]\cup V_{0}(u)$ by
Observation~\ref{V(G0)}, it follows that also $V(G_{0}^{\prime })\subseteq
N[u]\cup V_{0}(u)$. Furthermore, Observation~\ref{V(G0)} implies that the
set $V(G_{0}^{\prime })\setminus \{u\}$ has only bounded vertices, since $%
u^{\ast }$ is also bounded. Furthermore, since $N_{1}(u)\neq N$ by Lemma~\ref%
{N-Cu} (where $N=N(X_{1})$), there exists at least one vertex $q\in
N\setminus N(u)$, which is bounded by Lemma~\ref{N-H-C2-Cu-bounded}.
Moreover $q\in N(x_{2})$, since $N=N(X_{1})\subseteq N(x_{2})$ by Lemma~\ref%
{N-Cu}. Therefore, $P_{q}$ intersects $P_{u}$ in $R$, since $q\in
N(X_{1})\cap N(x_{2})$ and $P_{x}\ll _{R}P_{u}\ll _{R}P_{x_{2}}$ for every $%
x\in X_{1}$. Furthermore, $\phi _{q}<\phi _{u}$ in $R$, since otherwise $%
q\in N(u)$, which is a contradiction. Thus, $N(u)\subseteq N(q)$ by Lemma~%
\ref{intersecting-unbounded}, i.e.~$q$ is a covering vertex of $u$.
Furthermore $q\notin V(G_{0})$, since $q\in N=N(X_{1})$. Then, $q$ is
adjacent to all vertices of $C_{2}\cup C_{u}\setminus \{u\}$ by Lemmas~\ref%
{N-Cu-unbounded} and~\ref{N-C2}. Furthermore, $q\in N$ is adjacent to all
vertices of $H\setminus \bigcup\nolimits_{i=1}^{\infty }H_{i}\setminus
\bigcup\nolimits_{i=0}^{\infty }H_{i}^{\prime }$ by Definition~\ref{Hi},
since no vertex of $H_{1}$ is included in $H\setminus
\bigcup\nolimits_{i=1}^{\infty }H_{i}\setminus
\bigcup\nolimits_{i=0}^{\infty }H_{i}^{\prime }$. Summarizing, $q$ is a
bounded covering vertex of $u$, $P_{q}$ intersects $P_{u}$ in $R$, and $\phi
_{q}<\phi _{u}$ in $R$, and thus we may assume w.l.o.g.~that $u^{\ast }=q$,
as the next observation states.

\begin{observation}
\label{bounded-hovering-obs}Without loss of generality, we may assume that $%
u^{\ast }\in N=N(X_{1})$, i.e.~$u^{\ast }\notin V(G_{0})$, and that $u^{\ast
}$ is adjacent to every vertex of $V(G_{0})\setminus \{u\}$; thus, in
particular, $G_{0}^{\prime }$ is connected.
\end{observation}

Moreover, $G_{0}^{\prime }=G[V(G_{0})\cup \{u^{\ast }\}]$ has strictly less
vertices than $G$, since no vertex of $X_{1}\neq \emptyset $ is included in $%
G_{0}^{\prime }$. We assume now that the following condition holds. Its
correctness will be proved later, in Lemma~\ref{lem-cond1}.

\begin{condition}
\label{cond1}Let $G=(V,E)$ be a connected graph in \textsc{Tolerance }$\cap $ \textsc{%
Trapezoid}, $R$ be a projection representation of $G$ with $u$ as the
only unbounded vertex, such that $V_{0}(u)\neq \emptyset $ is connected and $%
V=N[u]\cup V_{0}(u)$. Then, there exists a projection representation $%
R^{\ast \ast }$ of $G$ with $u$ as the only unbounded vertex, such that $u$
has the right border property in $R^{\ast \ast }$.
\end{condition}

\subsection*{The projection representation $R_{\ell }$}

We define now the line segment $\ell $ with one endpoint $a_{\ell }$ on $%
L_{1}$ and the other endpoint $b_{\ell }$ on $L_{2}$ as follows. 
First recall that ${r(w)>_{R}r(u)}$ by Lemma~\ref{unbounded-bounded}, since ${w\in N(u)}$. 
Let ${\Delta = r(w)-r(u)>_{R}0}$ be the distance on $L_{2}$ between the lower right endpoints of $P_{w}$ and $P_{u}$ in $R$. 
Define in $R$ the values ${a_{\ell }=\min \{L(x_{2}),L(u)+\Delta \}}$ and 
${b_{\ell }=r(w)}$ as the endpoints of the line segment $\ell $ on $L_{1}$ and $L_{2}$,
respectively. Note that ${\phi _{\ell }\geq \phi _{u}}$ in $R$, where $\phi_{\ell }$ 
denotes the slope of the line segment $\ell $. 
Recall that~${\phi_{w}>\phi_{u}}$ in $R$ (since ${w\in N(u)}$), and thus in particular ${R(w)<_{R}L(u)+\Delta}$. 
Therefore, since ${P_{u}\ll_{R} P_{x_{2}}}$ and ${P_{w}\ll_{R} P_{x_{2}}}$, it follows that 
the line segment $\ell $ lies between $P_{u}$ and $P_{x_{2}}$ in $R$, 
as well as between $P_{w}$ and $P_{x_{2}}$ in $R$. 
Denote by $a_{u}$ and $b_{u}$ the upper and the lower endpoint of $P_{u}$ in~$R$,
respectively. Then, always ${a_{\ell }>a_{u}}$ and ${b_{\ell }>b_{u}}$
by definition of the line segment $\ell $.

Note that $G_{0}^{\prime }$ satisfies the requirements of Condition~\ref{cond1}.
Thus, since we assumed that Condition~\ref{cond1} holds, there exists a
representation $R_{0}^{\prime }$ of $G_{0}^{\prime }$ with $u$ as the only
unbounded vertex, where $u$ has the right border property in $R_{0}^{\prime
} $. Let $R_{0}^{\prime \prime }$ be the projection representation of $G_{0}$
that is obtained if we remove from $R_{0}^{\prime }$ the parallelogram that
corresponds to $u^{\ast }$. Let $\varepsilon >0$ be a sufficiently small
positive number. Consider now the $\varepsilon $-squeezed projection
representation $R_{0}$ of $G_{0}$ with respect to the line segment $\ell $,
which is obtained from $R_{0}^{\prime \prime }$. Then, replace the
parallelograms of the vertices of $G_{0}$ in $R$ by the projection
representation $R_{0}$, and denote the resulting projection representation
by $R_{\ell }$.

\begin{remark}
\label{R-ell-slopes}Recall that w.l.o.g.~all slopes of the parallelograms in
the projection representation~$R$ are 
distinct~\cite{GolTol04,IsaakNT03,MSZ-Model-SIDMA-09}. %FishburnTrotter99
Therefore, since ${\varepsilon >0}$ is assumed to be sufficiently small, we can assume
w.l.o.g.~that, for every vertex ${x\in V(G_{0})}$, the slopes $\phi_{x}$ are
arbitrarily ``close'' to ${\phi_{\ell}}$ (and to each other) in $R_{\ell}$.
That is, we can assume w.l.o.g.~that for every vertex ${v\notin V(G_{0})}$,
if ${\phi_{v}>\phi_{\ell}}$ (resp.~${\phi_{v}<\phi_{\ell}}$) in $R_{\ell}$,
then also ${\phi_{v}>\phi_{x}}$ (resp.~${\phi_{v}<\phi_{x}}$) in $R_{\ell}$
for every vertex ${x\in V(G_{0})}$.
\end{remark}

\begin{remark}
\label{trans4-remark-1}Recall that the vertices of $G_{0}$ in $R_{\ell }$
lie on an $\varepsilon $-squeezed projection representation~$R_{0}$ with
respect to the line segment $\ell $, where $\varepsilon >0$ is a
sufficiently (very) small positive number. Therefore, in particular $b_{\ell
}-\varepsilon <_{R_{\ell }}l(v)\leq _{R_{\ell }}r(v)<_{R_{\ell }}b_{\ell
}+\varepsilon $ and $a_{\ell }-\varepsilon <_{R_{\ell }}L(v)\leq _{R_{\ell
}}R(v)<_{R_{\ell }}a_{\ell }+\varepsilon $ for every vertex $v\in V(G_{0})$.
On the other hand, since $\varepsilon $ has been chosen to be sufficiently
small, we may assume w.l.o.g.~that for every vertex $z\notin V(G_{0})$, the
lower right endpoint $r(z)$ (resp.~the lower left endpoint $l(z)$) of $P_{z}$
in $R_{\ell }$ does not lie between $b_{\ell }-\varepsilon $ and $b_{\ell
}+\varepsilon $, i.e.~either~$r(z)<_{R_{\ell }}b_{\ell }-\varepsilon $ or $%
r(z)>_{R_{\ell }}b_{\ell }+\varepsilon $ (resp.~either $l(z)<_{R_{\ell
}}b_{\ell }-\varepsilon $ or $l(z)>_{R_{\ell }}b_{\ell }+\varepsilon $).
Similarly, for every vertex $z\notin V(G_{0})$, the upper right endpoint $%
R(z)$ (resp.~the upper left endpoint $L(z)$) of~$P_{z}$ in $R_{\ell }$ does
not lie between $a_{\ell }-\varepsilon $ and $a_{\ell }+\varepsilon $, i.e.
either $R(z)<_{R_{\ell }}a_{\ell }-\varepsilon $ or $R(z)>_{R_{\ell
}}a_{\ell }+\varepsilon $ (resp.~either $L(z)<_{R_{\ell }}a_{\ell
}-\varepsilon $ or $L(z)>_{R_{\ell }}a_{\ell }+\varepsilon $).
\end{remark}

\subsection*{Properties of $R_{\ell }$}

\begin{lemma}
\label{R-ell}$R_{\ell }\setminus \{u\}$ is a projection representation of $%
G\setminus \{u\}$.
\end{lemma}

\begin{proof}
Recall that all vertices of $G_{0}\setminus \{u\}$ are bounded by
Observation~\ref{V(G0)} and that $N(v)\setminus V(G_{0})=N(X_{1})\cup
\bigcup\nolimits_{i=1}^{\infty }H_{i}\cup \bigcup\nolimits_{i=0}^{\infty
}H_{i}^{\prime }$ for every vertex $v\in V(G_{0})\setminus \{u\}$ by Lemma %
\ref{module-2}. We will prove that for a vertex $z\in V(G\setminus G_{0})$
and a vertex $v\in V(G_{0})\setminus \{u\}$, $z$ is adjacent to $v$ in $%
R_{\ell }$ if and only if $z\in N(X_{1})\cup \bigcup\nolimits_{i=1}^{\infty
}H_{i}\cup \bigcup\nolimits_{i=0}^{\infty }H_{i}^{\prime }$.

Consider a vertex $z\in N(X_{1})\cup \bigcup\nolimits_{i=1}^{\infty
}H_{i}\cup \bigcup\nolimits_{i=0}^{\infty }H_{i}^{\prime }$. Then $z$ is a
vertex of $G\setminus G_{0}$ by definition of $G_{0}$. Furthermore, $z$ is
bounded by Lemma~\ref{N-H-C2-Cu-bounded}. If $z\in
\bigcup\nolimits_{i=1}^{\infty }H_{i}\cup \bigcup\nolimits_{i=0}^{\infty
}H_{i}^{\prime }$, then $z\in N(w)\cap N(x_{2})$ by the definition of $H$.
Let $z\in N(X_{1})$. Then again $z\in N(x_{2})$, since $%
N_{1}(x_{2})=N=N(X_{1})$ by Lemma~\ref{N-Cu}. Furthermore $z\in N(w)$, since 
$N_{1}(w)=N(X_{1})$ by Lemma~\ref{N(w)-1}. That is, $z\in N(w)\cap N(x_{2})$
for every case regarding $z$, and thus $P_{z}$ intersects both $P_{w}$ and $%
P_{x_{2}}$ in $R$. Recall now by definition of the line segment $\ell $ that 
$\ell $ lies between $P_{w}$ and $P_{x_{2}}$ in $R$. Therefore, since $P_{z}$
intersects both $P_{w}$ and $P_{x_{2}}$ in $R$, it follows that also $P_{z}$
intersects $\ell $ in $R$. Thus, $z$ is adjacent in $R_{\ell }$ to every
vertex $v\in V(G_{0})\setminus \{u\}$, since both $z$ and $v$ are bounded.

Conversely, consider a vertex $z\in V(G\setminus G_{0})$ and a vertex $v\in
V(G_{0})\setminus \{u\}$, such that $z$ is adjacent to $v$ in $R_{\ell }$.
Then, in particular $P_{z}$ intersects $\ell $ in $R$. Recall that $v$ is
bounded by Observation~\ref{V(G0)}. Therefore, either $z$ is bounded or $z$
is unbounded and $\phi _{z}<\phi _{\ell }$ (in both $R$ and $R_{\ell }$).
Furthermore, observe that $z\notin X_{1}$, since $P_{x}\ll _{R}P_{u}$ for
every $x\in X_{1}$, and since $P_{z}$ intersects $\ell $ in $R$. Suppose
that $z\in {V(\mathcal{B}_{1})}$, and thus $z\in {V\setminus Q_{u}\setminus
N[u]\setminus V_{0}(u)}$ by Observation~\ref{V(B1)}. Then, either $P_{u}\ll
_{R}P_{x_{2}}\ll _{R}P_{z}$ or $P_{z}\ll _{R}P_{x}\ll _{R}P_{u}\ll
_{R}P_{x_{2}}$ for every $x\in X_{1}$ by Lemma~\ref{foreigner-not-inside-1}.
Thus, $P_{z}$ does not intersect the line segment $\ell $ in $R$, since $%
\ell $ lies between $P_{u}$ and $P_{x_{2}}$ in $R$ by definition of $\ell $,
which is a contradiction. Thus, $z\notin {V(\mathcal{B}_{1})}$.

Suppose first that $z$ is bounded, and thus also $z\notin Q_{u}$. We will
prove that $z\in N(X_{1})\cup \bigcup\nolimits_{i=1}^{\infty }H_{i}\cup
\bigcup\nolimits_{i=0}^{\infty }H_{i}^{\prime }$. To this end, we
distinguish the cases where $z\in V_{0}(u)$, $z\in N(u)$, and $z\in
V\setminus N[u]\setminus V_{0}(u)$. Recall by Lemma~\ref{module-1} that $V({%
C_{u}\cup C_{2}\cup H)}$ induces a subgraph of ${G\setminus Q_{u}\setminus
N[X_{1}]\setminus \mathcal{B}_{1}}$ that includes all connected components
of ${G\setminus Q_{u}\setminus N[X_{1}]\setminus \mathcal{B}_{1}}$, in which
the vertices of~${S_{2}\cup \{u\}}$ belong. Let first $z\in V\setminus
N[u]\setminus V_{0}(u)$, i.e.~$z\in V\setminus Q_{u}\setminus N[u]\setminus
V_{0}(u)$. Then either $P_{u}\ll _{R}P_{x_{2}}\ll _{R}P_{z}$ or $P_{z}\ll
_{R}P_{x}\ll _{R}P_{u}\ll _{R}P_{x_{2}}$ for every $x\in X_{1}$ by Lemma~\ref%
{foreigner-not-inside-1}, and thus $P_{z}$ does not intersect $\ell $ in $R$%
, which is a contradiction. Let now $z\in V_{0}(u)$; then~$z\in S_{2}$,
since $P_{z}$ intersects $\ell $ in $R$ (i.e.~$P_{z}\not\ll _{R}P_{u}$).
Then, since $z\notin X_{1}\cup Q_{u}\cup V(\mathcal{B}_{1})$, it follows
that either $z\in N(X_{1})$ or $z\in V(C_{u}\cup C_{2}\cup H)$. Therefore,
since we assumed that $z\notin V(G_{0})$, it follows that either $z\in
N(X_{1})$ or $z\in \bigcup\nolimits_{i=1}^{\infty }H_{i}\cup
\bigcup\nolimits_{i=0}^{\infty }H_{i}^{\prime }$, i.e.~$z\in N(X_{1})\cup
\bigcup\nolimits_{i=1}^{\infty }H_{i}\cup \bigcup\nolimits_{i=0}^{\infty
}H_{i}^{\prime }$. Let finally $z\in N(u)$. If $z\notin N(X_{1})$, then $%
z\in V(C_{u}\cup H)$ by the definition of $H$ and by Lemma~\ref{N(w)-1}.
That is, either $z\in N(X_{1})$ or $z\in V(C_{u}\cup H)$. Thus, since we
assumed that $z\notin V(G_{0})$, it follows again that either $z\in N(X_{1})$
or $z\in \bigcup\nolimits_{i=1}^{\infty }H_{i}\cup
\bigcup\nolimits_{i=0}^{\infty }H_{i}^{\prime }$, i.e.~$z\in N(X_{1})\cup
\bigcup\nolimits_{i=1}^{\infty }H_{i}\cup \bigcup\nolimits_{i=0}^{\infty
}H_{i}^{\prime }$. Summarizing, if $z$ is bounded, then $z\in N(X_{1})\cup
\bigcup\nolimits_{i=1}^{\infty }H_{i}\cup \bigcup\nolimits_{i=0}^{\infty
}H_{i}^{\prime }$.

Suppose now that $z$ is unbounded and $\phi _{z}<\phi _{\ell }$ (in both $R$
and $R_{\ell }$). Then, $a_{\ell }<_{R}L(z)$ and $l(z)<_{R}b_{\ell }$.
Recall that $z\notin X_{1}$; furthermore also $z\notin N(X_{1})$, since $z$
is unbounded and every vertex of $N=N(X_{1})$ is bounded by Lemma~\ref%
{N-H-C2-Cu-bounded}. Therefore, $z\notin N[X_{1}]$. We distinguish now in
the definition of the line segment $\ell $, the cases where $a_{\ell
}<_{R}L(x_{2})$ and $a_{\ell }=_{R}L(x_{2})$ in $R$.

\emph{Case 1.} $a_{\ell }<_{R}L(x_{2})$. Then $a_{\ell }=_{R}L(u)+\Delta $
in $R$, and thus $\phi _{\ell }=\phi _{u}$ in $R$ by definition of the line
segment $\ell $. Therefore, $\phi _{z}<\phi _{\ell }=\phi _{u}$ in $R$ for
some unbounded vertex $z$, since we assumed that $\phi _{z}<\phi _{\ell }$
in $R$. This is a contradiction, since $\phi _{u}=\min \{\phi _{x}$ in $R\
|\ x\in V_{U}\}$ by our initial assumption on $u$.

\emph{Case 2.} $a_{\ell }=_{R}L(x_{2})$. Recall that $P_{w}\ll _{R}P_{x_{2}}$%
. Then, $R(w)<_{R}L(x_{2})=_{R}a_{\ell }<_{R}L(z)$ and $l(z)<_{R}b_{\ell
}=_{R}r(w)<_{R}l(x_{2})$, since we assumed that $\phi _{z}<\phi _{\ell }$.
Therefore, $P_{z}$ intersects both $P_{w}$ and $P_{x_{2}}$ in $R$, while
also $\phi _{z}<\phi _{w}$ and $\phi _{z}<\phi _{x_{2}}$ in $R$. Thus $z\in
N(w)\cap N(x_{2})$, since both $w$ and $x_{2}$ are bounded. Therefore, since
also $z\notin N[X_{1}]$, it follows that $z\in H$ by definition of $H$. If $%
z\in H\setminus \bigcup\nolimits_{i=1}^{\infty }H_{i}\setminus
\bigcup\nolimits_{i=0}^{\infty }H_{i}^{\prime }$, then $z\in V(G_{0})$,
which is a contradiction. Therefore, $z\in \bigcup\nolimits_{i=1}^{\infty
}H_{i}\setminus \bigcup\nolimits_{i=0}^{\infty }H_{i}^{\prime }$.

Summarizing, if $z$ is adjacent to $v$ in $R_{\ell }$ for a vertex $z\in
V(G\setminus G_{0})$ and a vertex $v\in V(G_{0})\setminus \{u\}$, then $z\in
N(X_{1})\cup \bigcup\nolimits_{i=1}^{\infty }H_{i}\cup
\bigcup\nolimits_{i=0}^{\infty }H_{i}^{\prime }$. This completes the proof
of the lemma.
\end{proof}

\begin{corollary}
\label{R-ell-N(u)-intersect}For every $z\in N(u)$, $P_{z}$ intersects $P_{u}$
in $R_{\ell }$.
\end{corollary}

\begin{proof}
If $z\in V(G_{0})$, then $P_{z}$ intersects $P_{u}$ in $R_{0}$, since $R_{0}$
is a projection representation of $G_{0}$. Therefore, $P_{z}$ intersects $%
P_{u}$ also in $R_{\ell }$, since $R_{0}$ is a sub-representation of $%
R_{\ell }$. Suppose now that $z\notin V(G_{0})$. Then, either $z\in N(X_{1})$
or $z\in V(C_{u}\cup H)$, since we assumed that $z\in N(u)$. Thus, either $%
z\in N(X_{1})$ or $z\in \bigcup\nolimits_{i=1}^{\infty }H_{i}\cup
\bigcup\nolimits_{i=0}^{\infty }H_{i}^{\prime }$, since $z\notin V(G_{0})$,
and thus $z$ is adjacent to every vertex $v$ of $G_{0}\setminus \{u\}$ by
Lemma~\ref{module-2}. Therefore, $P_{z}$ intersects the line segment $\ell $
in both $R$ and $R_{\ell }$ (cf.~the proof of Lemma~\ref{R-ell}), and thus
in particular $P_{z}$ intersects also $P_{u}$ in $R_{\ell }$.
\end{proof}

\medskip

Note that, since the position and the slope of $P_{u}$ is not the same in $R$
and in $R_{\ell }$, the projection representation $R_{\ell }$ may be \emph{%
not} a projection representation of $G$. Similarly to the Transformations~%
\ref{trans1},~\ref{trans2}, and~\ref{trans3} in the proof of Theorem~\ref%
{right-property-thm}, we define in the sequel the Transformations~\ref%
{trans4},~\ref{trans5}, and~\ref{trans6}. After applying these
transformations to $R_{\ell }$, we obtain eventually a projection
representation $R^{\ast }$ of $G$ with $k-1$ unbounded vertices. The
following lemma will be mainly used in the remaining part of the proof of
Theorem~\ref{no-property-thm}.

\begin{lemma}
\label{R-ell-right-property}$u$ has the right border property in $R_{\ell }$.
\end{lemma}

\begin{proof}
Recall first that $u$ has the right border property in $R_{0}$. Suppose for
the sake of contradiction that $u$ has not the right border property in $%
R_{\ell }$. Then, there exist vertices $z\in N(u)$ and $y\in V_{0}(u)$, such
that $P_{z}\ll _{R_{\ell }}P_{y}$. We will now prove that $b_{u}<_{R_{\ell
}}r(z)$ for the lower right endpoint $r(z)$ of every $z\in N(u)$. If $z\in
V(G_{0})$, then clearly $b_{u}<_{R_{\ell }}r(z)$, since $b_{u}<b_{\ell }$
and $R_{0}$ is an $\varepsilon $-squeezed projection representation of $%
G_{0} $ with respect to $\ell $, where $\varepsilon >0$ is sufficiently
small. If $z\notin V(G_{0})$, then $b_{u}=r(u)<_{R}r(z)$ in $R$ by Lemma~\ref%
{unbounded-bounded}, and thus also $b_{u}<_{R_{\ell }}r(z)$, since the
endpoints of $P_{z}$ remain the same in both $R$ and $R_{\ell }$. That is, $%
b_{u}<_{R_{\ell }}r(z)$ for every $z\in N(u)$.

\emph{Case 1.} Let first $z\in V(G_{0})$. Then, $y\notin V(G_{0})$, since $u$
has the right border property in $R_{0}$. Furthermore $b_{u}<_{R_{\ell
}}r(z)<_{R_{\ell }}r(y)$, since $P_{z}\ll _{R_{\ell }}P_{y}$. Therefore,
since $y\notin V(G_{0})$, i.e.~since the endpoints of $P_{y}$ remain the
same in both $R$ and $R_{\ell }$, it follows that also $b_{u}<_{R}r(y)$.
Thus $y\in S_{2}$, since we assumed that $y\in V_{0}(u)$; therefore in
particular $y\notin X_{1}$, since $X_{1}\subseteq D_{1}$ by Lemma~\ref%
{N(w)-1}. Furthermore, $y\notin Q_{u}$ by Lemma~\ref{Qu-1} and $y\notin V({%
\mathcal{B}_{1}})$ by definition of ${\mathcal{B}_{1}}$, since $y\in
V_{0}(u) $. Recall now by Lemma~\ref{module-1} that $V({C_{u}\cup C_{2}\cup
H)}$ induces a subgraph of ${G\setminus Q_{u}\setminus N[X_{1}]\setminus 
\mathcal{B}_{1}}$ that includes all connected components of ${G\setminus
Q_{u}\setminus N[X_{1}]\setminus \mathcal{B}_{1}}$, in which the vertices of~%
${S_{2}\cup \{u\}}$ belong. Therefore, since $y\in S_{2}$ and $y\notin
Q_{u}\cup X_{1}\cup V({\mathcal{B}_{1}})$, it follows that $y\in
N(X_{1})\cup V({C_{u}\cup C_{2}\cup H)}$. Thus $y\in N(X_{1})\cup
\bigcup\nolimits_{i=1}^{\infty }H_{i}\cup \bigcup\nolimits_{i=0}^{\infty
}H_{i}^{\prime }$, since otherwise $y\in V(G_{0})$, which is a
contradiction. Therefore, $y$ is adjacent to every vertex $v\in
V(G_{0})\setminus \{u\}$ by Lemma~\ref{module-2}. Thus, in particular, $%
P_{y} $ intersects $P_{z}$ in $R_{\ell }$, since $z\in V(G_{0})\setminus
\{u\}$ and $R_{\ell }\setminus \{u\}$ is a projection representation of $%
G\setminus \{u\}$ by Lemma~\ref{R-ell}. This is a contradiction, since we
assumed that $P_{z}\ll _{R_{\ell }}P_{y}$.

\emph{Case 2.} Let now $z\notin V(G_{0})$. Since we assumed that $z\in N(u)$%
, it follows that either $z\in N(X_{1})$ or $z\in V(C_{u}\cup H)$.
Therefore, either $z\in N(X_{1})$ or $z\in \bigcup\nolimits_{i=1}^{\infty
}H_{i}\cup \bigcup\nolimits_{i=0}^{\infty }H_{i}^{\prime }$, since $z\notin
V(G_{0})$, and thus $z$ is adjacent to every vertex $v\in V(G_{0})\setminus
\{u\}$ by Lemma~\ref{module-2}. Then, in particular, $P_{z}$ intersects $%
P_{v}$ in $R_{\ell }$, for every vertex $v\in V(G_{0})\setminus \{u\}$, and
thus $y\notin V(G_{0})$, since we assumed that $P_{z}\ll _{R_{\ell }}P_{y}$.
Therefore, since both $y,z\notin V(G_{0})$ and $P_{z}\ll _{R_{\ell }}P_{y}$,
it follows that also $P_{z}\ll _{R}P_{y}$, and thus in particular $%
b_{u}<_{R}r(z)<_{R}r(y)$ by Lemma~\ref{unbounded-bounded}. Thus $y\in S_{2}$%
, since we assumed that $y\in V_{0}(u)$; therefore in particular $y\notin
X_{1}$, since $X_{1}\subseteq D_{1}$ by Lemma~\ref{N(w)-1}. Furthermore, $%
y\notin Q_{u}$ by Lemma~\ref{Qu-1} and $y\notin V({\mathcal{B}_{1}})$ by
definition of ${\mathcal{B}_{1}}$, since $y\in V_{0}(u)$. Therefore, since $%
y\in S_{2}$ and $y\notin Q_{u}\cup X_{1}\cup V({\mathcal{B}_{1}})$, it
follows (similarly to the previous paragraph) that $y\in N(X_{1})\cup V({%
C_{u}\cup C_{2}\cup H)}$. Thus $y\in N(X_{1})\cup
\bigcup\nolimits_{i=1}^{\infty }H_{i}\cup \bigcup\nolimits_{i=0}^{\infty
}H_{i}^{\prime }$, since otherwise $y\in V(G_{0})$, which is a contradiction.

Suppose that $y\in N(X_{1})$, i.e.~$y\in N(x)$ for some $x\in X_{1}$. Recall
that $P_{x}\ll _{R}P_{u}$, since $X_{1}\subseteq D_{1}$ by Lemma~\ref{N(w)-1}%
. If $P_{u}\ll _{R}P_{y}$, then $P_{x}\ll _{R}P_{u}\ll _{R}P_{y}$, i.e.~$%
y\notin N(x)$, which is a contradiction. Thus $P_{u}\not\ll _{R}P_{y}$, i.e.
either $P_{y}$ intersects $P_{u}$ in $R$ or $P_{y}\ll _{R}P_{u}$. Suppose
that $P_{y}$ intersects $P_{u}$ in $R$, and thus either $N(y)\subseteq N(u)$
or $N(u)\subseteq N(y)$ by Lemma~\ref{intersecting-unbounded}, since $%
y\notin N(u)$. If $N(y)\subseteq N(u)$, then $x\in N(u)$, where $x\in X_{1}$%
, which is a contradiction. If $N(u)\subseteq N(y)$, then $z\in N(y)$, which
is a contradiction, since we assumed that $P_{z}\ll _{R_{\ell }}P_{y}$.
Therefore, $P_{y}$ does not intersect $P_{u}$ in $R$, and thus $P_{y}\ll
_{R}P_{u}$, i.e.~$P_{z}\ll _{R}P_{y}\ll _{R}P_{u}$. Then $z\notin N(u)$,
which is a contradiction. Therefore, $y\notin N(X_{1})$, and thus $y\in
\bigcup\nolimits_{i=1}^{\infty }H_{i}\cup \bigcup\nolimits_{i=0}^{\infty
}H_{i}^{\prime }$. On the other hand $y\notin \bigcup\nolimits_{i=0}^{\infty
}H_{i}^{\prime }$, since otherwise $y\in N(u)$ by Lemma~\ref{tilde-neighbors}%
, which is a contradiction. Thus $y\in \bigcup\nolimits_{i=1}^{\infty }H_{i}$%
. Summarizing, $z\notin V(G_{0})$ and $y=v_{i}\in H_{i}$ for some $i\geq 1$.

We will now prove by induction on $i$ that $v_{i}\in N(u)$ or $P_{z}\not\ll
_{R}P_{v_{i}}$, for every vertex $v_{i}\in H_{i}$, $i\geq 1$. This then
completes the proof of the lemma, since $v_{i}=y\notin N(u)$ (by the
assumption that $y\in V_{0}(u)$), and thus $P_{z}\not\ll _{R}P_{v_{i}}=P_{y}$%
, which is a contradiction (since we assumed that $P_{z}\ll _{R_{\ell
}}P_{y} $, and thus also $P_{z}\ll _{R}P_{y}$).

For the sake of contradiction, suppose that $v_{i}\notin N(u)$ and $P_{z}\ll
_{R}P_{v_{i}}$ for some $i\geq 1$. Then, note that $z\notin N(v_{i})$.
Recall that $v_{i}\in N(x_{2})$ due to the definition of $H$, and since $%
v_{i}\in H$. Therefore, since $v_{i}\notin N(u)$ and $x_{2}\in V_{0}(u)$, it
follows that $v_{i}\in V_{0}(u)$, and thus $T_{v_{i}}\ll _{R_{T}}T_{u}$ in
the trapezoid representation $R_{T}$. Therefore, also $T_{v_{i}}\ll
_{R_{T}}T_{z}$, since $z\in N(u)\setminus N(v_{i})$. Recall now that $%
T_{x}\ll _{R_{T}}T_{x_{2}}$ for every $x\in X_{1}$ by Lemma~\ref{N(w)-2}.
Thus, since $v_{i}\in N(x_{2})$ and $v_{i}\notin N(X_{1})$ by definition of $%
H$, it follows that $T_{x}\ll _{R_{T}}T_{v_{i}}$ for every $x\in X_{1}$,
i.e.~$T_{x}\ll _{R_{T}}T_{v_{i}}\ll _{R_{T}}T_{z}$ for every $x\in X_{1}$.
Thus, in particular, $z\notin N(X_{1})$.

For the induction basis, let $i=1$. Suppose that $N_{1}(z)=N$. Then, for
every $v\in N$, $T_{v}$ intersects $T_{v_{1}}$ in $R_{T}$, i.e.~$v\in
N(v_{1})$, since $v\in N(X_{1})\cap N(z)$ and $T_{x}\ll _{R_{T}}T_{v_{1}}\ll
_{R_{T}}T_{z}$ for every $x\in X_{1}$. Thus, $N_{1}(v_{1})=N$, i.e.~$%
N=H_{0}\subseteq N(v_{1})$, which is a contradiction by Definition~\ref{Hi},
since $v_{1}\in H_{1}$.

Therefore $N_{1}(z)\neq N$, and thus there exists a vertex $v\in N\setminus
N(z)$, i.e.~$v\in N(x)\setminus N(z)$ for some $x\in X_{1}$. Then $v\in
N(x_{2})$, since $N_{1}(x_{2})=N=N(X_{1})$ by Lemma~\ref{N-Cu}. Thus, since $%
v\in N(x)\cap N(x_{2})$ and $P_{x}\ll _{R}P_{u}\ll _{R}P_{x_{2}}$, it
follows that $P_{v}$ intersects $P_{u}$ in $R$. If $v\notin N(u)$, then
either $N(v)\subseteq N(u)$ or $N(u)\subseteq N(v)$ by Lemma~\ref%
{intersecting-unbounded}. If $N(v)\subseteq N(u)$, then $x_{2}\in N(u)$,
which is a contradiction. If $N(u)\subseteq N(v)$, then $z\in N(v)$, which
is again a contradiction. Therefore, $v\in N(u)$ for all vertices $v\in
N\setminus N(z)$.

Consider now the trapezoid representation $R_{T}$. Recall that $T_{x}\ll
_{R_{T}}T_{v_{1}}\ll _{R_{T}}T_{u}$ and $T_{x}\ll _{R_{T}}T_{v_{1}}\ll
_{R_{T}}T_{z}$ for every $x\in X_{1}$. Consider an arbitrary vertex $v\in
N=N(X_{1})$. If $v\in N(z)$, then $T_{v}$ intersects $T_{v_{1}}$ in $R_{T}$,
since $v\in N(X_{1})\cap N(z)$ and $T_{x}\ll _{R_{T}}T_{v_{1}}\ll
_{R_{T}}T_{z}$ for every $x\in X_{1}$; therefore $v\in N(v_{1})$. Otherwise,
if $v\notin N(z)$, then $v\in N(u)$, as we proved in the previous paragraph.
Then, $T_{v}$ intersects $T_{v_{1}}$ in $R_{T}$, since $v\in N(X_{1})\cap
N(u)$ and $T_{x}\ll _{R_{T}}T_{v_{1}}\ll _{R_{T}}T_{u}$ for every $x\in
X_{1} $; therefore again $v\in N(v_{1})$. Thus, $v\in N(v_{1})$ for every $%
v\in N$, i.e.~$N=H_{0}\subseteq N(v_{1})$, which is a contradiction by
Definition~\ref{Hi}, since $v_{1}\in H_{1}$. Therefore, $v_{1}\in N(u)$ or $%
P_{z}\not\ll _{R}P_{v_{1}}$ for every vertex $v_{1}\in H_{1}$. This proves
the induction basis.

For the induction step, let $i\geq 2$. Let $(v_{0},v_{1},\ldots
,v_{i-2},v_{i-1},v_{i})$ be an $H_{i}$-chain of $v_{i}$. By the induction
hypothesis, $v_{i-1}\in N(u)$ or $P_{z}\not\ll _{R}P_{v_{i-1}}$. Recall that 
$T_{v_{i}}\ll _{R_{T}}T_{z}$, as we proved above. Assume that $z\in
N(v_{i-1})$. Then, since $z\in N(v_{i-1})\setminus N(v_{i})$ and $v_{i-2}\in
N(v_{i})\setminus N(v_{i-1})$, $P_{v_{i}}$ does not intersect $P_{v_{i-1}}$
in $R$ by Lemma~\ref{intersecting-unbounded}. Suppose first that $i$ is
even. Then, $P_{v_{i-2}}\ll _{R}P_{v_{i-1}}$ by Lemmas~\ref%
{alternating-bounded-chain-1} and~\ref{alternating-bounded-chain}. Thus,
since $v_{i}\in N(v_{i-2})$ and $P_{v_{i}}$ does not intersect $P_{v_{i-1}}$
in $R$, it follows that $P_{v_{i}}\ll _{R}P_{v_{i-1}}$. Then, since we
assumed that $P_{z}\ll _{R}P_{v_{i}}$, it follows that $P_{z}\ll
_{R}P_{v_{i}}\ll _{R}P_{v_{i-1}}$, i.e.~$z\notin N(v_{i-1})$. This is a
contradiction to the assumption that $z\in N(v_{i-1})$. Suppose now that $i$
is odd, i.e.~$i\geq 3$. Then, $T_{v_{i-1}}\ll _{R_{T}}T_{v_{i-2}}$ by Lemma~%
\ref{alternating-bounded-chain}. Thus, since $v_{i}\in N(v_{i-2})\setminus
N(v_{i-1})$, it follows that $T_{v_{i-1}}\ll _{R_{T}}T_{v_{i}}$. Then, since 
$T_{v_{i}}\ll _{R_{T}}T_{z}$, it follows that $T_{v_{i-1}}\ll
_{R_{T}}T_{v_{i}}\ll _{R_{T}}T_{z}$, i.e.~$z\notin N(v_{i-1})$. This is
again a contradiction to the assumption that $z\in N(v_{i-1})$.

Therefore $z\notin N(v_{i-1})$. Recall that $v_{i-1}$ is a bounded vertex by
Lemma~\ref{N-H-C2-Cu-bounded}. Furthermore, $z$ is a bounded vertex, since $%
z\in N(u)$. Therefore, since $z\notin N(v_{i-1})$, it follows that $%
P_{v_{i-1}}$ does not intersect $P_{z}$ in $R$, i.e.~either $P_{v_{i-1}}\ll
_{R}P_{z}$ or $P_{z}\ll _{R}P_{v_{i-1}}$.

\emph{Case 2a.} $P_{v_{i-1}}\ll _{R}P_{z}$. Then, since $z\in N(u)$ and $%
P_{u}\ll _{R}P_{x_{2}}$, it follows by Lemma~\ref{unbounded-bounded} that $%
R(v_{i-1})<_{R}L(z)<_{R}L(u)<_{R}L(x_{2})$, i.e.~$R(v_{i-1})<_{R}L(x_{2})$.
Thus, since $v_{i-1}\in N(x_{2})$ and $P_{u}\ll _{R}P_{x_{2}}$, it follows
that $r(u)<_{R}l(x_{2})<_{R}r(v_{i-1})$. That is, $R(v_{i-1})<_{R}L(u)=R(u)$
and $r(u)<_{R}r(v_{i-1})$, i.e.~$P_{v_{i-1}}$ intersects $P_{u}$ in $R$ and $%
\phi _{v_{i-1}}>\phi _{u}$. If $v_{i-1}\notin N(u)$, then $%
N(v_{i-1})\subseteq N(u)$ by Lemma~\ref{intersecting-unbounded}, and thus $%
x_{2}\in N(u)$, which is a contradiction. Thus, $v_{i-1}\in N(u)$.

Since $P_{v_{i-1}}\ll _{R}P_{z}$ and $P_{z}\ll _{R}P_{v_{i}}$ by assumption,
it follows that $P_{v_{i-1}}\ll _{R}P_{v_{i}}$. Recall by Lemmas~\ref%
{alternating-bounded-chain-1} and~\ref{alternating-bounded-chain} that
either $P_{v_{i-2}}\ll _{R}P_{v_{i-1}}$ or $P_{v_{i-1}}\ll _{R}P_{v_{i-2}}$.
If $P_{v_{i-2}}\ll _{R}P_{v_{i-1}}$, then $P_{v_{i-2}}\ll _{R}P_{v_{i-1}}\ll
_{R}P_{v_{i}}$, i.e.~$v_{i-2}v_{i}\notin E$, which is a contradiction.
Therefore, $P_{v_{i-1}}\ll _{R}P_{v_{i-2}}$ and $i$ is odd, and thus $%
T_{v_{i-1}}\ll _{R_{T}}T_{v_{i-2}}$ by Lemmas~\ref%
{alternating-bounded-chain-1} and~\ref{alternating-bounded-chain}. Thus,
since $v_{i}\in N(v_{i-2})\setminus N(v_{i-1})$, it follows that also $%
T_{v_{i-1}}\ll _{R_{T}}T_{v_{i}}$. Recall now that $T_{v_{i}}\ll
_{R_{T}}T_{u}$, as we proved above. Therefore, it follows that $%
T_{v_{i-1}}\ll _{R_{T}}T_{v_{i}}\ll _{R_{T}}T_{u}$, and thus $v_{i-1}\notin
N(u)$, which is a contradiction by the previous paragraph.

\emph{Case 2b.} $P_{z}\ll _{R}P_{v_{i-1}}$. Then, $v_{i-1}\in N(u)$ by the
induction hypothesis, and thus~$v_{i-1}$ is bounded. Furthermore, $v_{i}$ is
also bounded by Lemma~\ref{N-H-C2-Cu-bounded}, since $v_{i}\in H$.
Therefore, $P_{v_{i}}$ does not intersect $P_{v_{i-1}}$ in $R$, since $%
v_{i-1}v_{i}\notin E$, and thus either $P_{v_{i}}\ll _{R}P_{v_{i-1}}$ or $%
P_{v_{i-1}}\ll _{R}P_{v_{i}}$. Recall that $v_{i}\notin N(u)$ and $P_{z}\ll
_{R}P_{v_{i}}$ by assumption. Suppose first that $P_{v_{i}}\ll
_{R}P_{v_{i-1}}$, that is, $P_{z}\ll _{R}P_{v_{i}}\ll _{R}P_{v_{i-1}}$.
Then, since $z\in N(u)$ and $v_{i-1}\in N(u)$, it follows that $P_{u}$
intersects $P_{v_{i}}$ in $R$. Since $v_{i}\notin N(u)$, either $%
N(v_{i})\subseteq N(u)$ or $N(u)\subseteq N(v_{i})$ by Lemma~\ref%
{intersecting-unbounded}. If $N(v_{i})\subseteq N(u)$, then $x_{2}\in N(u)$,
which is a contradiction. If $N(u)\subseteq N(v_{i})$, then $v_{i-1}\in
N(v_{i})$, which is again a contradiction.

Suppose now that $P_{v_{i-1}}\ll _{R}P_{v_{i}}$. Recall by Lemmas~\ref%
{alternating-bounded-chain-1} and~\ref{alternating-bounded-chain} that
either $P_{v_{i-2}}\ll _{R}P_{v_{i-1}}$ or $P_{v_{i-1}}\ll _{R}P_{v_{i-2}}$.
If $P_{v_{i-2}}\ll _{R}P_{v_{i-1}}$, then $P_{v_{i-2}}\ll _{R}P_{v_{i-1}}\ll
_{R}P_{v_{i}}$, i.e.~$v_{i-2}v_{i}\notin E$, which is a contradiction.
Therefore, $P_{v_{i-1}}\ll _{R}P_{v_{i-2}}$ and $i$ is odd, and thus $%
T_{v_{i-1}}\ll _{R_{T}}T_{v_{i-2}}$ by Lemmas~\ref%
{alternating-bounded-chain-1} and~\ref{alternating-bounded-chain}. Thus,
since $v_{i}\in N(v_{i-2})\setminus N(v_{i-1})$, it follows that also $%
T_{v_{i-1}}\ll _{R_{T}}T_{v_{i}}$. Recall now that $T_{v_{i}}\ll
_{R_{T}}T_{u}$, as we proved above. Therefore, $T_{v_{i-1}}\ll
_{R_{T}}T_{v_{i}}\ll _{R_{T}}T_{u}$, and thus $v_{i-1}\notin N(u)$, which is
a contradiction. This completes the induction step and the lemma follows.
\end{proof}

\subsection*{The projection representations $R_{\ell }^{\prime }$, 
$R_{\ell }^{\prime\prime }$, and $R_{\ell }^{\prime \prime \prime }$}

\begin{notation}
\label{Notation-N(u)}In the following, whenever we refer to $N(u)$, we will
mean $N_{G}(u)$, i.e.~the neighborhood set of vertex $u$ in $G$. Note that,
since $R_{\ell }$ may be \emph{not} a projection representation of $G$
(although $R_{\ell }\setminus \{u\}$ is a projection representation of $%
G\setminus \{u\}$ by Lemma~\ref{R-ell}), the set $N_{G}(u)$ does not
coincide necessarily with the set of adjacent vertices of $u$ in the graph
induced by $R_{\ell }$.
\end{notation}

Similarly to the proof of Theorem~\ref{right-property-thm}, we add to $G$ an
isolated bounded vertex $t$. This isolated vertex $t$ corresponds to a
parallelogram $P_{t}$, such that $P_{v}\ll _{R}P_{t}$ and $P_{v}\ll
_{R_{\ell }}P_{t}$ for every other vertex $v$ of $G$. Denote by $V_{B}$ and $%
V_{U}$ the set of bounded and unbounded vertices of $G$ in $R_{\ell }$,
after the addition of the auxiliary vertex $t$ to $G$ (note that $t\in V_{B}$%
).

Now, we define for every $z\in N(u)$ the value $L_{0}(z)=\min_{R_{\ell
}}\{L(x)\ |\ x\in V_{B}\setminus N(u),P_{z}\ll _{R_{\ell }}P_{x}\}$. For
every vertex $x\in V_{B}\setminus N(u)$, such that $P_{z}\ll _{R_{\ell
}}P_{x}$ for some $z\in N(u)$, it follows that $x\notin V_{0}(u)$, since $u$
has the right border property in $R_{\ell }$ by Lemma~\ref%
{R-ell-right-property}. Thus, for every $z\in N(u)$, $L_{0}(z)=\min_{R_{\ell
}}\{L(x)\ |\ x\in V_{B}\setminus N(u)\setminus V_{0}(u),P_{z}\ll _{R_{\ell
}}P_{x}\}$. Note that the value $L_{0}(z)$ is well defined for every $z\in
N(u)$, since in particular $t\in V_{B}\setminus N(u)$ and $P_{z}\ll
_{R_{\ell }}P_{t}$. Furthermore, note that for every every $z\in N(u)$, the
endpoint $L_{0}(z)$ does not correspond to any vertex of $G_{0}$, since $%
V(G_{0})\subseteq N[u]\cup V_{0}(u)$ by Observation~\ref{V(G0)}. Define now
the the value $\ell _{0}=\max_{R_{\ell }}\{l(x)\ |\ x\in V_{0}(u)\}$ and the
subset $N_{1}=\{z\in N(u)\ |\ r(z)<_{R_{\ell }}\ell _{0}\}$ of neighbors of $%
u$ (in $G$, and not in $R_{\ell }$). Similarly to Transformation~\ref{trans1}
in the proof of Theorem~\ref{right-property-thm}, we construct now the
projection representation $R_{\ell }^{\prime }$ from $R_{\ell }$ as follows.

\begin{transformation}
\label{trans4}For every $z\in N_{1}$, move the right line of $P_{z}$
parallel to the right, until either $r(z)$ comes immediately after $\ell
_{0} $ on $L_{2}$, or $R(z)$ comes immediately before $L_{0}(z)$ on $L_{1}$.
Denote the resulting projection representation by $R_{\ell }^{\prime }$.
\end{transformation}

\begin{remark}
\label{trans4-remark-2}Suppose now that the endpoint $\ell _{0}$ corresponds
to a vertex of $V(G_{0})$, i.e.~${b_{\ell }-\varepsilon <_{R_{\ell }}\ell
_{0}<_{R_{\ell }}b_{\ell }+\varepsilon }$ by Remark~\ref{trans4-remark-1}.
Then, since $\varepsilon $ has been chosen to be sufficiently small, we make
w.l.o.g.~the following convention in the statement of Transformation~\ref%
{trans4}: for every vertex $z\in N_{1}$, such that $z\notin V(G_{0})$,
either $r(z)<_{R_{\ell }^{\prime }}b_{\ell }-\varepsilon $ (in the case
where $r(z)<_{R_{\ell }^{\prime }}\ell _{0}$) or~$r(z)$ comes immediately
after $b_{\ell }+\varepsilon $ on $L_{2}$, i.e.~$r(z)>_{R_{\ell }^{\prime
}}b_{\ell }+\varepsilon $ (in the case where $r(z)>_{R_{\ell }^{\prime
}}\ell _{0}$). Summarizing, similarly to $R_{\ell }$, we may assume in $%
R_{\ell }^{\prime }$ w.l.o.g.~that for every vertex $z\in N(u)$, such that $%
z\notin V(G_{0})$, either $r(z)<_{R_{\ell }^{\prime }}b_{\ell }-\varepsilon $
or $r(z)>_{R_{\ell }^{\prime }}b_{\ell }+\varepsilon $.
\end{remark}

Note that the left lines of all parallelograms do not move during
Transformation~\ref{trans4}. Thus, in particular, the value of $\ell _{0}$
is the same in $R_{\ell }$ and in $R_{\ell }^{\prime }$, i.e.~$\ell
_{0}=\max_{R_{\ell }^{\prime }}\{l(x)\ |\ x\in V_{0}(u)\}$. As we will prove
in Lemma~\ref{R'-ell}, the representation $R_{\ell }^{\prime }\setminus
\{u\} $ is a projection representation of the graph $G\setminus \{u\}$, and
thus the parallelograms of two bounded vertices intersect in $R_{\ell }$ if
and only if they intersect also in $R_{\ell }^{\prime }$. Therefore, for
every $z\in N(u)$, the value $L_{0}(z)$ remains the same in $R_{\ell }$ and
in $R_{\ell }^{\prime }$, i.e.~$L_{0}(z)=\min_{R_{\ell }^{\prime }}\{L(x)\
|\ x\in V_{B}\setminus N(u)\setminus V_{0}(u),P_{z}\ll _{R_{\ell }^{\prime
}}P_{x}\}$ for every $z\in N(u)$. Similarly to the proof of Theorem~\ref%
{right-property-thm}, we define now the subset $N_{2}=\{z\in N(u)\ |\ \ell
_{0}<_{R_{\ell }^{\prime }}r(z)\}$ of neighbors of $u$. Since the lower
right endpoint $r(z)$ of all parallelograms $P_{z}$ in $R_{\ell }^{\prime }$
is greater than or equal to the corresponding value $r(z)$ in $R_{\ell }$,
it follows that $N(u)\setminus N_{1}=\{z\in N(u)\ |\ \ell _{0}<_{R_{\ell
}}r(z)\}\subseteq \{z\in N(u)\ |\ \ell _{0}<_{R_{\ell }^{\prime
}}r(z)\}=N_{2}$. Thus, $N(u)\setminus N_{2}\subseteq N_{1}$ and $N_{2}\cup
(N_{1}\setminus N_{2})=N(u)$. If $N_{2}\neq \emptyset $, we define the value 
$r_{0}=\min_{R_{\ell }^{\prime }}\{r(z)\ |\ z\in N_{2}\}$.

\begin{lemma}
\label{r(u)<r0-in-R-ell'}If $N_{2}\neq \emptyset $, i.e.~if the value $r_{0}$
can be defined, then $r(u)<_{R_{\ell }^{\prime }}r_{0}$.
\end{lemma}

\begin{proof}
Denote by $z_{0}$ the vertex of $N_{2}$, such that $r_{0}=r(z_{0})$. Let
first $z_{0}\in V(G_{0})$. Then $r(z_{0})>_{R_{0}}r(u)$ by Lemma~\ref%
{unbounded-bounded}, since $N_{2}\subseteq N(u)$, and since $R_{0}$ is a
projection representation of $G_{0}$. Thus, also $r(z_{0})>_{R_{\ell }}r(u)$%
, since $R_{0}$ is a sub-representation of $R_{\ell }$. Furthermore, $%
r_{0}=r(z_{0})>_{R_{\ell }^{\prime }}r(u)$, since the lower right endpoints $%
r(z)$ do not decrease by Transformation~\ref{trans4}. Let now $z_{0}\notin
V(G_{0})$. Then, either $r(z_{0})<_{R_{\ell }^{\prime }}b_{\ell
}-\varepsilon $ or $r(z_{0})>_{R_{\ell }^{\prime }}b_{\ell }+\varepsilon $
by Remark~\ref{trans4-remark-2}. Recall that $x_{2}\in V(G_{0})$, and thus $%
b_{\ell }-\varepsilon <_{R_{\ell }}l(x_{2})<_{R_{\ell }}b_{\ell
}+\varepsilon $ by Remark~\ref{trans4-remark-1}. Thus, since also $x_{2}\in
V_{0}(u)$, it follows by definition of $\ell _{0}$ that $b_{\ell
}-\varepsilon <_{R_{\ell }}l(x_{2})\leq _{R_{\ell }}\ell _{0}$. Therefore $%
b_{\ell }-\varepsilon <_{R_{\ell }^{\prime }}\ell _{0}<_{R_{\ell }^{\prime
}}r(z_{0})$, since $z_{0}\in N_{2}$. Thus $r(z_{0})>_{R_{\ell }^{\prime
}}b_{\ell }+\varepsilon $ by Remark~\ref{trans4-remark-2} (since $%
z_{0}\notin V(G_{0})$), i.e.~$r(z_{0})>_{R_{\ell }^{\prime }}b_{\ell
}+\varepsilon >_{R_{\ell }^{\prime }}r(u)$. Summarizing, $%
r_{0}=r(z_{0})>_{R_{\ell }^{\prime }}r(u)$ in all cases.
\end{proof}

\medskip

Define now the value $L_{0}=\min_{R_{\ell }}\{L(x)\ |\ x\in V_{B}\setminus
N(u)\setminus V_{0}(u),P_{u}\ll _{R_{\ell }}P_{x}\}$; again, $L_{0}$ is well
defined, since in particular $t\in V_{B}\setminus N(u)\setminus V_{0}(u)$
and $P_{u}\ll _{R_{\ell }}P_{t}$. Then, since by Transformation~\ref{trans4}
only some endpoints of vertices $z\in N(u)$ are moved, it follows that the
value $L_{0}$ does not change in $R_{\ell }^{\prime }$, i.e.~$%
L_{0}=\min_{R_{\ell }^{\prime }}\{L(x)\ |\ x\in V_{B}\setminus N(u)\setminus
V_{0}(u),P_{u}\ll _{R_{\ell }^{\prime }}P_{x}\}$. The following property of
the projection representation $R_{\ell }^{\prime }$ can be obtained easily
by Transformation~\ref{trans4}.

\begin{lemma}
\label{property-trans4}For all vertices $z\in N_{1}\setminus N_{2}$, for
which $R(z)<_{R_{\ell }^{\prime }}L_{0}$, the values $R(z)$ lie immediately
before $L_{0}$ in $R_{\ell }^{\prime }$.
\end{lemma}

\begin{proof}
Let $z\in N_{1}\setminus N_{2}$. By definition of the sets $N_{1}$ and $%
N_{2} $, it follows that $r(z)<_{R_{\ell }}\ell _{0}$ and $r(z)<_{R_{\ell
}^{\prime }}\ell _{0}$ in both $R_{\ell }$ and $R_{\ell }^{\prime }$. Thus, $%
R(z)$ comes immediately before $L_{0}(z)$ in $R_{\ell }^{\prime }$ during
Transformation~\ref{trans4}. We will now prove that $L_{0}\leq _{R_{\ell
}}L_{0}(z)$. Consider a vertex $x\in V_{B}\setminus N(u)\setminus V_{0}(u)$,
such that $P_{z}\ll _{R_{\ell }}P_{x}$, i.e.~$r(z)<_{R_{\ell }}l(x)$ and $%
R(z)<_{R_{\ell }}L(x)$. Then, in particular $x\notin V(G_{0})$, since $%
x\notin N(u)\cup V_{0}(u)$ and $V(G_{0})\subseteq N[u]\cup V_{0}(u)$ by
Observation~\ref{V(G0)}. Suppose that $P_{x}$ intersects $P_{u}$ in $R_{\ell
}$, i.e.~$P_{x}$ intersects the line segment $\ell $ in $R_{\ell }$. Then,
in particular $P_{x}$ intersects also $P_{x_{2}}$ in $R_{\ell }$, since $%
x_{2}\in V(G_{0})$, and thus $x\in N(x_{2})$, since both $x$ and $x_{2}$ are
bounded in $R_{\ell }$. Therefore $x\in V_{0}(u)$, since $x_{2}\in V_{0}(u)$
and $x\notin N(u)$, which is a contradiction. Thus, $P_{x}$ does not
intersect $P_{u}$ in $R_{\ell }$, i.e.~either $P_{x}\ll _{R_{\ell }}P_{u}$
or $P_{u}\ll _{R_{\ell }}P_{x}$. If $P_{x}\ll _{R_{\ell }}P_{u}$, then $%
P_{z}\ll _{R_{\ell }}P_{x}\ll _{R_{\ell }}P_{u}$, which is a contradiction,
since $P_{z}$ intersects $P_{u}$ in $R_{\ell }$ by Corollary~\ref%
{R-ell-N(u)-intersect}. Therefore, $P_{u}\ll _{R_{\ell }}P_{x}$. That is,
for every $x\in V_{B}\setminus N(u)\setminus V_{0}(u)$, for which $P_{z}\ll
_{R_{\ell }}P_{x}$, it follows that also $P_{u}\ll _{R_{\ell }}P_{x}$. Thus,
it follows by the definitions of $L_{0}$ and of $L_{0}(z)$ that $L_{0}\leq
_{R_{\ell }}L_{0}(z)$.

Furthermore, also $L_{0}\leq _{R_{\ell }^{\prime }}L_{0}(z)$ in $R_{\ell
}^{\prime }$, since by Transformation~\ref{trans4} only some endpoints of
vertices $z\in N(u)$ are moved. Therefore, since $R(z)$ comes immediately
before $L_{0}(z)$ in $R_{\ell }^{\prime }$ during Transformation~\ref{trans4}%
, it follows that either $R(z)$ comes immediately before $L_{0}$ in $R_{\ell
}^{\prime }$ during Transformation~\ref{trans4} (in the case where $%
L_{0}=_{R_{\ell }^{\prime }}L_{0}(z)$) or $R(z)>_{R_{\ell }^{\prime }}L_{0}$
(in the case where $L_{0}<_{R_{\ell }^{\prime }}L_{0}(z)$).
\end{proof}

\medskip

If $N_{2}=\emptyset $, then we set $R_{\ell }^{\prime \prime }=R_{\ell
}^{\prime }$; otherwise, if $N_{2}\neq \emptyset $, we construct the
projection representation $R_{\ell }^{\prime \prime }$ from $R_{\ell
}^{\prime }$ as follows.

\begin{transformation}
\label{trans5}For every $v\in V_{0}(u)\cap V_{B}$, such that $r(v)>_{R_{\ell
}^{\prime }}r_{0}$, we move the right line of $P_{v}$ in $R_{\ell }^{\prime
} $ to the left, such that $r(v)$ comes immediately before $r_{0}$ in $L_{2}$%
. Denote the resulting projection representation by $R_{\ell }^{\prime
\prime } $.
\end{transformation}

Since by Transformation~\ref{trans5} only some endpoints of vertices $v\in
V_{0}(u)\cap V_{B}$ are moved, it follows that the value $L_{0}$ does not
change in $R_{\ell }^{\prime \prime }$, i.e.~$L_{0}=\min_{R_{\ell }^{\prime
\prime }}\{L(x)\ |\ x\in V_{B}\setminus N(u)\setminus V_{0}(u),P_{u}\ll
_{R_{\ell }^{\prime \prime }}P_{x}\}$. The next property of the projection
representation $R_{\ell }^{\prime \prime }$ follows by Lemma~\ref%
{property-trans4}.

\begin{corollary}
\label{property-trans5}For all vertices $z\in N_{1}\setminus N_{2}$, for
which $R(z)<_{R_{\ell }^{\prime \prime }}L_{0}$, the values $R(z)$ lie
immediately before $L_{0}$ in $R_{\ell }^{\prime \prime }$.
\end{corollary}

\begin{proof}
Let $x_{0}$ be the vertex of $V_{B}\setminus N(u)\setminus V_{0}(u)$, such
that $L_{0}=L(x_{0})$. Recall by Lemma~\ref{property-trans4} that for all
vertices $z\in N_{1}\setminus N_{2}$, for which $R(z)<_{R_{\ell }^{\prime
}}L_{0}$, the values $R(z)$ lie immediately before~$L_{0}$ in $R_{\ell
}^{\prime }$. Furthermore, note that the parallelograms of all neighbors $%
z\in N(u)$ of $u$ do not move by Transformation~\ref{trans5}. Therefore,
since also the value $L_{0}$ is the same in both $R_{\ell }^{\prime }$ and~$%
R_{\ell }^{\prime \prime }$, it suffices to prove that there do not exist
vertices $v\in V_{0}(u)\cap V_{B}$ and $z\in N_{1}\setminus N_{2}$, such
that $R(z)<_{R_{\ell }^{\prime \prime }}R(v)<_{R_{\ell }^{\prime \prime
}}L_{0}$ in $R_{\ell }^{\prime \prime }$. Suppose otherwise that $%
R(z)<_{R_{\ell }^{\prime \prime }}R(v)<_{R_{\ell }^{\prime \prime
}}L_{0}=L(x_{0})$ for two vertices $v\in V_{0}(u)\cap V_{B}$ and $z\in
N_{1}\setminus N_{2}$. Thus, since only the right lines of some
parallelograms $P_{v}$, where $v\in V_{0}(u)\cap V_{B}$, are moved to the
left by Transformation~\ref{trans5}, it follows that $R(z)<_{R_{\ell
}^{\prime }}L_{0}=L(x_{0})<_{R_{\ell }^{\prime }}R(v)$ in $R_{\ell }^{\prime
}$. Therefore, in particular $P_{v}$ intersects $P_{x_{0}}$ in $R_{\ell
}^{\prime }$, and thus $v\in N(x_{0})$, since both $v$ and $x_{0}$ are
bounded. Thus $x_{0}\in V_{0}(u)$, since also $v\in V_{0}(u)$. This is a
contradiction, since $x_{0}\in V_{B}\setminus N(u)\setminus V_{0}(u)$. This
completes the proof.
\end{proof}

\medskip

We construct now the projection representation $R_{\ell }^{\prime \prime
\prime }$ from $R_{\ell }^{\prime \prime }$ as follows.

\begin{transformation}
\label{trans6}Move the line $P_{u}$ in $R_{\ell }^{\prime \prime }$, such
that its upper endpoint $L(u)=R(u)$ comes immediately before $\min_{R_{\ell
}^{\prime \prime }}\{L_{0},R(z)\ |\ z\in N_{1}\setminus N_{2}\}$ and its
lower endpoint $l(u)=r(u)$ comes immediately after $\max_{R_{\ell }^{\prime
\prime }}\{r(v)\ |\ v\in V_{0}(u)\cap V_{B}\}$. Finally, make $u$ a bounded
vertex. Denote the resulting projection representation by $R_{\ell }^{\prime
\prime \prime }$.
\end{transformation}

Note by the statement of Transformation~\ref{trans6} that $R_{\ell }^{\prime
\prime \prime }$ is a projection representation with $k-1$ unbounded
vertices, since $u$ is a bounded vertex in $R_{\ell }^{\prime \prime \prime
} $.

\subsection*{Properties of $R_{\ell }^{\prime }$, 
$R_{\ell }^{\prime \prime }$, and $R_{\ell }^{\prime \prime \prime }$}

In the following (in Lemmas~\ref{R'-ell},~\ref{R''-ell}), we prove that the
projection representations $R_{\ell }^{\prime }\setminus \{u\}$ and $R_{\ell
}^{\prime \prime }\setminus \{u\}$ (constructed by Transformations~\ref%
{trans4} and~\ref{trans5}, respectively) are both projection representations
of $G\setminus \{u\}$. Furthermore, we prove in Lemma~\ref{R'''-ell} that $%
R_{\ell }^{\prime \prime \prime }$ is a projection representation of $G$;
that is, $R^{\ast }=R_{\ell }^{\prime \prime \prime }$ is a projection
representation of $G$ with $k-1$ unbounded vertices, as Theorem~\ref%
{no-property-thm} states.

\begin{lemma}
\label{R'-ell}$R_{\ell }^{\prime }\setminus \{u\}$ is a projection
representation of $G\setminus \{u\}$.
\end{lemma}

\begin{proof}
Denote by $x_{0}$ the vertex of $V_{0}(u)$, such that $\ell _{0}=l(x_{0})$.
Since we move the right line of some parallelograms to the right, i.e.~we
increase some parallelograms, all adjacencies of $R_{\ell }$ are kept in $%
R_{\ell }^{\prime }$. Suppose that $R_{\ell }^{\prime }$ has the new
adjacency $zv$ that is not an adjacency in $R_{\ell }$, for some $z\in N_{1}$%
. Therefore, since perform parallel movements of lines, i.e.~since every
slope $\phi _{x}$ in $R_{\ell }^{\prime }$ equals the value of $\phi _{x}$
in $R_{\ell }$ for every vertex $x$ of $G$, it follows that $P_{z}\ll
_{R_{\ell }}P_{v}$ and $P_{z}$ intersects $P_{v}$ in $R_{\ell }^{\prime }$.
Thus, $v\notin V_{0}(u)$, since $u$ has the right border property in $%
R_{\ell }$ by Lemma~\ref{R-ell-right-property}. Furthermore, $r(z)<_{R_{\ell
}}\ell _{0}=l(x_{0})$, since $z\in N_{1}$. However, since $x_{0}\in V_{0}(u)$%
, and since $u$ has the right border property in $R_{\ell }$, it follows
that $P_{z}$ intersects $P_{x_{0}}$ in $R_{\ell }$, and thus $%
L(x_{0})<_{R_{\ell }}R(z)$. We distinguish in the following the cases where $%
v\notin N(u)$ and $v\in N(u)$.

\emph{Case 1.} $v\notin N(u)$. Then, since also $v\notin V_{0}(u)$, it
follows by Observation~\ref{V(G0)} that $v\notin V(G_{0})$. We will derive a
contradiction to the assumption that $R_{\ell }^{\prime }$ has the new
adjacency $zv$ that is not an adjacency in $R_{\ell }$, for some $z\in N_{1}$%
. Recall that every slope $\phi _{x}$ in $R_{\ell }^{\prime }$ equals the
value of $\phi _{x}$ in $R_{\ell }$ for every vertex $x$ of $G$. Suppose
first that $r(z)<_{R_{\ell }^{\prime }}l(v)$. Then, since $P_{z}$ intersects 
$P_{v}$ in $R_{\ell }^{\prime }$, it follows that $L(v)<_{R_{\ell }^{\prime
}}R(z)$, and thus $\phi _{v}>\phi _{z}$ in $R_{\ell }^{\prime }$. If $v$ is
unbounded, then $z$ is not adjacent to $v$ in $R_{\ell }^{\prime }$, which
is a contradiction to the assumption. Thus $v$ is bounded, i.e.~$v\in
V_{B}\setminus N(u)$ and $P_{z}\ll _{R_{\ell }}P_{v}$, and thus $%
L_{0}(z)\leq _{R_{\ell }}L(v)$ by definition of $L_{0}(z)$. Furthermore,
since all left lines of the parallelograms in $R_{\ell }$ do not move during
Transformation~\ref{trans4}, it follows that also $L_{0}(z)\leq _{R_{\ell
}^{\prime }}L(v)$. Thus, $R(z)<_{R_{\ell }^{\prime }}L_{0}(z)\leq _{R_{\ell
}^{\prime }}L(v)$ by the statement of Transformation~\ref{trans4}, which is
a contradiction, since $L(v)<_{R_{\ell }^{\prime }}R(z)$.

Suppose now that $l(v)<_{R_{\ell }^{\prime }}r(z)$. We will first prove that
in this case $l(v)<_{R_{\ell }}l(x_{0})$. Suppose otherwise that $%
l(x_{0})<_{R_{\ell }}l(v)$. Let $x_{0}\notin V(G_{0})$. Then, since $r(z)$
comes in $R_{\ell }^{\prime }$ at most immediately after $\ell _{0}=l(x_{0})$
on $L_{2}$, it follows that $l(x_{0})<_{R_{\ell }^{\prime }}r(z)<_{R_{\ell
}^{\prime }}l(v)$. This is a contradiction to the assumption that $%
l(v)<_{R_{\ell }^{\prime }}r(z)$. Let $x_{0}\in V(G_{0})$. Then, $b_{\ell
}-\varepsilon <_{R_{\ell }}l(x_{0})<_{R_{\ell }}b_{\ell }+\varepsilon $ by
Remark~\ref{trans4-remark-1}. Furthermore, since $v\notin V(G_{0})$, and
since we assumed that $l(x_{0})<_{R_{\ell }}l(v)$, it follows that $%
l(x_{0})<_{R_{\ell }}b_{\ell }+\varepsilon <_{R_{\ell }}l(v)$ by Remark~\ref%
{trans4-remark-1}. If $z\in V(G_{0})$, then $r(z)$ comes in $R_{\ell
}^{\prime }$ (due to the statement of Transformation~\ref{trans4}) at most
immediately after $\ell _{0}=l(x_{0})$ on $L_{2}$, and thus in this case $%
l(x_{0})<_{R_{\ell }^{\prime }}r(z)<_{R_{\ell }^{\prime }}b_{\ell
}+\varepsilon <_{R_{\ell }^{\prime }}l(v)$. This is a contradiction to the
assumption that $l(v)<_{R_{\ell }^{\prime }}r(z)$. Otherwise, if $z\notin
V(G_{0})$, then $r(z)$ comes in $R_{\ell }^{\prime }$ (due to Remark~\ref%
{trans4-remark-2}) immediately after $b_{\ell }+\varepsilon $ on $L_{2}$,
and thus in this case $l(x_{0})<_{R_{\ell }^{\prime }}b_{\ell }+\varepsilon
<_{R_{\ell }^{\prime }}r(z)<_{R_{\ell }^{\prime }}l(v)$. This is again a
contradiction to the assumption that $l(v)<_{R_{\ell }^{\prime }}r(z)$.
Therefore $l(v)<_{R_{\ell }}l(x_{0})$.

Recall that $L(x_{0})<_{R_{\ell }}R(z)$, and thus also $L(x_{0})<_{R_{\ell
}}R(z)<_{R_{\ell }}L(v)$, since $P_{z}\ll _{R_{\ell }}P_{v}$. Therefore,
since also $l(v)<_{R_{\ell }}l(x_{0})$ by the previous paragraph, it follows
that $P_{x_{0}}$ intersects $P_{v}$ in $R_{\ell }$ and $\phi _{x_{0}}>\phi
_{v}$ in $R_{\ell }$. If $x_{0}$ is bounded, then $x_{0}v\in E$, and thus $%
v\in V_{0}(u)$, since $x_{0}\in V_{0}(u)$ and $v\notin N(u)$, which is a
contradiction. Therefore, $x_{0}$ is unbounded, and thus $x_{0}v\notin E$.
Therefore, $N(x_{0})\subseteq N(v)$ by Lemma~\ref{intersecting-unbounded}.
Recall now that there exists a bounded covering vertex $u^{\ast }$ of $u$ in 
$G$, and thus $u^{\ast },x_{0}\in V_{0}(u)$. Furthermore, $u^{\ast }\neq
x_{0}$, since $x_{0}$ is unbounded. Therefore, since $V_{0}(u)$ is connected
with at least two vertices, $x_{0}$ is adjacent to at least one other vertex 
$y\in V_{0}(u)$, and thus $y\in N(v)$, since $N(x_{0})\subseteq N(v)$. Thus $%
v\in V_{0}(u)$, since $v\notin N(u)$, which is again a contradiction.
Summarizing, $R_{\ell }^{\prime }$ has no new adjacency $zv$ that is not an
adjacency in $R_{\ell }$, for any $v\notin N(u)$ and any $z\in N_{1}$.

\emph{Case 2.} $v\in N(u)$. We distinguish in the following the cases where $%
z\notin V(G_{0})$ and $z\in V(G_{0})$.

\emph{Case 2a.} $z\notin V(G_{0})$. Since $z\in N(u)$, it follows that $%
P_{z} $ intersects $P_{u}$ in $R_{\ell }$ by Corollary~\ref%
{R-ell-N(u)-intersect}, and thus $P_{z}$ intersects the line segment $\ell $
in $R_{\ell }$. If $v\in V(G_{0})$, then $P_{z}$ intersects $P_{v}$ in $%
R_{\ell }$ (since $v\in N(u)$), which is a contradiction. Thus, $v\notin
V(G_{0})$. Therefore, since both $z,v\notin V(G_{0})$, and since $P_{z}\ll
_{R_{\ell }}P_{v}$, it follows that also $P_{z}\ll _{R}P_{v}$. Therefore,
since $v\in N(u)$, it follows that $R(z)<_{R}L(v)<_{R}a_{u}=_{R}L(u)$ by
Lemma~\ref{unbounded-bounded}, and thus $L(x_{0})<_{R_{\ell }}R(z)<_{R_{\ell
}}L(v)<_{R_{\ell }}a_{u}$, since the endpoints of $P_{z}$ and $P_{v}$ remain
the same in both $R$ and $R_{\ell }$. Therefore $x_{0}\notin V(G_{0})$,
since otherwise $L(x_{0})>_{R_{\ell }}a_{\ell }-\varepsilon >_{R_{\ell
}}a_{u}$ (by definition of the line segment $\ell $). Thus, also $%
L(x_{0})<_{R}R(z)<_{R}L(v)<_{R}a_{u}$. Furthermore $%
b_{u}=_{R}r(u)<_{R}r(z)<_{R}\ell _{0}=l(x_{0})$ due to Lemma~\ref%
{unbounded-bounded}, since $z\in N_{1}$. Then, $P_{x_{0}}$ intersects $P_{u}$
in $R$ and $\phi _{x_{0}}>\phi _{u}$, since $L(x_{0})<_{R}a_{u}$ and $%
b_{u}<_{R}l(x_{0})$. If $x_{0}\notin N(u)$, then $N(x_{0})\subseteq N(u)$ by
Lemma~\ref{intersecting-unbounded}, and thus $x_{0}\in Q_{u}$. This is a
contradiction by Lemma~\ref{Qu-1}, since $x_{0}\in V_{0}(u)$ by assumption.
Thus $x_{0}\in N(u)$, which is again a contradiction, since $x_{0}\in
V_{0}(u)$.

\emph{Case 2b.} $z\in V(G_{0})$. Then, note that $r(u)<_{R_{0}}r(z)$ by
Lemma~\ref{unbounded-bounded}, and thus also $b_{u}<_{R_{\ell
}}r(u)<_{R_{\ell }}r(z)$, since $R_{0}$ is a projection representation of $%
G_{0}$ (and a sub-representation of $R_{\ell }$). Suppose that $v\notin
V(G_{0})$. Then, since we assumed that $v\in N(u)$, it follows by Corollary %
\ref{R-ell-N(u)-intersect} that $P_{v}$ intersects $P_{u}$ in $R_{\ell }$.
That is, $P_{v}$ intersects the line segment $\ell $ in $R_{\ell }$, and
thus $P_{v}$ intersects $P_{z}$ in $R_{\ell }$, which is a contradiction,
since $P_{z}\ll _{R_{\ell }}P_{v}$. Therefore, $v\in V(G_{0})$.

Consider the projection representation $R_{0}$ of $G_{0}$ (which is a
sub-representation of $R_{\ell }$) and suppose that $x_{0}\in V(G_{0})$.
Then, $r(u)<_{R_{0}}r(z)<_{R_{0}}\ell _{0}=l(x_{0})$ and $%
L(z)<_{R_{0}}L(u)=R(u)$ by Lemma~\ref{unbounded-bounded}. If $%
L(x_{0})<_{R_{0}}R(u)$, then $P_{u}$ intersects $P_{x_{0}}$ in $R_{0}$ and $%
\phi _{x_{0}}>\phi _{u}$ in $R_{0}$. Thus, since $x_{0}\in V(G_{0})\setminus
\{u\}$ and every vertex of $G_{0}\setminus \{u\}$ is bounded by Lemma~\ref%
{N-H-C2-Cu-bounded}, it follows that $x_{0}\in N(u)$. This is a
contradiction, since $x_{0}\in V_{0}(u)$ by definition of $x_{0}$. Therefore 
$R(u)<_{R_{0}}L(x_{0})$. Recall now that $L(x_{0})<_{R_{\ell }}R(z)$ and $%
P_{z}\ll _{R_{\ell }}P_{v}$; thus, also $L(x_{0})<_{R_{0}}R(z)$ and $%
P_{z}\ll _{R_{0}}P_{v}$, since $R_{0}$ is a sub-representation of $R_{\ell }$%
. Therefore, $R(u)<_{R_{0}}L(x_{0})<_{R_{0}}R(z)<_{R_{0}}L(v)$ and $%
r(u)<_{R_{0}}r(z)<_{R_{0}}l(v)$. That is, $R(u)<_{R_{0}}L(v)$ and $%
r(u)<_{R_{0}}l(v)$, i.e.~$P_{u}\ll _{R_{0}}P_{v}$, and thus $v\notin N(u)$,
which is a contradiction to the assumption of Case 2. Therefore, $%
x_{0}\notin V(G_{0})$.

Since $x_{0}\notin V(G_{0})$, i.e.~the endpoints of $P_{x_{0}}$ remain the
same in both $R$ and $R_{\ell }$, and since $b_{u}<_{R_{\ell
}}r(z)<_{R_{\ell }}\ell _{0}=l(x_{0})$, it follows that also $%
b_{u}<_{R}l(x_{0})$. Suppose that $L(x_{0})<_{R}a_{u}$. Then, $P_{x_{0}}$
intersects $P_{u}$ in $R$ and $\phi _{x_{0}}>\phi _{u}$. Thus, $x_{0}$ is
unbounded, since otherwise $x_{0}\in N(u)$, which is a contradiction.
Furthermore, $N(x_{0})\subseteq N(u)$ by Lemma~\ref{intersecting-unbounded},
and thus $x_{0}\in Q_{u}$, which is a contradiction by Lemma~\ref{Qu-1},
since $x_{0}\in V_{0}(u)$ by assumption. Therefore $a_{u}<_{R}L(x_{0})$,
i.e.~$P_{u}\ll _{R}P_{x_{0}}$, since also $b_{u}<_{R}l(x_{0})$. Thus $%
x_{0}\in D_{2}\subseteq S_{2}$, since $x_{0}\in V_{0}(u)$. Furthermore~$%
x_{0}\notin N[X_{1}]$, since $P_{x}\ll _{R}P_{u}\ll _{R}P_{x_{0}}$ for every 
$x\in X_{1}$. Moreover, $x_{0}\notin Q_{u}$ by Lemma~\ref{Qu-1} and $%
x_{0}\notin V({\mathcal{B}_{1}})$ by definition of ${\mathcal{B}_{1}}$,
since $x_{0}\in V_{0}(u)$. Recall now by Lemma~\ref{module-1} that $V({%
C_{u}\cup C_{2}\cup H)}$ induces a subgraph of ${G\setminus Q_{u}\setminus
N[X_{1}]\setminus \mathcal{B}_{1}}$ that includes all connected components
of ${G\setminus Q_{u}\setminus N[X_{1}]\setminus \mathcal{B}_{1}}$, in which
the vertices of~${S_{2}\cup \{u\}}$ belong. Therefore, since $x_{0}\in S_{2}$
and $x_{0}\notin Q_{u}\cup N[X_{1}]\cup V({\mathcal{B}_{1}})$, it follows
that $x_{0}\in V({C_{u}\cup C_{2}\cup H)}$. Thus $x_{0}\in
\bigcup\nolimits_{i=1}^{\infty }H_{i}\cup \bigcup\nolimits_{i=0}^{\infty
}H_{i}^{\prime }$, since otherwise $x_{0}\in V(G_{0})$, which is a
contradiction. If $x_{0}\in \bigcup\nolimits_{i=0}^{\infty }H_{i}^{\prime }$%
, then $x_{0}\in N(u)$ by Lemma~\ref{tilde-neighbors}, which is a
contradiction, since $x_{0}\in V_{0}(u)$. Therefore $x_{0}\in
\bigcup\nolimits_{i=1}^{\infty }H_{i}$.

Let $x_{0}=v_{i}\in H_{i}$, for some $i\geq 1$, and let $(v_{0},v_{1},\ldots
,v_{i})$ be an $H_{i}$-chain of $v_{i}$. Note that $v_{j}\in N(u)\cup
V_{0}(u)$ for every vertex $v_{j}$, where $0\leq j\leq i$; indeed, if $%
v_{j}\notin N(u)$, then $v_{j}\in V_{0}(u)$, since $x_{2}\in V_{0}(u)$ and $%
v_{j}\in N(x_{2})$ by definition of $H$. Furthermore, recall that every
vertex $v_{j}$, where $0\leq j\leq i$, is a bounded vertex by Lemma~\ref%
{N-H-C2-Cu-bounded}. Therefore, since $v_{i}v_{i-1}\notin E$, it follows
that $P_{v_{i}}$ does not intersect $P_{v_{i-1}}$ in $R_{\ell }$, i.e.
either $P_{v_{i}}\ll _{R_{\ell }}P_{v_{i-1}}$ or $P_{v_{i-1}}\ll _{R_{\ell
}}P_{v_{i}}$. Moreover, either $P_{v_{j}}\ll _{R_{\ell }}P_{v_{j-1}}$ or $%
P_{v_{j-1}}\ll _{R_{\ell }}P_{v_{j}}$ for every $j\in \{1,2,\ldots ,i-1\}$
by Lemma~\ref{alternating-bounded-chain}. Thus, either $P_{v_{j-1}}\ll
_{R_{\ell }}P_{v_{j}}$ or $P_{v_{j}}\ll _{R_{\ell }}P_{v_{j-1}}$ for every $%
j\in \{1,2,\ldots ,i\}$.

We will prove by induction on $j$ that $v_{j}\in V_{0}(u)$, $b_{\ell
}-\varepsilon <_{R_{\ell }}r(v_{j})$, and $L(v_{j})<_{R_{\ell }}a_{\ell
}-\varepsilon $, for every $j\in \{0,1,\ldots ,i\}$. Recall first that every 
$v_{j}$, where $0\leq j\leq i$, is adjacent to every vertex of $%
G_{0}\setminus \{u\}$ by Lemma~\ref{module-2}. Thus, in particular every $%
P_{v_{j}}$, where $0\leq j\leq i$, intersects the line segment $\ell $ in $%
R_{\ell }$, since $R_{\ell }\setminus \{u\}$ is a projection representation
of $G\setminus \{u\}$ by Lemma~\ref{R-ell}. Furthermore, recall that $%
v_{j}\notin V(G_{0})$ by definition of $G_{0}$, for every $j\in \{0,1,\ldots
,i\}$, and thus the endpoints of every $P_{v_{j}}$, $j\in \{0,1,\ldots ,i\}$%
, remain the same in both $R$ and $R_{\ell }$. Furthermore, since $%
v_{j}\notin V(G_{0})$, either $l(v_{j})<_{R_{\ell }}b_{\ell }-\varepsilon $
or $l(v_{j})>_{R_{\ell }}b_{\ell }+\varepsilon $ by Remark~\ref%
{trans4-remark-1}, for every $v_{j}$, where $0\leq j\leq i$.

For the induction basis, let $j=i$. Then, $x_{0}=v_{i}\in V_{0}(u)$ by
definition of $x_{0}$. If $l(x_{0})<_{R_{\ell }}b_{\ell }-\varepsilon $,
then $l(x_{0})<_{R_{\ell }}b_{\ell }-\varepsilon <_{R_{\ell }}r(z)<_{R_{\ell
}}b_{\ell }+\varepsilon $, since $x_{0}\notin V(G_{0})$ and $z\in V(G_{0})$
(cf.~Remark~\ref{trans4-remark-1}). This is a contradiction, since $%
r(z)<_{R_{\ell }}\ell _{0}=l(x_{0})$ by definition of $N_{1}$. Therefore $%
b_{\ell }+\varepsilon <_{R_{\ell }}l(x_{0})\leq _{R_{\ell }}r(x_{0})$. Thus,
since $P_{x_{0}}=P_{v_{i}}$ intersects the line segment $\ell $ in $R_{\ell
} $, it follows that $L(x_{0})<_{R_{\ell }}a_{\ell }-\varepsilon $. That is, 
$v_{i}\in V_{0}(u)$, $b_{\ell }+\varepsilon <_{R_{\ell }}r(v_{i})$, and $%
L(v_{i})<_{R_{\ell }}a_{\ell }-\varepsilon $. This completes the induction
basis.

For the induction step, assume that $v_{j}\in V_{0}(u)$, $b_{\ell
}+\varepsilon <_{R_{\ell }}r(v_{j})$, and $L(v_{j})<_{R_{\ell }}a_{\ell
}-\varepsilon $, for some $j\in $ $\{1,2,\ldots ,i\}$. We will prove that
also $v_{j-1}\in V_{0}(u)$, $b_{\ell }+\varepsilon <_{R_{\ell }}r(v_{j-1})$,
and $L(v_{j-1})<_{R_{\ell }}a_{\ell }-\varepsilon $. Let first $%
P_{v_{j-1}}\ll _{R_{\ell }}P_{v_{j}}$. Suppose that $v_{j-1}\notin V_{0}(u)$%
. Then, since $v_{j-1}\in N(u)\cup V_{0}(u)$, it follows that $v_{j-1}\in
N(u)$. That is, $P_{v_{j-1}}\ll _{R_{\ell }}P_{v_{j}}$, where $v_{j-1}\in
N(u)$ and $v_{j}\in V_{0}(u)$. This is a contradiction, since $u$ has the
right border property in $R_{\ell }$ by Lemma~\ref{R-ell-right-property}.
Therefore $v_{j-1}\in V_{0}(u)$. Furthermore, since we assumed that $%
P_{v_{j-1}}\ll _{R_{\ell }}P_{v_{j}}$, and since $L(v_{j})<_{R_{\ell
}}a_{\ell }-\varepsilon $ by the induction hypothesis, it follows that $%
R(v_{j-1})<_{R_{\ell }}L(v_{j})<_{R_{\ell }}a_{\ell }-\varepsilon $. Thus,
also $L(v_{j-1})<_{R_{\ell }}a_{\ell }-\varepsilon $, since $L(v_{j-1})\leq
_{R_{\ell }}R(v_{j-1})$. Furthermore, since $P_{v_{j-1}}$ intersects the
line segment $\ell $ in $R_{\ell }$, it follows that $b_{\ell }+\varepsilon
<_{R_{\ell }}r(v_{j-1})$. That is, $v_{j-1}\in V_{0}(u)$, $b_{\ell
}+\varepsilon <_{R_{\ell }}r(v_{j-1})$, and $L(v_{j-1})<_{R_{\ell }}a_{\ell
}-\varepsilon $.

Let now $P_{v_{j}}\ll _{R_{\ell }}P_{v_{j-1}}$, and thus also $P_{v_{j}}\ll
_{R}P_{v_{j-1}}$, since $v_{j-1},v_{j}\notin V(G_{0})$. Then, since $b_{\ell
}+\varepsilon <_{R_{\ell }}r(v_{j})$ (and thus also $b_{\ell }+\varepsilon
<_{R}r(v_{j})$) by the induction hypothesis, it follows that $b_{\ell
}+\varepsilon <_{R_{\ell }}r(v_{j})<_{R_{\ell }}l(v_{j-1})$. Therefore $%
b_{\ell }+\varepsilon <_{R_{\ell }}r(v_{j-1})$, since $l(v_{j-1})\leq
_{R_{\ell }}r(v_{j-1})$. Furthermore, since $b_{\ell }+\varepsilon
<_{R_{\ell }}l(v_{j-1})$, and since $P_{v_{j-1}}$ intersects the line
segment $\ell $ in $R_{\ell }$, it follows that $R(v_{j-1})<_{R_{\ell
}}a_{\ell }-\varepsilon $. Therefore $L(v_{j-1})<_{R_{\ell }}a_{\ell
}-\varepsilon $, since $L(v_{j-1})\leq _{R_{\ell }}R(v_{j-1})$. That is, $%
b_{\ell }+\varepsilon <_{R_{\ell }}r(v_{j-1})$ and $L(v_{j-1})<_{R_{\ell
}}a_{\ell }-\varepsilon $. Recall that also $b_{\ell }+\varepsilon
<_{R_{\ell }}l(v_{j-1})$. Thus $b_{u}<_{R}b_{\ell }+\varepsilon
<_{R}l(v_{j-1})$, since $b_{u}<_{R}b_{\ell }$ (by definition of the line
segment $\ell $), and since the endpoints of $P_{v_{j-1}}$ remain the same
in both $R$ and $R_{\ell }$. Suppose now that $v_{j-1}\notin V_{0}(u)$.
Then, since $v_{j-1}\in N(u)\cup V_{0}(u)$, it follows that $v_{j-1}\in N(u)$%
, i.e.~in particular $P_{v_{j-1}}$ intersects $P_{u}$ in $R$. Thus, since $%
b_{u}=_{R}r(u)<_{R}l(v_{j-1})$, it follows that $%
L(v_{j-1})<_{R}a_{u}=_{R}L(u)$. Therefore $R(v_{j})<_{R}L(v_{j-1})<_{R}a_{u}$%
, since we assumed that $P_{v_{j}}\ll _{R}P_{v_{j-1}}$. Then, since $%
R(v_{j})<_{R}a_{u}$ and $b_{u}<_{R}b_{\ell }+\varepsilon <_{R}r(v_{j})$, it
follows that $P_{v_{j}}$ intersects $P_{u}$ in $R$ and $\phi _{v_{j}}>\phi
_{u}$. Thus $v_{j}\in N(u)$, since $v_{j}$ is bounded in $R$, which is a
contradiction to the induction hypothesis that $v_{j}\in V_{0}(u)$.
Therefore, $v_{j-1}\in V_{0}(u)$. This completes the induction step, and
thus $v_{j}\in V_{0}(u)$, $b_{\ell }-\varepsilon <_{R_{\ell }}r(v_{j})$, and 
$L(v_{j})<_{R_{\ell }}a_{\ell }-\varepsilon $, for every $j\in \{0,1,\ldots
,i\}$.

Consider now the vertex $v_{0}\in H_{0}=N$. Then $P_{v_{0}}$ intersects $%
P_{u}$ in $R$, since $v_{0}\in N(X_{1})\cap N(x_{2})$ by Lemma~\ref{N-Cu},
and since $P_{x}\ll _{R}P_{u}\ll _{R}P_{x_{2}}$ for every $x\in X_{1}$.
Recall that $x_{0}=v_{i}\in H_{i}$, for some $i\geq 1$, and that $%
(v_{0},v_{1},\ldots ,v_{i})$ is an $H_{i}$-chain of $v_{i}$. Thus, in
particular, $v_{1}$ exists, since $i\geq 1$. Furthermore, $%
L(v_{1})<_{R_{\ell }}a_{\ell }-\varepsilon $ by the previous paragraph. Thus
also $L(v_{1})<_{R}a_{\ell }-\varepsilon $, since the endpoints of $%
P_{v_{1}} $ remain the same in both $R$ and $R_{\ell }$. Therefore, since $%
P_{v_{0}}\ll _{R}P_{v_{1}}$ by Lemma~\ref{alternating-bounded-chain-1}, it
follows that $R(v_{0})<_{R}L(v_{1})<_{R}a_{\ell }-\varepsilon $. On the
other hand, $b_{\ell }-\varepsilon <_{R_{\ell }}r(v_{0})$ by the previous
paragraph, and thus also $b_{\ell }-\varepsilon <_{R}r(v_{0})$. That is, $%
R(v_{0})<_{R}a_{\ell }-\varepsilon $ and $b_{\ell }-\varepsilon
<_{R}r(v_{0}) $, and thus in particular $\phi _{v_{0}}>\phi _{\ell }$ in $R$%
. Therefore $\phi _{v_{0}}>\phi _{\ell }\geq\phi _{u}$ in $R$, 
since $\phi_{\ell }\geq\phi _{u} $ in $R$ by the definition of the line segment $\ell $.
Thus, since $P_{v_{0}}$ intersects $P_{u}$ in $R$, it follows that $v_{0}\in
N(u)$. This is a contradiction, since $v_{0}\in V_{0}(u)$ by the previous
paragraph.

This completes Case 2b, and thus also due to Cases 1 and 2a, it follows that 
$R_{\ell }^{\prime }$ has no new adjacency $zv$ that is not an adjacency in $%
R_{\ell }$, for any $z\in N_{1}$, i.e.~$R_{\ell }^{\prime }\setminus \{u\}$
is a projection representation of $G\setminus \{u\}$. This completes the
proof of the lemma.
\end{proof}

\begin{lemma}
\label{R''-ell}$R_{\ell }^{\prime \prime }\setminus \{u\}$ is a projection
representation of $G\setminus \{u\}$.
\end{lemma}

\begin{proof}
Denote by $z_{0}$ the vertex of $N_{2}$, such that $r_{0}=r(z_{0})$. Since
during Transformation~\ref{trans5} we move the right line of some
parallelograms to the left, i.e.~we decrease some parallelograms, no new
adjacencies are introduced in $R_{\ell }^{\prime \prime }$ in comparison to $%
R_{\ell }^{\prime }$. Suppose that $vx\in E$ and that the adjacency $vx$ has
been removed from $R_{\ell }^{\prime }$ in $R_{\ell }^{\prime \prime }$, for
some $v\in V_{0}(u)\cap V_{B}$, such that $r(v)>_{R_{\ell }^{\prime
}}r_{0}=r(z_{0})$. Therefore, since we perform parallel movements of lines
in $R_{\ell }^{\prime }$, i.e.~since every slope $\phi _{y}$ in $R_{\ell
}^{\prime \prime }$ equals the value of $\phi _{y}$ in $R_{\ell }^{\prime }$
for every vertex $y$ of $G$, it follows that $P_{v}\ll _{R_{\ell }^{\prime
\prime }}P_{x}$ and that $P_{v}$ intersects $P_{x}$ in $R_{\ell }^{\prime }$%
. Note that $l(v)\leq _{R_{\ell }^{\prime }}\ell _{0}$, since $v\in V_{0}(u)$
and $\ell _{0}=\max_{R_{\ell }^{\prime }}\{l(x)\ |\ x\in V_{0}(u)\}$.

We first assume that $x\notin N(u)$. Since $r(v)$ comes in $R_{\ell
}^{\prime \prime }$ immediately before $r_{0}$, and since $P_{v}\ll
_{R_{\ell }^{\prime \prime }}P_{x}$, it follows that $r(v)<_{R_{\ell
}^{\prime \prime }}r_{0}<_{R_{\ell }^{\prime \prime }}l(x)$, and thus also $%
r_{0}<_{R_{\ell }^{\prime }}l(x)$. Furthermore, since $vx\in E$ by
assumption, and since $v\in V_{0}(u)$, it follows that $x\in V_{0}(u)$.
Therefore $l(x)\leq _{R_{\ell }^{\prime }}\ell _{0}$, since $\ell
_{0}=\max_{R_{\ell }^{\prime }}\{l(x)\ |\ x\in V_{0}(u)\}$, and thus $%
r_{0}=r(z_{0})<_{R_{\ell }^{\prime }}l(x)\leq _{R_{\ell }^{\prime }}\ell
_{0} $, i.e.~$r(z_{0})<_{R_{\ell }^{\prime }}\ell _{0}$. This is a
contradiction, since $z_{0}\in N_{2}$. Therefore, no adjacency $vx$ has been
removed from $R_{\ell }^{\prime }$ in $R_{\ell }^{\prime \prime }$ in the
case where $x\notin N(u)$.

Assume now that $x\in N(u)$, and thus the endpoints of $P_{x}$ in $R_{\ell
}^{\prime }$ remain the same also in $R_{\ell }^{\prime \prime }$.

\emph{Case 1.} $v\in V(G_{0})$. Then, since the endpoints of $P_{v}$ do not
move during Transformation~\ref{trans4}, it follows by Remark~\ref%
{trans4-remark-1} that $b_{\ell }-\varepsilon <_{R_{\ell }^{\prime
}}l(v)\leq _{R_{\ell }^{\prime }}r(v)<_{R_{\ell }^{\prime }}b_{\ell
}+\varepsilon $ and $a_{\ell }-\varepsilon <_{R_{\ell }^{\prime }}L(v)\leq
_{R_{\ell }^{\prime }}R(v)<_{R_{\ell }^{\prime }}a_{\ell }+\varepsilon $ in $%
R_{\ell }^{\prime }$. Thus, in particular also $b_{\ell }-\varepsilon
<_{R_{\ell }^{\prime \prime }}l(v)$ and $a_{\ell }-\varepsilon <_{R_{\ell
}^{\prime \prime }}L(v)$ in $R_{\ell }^{\prime \prime }$, since the left
lines of all parallelograms do not move during Transformation~\ref{trans5}.
Therefore $b_{\ell }-\varepsilon <_{R_{\ell }^{\prime \prime
}}l(v)<_{R_{\ell }^{\prime \prime }}l(x)$ and $a_{\ell }-\varepsilon
<_{R_{\ell }^{\prime \prime }}L(v)<_{R_{\ell }^{\prime \prime }}L(x)$, since 
$P_{v}\ll _{R_{\ell }^{\prime \prime }}P_{x}$. Furthermore, also $b_{\ell
}-\varepsilon <_{R_{\ell }}l(x)$ and $a_{\ell }-\varepsilon <_{R_{\ell
}}L(x) $ in $R_{\ell }$, since left lines of all parallelograms do not move
during Transformations~\ref{trans4} and~\ref{trans5}. We distinguish in the
following the cases where $x\notin V(G_{0})$ and $x\in V(G_{0})$.

\emph{Case 1a.} $x\notin V(G_{0})$. Then, either $l(x)<_{R_{\ell }}b_{\ell
}-\varepsilon $ or $l(x)>_{R_{\ell }}b_{\ell }+\varepsilon $ (resp.~either $%
L(x)<_{R_{\ell }}a_{\ell }-\varepsilon $ or $L(x)>_{R_{\ell }}a_{\ell
}+\varepsilon $) by Remark~\ref{trans4-remark-1}. Thus, since $b_{\ell
}-\varepsilon <_{R_{\ell }}l(x)$ and $a_{\ell }-\varepsilon <_{R_{\ell
}}L(x) $ by the previous paragraph, it follows that $l(x)>_{R_{\ell
}}b_{\ell }+\varepsilon $ and $L(x)>_{R_{\ell }}a_{\ell }+\varepsilon $.
Therefore $r(v)<_{R_{\ell }}b_{\ell }+\varepsilon <_{R_{\ell }}l(x)$ and $%
R(v)<_{R_{\ell }}a_{\ell }+\varepsilon <_{R_{\ell }}L(x)$ by Remark~\ref%
{trans4-remark-1}, i.e.~$P_{v}\ll _{R_{\ell }}P_{x}$ in $R_{\ell }$, and
thus $vx\notin E$. This is a contradiction, since we assumed that $vx\in E$.

\emph{Case 1b.} $x\in V(G_{0})$. Recall by Lemma~\ref{r(u)<r0-in-R-ell'}
that $r(u)<_{R_{\ell }^{\prime }}r_{0}=r(z_{0})$, and thus $r(u)<_{R_{\ell
}^{\prime }}r_{0}<_{R_{\ell }^{\prime }}r(v)$. Therefore, since $r(v)$ comes
immediately before $r_{0}$ in $R_{\ell }^{\prime \prime }$ during
Transformation~\ref{trans5}, it follows that $r(u)<_{R_{\ell }^{\prime
\prime }}r(v)<_{R_{\ell }^{\prime \prime }}r_{0}$. Therefore, $%
r(u)<_{R_{\ell }^{\prime \prime }}r(v)<_{R_{\ell }^{\prime \prime }}l(x)$,
since $P_{v}\ll _{R_{\ell }^{\prime \prime }}P_{x}$. Suppose that $P_{x}$
intersects $P_{u}$ in $R_{\ell }^{\prime \prime }$. Then, since $%
r(u)<_{R_{\ell }^{\prime \prime }}l(x)$, it follows that $L(x)<_{R_{\ell
}^{\prime \prime }}R(u)$; thus $R(v)<_{R_{\ell }^{\prime \prime
}}L(x)<_{R_{\ell }^{\prime \prime }}R(u)$, since $P_{v}\ll _{R_{\ell
}^{\prime \prime }}P_{x}$. That is, $r(u)<_{R_{\ell }^{\prime \prime }}r(v)$
and $R(v)<_{R_{\ell }^{\prime \prime }}R(u)$, i.e.~$P_{v}$ intersects $P_{u}$
in $R_{\ell }^{\prime \prime }$ and $\phi _{v}>\phi _{u}$ in $R_{\ell
}^{\prime \prime }$. Therefore, $P_{v}$ intersects $P_{u}$ and $\phi
_{v}>\phi _{u}$ also in $R_{\ell }^{\prime }$ and in $R_{\ell }$. Thus,
since $v\in V(G_{0})$, and since $R_{0}$ is a sub-representation of $R_{\ell
}$, $P_{v}$ intersects $P_{u}$ in $R_{0}$ and $\phi _{v}>\phi _{u}$ in $%
R_{0} $. Therefore, since $v$ is bounded (recall that $v\in V_{0}(u)\cap
V_{B}$ by our initial assumption on $v$), it follows that $v\in N(u)$, which
is a contradiction. Therefore, $P_{x}$ does not intersect $P_{u}$ in $%
R_{\ell }^{\prime \prime }$, and thus $P_{u}\ll _{R_{\ell }^{\prime \prime
}}P_{x}$, since $r(u)<_{R_{\ell }^{\prime \prime }}l(x)$. Thus also $%
P_{u}\ll _{R_{\ell }^{\prime }}P_{x}$ and $P_{u}\ll _{R_{\ell }}P_{x}$,
since the left line of $P_{x}$ does not move by Transformations~\ref{trans4}
and~\ref{trans5}. Therefore $P_{u}\ll _{R_{0}}P_{x}$, since $x\in V(G_{0})$
and $R_{0}$ is a sub-representation of $R_{\ell }$. Thus $x\notin N(u)$,
which is a contradiction to our assumption on $x$.

\emph{Case 2.} $v\notin V(G_{0})$.

\emph{Case 2a.} $x\notin V(G_{0})$. We will now prove that $b_{u}<_{R_{\ell
}^{\prime \prime }}r(v)<_{R_{\ell }^{\prime \prime }}l(x)$. Recall that $%
z_{0}\in N(u)$. Thus, if $z_{0}\in V(G_{0})$, then $r(u)<_{R_{0}}r(z_{0})$
by Lemma~\ref{unbounded-bounded}, and thus also $r(u)<_{R_{\ell }}r(z_{0})$,
since $R_{0}$ is a sub-representation of $R_{\ell }$. Furthermore $%
b_{u}<_{R_{\ell }^{\prime }}r(u)<_{R_{\ell }^{\prime }}r(z_{0})$, since the
right endpoint $r(z_{0})$ of $P_{z_{0}}$ does not decrease by Transformation~%
\ref{trans4}. On the other hand, let $z_{0}\notin V(G_{0})$. Then $%
b_{u}<_{R}r(z_{0})$ by Lemma~\ref{unbounded-bounded}, and thus also $%
b_{u}<_{R_{\ell }}r(z_{0})$, since $z_{0}\notin V(G_{0})$ (i.e.~the
endpoints of $P_{z_{0}}$ are the same in both $R$ and $R_{\ell }$).
Furthermore $b_{u}<_{R_{\ell }^{\prime }}r(z_{0})$, since $r(z_{0})$ does
not decrease by Transformation~\ref{trans4}. That is, $b_{u}<_{R_{\ell
}^{\prime }}r(z_{0})=r_{0}<_{R_{\ell }^{\prime }}r(v)$ in both cases where $%
z_{0}\in V(G_{0})$ and $z_{0}\notin V(G_{0})$. Therefore, since $r(v)$ comes
immediately before $r_{0}=r(z_{0})$ in $R_{\ell }^{\prime \prime }$ by
Transformation~\ref{trans5}, it follows that $b_{u}<_{R_{\ell }^{\prime
\prime }}r(v)<_{R_{\ell }^{\prime \prime }}r_{0}$. Thus, $b_{u}<_{R_{\ell
}^{\prime \prime }}r(v)<_{R_{\ell }^{\prime \prime }}l(x)$, since $P_{v}\ll
_{R_{\ell }^{\prime \prime }}P_{x}$.

Furthermore, since the left lines of the parallelograms do not move by
Transformations~\ref{trans4} and~\ref{trans5}, it follows that also $%
b_{u}<_{R_{\ell }}l(x)$. Therefore $r(u)=_{R}b_{u}<_{R}l(x)$, since $x\notin
V(G_{0})$ (i.e.~the endpoints of $P_{x}$ are the same in both $R$ and $%
R_{\ell }$). Thus, since we assumed that $x\in N(u)$, it follows that $%
L(x)<_{R}a_{u}=_{R}L(u)$. Similarly, since the left lines of the
parallelograms do not move by Transformations~\ref{trans4} and~\ref{trans5},
and since $x\notin V(G_{0})$, it follows that also $L(x)<_{R_{\ell }}a_{u}$
and $L(x)<_{R_{\ell }^{\prime \prime }}a_{u}$. Thus, $R(v)<_{R_{\ell
}^{\prime \prime }}L(x)<_{R_{\ell }^{\prime \prime }}a_{u}$, since $P_{v}\ll
_{R_{\ell }^{\prime \prime }}P_{x}$. That is, $b_{u}<_{R_{\ell }^{\prime
\prime }}r(v)$ (by the previous paragraph) and $L(v)\leq _{R_{\ell }^{\prime
\prime }}R(v)<_{R_{\ell }^{\prime \prime }}a_{u}$. Therefore, since the
slope $\phi _{v}$ of $P_{v}$ (where $v\notin V(G_{0})$) remains the same in
the representations $R$, $R_{\ell }$, $R_{\ell }^{\prime }$, and $R_{\ell
}^{\prime \prime }$, and since the lower right endpoint $r(v)$ in $R$ is
greater than or equal to the corresponding value $r(v)$ in $R_{\ell
}^{\prime \prime }$, it follows that $P_{v}$ intersects $P_{u}$ in $R$ and $%
\phi _{v}>\phi _{u}$ in $R$. Thus $v\in N(u)$, since $v$ is bounded (recall
that $v\in V_{0}(u)\cap V_{B}$), which is a contradiction to the assumption
that $v\in V_{0}(u)$.

\emph{Case 2b.} $x\in V(G_{0})$. Recall that $v\notin V(G_{0})$ by the
assumption of Case 2. Therefore, since $vx\notin E$, it follows by Lemma~\ref%
{module-2} that $v\in N(X_{1})\cup \bigcup\nolimits_{i=1}^{\infty
}H_{i}\bigcup\nolimits_{i=0}^{\infty }H_{i}^{\prime }$. Recall that $v\in
V_{0}(u)\cap V_{B}$, and thus in particular $v\notin N(u)$. Therefore $%
v\notin \bigcup\nolimits_{i=0}^{\infty }H_{i}^{\prime }$ by Lemma~\ref%
{tilde-neighbors}, and thus $v\in N(X_{1})\cup
\bigcup\nolimits_{i=1}^{\infty }H_{i}$. We distinguish in the following the
cases where $v\in N(X_{1})$ and $v\in \bigcup\nolimits_{i=1}^{\infty }H_{i}$.

\emph{Case 2b-i.} $v\in N=N(X_{1})$. Then, $P_{v}$ intersects $P_{u}$ in $R$%
, since $v\in N(X_{1})\cap N(x_{2})$ by Lemma~\ref{N-Cu}, and since $%
P_{x}\ll _{R}P_{u}\ll _{R}P_{x_{2}}$ for every $x\in X_{1}$. Recall that $v$
is bounded and $v\notin N(u)$, since $v\in V_{0}(u)\cap V_{B}$ by our
initial assumption on $v$, and thus $\phi _{v}<\phi _{u}\leq \phi _{\ell }$
in $R$. Therefore, $\phi _{v}<\phi _{\ell }$ also in $R_{\ell }$, since $%
v\notin V(G_{0})$ (i.e.~the endpoints of $P_{v}$ remain the same in both $R$
and $R_{\ell }$). On the other hand, since $z_{0}\in N(u)$, it follows that $%
\phi _{z_{0}}>\phi _{u}$ in $R$, and thus $\phi _{v}<\phi _{u}<\phi _{z_{0}}$
in $R$. Furthermore, recall by Remark~\ref{trans4-remark-1} that $b_{\ell
}-\varepsilon <_{R_{\ell }}l(x)<_{R_{\ell }}b_{\ell }+\varepsilon $ in $%
R_{\ell }$, since $x\in V(G_{0})$ by the assumption of Case 2b. Therefore,
since the left lines of the parallelograms do not move by Transformations~%
\ref{trans4} and~\ref{trans5}, it follows that also $b_{\ell }-\varepsilon
<_{R_{\ell }^{\prime \prime }}l(x)<_{R_{\ell }^{\prime \prime }}b_{\ell
}+\varepsilon $ in $R_{\ell }^{\prime \prime }$. Similarly, it follows by to
Remark~\ref{trans4-remark-1} that $a_{\ell }-\varepsilon <_{R_{\ell
}^{\prime \prime }}L(x)<_{R_{\ell }^{\prime \prime }}a_{\ell }+\varepsilon $
in $R_{\ell }^{\prime \prime }$.

Let first $z_{0}\notin V(G_{0})$. Then, either $r(z_{0})>_{R_{\ell }^{\prime
}}b_{\ell }+\varepsilon $ or $r(z_{0})<_{R_{\ell }^{\prime }}b_{\ell
}-\varepsilon $ by Remark~\ref{trans4-remark-2}. Suppose that $%
r(z_{0})>_{R_{\ell }^{\prime }}b_{\ell }+\varepsilon $. Then, since $r(v)$
comes by Transformation~\ref{trans5} immediately before $r_{0}=r(z_{0})$ in $%
R_{\ell }^{\prime \prime }$, it follows that $b_{\ell }+\varepsilon
<_{R_{\ell }^{\prime \prime }}r(v)<_{R_{\ell }^{\prime \prime }}r(z_{0})$.
Thus $b_{\ell }+\varepsilon <_{R_{\ell }^{\prime \prime }}r(v)<_{R_{\ell
}^{\prime \prime }}l(x)$, since $P_{v}\ll _{R_{\ell }^{\prime \prime }}P_{x}$%
. This is a contradiction, since $b_{\ell }-\varepsilon <_{R_{\ell }^{\prime
\prime }}l(x)<_{R_{\ell }^{\prime \prime }}b_{\ell }+\varepsilon $.
Therefore $r(z_{0})<_{R_{\ell }^{\prime }}b_{\ell }-\varepsilon $.

Recall now by Corollary~\ref{R-ell-N(u)-intersect} that $P_{z_{0}}$
intersects $P_{u}$ in $R_{\ell }$, since $z_{0}\in N(u)$. Therefore, since $%
P_{z_{0}}$ does not decrease during Transformation~\ref{trans4}, $P_{z_{0}}$
intersects $P_{u}$ also in $R_{\ell }^{\prime }$, i.e.~$P_{z_{0}}$
intersects the line segment $\ell $ in $R_{\ell }^{\prime }$. Furthermore,
since $z_{0}\notin V(G_{0})$, either $R(z_{0})>_{R_{\ell }^{\prime }}a_{\ell
}+\varepsilon $ or $R(z_{0})<_{R_{\ell }^{\prime }}a_{\ell }-\varepsilon $
by Remark~\ref{trans4-remark-2}. Therefore, since $r(z_{0})<_{R_{\ell
}^{\prime }}b_{\ell }-\varepsilon $ and $P_{z_{0}}$ intersects the line
segment $\ell $ in $R_{\ell }^{\prime }$, it follows that $%
R(z_{0})>_{R_{\ell }^{\prime }}a_{\ell }+\varepsilon $; thus also $%
R(z_{0})>_{R_{\ell }^{\prime \prime }}a_{\ell }+\varepsilon $, since the
endpoints of $P_{z_{0}}$ do not change by Transformation~\ref{trans5}.
Recall now that $\phi _{v}<\phi _{z_{0}}$ in $R$. Therefore also $\phi
_{v}<\phi _{z_{0}}$ in $R_{\ell }^{\prime \prime }$, since $v,z_{0}\notin
V(G_{0})$ (i.e.~the slopes $\phi _{z_{0}}$ and $\phi _{v}$ remain the same
in both $R$ and $R_{\ell }^{\prime \prime }$). Furthermore, recall that $%
r(v) $ comes by Transformation~\ref{trans5} immediately before $r(z_{0})$
(i.e.~sufficiently close to $r(z_{0})$) in $R_{\ell }^{\prime \prime }$.
Therefore, since $a_{\ell }+\varepsilon <_{R_{\ell }^{\prime \prime
}}R(z_{0})$ and $\phi _{v}<\phi _{z_{0}}$ in $R_{\ell }^{\prime \prime }$,
it follows that $a_{\ell }+\varepsilon <_{R_{\ell }^{\prime \prime
}}R(z_{0})<_{R_{\ell }^{\prime \prime }}R(v)$. Thus $a_{\ell }+\varepsilon
<_{R_{\ell }^{\prime \prime }}R(v)<_{R_{\ell }^{\prime \prime }}L(x)$, since 
$P_{v}\ll _{R_{\ell }^{\prime \prime }}P_{x}$. This is a contradiction,
since $a_{\ell }-\varepsilon <_{R_{\ell }^{\prime \prime }}L(x)<_{R_{\ell
}^{\prime \prime }}a_{\ell }+\varepsilon $ in $R_{\ell }^{\prime \prime }$.

Let now $z_{0}\in V(G_{0})$. Then $r(u)<_{R_{0}}r(z_{0})$ by Lemma~\ref%
{unbounded-bounded}, since $z_{0}\in N(u)$. Thus, also $r(u)<_{R_{\ell
}}r(z_{0})$, since $R_{0}$ is a sub-representation of $R_{\ell }$.
Furthermore $r(u)<_{R_{\ell }^{\prime \prime }}r(z_{0})$, since the value $%
r(z_{0})$ does not decrease by Transformations~\ref{trans4} and~\ref{trans5}%
. Therefore, since $r(v)$ comes by Transformation~\ref{trans5} immediately
before $r(z_{0})$, it follows that $r(u)<_{R_{\ell }^{\prime \prime
}}r(v)<_{R_{\ell }^{\prime \prime }}r(z_{0})$. Similarly, $L(x)<_{R_{0}}L(u)$
by Lemma~\ref{unbounded-bounded}, since $x\in N(u)$, and thus also $%
L(x)<_{R_{\ell }}L(u)$. Furthermore $L(x)<_{R_{\ell }^{\prime \prime }}L(u)$%
, since the left lines of the parallelograms do not move by Transformations~%
\ref{trans4} and~\ref{trans5}. Therefore $R(v)<_{R_{\ell }^{\prime \prime
}}L(x)<_{R_{\ell }^{\prime \prime }}L(u)$, since $P_{v}\ll _{R_{\ell
}^{\prime \prime }}P_{x}$. That is, $r(u)<_{R_{\ell }^{\prime \prime }}r(v)$
and $R(v)<_{R_{\ell }^{\prime \prime }}L(u)=R(u)$, and thus $\phi _{v}>\phi
_{u}$ in $R_{\ell }^{\prime \prime }$. Therefore, $\phi _{v}>\phi _{u}$ also
in $R_{\ell }$, since all the slopes are the same in both $R_{\ell }$ and $%
R_{\ell }^{\prime \prime }$. However, recall that $\phi _{v}<\phi _{\ell }$
in $R_{\ell }$ (as we proved in the beginning of Case 2b-i), and thus $\phi
_{v}<\phi _{u}$ in $R_{\ell }$ by Remark~\ref{R-ell-slopes}, since $u\in
V(G_{0})$. This is a contradiction, since $\phi _{v}>\phi _{u}$ in $R_{\ell
} $.

\emph{Case 2b-ii.} $v\in \bigcup\nolimits_{i=1}^{\infty }H_{i}$. Let $%
v=v_{i}\in H_{i}$ for some $i\geq 1$ and let $(v_{0},v_{1},\ldots ,v_{i})$
be an $H_{i}$-chain of $v_{i}$. Recall that $P_{v}\ll _{R_{\ell }^{\prime
\prime }}P_{x}$ and that $P_{v}$ intersects $P_{x}$ in $R_{\ell }^{\prime }$
by our initial assumption on $v$ and on $x$. Assume w.l.o.g.~that $i\geq 1$
is the smallest index, such that $P_{v}=P_{v_{i}}$ does not intersect $P_{x}$
in $R_{\ell }^{\prime \prime }$, i.e.~in particular $P_{v_{i-1}}$ intersects 
$P_{x}$ in $R_{\ell }^{\prime \prime }$. Recall that both $v_{i}$ and $%
v_{i-1}$ are bounded by Lemma~\ref{N-H-C2-Cu-bounded}, and thus $P_{v_{i}}$
does not intersect $P_{v_{i-1}}$ in $R_{\ell }^{\prime }$, i.e.~either $%
P_{v_{i-1}}\ll _{R_{\ell }^{\prime }}P_{v_{i}}$ or $P_{v_{i}}\ll _{R_{\ell
}^{\prime }}P_{v_{i-1}}$. Let first $P_{v_{i-1}}\ll _{R_{\ell }^{\prime
}}P_{v_{i}}$. Recall that the left line of $P_{v_{i}}$ does not move by
Transformation~\ref{trans5} and that the right line of $P_{v_{i-1}}$ is
possibly moved to the left by Transformation~\ref{trans5}. Thus, also $%
P_{v_{i-1}}\ll _{R_{\ell }^{\prime \prime }}P_{v_{i}}$ in $R_{\ell }^{\prime
\prime }$. Furthermore, since $P_{v_{i}}=P_{v}\ll _{R_{\ell }^{\prime \prime
}}P_{x}$ by our assumption on $v$, it follows that $P_{v_{i-1}}\ll _{R_{\ell
}^{\prime }}P_{x}$. This is a contradiction, since $P_{v_{i-1}}$ intersects $%
P_{x}$ in $R_{\ell }^{\prime \prime }$.

Let now $P_{v_{i}}\ll _{R_{\ell }^{\prime }}P_{v_{i-1}}$, and thus in
particular $l(v_{i})<_{R_{\ell }^{\prime }}l(v_{i-1})$. Thus also $%
l(v_{i})<_{R_{\ell }}l(v_{i-1})$, since the left lines of $P_{v_{i}}$ and $%
P_{v_{i-1}}$ do not move by Transformation~\ref{trans4}. Furthermore $%
l(v_{i})<_{R}l(v_{i-1})$, since $v_{i},v_{i-1}\notin V(G_{0})$ (i.e.~$%
P_{v_{i}}$ and $P_{v_{i-1}}$ remain the same in both $R$ and $R_{\ell }$).
Recall now that $v_{i}$ and $v_{i-1}$ are bounded by Lemma~\ref%
{N-H-C2-Cu-bounded}, and thus $P_{v_{i}}$ does not intersect $P_{v_{i-1}}$
in $R$, i.e.~either $P_{v_{i-1}}\ll _{R}P_{v_{i}}$ or $P_{v_{i}}\ll
_{R}P_{v_{i-1}}$. Therefore, since $l(v_{i})<_{R}l(v_{i-1})$, it follows
that $P_{v_{i}}\ll _{R}P_{v_{i-1}}$.

We will now prove that $b_{u}<_{R}r(v_{i})<_{R}l(v_{i-1})$. Recall that $%
z_{0}\in N(u)$. Thus, if $z_{0}\in V(G_{0})$, then $r(u)<_{R_{0}}r(z_{0})$
by Lemma~\ref{unbounded-bounded}, and thus also $r(u)<_{R_{\ell }}r(z_{0})$,
since $R_{0}$ is a sub-representation of $R_{\ell }$. Furthermore $%
b_{u}<_{R_{\ell }^{\prime }}r(u)<_{R_{\ell }^{\prime }}r(z_{0})$, since the
right endpoint $r(z_{0})$ of $P_{z_{0}}$ does not decrease by Transformation~%
\ref{trans4}. On the other hand, let $z_{0}\notin V(G_{0})$. Then $%
b_{u}<_{R}r(z_{0})$ by Lemma~\ref{unbounded-bounded}, and thus also $%
b_{u}<_{R_{\ell }}r(z_{0})$, since $z_{0}\notin V(G_{0})$ (i.e.~the
endpoints of $P_{z_{0}}$ are the same in both $R$ and $R_{\ell }$).
Furthermore $b_{u}<_{R_{\ell }^{\prime }}r(z_{0})$, since $r(z_{0})$ does
not decrease by Transformation~\ref{trans4}. That is, in both cases where $%
z_{0}\in V(G_{0})$ and $z_{0}\notin V(G_{0})$, it follows that $%
b_{u}<_{R_{\ell }^{\prime }}r(z_{0})=r_{0}<_{R_{\ell }^{\prime }}r(v)$
(since $r_{0}<_{R_{\ell }^{\prime }}r(v)$ by our initial assumption on $v$),
and thus $b_{u}<_{R_{\ell }^{\prime }}r(v)=r(v_{i})$. Furthermore, $%
b_{u}<_{R_{\ell }^{\prime }}r(v_{i})<_{R_{\ell }^{\prime }}l(v_{i-1})$,
since we assumed that $P_{v_{i}}\ll _{R_{\ell }^{\prime }}P_{v_{i-1}}$.
Recall now that the value $r(v_{i})$ remains the same in both $R_{\ell }$
and $R_{\ell }^{\prime }$, since $v_{i}\notin N(u)$ and by Transformation~%
\ref{trans4} only some endpoints of vertices of $N(u)$ are moved.
Furthermore, the value $l(v_{i-1})$ remains the same in both $R_{\ell }$ and 
$R_{\ell }^{\prime }$, since the left lines of the parallelograms do not
move by Transformation~\ref{trans4}. Therefore $b_{u}<_{R_{\ell
}}r(v_{i})<_{R_{\ell }}l(v_{i-1})$, since also $b_{u}<_{R_{\ell }^{\prime
}}r(v_{i})<_{R_{\ell }^{\prime }}l(v_{i-1})$. Moreover, since $%
v_{i},v_{i-1}\notin V(G_{0})$ (i.e.~the endpoints of $P_{v_{i}}$ and $%
P_{v_{i-1}}$ remain the same in both $R$ and $R_{\ell }$), it follows that $%
b_{u}<_{R}r(v_{i})<_{R}l(v_{i-1})$.

Suppose that $v_{i-1}\in N(u)$. Then $L(v_{i-1})<_{R}L(u)=a_{u}$ by Lemma %
\ref{unbounded-bounded}, and thus $R(v_{i})<_{R}L(v_{i-1})<_{R}a_{u}$, since 
$P_{v_{i}}\ll _{R}P_{v_{i-1}}$. That is, $R(v_{i})<_{R}a_{u}$ and $%
b_{u}<_{R}r(v_{i})$ (by the previous paragraph). Therefore, $P_{v_{i}}$
intersects $P_{u}$ in $R$ and $\phi _{v_{i}}>\phi _{u}$ in $R$. Thus, since $%
v_{i}$ is bounded, it follows that $v_{i}\in N(u)$. This is a contradiction
to the assumption that $v_{i}=v\in V_{0}(u)$. Therefore $v_{i-1}\notin N(u)$%
. Thus, since $v_{i-1}\in N(x_{2})$ (by definition of $H$) and $x_{2}\in
V_{0}(u)$, it follows that $v_{i-1}\in V_{0}(u)$. Therefore, in particular $%
l(v_{i-1})\leq _{R_{\ell }^{\prime }}\ell _{0}$, since $\ell
_{0}=\max_{R_{\ell }^{\prime }}\{l(x)\ |\ x\in V_{0}(u)\}$.

Recall now that $P_{v_{i}}\ll _{R_{\ell }^{\prime }}P_{v_{i-1}}$ (as we
assumed) and that $r_{0}=r(z_{0})<_{R_{\ell }^{\prime }}r(v)=r(v_{i})$ (by
our initial assumption on $v$). Therefore $r(z_{0})<_{R_{\ell }^{\prime
}}r(v_{i})<_{R_{\ell }^{\prime }}l(v_{i-1})\leq _{R_{\ell }^{\prime }}\ell
_{0}$, i.e.~$r(z_{0})<_{R_{\ell }^{\prime }}\ell _{0}$. This is a
contradiction, since $z_{0}\in N_{2}$.

Summarizing Cases 1 and 2, it follows that no adjacency $vx$ has been
removed from $R_{\ell }^{\prime }$ in $R_{\ell }^{\prime \prime }$ in the
case where $x\in N(u)$. This completes the proof of the lemma.
\end{proof}

\begin{lemma}
\label{R'''-ell}$R_{\ell }^{\prime \prime \prime }$ is a projection
representation of $G$.
\end{lemma}

\begin{proof}
The proof is done in two parts. In Part 1 we prove that $u$ is adjacent in $%
R_{\ell }^{\prime \prime \prime }$ to all vertices of $N(u)$, while in Part
2 we prove that $u$ is not adjacent in $R_{\ell }^{\prime \prime \prime }$
to any vertex of~${V\setminus N[u]}$.

\medskip

\emph{Part 1.} In this part we prove that $u$ is adjacent in $R_{\ell
}^{\prime \prime \prime }$ to all vertices of $N(u)$. Denote by $\widehat{a}%
_{u}$ and $\widehat{b}_{u}$ the coordinates of the upper and lower endpoint
of $P_{u}$ in the projection representation $R_{\ell }$ on $L_{1}$ and on $%
L_{2}$, respectively. Then, since the endpoints of $P_{u}$ do not move by
Transformations~\ref{trans4} and~\ref{trans5}, $\widehat{a}_{u}$ and $%
\widehat{b}_{u}$ remain the endpoints of $P_{u}$ also in the representations 
$R_{\ell }^{\prime }$ and $R_{\ell }^{\prime \prime }$. Let $z\in N(u)$ be
arbitrary. Suppose that $z\notin V(G_{0})$. Then, the left line of $P_{z}$
remains the same in the representations $R$, $R_{\ell }$, $R_{\ell }^{\prime
}$, and $R_{\ell }^{\prime \prime }$. Therefore, since $%
L(z)<_{R}a_{u}=_{R}L(u)$ by Lemma~\ref{unbounded-bounded}, it follows that
also $L(z)<_{R_{\ell }^{\prime \prime }}a_{u}<_{R_{\ell }^{\prime \prime
}}L(u)=\widehat{a}_{u}$. Suppose that $z\in V(G_{0})$. Then, $%
L(z)<_{R_{0}}L(u)$ by Lemma~\ref{unbounded-bounded}, since $R_{0}$ is a
projection representation of $G_{0}$, and thus also $L(z)<_{R_{\ell }}L(u)=%
\widehat{a}_{u}$, since $R_{0}$ is a sub-representation of $R_{\ell }$.
Furthermore $L(z)<_{R_{\ell }^{\prime \prime }}L(u)=\widehat{a}_{u}$, since
the left line of $P_{z}$ remains the same in the representations $R_{\ell }$%
, $R_{\ell }^{\prime }$, and $R_{\ell }^{\prime \prime }$. Summarizing, $%
L(z)<_{R_{\ell }^{\prime \prime }}\widehat{a}_{u}$ for every vertex $z\in
N(u)$. Therefore, since the endpoint $L(z)$ does not move by Transformation~%
\ref{trans6}, it follows that also $L(z)<_{R_{\ell }^{\prime \prime \prime }}%
\widehat{a}_{u}$ for every vertex $z\in N(u)$.

Note now that $\widehat{a}_{u}<_{R_{\ell }^{\prime \prime }}L_{0}$, since $%
L_{0}=\min_{R_{\ell }^{\prime \prime }}\{L(x)\ |\ x\in V_{B}\setminus
N(u)\setminus V_{0}(u),P_{u}\ll _{R_{\ell }^{\prime \prime }}P_{x}\}$.
Furthermore, recall by Corollary~\ref{property-trans5} that for all vertices 
$z\in N_{1}\setminus N_{2}$, for which $R(z)<_{R_{\ell }^{\prime \prime
}}L_{0}$, the values $R(z)$ lie immediately before $L_{0}$ in $R_{\ell
}^{\prime \prime }$. Therefore, since $\widehat{a}_{u}<_{R_{\ell }^{\prime
\prime }}L_{0}$, it follows in particular that $\widehat{a}_{u}<_{R_{\ell
}^{\prime \prime }}R(z)$ for every $z\in N_{1}\setminus N_{2}$, and thus $%
L(z)<_{R_{\ell }^{\prime \prime }}\widehat{a}_{u}<_{R_{\ell }^{\prime \prime
}}R(z)$ for every $z\in N_{1}\setminus N_{2}\subseteq N(u)$ by the previous
paragraph. Therefore, since $\widehat{a}_{u}<_{R_{\ell }^{\prime \prime
}}L_{0}$, and since the upper endpoint $R(u)$ of the line $P_{u}$ lies in $%
R_{\ell }^{\prime \prime }$ immediately before $\min_{R_{\ell }^{\prime
\prime }}\{L_{0},R(z)\ |\ z\in N_{1}\setminus N_{2}\}$, cf.~the statement of
Transformation~\ref{trans6}, it follows that also $L(z)<_{R_{\ell }^{\prime
\prime \prime }}\widehat{a}_{u}<_{R_{\ell }^{\prime \prime \prime
}}R(u)<_{R_{\ell }^{\prime \prime \prime }}R(z)$ for every $z\in
N_{1}\setminus N_{2}$. That is, $L(z)<_{R_{\ell }^{\prime \prime \prime
}}R(u)<_{R_{\ell }^{\prime \prime \prime }}R(z)$ for every $z\in
N_{1}\setminus N_{2}$, and thus $P_{u}$ intersects $P_{z}$ in $R_{\ell
}^{\prime \prime \prime }$ for every $z\in N_{1}\setminus N_{2}$. Therefore,
since all vertices of $\{u\}\cup N_{1}\setminus N_{2}$ are bounded in $%
R_{\ell }^{\prime \prime \prime }$, $u$ is adjacent in $R_{\ell }^{\prime
\prime \prime }$ to all vertices of $N_{1}\setminus N_{2}$.

Consider now an arbitrary vertex $z\in N_{2}$. Recall that $%
r_{0}=\min_{R_{\ell }^{\prime }}\{r(z)\ |\ z\in N_{2}\}$, i.e.~$r_{0}\leq
_{R_{\ell }^{\prime }}r(z)$. Thus, since the endpoint $r(z)$ does not move
by Transformation~\ref{trans5}, it follows that also $r_{0}\leq _{R_{\ell
}^{\prime \prime }}r(z)$. Furthermore, by Transformation~\ref{trans5}, $%
r(v)<_{R_{\ell }^{\prime \prime }}r_{0}\leq _{R_{\ell }^{\prime \prime
}}r(z) $ for every $v\in V_{0}(u)\cap V_{B}$. This holds clearly also in $%
R_{\ell }^{\prime \prime \prime }$, i.e.~$r(v)<_{R_{\ell }^{\prime \prime
\prime }}r(z)$ for every $v\in V_{0}(u)\cap V_{B}$. Since the lower endpoint
of the line $P_{u}$ comes immediately after $\max_{R_{\ell }^{\prime \prime
\prime }}\{r(v)\ |\ V_{0}(u)\cap V_{B}\}$, it follows that $r(v)<_{R_{\ell
}^{\prime \prime \prime }}l(u)=r(u)<_{R_{\ell }^{\prime \prime \prime }}r(z)$
for every $v\in V_{0}(u)\cap V_{B}$ and every $z\in N_{2}$. Thus, since also 
$L(z)<_{R_{\ell }^{\prime \prime \prime }}\widehat{a}_{u}<_{R_{\ell
}^{\prime \prime \prime }}R(u)$ for every $z\in N(u)$, it follows that $%
P_{u} $ intersects $P_{z}$ in $R_{\ell }^{\prime \prime \prime }$ for every $%
z\in N_{2}$. Therefore, since all vertices of $\{u\}\cup N_{2}$ are bounded
in $R_{\ell }^{\prime \prime \prime }$, $u$ is adjacent in $R_{\ell
}^{\prime \prime \prime }$ to all vertices of $N_{2}$. Thus, since $%
N_{2}\cup (N_{1}\setminus N_{2})=N(u)$, $u$ is adjacent in $R_{\ell
}^{\prime \prime \prime }$ to all vertices of $N(u)$.

\medskip

\emph{Part 2.} In this part we prove that $u$ is not adjacent in $R_{\ell
}^{\prime \prime \prime }$ to any vertex of $V\setminus N[u]$. To this end,
recall first by Lemma~\ref{bounded-hovering} that $u^{\ast }$ is a bounded
covering vertex of $u$ in $G$ (and thus~$u^{\ast }\in V_{0}(u)\cap V_{B}$),
such that $P_{u}$ intersects $P_{u^{\ast }}$ in the initial projection
representation $R$ and $\phi _{u^{\ast }}<\phi _{u}$ in $R$. Therefore, $%
l(u^{\ast })<_{R}b_{u}=_{R}r(u)$ by Lemma~\ref{unbounded-hovering}.
Furthermore, $u^{\ast }\notin V(G_{0})$ by Observation~\ref%
{bounded-hovering-obs}. Therefore, the endpoint $l(u^{\ast })$ remains the
same in the representations $R$, $R_{\ell }$, $R_{\ell }^{\prime }$, and $%
R_{\ell }^{\prime \prime }$, and thus $l(u^{\ast })<_{R_{\ell }^{\prime
\prime }}b_{u}$, since also $l(u^{\ast })<_{R}b_{u}$. Therefore, since $%
b_{u}<_{R_{\ell }^{\prime \prime }}\widehat{b}_{u}=_{R_{\ell }^{\prime
\prime }}r(u)$, it follows that also $l(u^{\ast })<_{R_{\ell }^{\prime
\prime }}\widehat{b}_{u}=_{R_{\ell }^{\prime \prime }}r(u)$. Recall now that 
$L_{0}=\min_{R_{\ell }^{\prime \prime }}\{L(x)\ |\ x\in V_{B}\setminus
N(u)\setminus V_{0}(u),P_{u}\ll _{R_{\ell }^{\prime \prime }}P_{x}\}$.
Denote by $y_{0}$ the vertex of $V_{B}\setminus N(u)\setminus V_{0}(u)$,
such that $L_{0}=L(y_{0})$ in $R_{\ell }^{\prime \prime }$, and thus $%
P_{u}\ll _{R_{\ell }^{\prime \prime }}P_{y_{0}}$. Therefore, since $%
l(u^{\ast })<_{R_{\ell }^{\prime \prime }}r(u)$, it follows that $l(u^{\ast
})<_{R_{\ell }^{\prime \prime }}r(u)<_{R_{\ell }^{\prime \prime }}l(y_{0})$.
Now, since $u^{\ast }\in V_{0}(u)$ and $y_{0}\notin N(u)\cup V_{0}(u)$, it
follows that $u^{\ast }y_{0}\notin E$. Thus, $P_{u^{\ast }}\ll _{R_{\ell
}^{\prime \prime }}P_{y_{0}}$, since both $u^{\ast }$ and $y_{0}$ are
bounded vertices and $l(u^{\ast })<_{R_{\ell }^{\prime \prime }}l(y_{0})$.
Moreover, since by Transformation~\ref{trans6} only the line $P_{u}$ is
moved, it follows that also $P_{u^{\ast }}\ll _{R_{\ell }^{\prime \prime
\prime }}P_{y_{0}}$.

Recall that $u^{\ast }\notin V(G_{0})$ and that $u^{\ast }$ is adjacent to
every vertex of $V(G_{0})\setminus \{u\}$ by Observation~\ref%
{bounded-hovering-obs}. Therefore $u^{\ast }\in N(x_{2})$, since $x_{2}\in
V(G_{0})\setminus \{u\}$, and thus $P_{u^{\ast }}$ intersects the line
segment $\ell $ in $R_{\ell }$; in particular, $P_{u^{\ast }}$ intersects $%
P_{u}$ in $R_{\ell }$. Moreover, since by Transformation~\ref{trans4} the
parallelogram $P_{u^{\ast }}$ is not modified, $P_{u^{\ast }}$ intersects $%
P_{u}$ also in $R_{\ell }^{\prime }$. Denote by $z_{0}$ the vertex of $N_{2}$%
, such that $r_{0}=r(z_{0})$. We will now prove that $r(u)<_{R_{\ell
}^{\prime }}r_{0}=r(z_{0})$. Suppose first that $z_{0}\notin V(G_{0})$.
Then, in particular, either $r(z_{0})<_{R_{\ell }^{\prime }}b_{\ell
}-\varepsilon <_{R_{\ell }^{\prime }}l(x_{2})$ or $r(x_{2})<_{R_{\ell
}^{\prime }}b_{\ell }+\varepsilon <_{R_{\ell }^{\prime }}r(z_{0})$ by
Remarks~\ref{trans4-remark-1} and~\ref{trans4-remark-2}. Recall that $\ell
_{0}=\max_{R_{\ell }^{\prime }}\{l(x)\ |\ x\in V_{0}(u)\}$ and that $%
z_{0}\in N_{2}$, and thus $l(x_{2})\leq _{R_{\ell }^{\prime }}\ell
_{0}<_{R_{\ell }^{\prime }}r(z_{0})$. Therefore $r(x_{2})<_{R_{\ell
}^{\prime }}b_{\ell }+\varepsilon <_{R_{\ell }^{\prime }}r(z_{0})$. Thus,
since $u\in V(G_{0})$, also $r(u)<_{R_{\ell }^{\prime }}b_{\ell
}+\varepsilon <_{R_{\ell }^{\prime }}r(z_{0})$ in the case where $%
z_{0}\notin V(G_{0})$. Suppose now that $z_{0}\in V(G_{0})$; then $%
r(u)<_{R_{0}}r(z_{0})$ by Lemma~\ref{unbounded-bounded}. Thus, since $R_{0}$
is a sub-representation of $R_{\ell }^{\prime }$, and since $r(z_{0})$ does
not decrease by Transformation~\ref{trans4}, it follows that $r(u)<_{R_{\ell
}^{\prime }}r(z_{0})=r_{0}$ in the case where $z_{0}\in V(G_{0})$. That is, $%
r(u)<_{R_{\ell }^{\prime }}r_{0}=r(z_{0})$ in both cases, where $z_{0}\in
V(G_{0})$ and $z_{0}\notin V(G_{0})$.

We will now prove that $P_{u^{\ast }}$ intersects $P_{u}$ also in $R_{\ell
}^{\prime \prime }$. This holds clearly in the case where the right line of $%
P_{u^{\ast }}$ is not moved during Transformation~\ref{trans5}, since $%
P_{u^{\ast }}$ intersects $P_{u}$ in $R_{\ell }^{\prime }$ by the previous
paragraph. Suppose now that the right line of $P_{u^{\ast }}$ is moved
during Transformation~\ref{trans5}. Then, $r(u)<_{R_{\ell }^{\prime
}}r_{0}<_{R_{\ell }^{\prime }}r(u^{\ast })$, while $r(u^{\ast })$ comes
immediately before $r_{0}$ in $R_{\ell }^{\prime \prime }$, i.e.~$%
r(u)<_{R_{\ell }^{\prime \prime }}r(u^{\ast })<_{R_{\ell }^{\prime \prime
}}r_{0}$, since $r_{0}=r(z_{0})$ does not move during Transformation~\ref%
{trans5}. Therefore, since the left line of $P_{u^{\ast }}$ does not move
during Transformation~\ref{trans5}, and since $P_{u^{\ast }}$ intersects $%
P_{u}$ in $R_{\ell }^{\prime }$, it follows that $P_{u^{\ast }}$ intersects $%
P_{u}$ also in $R_{\ell }^{\prime \prime }$.

Denote by $v_{0}$ the vertex of $V_{0}(u)\cap V_{B}$, such that $%
r(v_{0})=\max_{R_{\ell }^{\prime \prime }}\{r(v)\ |\ v\in V_{0}(u)\cap
V_{B}\}$, cf.~the statement of Transformation~\ref{trans6}. Since $v_{0}\in
V_{0}(u)$ and $y_{0}\notin N(u)\cup V_{0}(u)$, it follows that $%
v_{0}y_{0}\notin E$. Therefore, since both $v_{0}$ and $y_{0}$ are bounded
vertices, either $P_{y_{0}}\ll _{R_{\ell }^{\prime \prime }}P_{v_{0}}$ or $%
P_{v_{0}}\ll _{R_{\ell }^{\prime \prime }}P_{y_{0}}$. Suppose that $%
P_{y_{0}}\ll _{R_{\ell }^{\prime \prime }}P_{v_{0}}$, and thus $P_{u^{\ast
}}\ll _{R_{\ell }^{\prime \prime }}P_{y_{0}}\ll _{R_{\ell }^{\prime \prime
}}P_{v_{0}}$. Then, since $u^{\ast },v_{0}\in V_{0}(u)$ and since $V_{0}(u)$
is connected, there exists at least one vertex $v\in V_{0}(u)$, such that $%
P_{v}$ intersects $P_{y_{0}}$ in $R_{\ell }^{\prime \prime }$. Similarly $%
vy_{0}\notin E$, since $y_{0}\notin N(u)\cup V_{0}(u)$. Therefore, since $%
y_{0}$ is a bounded vertex, $v$ must be an unbounded vertex with $\phi
_{v}>\phi _{y_{0}}$ in $R_{\ell }^{\prime \prime }$, and thus $N(v)\subseteq
N(y_{0})$ by Lemma~\ref{intersecting-unbounded}. Then, $N(v)$ includes at
least one vertex $v^{\prime }\in V_{0}(u)$, and thus $v^{\prime }\in
N(y_{0}) $. Therefore, $y_{0}\in V_{0}(u)$, which is a contradiction. Thus, $%
P_{v_{0}}\ll _{R_{\ell }^{\prime \prime }}P_{y_{0}}$. Moreover, since by
Transformation~\ref{trans6} only the line $P_{u}$ is moved, it follows that
also $P_{v_{0}}\ll _{R_{\ell }^{\prime \prime \prime }}P_{y_{0}}$.

We will prove in the following that $u$ is not adjacent in $R_{\ell
}^{\prime \prime \prime }$ to any vertex $x\notin N(u)$. For the sake of
contradiction, suppose that $P_{x}$ intersects $P_{u}$ in $R_{\ell }^{\prime
\prime \prime }$. We distinguish in the following the cases regarding $x$.

\emph{Case 2a.} $x\in V_{B}\setminus N(u)$ (i.e.~$x$ is bounded) and $x\in
V_{0}(u)$. Then, $r(x)\leq _{R_{\ell }^{\prime \prime }}r(v_{0})$ and $%
r(u^{\ast })\leq _{R_{\ell }^{\prime \prime }}r(v_{0})$ by definition of $%
v_{0}$, and thus also $r(x)\leq _{R_{\ell }^{\prime \prime \prime }}r(v_{0})$
and $r(u^{\ast })\leq _{R_{\ell }^{\prime \prime \prime }}r(v_{0})$.
Therefore, by Transformation~\ref{trans6}, $r(x)\leq _{R_{\ell }^{\prime
\prime \prime }}r(v_{0})<_{R_{\ell }^{\prime \prime \prime }}l(u)$, i.e.~$%
r(x)<_{R_{\ell }^{\prime \prime \prime }}l(u)$. Thus $L(u)<_{R_{\ell
}^{\prime \prime \prime }}R(x)$, since we assumed that $P_{x}$ intersects $%
P_{u}$ in $R_{\ell }^{\prime \prime \prime }$. Furthermore, $r(x)\leq
_{R_{\ell }^{\prime \prime \prime }}r(v_{0})<_{R_{\ell }^{\prime \prime
\prime }}l(y_{0})$, i.e.~$r(x)<_{R_{\ell }^{\prime \prime \prime }}l(y_{0})$%
, since $P_{v_{0}}\ll _{R_{\ell }^{\prime \prime \prime }}P_{y_{0}}$. Recall
by Corollary~\ref{property-trans5} that for all vertices $z\in
N_{1}\setminus N_{2}$, for which $R(z)<_{R_{\ell }^{\prime \prime
}}L_{0}=L(y_{0})$, the values $R(z)$ lie immediately before $L_{0}$ in $%
R_{\ell }^{\prime \prime }$, and thus also in $R_{\ell }^{\prime \prime
\prime }$. Thus, since $L(u)<_{R_{\ell }^{\prime \prime \prime }}R(x)$, and
since the upper point $L(u)=R(u)$ lies immediately before $\min
\{L_{0},R(z)\ |\ z\in N_{1}\setminus N_{2}\}$ in $R_{\ell }^{\prime \prime
\prime }$, it follows that $L(u)<_{R_{\ell }^{\prime \prime \prime
}}L_{0}=L(y_{0})<_{R_{\ell }^{\prime \prime \prime }}R(x)$. Therefore, since
also $r(x)<_{R_{\ell }^{\prime \prime \prime }}l(y_{0})$, $P_{x}$ intersects 
$P_{y_{0}}$ in $R_{\ell }^{\prime \prime \prime }$, and thus also in $%
R_{\ell }^{\prime \prime }$. Thus $xy_{0}\in E$, since both $x$ and $y_{0}$
are bounded, and therefore $y_{0}\in V_{0}(u)$, which is a contradiction.
Therefore, $P_{x}$ does not intersect $P_{u}$ in $R_{\ell }^{\prime \prime
\prime }$, for every $x\in V_{B}\setminus N(u)$, such that $x\in V_{0}(u)$.
In particular, since $u^{\ast },v_{0}\in V_{B}\setminus N(u)$ and $u^{\ast
},v_{0}\in V_{0}(u)$, it follows that neither $P_{u^{\ast }}$ nor $P_{v_{0}}$
intersects $P_{u}$ in $R_{\ell }^{\prime \prime \prime }$. Therefore, since $%
r(u^{\ast })\leq _{R_{\ell }^{\prime \prime \prime }}r(v_{0})<_{R_{\ell
}^{\prime \prime \prime }}l(u)$ by Transformation~\ref{trans6}, it follows
that $P_{u^{\ast }}\ll _{R_{\ell }^{\prime \prime \prime }}P_{u}$ and $%
P_{v_{0}}\ll _{R_{\ell }^{\prime \prime \prime }}P_{u}$.

\emph{Case 2b.} $x\in V_{B}\setminus N(u)$ (i.e.~$x$ is bounded) and $%
x\notin V_{0}(u)$. Then $u^{\ast }x\notin E$, since $u^{\ast }\in V_{0}(u)$.
Furthermore, since both $x$ and $u^{\ast }$ (resp.~$v_{0}$) are bounded
vertices, either $P_{x}\ll _{R_{\ell }^{\prime \prime \prime }}P_{u^{\ast }}$
or $P_{u^{\ast }}\ll _{R_{\ell }^{\prime \prime \prime }}P_{x}$ (resp.
either $P_{x}\ll _{R_{\ell }^{\prime \prime \prime }}P_{v_{0}}$ or $%
P_{v_{0}}\ll _{R_{\ell }^{\prime \prime \prime }}P_{x}$). If $P_{x}\ll
_{R_{\ell }^{\prime \prime \prime }}P_{u^{\ast }}$ (resp.~$P_{x}\ll
_{R_{\ell }^{\prime \prime \prime }}P_{v_{0}}$), then $P_{x}\ll _{R_{\ell
}^{\prime \prime \prime }}P_{u^{\ast }}\ll _{R_{\ell }^{\prime \prime \prime
}}P_{u}$ (resp.~$P_{x}\ll _{R_{\ell }^{\prime \prime \prime }}P_{v_{0}}\ll
_{R_{\ell }^{\prime \prime \prime }}P_{u}$) by the previous paragraph. This
is a contradiction to the assumption that $P_{x}$ intersects $P_{u}$ in $%
R_{\ell }^{\prime \prime \prime }$. Therefore $P_{u^{\ast }}\ll _{R_{\ell
}^{\prime \prime \prime }}P_{x}$ and $P_{v_{0}}\ll _{R_{\ell }^{\prime
\prime \prime }}P_{x}$, and thus also $P_{u^{\ast }}\ll _{R_{\ell }^{\prime
\prime }}P_{x}$ and $P_{v_{0}}\ll _{R_{\ell }^{\prime \prime }}P_{x}$. Thus,
in particular $r(v_{0})<_{R_{\ell }^{\prime \prime \prime }}l(x)$.
Furthermore, the lower endpoint $l(u)=r(u)$ of $P_{u}$ comes by
Transformation~\ref{trans6} immediately after $r(v_{0})$ in $R_{\ell
}^{\prime \prime \prime }$, and thus $r(v_{0})<_{R_{\ell }^{\prime \prime
\prime }}r(u)<_{R_{\ell }^{\prime \prime \prime }}l(x)$. Then, $%
L(x)<_{R_{\ell }^{\prime \prime \prime }}R(u)$, since we assumed that $P_{x}$
intersects $P_{u}$ in $R_{\ell }^{\prime \prime \prime }$.

We distinguish now the cases according to the relative positions of $P_{u}$
and $P_{x}$ in $R_{\ell }^{\prime \prime }$. If $P_{x}\ll _{R_{\ell
}^{\prime \prime }}P_{u}$, then $P_{u^{\ast }}\ll _{R_{\ell }^{\prime \prime
}}P_{x}\ll _{R_{\ell }^{\prime \prime }}P_{u}$ by the previous paragraph,
which is a contradiction, since $P_{u^{\ast }}$ intersects $P_{u}$ in $%
R_{\ell }^{\prime \prime }$, as we proved above. If $P_{u}\ll _{R_{\ell
}^{\prime \prime }}P_{x}$, then $L_{0}\leq _{R_{\ell }^{\prime \prime }}L(x)$%
, since $x\in V_{B}\setminus N(u)\setminus V_{0}(u)$ and $%
L_{0}=\min_{R_{\ell }^{\prime \prime }}\{L(x)\ |\ x\in V_{B}\setminus
N(u)\setminus V_{0}(u),P_{u}\ll _{R_{\ell }^{\prime \prime }}P_{x}\}$. Thus $%
R(u)<_{R_{\ell }^{\prime \prime \prime }}L_{0}\leq _{R_{\ell }^{\prime
\prime \prime }}L(x)$ by Transformation~\ref{trans3}, which is a
contradiction, since $L(x)<_{R_{\ell }^{\prime \prime \prime }}R(u)$ by the
previous paragraph. Suppose that $P_{x}$ intersects $P_{u}$ in $R_{\ell
}^{\prime \prime }$. Note that $x\notin V(G_{0})$, since $x\notin N(u)\cup
V_{0}(u)$ and $V(G_{0})\subseteq N[u]\cup V_{0}(u)$ by Observation~\ref%
{V(G0)}. Thus, since we assumed that $P_{x}$ intersects $P_{u}$ in $R_{\ell
}^{\prime \prime }$, i.e.~$P_{x}$ intersects the line segment $\ell $ in $%
R_{\ell }^{\prime \prime }$, it follows that $P_{x}$ intersects also $%
P_{x_{2}}$ in $R_{\ell }^{\prime \prime }$. Therefore $x\in N(x_{2})$, since
both $x$ and $x_{2}$ are bounded, and thus $x\in V_{0}(u)$, since also $%
x_{2}\in V_{0}(u)$. This is a contradiction, since $x\notin V_{0}(u)$ by the
assumption of Case 2b. Therefore, $P_{x}$ does not intersect $P_{u}$ in $%
R_{\ell }^{\prime \prime \prime }$, for every $x\in V_{B}\setminus N(u)$,
such that $x\notin V_{0}(u)$.

\emph{Case 2c.} $x\in V_{U}$ (i.e.~$x$ is unbounded), such that $\phi
_{x}<\phi _{u}$ in $R_{\ell }^{\prime \prime \prime }$. Then, since both $%
P_{x}$ and $P_{u}$ are lines in $R_{\ell }^{\prime \prime \prime }$, it
follows that $l(x)<_{R_{\ell }^{\prime \prime \prime }}l(u)$ and $%
R(x)>_{R_{\ell }^{\prime \prime \prime }}R(u)$. Thus, by Transformation~\ref%
{trans6}, $l(x)<_{R_{\ell }^{\prime \prime \prime }}r(v_{0})<_{R_{\ell
}^{\prime \prime \prime }}l(u)$ and $R(u)<_{R_{\ell }^{\prime \prime \prime
}}L_{0}=L(y_{0})<_{R_{\ell }^{\prime \prime \prime }}R(x)$. Since $%
P_{v_{0}}\ll _{R_{\ell }^{\prime \prime \prime }}P_{y_{0}}$ (as we proved
above), it follows that $P_{x}$ intersects both $P_{v_{0}}$ and $P_{y_{0}}$
in $R_{\ell }^{\prime \prime \prime }$ (and thus also in $R_{\ell }^{\prime
\prime })$, and that $\phi _{x}<\phi _{v_{0}}$ and $\phi _{x}<\phi _{y_{0}}$
in both $R_{\ell }^{\prime \prime }$ and $R_{\ell }^{\prime \prime \prime }$%
. Therefore, since both $v_{0}$ and $y_{0}$ are bounded, it follows that $%
x\in N(v_{0})$ and $x\in N(y_{0})$. Thus $x,y_{0}\in V_{0}(u)$, since $%
v_{0}\in V_{0}(u)$. This is a contradiction, since $y_{0}\notin V_{0}(u)$ by
definition of $y_{0}$. Therefore, $P_{x}$ does not intersect $P_{u}$ in $%
R_{\ell }^{\prime \prime \prime }$, for every $x\in V_{U}$, for which $\phi
_{x}<\phi _{u}$ in $R_{\ell }^{\prime \prime \prime }$.

\medskip

Summarizing, due to Part 1 and due to Cases 2a, 2b, and 2c of Part 2, it
follows that $P_{u}$ intersects in $R_{\ell }^{\prime \prime \prime }$ only
the parallelograms $P_{z}$, for every $z\in N(u)$, and possibly some trivial
parallelograms (lines) $P_{x}$, where $x\in V_{U}$ and $\phi _{x}>\phi _{u}$
in $R_{\ell }^{\prime \prime \prime }$. However, since $\phi _{x}>\phi _{u}$
in $R_{\ell }^{\prime \prime \prime }$ for all these vertices $x$, it
follows that $u$ is not adjacent to these vertices in $R_{\ell }^{\prime
\prime \prime }$. Thus $R_{\ell }^{\prime \prime \prime }$ is a projection
representation of $G$, since $R_{\ell }^{\prime \prime }\setminus \{u\}$ is
a projection representation of $G\setminus \{u\}$ by Lemma~\ref{R''-ell}.
This completes the proof of the lemma.
\end{proof}

\medskip

The next lemma follows now easily by Lemma~\ref{R'''-ell} and by the fact
that $V_{0}(u)$ induces a connected subgraph of $G$.

\begin{lemma}
\label{V0(u)-R'''-ell}The (bounded) vertex $u$ has the right border property
in $R_{\ell }^{\prime \prime \prime }$, i.e.~there exists no pair of
vertices $z\in N(u)$ and $v\in V_{0}(u)$, such that $P_{z}\ll _{R_{\ell
}^{\prime \prime \prime }}P_{v}$.
\end{lemma}

\begin{proof}
Recall first that $u_{0}^{\ast }\in V_{0}(u)\cap V_{B}$ by Lemma~\ref%
{bounded-hovering}, i.e.~$V_{0}(u)\cap V_{B}\neq \emptyset $. Furthermore,
recall that by Transformation~\ref{trans6} the lower endpoint $l(u)=r(u)$ of 
$P_{u}$ comes immediately after $\max \{r(v)\ |\ v\in V_{0}(u)\cap V_{B}\}$
in $R_{\ell }^{\prime \prime \prime }$, and thus $r(v)<_{R_{\ell }^{\prime
\prime \prime }}r(u)$ for every $v\in V_{0}(u)\cap V_{B}$. Since $u$ is a
bounded vertex in $R_{\ell }^{\prime \prime \prime }$, and since $R_{\ell
}^{\prime \prime \prime }$ is a projection representation of $G$ by Lemma~%
\ref{R'''-ell}, $P_{u}$ does not intersect $P_{v}$ in $R_{\ell }^{\prime
\prime \prime }$, for any $v\in V_{0}(u)\cap V_{B}$. Therefore, for every $%
v\in V_{0}(u)\cap V_{B}$, either $P_{u}\ll _{R_{\ell }^{\prime \prime \prime
}}P_{v}$ or $P_{v}\ll _{R_{\ell }^{\prime \prime \prime }}P_{u}$. If $%
P_{u}\ll _{R_{\ell }^{\prime \prime \prime }}P_{v}$ for a vertex $v\in
V_{0}(u)\cap V_{B}$, then in particular $r(u)<_{R_{\ell }^{\prime \prime
\prime }}r(v)$, which is a contradiction. Therefore, $P_{v}\ll _{R_{\ell
}^{\prime \prime \prime }}P_{u}$ for every $v\in V_{0}(u)\cap V_{B}$.

Suppose now for the sake of contradiction that $P_{z}\ll _{R_{\ell }^{\prime
\prime \prime }}P_{v}$ for two vertices $z\in N(u)$ and $v\in V_{0}(u)$.
Suppose first that $v$ is a bounded vertex, i.e.~$v\in V_{0}(u)\cap V_{B}$.
Then, since $P_{v}\ll _{R_{\ell }^{\prime \prime \prime }}P_{u}$ by the
previous paragraph, it follows that $P_{z}\ll _{R_{\ell }^{\prime \prime
\prime }}P_{v}\ll _{R_{\ell }^{\prime \prime \prime }}P_{u}$, and thus $%
z\notin N(u)$, which is a contradiction.

Suppose now that $v$ is an unbounded vertex. Then, since $V_{0}(u)$ is
connected and $V_{0}(u)\cap V_{B}\neq \emptyset $, there exists at least one
bounded vertex $v^{\prime }\in V_{0}(u)\cap V_{B}$, such that $v^{\prime
}\in N(v)$. Then $P_{v^{\prime }}\ll _{R_{\ell }^{\prime \prime \prime
}}P_{u}$, as we proved above. We distinguish now the cases according to the
relative positions of $P_{v}$ and $P_{u}$ in $R_{\ell }^{\prime \prime
\prime }$. If $P_{v}\ll _{R_{\ell }^{\prime \prime \prime }}P_{u}$, then $%
P_{z}\ll _{R_{\ell }^{\prime \prime \prime }}P_{v}\ll _{R_{\ell }^{\prime
\prime \prime }}P_{u}$ by the assumption on $z$ and $v$, and thus $z\notin
N(u)$, which is a contradiction. If $P_{u}\ll _{R_{\ell }^{\prime \prime
\prime }}P_{v}$, then $P_{v^{\prime }}\ll _{R_{\ell }^{\prime \prime \prime
}}P_{u}\ll _{R_{\ell }^{\prime \prime \prime }}P_{v}$, and thus $v^{\prime
}\notin N(v)$, which is again a contradiction. Suppose that $P_{v}$
intersects $P_{u}$ in $R_{\ell }^{\prime \prime \prime }$. Then, $\phi
_{v}>\phi _{u}$ in $R_{\ell }^{\prime \prime \prime }$, since $u$ is bounded
in $R_{\ell }^{\prime \prime \prime }$ and $v\notin N(u)$. Therefore, in
particular $r(u)<_{R_{\ell }^{\prime \prime \prime }}r(v)$. Furthermore,
since $v$ is unbounded and $v^{\prime }\in N(u)$, it follows that $%
r(v)<_{R_{\ell }^{\prime \prime \prime }}r(v^{\prime })$ by Lemma~\ref%
{unbounded-bounded}, and thus $r(u)<_{R_{\ell }^{\prime \prime \prime
}}r(v)<_{R_{\ell }^{\prime \prime \prime }}r(v^{\prime })$, i.e.~$%
r(u)<_{R_{\ell }^{\prime \prime \prime }}r(v^{\prime })$. This is a
contradiction, since $P_{v^{\prime }}\ll _{R_{\ell }^{\prime \prime \prime
}}P_{u}$ for every $v^{\prime }\in V_{0}(u)\cap V_{B}$, as we proved above.
Summarizing, there exist no vertices $z\in N(u)$ and $v\in V_{0}(u)$, such
that $P_{z}\ll _{R_{\ell }^{\prime \prime \prime }}P_{v}$. This completes
the proof of the lemma.
\end{proof}

\subsection*{The correctness of Condition~\ref{cond1}}

Note now that the projection representation $R_{\ell }^{\prime \prime \prime
}$ of $G$ (cf.~Lemma~\ref{R'''-ell}) has $k-1$ unbounded vertices, since the
input graph $G$ has $k$ unbounded vertices, and since $u$ is bounded in $%
R_{\ell }^{\prime \prime \prime }$. Therefore, the projection representation 
$R^{\ast }=R_{\ell }^{\prime \prime \prime }$ satisfies the conditions of
Theorem~\ref{no-property-thm}. However, in order to complete the proof of
Theorem~\ref{no-property-thm}, we have to prove the correctness of Condition~%
\ref{cond1} (cf.~Lemma~\ref{lem-cond1}). To this end, we first prove Lemma~%
\ref{R'''-ell-only-N(u)-V0}.

Recall that for simplicity reasons, before applying Transformations~\ref%
{trans4},~\ref{trans5}, and~\ref{trans6}, we have added to $G$ an isolated
bounded vertex $t$, and thus also $t\in V_{B}\setminus N(u)\setminus
V_{0}(u) $. This isolated vertex $t$ corresponds to a parallelogram $P_{t}$,
such that $P_{v}\ll _{R}P_{t}$ and $P_{v}\ll _{R_{\ell }}P_{t}$ for every
other vertex $v$ of $G$; thus also $P_{v}\ll _{R_{\ell }^{\prime }}P_{t}$, $%
P_{v}\ll _{R_{\ell }^{\prime \prime }}P_{t}$, and $P_{v}\ll _{R_{\ell
}^{\prime \prime \prime }}P_{t}$ for every vertex $v\neq t$ of $G$. The next
lemma follows now easily by Transformation~\ref{trans6} and Lemma~\ref%
{R'''-ell}.

\begin{lemma}
\label{R'''-ell-only-N(u)-V0}If $V_{B}\setminus N(u)\setminus V_{0}(u)=\{t\}$%
, then there exists a projection representation $R^{\#}$ of $G$ with the
same unbounded vertices as in $R$, where $u$ has the right border property
in $R^{\#}$.
\end{lemma}

\begin{proof}
Suppose that $V_{B}\setminus N(u)\setminus V_{0}(u)=\{t\}$, i.e.~the set $%
V_{B}\setminus N(u)\setminus V_{0}(u)$ is empty in $G$ before the addition
of the isolated bounded vertex $t$. Then, the values $L_{0}$ and $L_{0}(z)$
for every $z\in N(u)$ are all equal to $L(t)$. Therefore, since we can place
the parallelogram $P_{t}$ that corresponds to $t$ arbitrarily much to the
right of every other parallelogram in the projection representation $R_{\ell
}$, these values can become arbitrarily big in $R_{\ell }$. Recall that $%
N_{1}=\{z\in N(u)\ |\ r(z)<_{R_{\ell }}\ell _{0}\}$ by definition. Then,
during Transformation~\ref{trans4}, $r(z)$ comes immediately after $\ell
_{0} $ on $L_{2}$ for every $z\in N_{1}$ (i.e.~$R(z)$ does not come
immediately before $L_{0}(z)$ on $L_{1}$, since $L_{0}(z)=L(t)$ is
arbitrarily big). Therefore, $\ell _{0}<_{R_{\ell }^{\prime }}r(z)$ for
every $z\in N_{1}$, and thus~$\ell _{0}<_{R_{\ell }^{\prime }}r(z)$ for
every $z\in N(u)$. That is, $N_{2}=N(u)$, since by definition $N_{2}=\{z\in
N(u)\ |\ \ell _{0}<_{R_{\ell }^{\prime }}r(z)\}$. Thus, in particular $%
N_{1}\setminus N_{2}=N_{1}\setminus N(u)=\emptyset $, since $N_{1}\subseteq
N(u)$ by definition.

Consider now the projection representation $R_{\ell }^{\prime \prime \prime
} $, which is obtained by applying Transformation~\ref{trans6} to $R_{\ell
}^{\prime \prime }$. Recall that by Transformation~\ref{trans6} the upper
endpoint $L(u)=R(u)$ of the line $P_{u}$ comes immediately before $\min
\{L_{0},R(z)\ |\ z\in N_{1}\setminus N_{2}\}=L_{0}$ in $R_{\ell }^{\prime
\prime \prime }$ (since $N_{1}\setminus N_{2}=\emptyset $ by the previous
paragraph). Then, since the value $L_{0}=L(t)$ has been chosen arbitrarily
big, the slope $\phi _{u}$ of $P_{u}$ becomes arbitrarily small in $R_{\ell
}^{\prime \prime \prime }$, i.e.~in particular smaller than all other slopes
in $R_{\ell }^{\prime \prime \prime }$. Furthermore, since $R_{\ell
}^{\prime \prime \prime }$ is a projection representation of $G$ by Lemma~%
\ref{R'''-ell}, it follows that $P_{u}$ intersects in $R_{\ell }^{\prime
\prime \prime }$ only the parallelograms $P_{z}$, for every $z\in N(u)$, and
possibly some trivial parallelograms (lines) $P_{x}$, where $x$ is an
unbounded vertex and $\phi _{x}>\phi _{u}$ in $R_{\ell }^{\prime \prime
\prime }$. Denote now by $R^{\#}$ the projection representation that is
obtained from $R_{\ell }^{\prime \prime \prime }$ if we make $u$ again an
unbounded vertex. Then, since the slope $\phi _{u}$ is smaller than all
other slopes in both $R_{\ell }^{\prime \prime \prime }$ and $R^{\#}$, it
follows in particular that $\phi _{u}<\phi _{z}$ in $R^{\#}$ for every $z\in
N(u)$. Therefore, $u$ remains adjacent to all vertices $z\in N(u)$ in the
graph induced by $R^{\#}$, and thus $R^{\#}$ is a projection representation
of $G$, in which $u$ is an unbounded vertex.

Finally, recall by Lemma~\ref{V0(u)-R'''-ell} that there exists no pair of
vertices $z\in N(u)$ and $v\in V_{0}(u)$, such that $P_{z}\ll _{R_{\ell
}^{\prime \prime \prime }}P_{v}$ in $R_{\ell }^{\prime \prime \prime }$.
Therefore, since the only difference between $R_{\ell }^{\prime \prime
\prime }$ and $R^{\#}$ is that $u$ is made bounded in $R^{\#}$, there exists
also in $R^{\#}$ no pair of vertices $z\in N(u)$ and $v\in V_{0}(u)$, such
that $P_{z}\ll _{R^{\#}}P_{v}$ in $R^{\#}$. That is, $u$ has the right
border property in $R^{\#}$. This completes the proof of the lemma.
\end{proof}

\medskip

Now we can prove the correctness of Condition~\ref{cond1}.

\begin{lemma}
\label{lem-cond1}Condition~\ref{cond1} is true.
\end{lemma}

\begin{proof}
Let $G=(V,E)$ be a connected graph in \textsc{Tolerance }$\cap $ \textsc{Trapezoid}
and $R$ be a projection representation of $G$ with $u$ as the only
unbounded vertex. Let furthermore $V_{0}(u)\neq \emptyset $ be connected and 
$V=N[u]\cup V_{0}(u)$. If $u$ has the right (resp.~the left) border property
in $R$, then $R$ (resp.~the reverse representation $\widehat{R}$ of $R$)
satisfies Condition~\ref{cond1}. Suppose now that $u$ has neither the left
nor the right border property in $R$, and suppose w.l.o.g.~that $G$ has the
smallest number of vertices among the graphs that satisfy the above
conditions. Then, since $V_{0}(u)\neq \emptyset $ is connected, the whole
proof of Theorem~\ref{no-property-thm} above applies to $G$. In particular,
we can construct similarly to the above the induced subgraphs $G_{0}$ and $%
G_{0}^{\prime }=G[V(G_{0})\cup \{u^{\ast }\}]$ of $G$. Then, $%
V(G_{0})\subseteq N[u]\cup V_{0}(u)$ by Observation~\ref{V(G0)}, and thus
also $V(G_{0}^{\prime })\subseteq N[u]\cup V_{0}(u)$, since $u^{\ast }\in
V_{0}(u)$. Furthermore, $u$ is the only unbounded vertex of $G_{0}^{\prime }$.

Recall that $G_{0}^{\prime }$ is a connected subgraph of $G$ by Observation~%
\ref{bounded-hovering-obs}. Furthermore, $G_{0}^{\prime }$ has strictly
smaller vertices than $G$, and thus Condition~\ref{cond1} applies to $%
G_{0}^{\prime }$, i.e.~we can construct the projection representations $%
R_{\ell }$, $R_{\ell }^{\prime }$, $R_{\ell }^{\prime \prime }$, and $%
R_{\ell }^{\prime \prime \prime }$, as above. Moreover, since $%
V=V(G)=N[u]\cup V_{0}(u)$ by assumption, it follows that $V_{B}\setminus
N(u)\setminus V_{0}(u)=\{t\}$ after adding an isolated bounded vertex $t$ to 
$R_{\ell }$. Thus, there exists by Lemma~\ref{R'''-ell-only-N(u)-V0} a
projection representation $R^{\ast \ast }=R^{\#}$ of $G$ with the same
unbounded vertices as in $R$ (i.e.~with $u$ as the only unbounded vertex),
such that $u$ has the right border property in $R^{\ast \ast }$. This
completes the proof of the lemma.
\end{proof}

\medskip

Summarizing, since also the correctness of Condition~\ref{cond1} has been
proved in Lemma~\ref{lem-cond1}, the projection representation $R^{\ast
}=R_{\ell }^{\prime \prime \prime }$ of $G$, cf.~Lemma~\ref{R'''-ell}, has $%
k-1$ unbounded vertices, since the input graph $G$ has $k$ unbounded
vertices, and since $u$ is bounded in $R_{\ell }^{\prime \prime \prime }$.
This completes the proof of Theorem~\ref{no-property-thm}.
\end{proof}

\end{document}